\documentclass[11pt]{article}
\usepackage{setspace} 
\usepackage{fancyhdr} 
\usepackage[font=small,labelfont=bf]{caption}
\usepackage{booktabs}
\usepackage{enumitem}
\usepackage{tabularx} 
\usepackage{amsmath}  
\usepackage{graphicx, subcaption} 
\usepackage{overpic}
\usepackage{algorithm}
\usepackage{algorithmicx}
\usepackage{algpseudocode}
\algnewcommand{\LeftComment}[1]{\Statex \(\triangleright\) #1}

\usepackage[margin=1in,letterpaper]{geometry} 
\usepackage{cite} 
\usepackage[final]{hyperref} 
\usepackage [english]{babel}
\usepackage [autostyle, english = american]{csquotes}
\usepackage{amsfonts}
\usepackage{amsmath, amssymb, amsthm}
\MakeOuterQuote{"}

\usepackage{libertine}\usepackage[libertine]{newtxmath}
\usepackage[scaled=0.96]{zi4}

\usepackage{float}
\usepackage{physics}
\usepackage{todonotes}
\usepackage{empheq}
 
\newlength\dlf  

\usepackage{gensymb}
\usepackage{tikz}
\usepackage{mathtools}
\usepackage{todonotes}

\newtheorem{theorem}{Theorem}[section]

\newtheorem{corollary}[theorem]{Corollary}

\newtheorem{lemma}[theorem]{Lemma}

\newtheorem{problem}[theorem]{Problem}

\theoremstyle{definition}

\newtheorem{defn}[theorem]{Definition}

\let\realbfseries=\bfseries
\def\bfseries{\realbfseries\boldmath}

\let\epsilon=\varepsilon
\def\defn#1{\textbf{\textit{\boldmath #1}}}

\newcommand{\OO}{\texorpdfstring{\,\vcenter{\hbox{\includegraphics[scale=0.2]{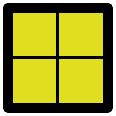}}}\,}O}
\newcommand{\TT}{\texorpdfstring{\,\vcenter{\hbox{\includegraphics[scale=0.2]{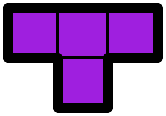}}}\,}T}
\newcommand{\LL}{\texorpdfstring{\,\vcenter{\hbox{\includegraphics[scale=0.2]{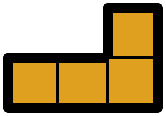}}}\,}L}
\newcommand{\JJ}{\texorpdfstring{\,\vcenter{\hbox{\includegraphics[scale=0.2]{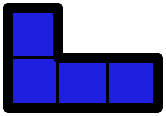}}}\,}J}
\renewcommand{\SS}{\texorpdfstring{\,\vcenter{\hbox{\includegraphics[scale=0.2]{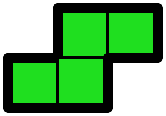}}}\,}S}
\newcommand{\ZZ}{\texorpdfstring{\,\vcenter{\hbox{\includegraphics[scale=0.2]{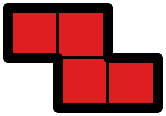}}}\,}Z}
\newcommand{\II}{\texorpdfstring{\,\vcenter{\hbox{\includegraphics[scale=0.2]{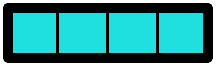}}}\,}I}
\newcommand{\ALL}{\II, \allowbreak \OO, \allowbreak \TT, \allowbreak \SS, \allowbreak \ZZ, \allowbreak \JJ, \allowbreak \LL}

\begin{document}
\title{Tetris is Hard with Just One Piece Type}
\author{%
  \hspace{2in}MIT Hardness Group\footnote{Artificial first author to highlight that the other authors (in alphabetical order) worked as an equal group. Please include all
authors (including this one) in your bibliography, and refer to the authors as “MIT Hardness Group” (without “et al.”).
}\hspace{2in}
\and
  Josh Brunner%
    \thanks{MIT Computer Science and Artificial Intelligence Laboratory,
      32 Vassar St., Cambridge, MA 02139, USA,
      \protect\url{{brunnerj,edemaine,della,jeli}@mit.edu}}
\and
  Erik D. Demaine\footnotemark[2]
\and
  Della Hendrickson\footnotemark[2]
\and
  Jeffery Li\footnotemark[2]
}
\date{}

\maketitle

\begin{abstract}
    We analyze the computational complexity of Tetris clearing (determining whether the player can clear an initial board using a given sequence of pieces) and survival (determining whether the player can avoid losing before placing all the given pieces in an initial board) when restricted to a single polyomino piece type.
    We prove, for any tetromino piece type $P$ except for $\OO$, the NP-hardness of Tetris clearing and survival under the standard Super Rotation System (SRS), even when the input sequence consists of only a specified number of $P$ pieces.
    These surprising results disprove a 23-year-old conjecture on the computational complexity of Tetris with only $\II$ pieces (although our result is only for a specific rotation system).
    As a corollary, we prove the NP-hardness of Tetris clearing when the sequence of pieces has to be able to be generated from a $7k$-bag randomizer for any positive integer $k\geq 1$.
    On the positive side, we give polynomial-time algorithms for Tetris clearing and survival when the input sequence consists of only
    dominoes,
    assuming a particular rotation model,
    solving a version of a 9-year-old open problem.
    Along the way, we give polynomial-time algorithms for Tetris clearing and survival with
    $1\times k$ pieces (for any fixed $k$), provided the top $k-1$ rows are initially empty,
    showing that our $\II$ NP-hardness result needs to have filled cells in the top three rows.

\end{abstract}

\section{Introduction}

Tetris is one of the oldest and most popular puzzle video games, originally created by Alexey Pajitnov in 1984.
Tetris has reached mainstream media many times,
most recently in the biopic \emph{Tetris} \cite{Tetris_movie}
and with the news of 16-year-old Michael Artiaga being the first person to go through all 255 levels of the NES version of Tetris \cite{Tetris_Beat}.

In the base version of Tetris, a player is given a rectangular grid of unit-square cells (\defn{squares}). In each round, a tetromino piece (one of $\ALL$)
spawns at the top of the grid and periodically moves down one unit, assuming the squares below the piece are empty.
The player can repeatedly move this piece one unit left, one unit right, one unit down, or rotate the piece by $\pm 90^\circ$.
When a piece attempts to move downwards while any part of the piece rests on top of a filled square, the piece ``locks'' in place, and stops moving.
If a piece ``locks'' in place above a certain height or where the next piece would spawn, the player loses; otherwise, the next piece spawns at the top of the grid, and play continues.
Completely filling a row causes the row to clear,
with all squares above that row moving downward by one unit.
For more detailed rules, see \cite{Tetris_Wikipedia}.

Tetris has many variations, both official and unofficial \cite{FanGames_TetrisWiki},
that change the piece types that spawn for the player.
For example, \emph{ntris} \cite{ntris} and \emph{Combinos} \cite{Combinos_TetrisWiki} generalize the piece types to \defn{$k$-ominoes} (connected shapes made from gluing $k$ unit squares edge-to-edge).

To study Tetris from a computational complexity perspective, we generally
assume that the player is given an initial board state
with some filled cells and either a finite sequence of pieces or an infinite number of pieces of a given type, making it a perfect-information game
(as introduced in \cite{Tetris_IJCGA}).
The two main objectives are ``clearing'' and ``survival''
(as introduced in \cite{TotalTetris_JIP}).
In \defn{Tetris clearing}, we want to determine whether we can clear the
entire board, either after placing all the given pieces in the finite piece sequence
or after placing some number of pieces of a given type.
In \defn{Tetris survival}, we want to determine whether
the player can avoid losing from the given initial board
while placing all the pieces in the given piece sequence.
Previous work shows that, if the player is given a finite sequence of pieces, these problems are NP-complete,
even to approximate various metrics within $n^{1-\epsilon}$ \cite{Tetris_IJCGA},
or with only $8$ columns or $4$ rows \cite{ThinTetris_JCDCGGG2019},
or with unrotatable dominoes or $k$-ominoes for $k \geq 3$ clearing or $k \geq 4$ survival
\cite{TotalTetris_JIP},
or when the pieces are limited to \emph{any two} of the seven piece types (assuming a specific rotation system)
\cite{mithg2024tetris},
also when restricted to dropping pieces or very high ``gravity''
\cite{mithg2024tetris}.

The \emph{online} version of Tetris, which hides most of the sequence of pieces from the player
(other than possibly giving the player a preview of some of the following pieces),
has also been studied before, both within the Tetris community \cite{playingforever_harddrop} and, most recently, in \cite{gehnen2024online},
which shows that no online algorithm can be competitive with any offline algorithm.

\subsection{Our Results}

In this paper, we analyze Tetris clearing and survival under the restriction that the input sequence contains exactly one polyomino piece type, resolving open problems posed in previous Tetris papers \cite{Tetris_IJCGA,TotalTetris_JIP,mithg2024tetris}.

On the positive side,
we give polynomial-time algorithms for Tetris clearing and survival when the input sequence consists of only dominoes.
These results assume a particularly simple \defn{falling rotation model}
for dominoes, where pieces move monotonically downward;
we leave open the complexity for other rotation models.
The clearing result also only applies to infinite (or sufficiently long) sequences of dominoes.
Other than these catches,
our results solve two of the three remaining open problems posed in \cite{TotalTetris_JIP} regarding $n$-tris for various~$n$.
(The only remaining open problem is surviving a sequence of unrotatable dominoes.)
Key subroutines we develop are strategies for both clearing and survival using only
$1\times k$ pieces, for any fixed $k$, provided the top $k-1$ rows of the board are empty.
(Intuitively, it is easy to survive with stick-shaped pieces provided you are not too close to losing.)

On the negative side,
we show that, for any tetromino piece type $P$ except for $\OO$,
Tetris clearing and survival are NP-hard even if the input sequence consists of only a specified number of $P$ pieces.
These results assume the \defn{Super Rotation System (SRS)} \cite{SRS_TetrisWiki},
defined by the Tetris Company's Tetris Guideline for how all modern (2001+) Tetris games should behave \cite{Guideline_TetrisWiki}.
These results address most of the open problems posed in \cite{mithg2024tetris} regarding singular tetromino piece types in $\{\ALL\}$, and disprove the 23-year-old conjecture that Tetris with only $\II$ pieces should be solvable in polynomial time \cite{Tetris_IJCGA}, in notable contrast to our positive result when the top three rows are empty.
However, our results rely crucially on the rotation system, in particular, the ability with SRS to climb up certain ``staircases'' using ``kicks''.

As a particular consequence, our NP-hardness results are the first to work
in modern Tetris games with a \defn{holding} feature,
where the player can put one piece aside for later use.
Because all the given pieces are of the same type, holding makes no difference on which pieces get placed in what order.

In addition, our NP-hardness result for Tetris clearing with only $\II$ pieces can be modified to prove NP-hardness for Tetris clearing if the sequence of pieces has to be able to be generated from a ``$7k$-bag randomizer'' for any positive integer $k\geq 1$. In a \defn{7-bag randomizer}, the sequence of pieces is divided into groups (or ``bags'') of 7 pieces, where each group contains exactly one of each tetromino in one of $7! = 5{,}040$ possible orderings. A \defn{$7k$-bag randomizer} is a natural extension of such a randomizer, where the sequence is now divided into groups of $7k$ pieces, and each group contains exactly $k$ of each tetromino in one of $\frac{(7k)!}{(k!)^7}$ possible orderings.

\subsection{Outline}\label{sec:outline}

The structure of the rest of the paper is as follows.
Section~\ref{sec:rs} details the rotation systems being considered in our paper, including
the Super Rotation System (SRS) for tetrominoes
and our falling rotation model for dominoes.
Section~\ref{sec:dominoes} discusses our positive results regarding Tetris survival and clearing with $1\times k$ pieces.
Section~\ref{sec:satgo} describes the 1-in-3SAT and Graph Orientation problems, which our hardness proofs reduce from.
Section~\ref{sec:itrisclearing} describes our hardness result for Tetris clearing with SRS with only $\II$ pieces, along with how the hardness result can be modified to get NP-hardness for Tetris clearing if the sequence of pieces has to be able to be generated from a 7-bag randomizer or a $7k$-bag randomizer for $k>1$.
Section~\ref{sec:jtrisclearing} describes our hardness result for Tetris clearing with SRS with only $\JJ$ pieces or with only $\LL$ pieces, Section~\ref{sec:ttrisclearing} describes our hardness result for Tetris clearing with SRS with only $\TT$ pieces, and
Section~\ref{sec:strisclearing} describes our hardness result for Tetris clearing with SRS with only $\SS$ pieces or with only $\ZZ$ pieces.
Section~\ref{sec:survival} describes the modifications to our constructions that yield corresponding one-piece-type hardness results for Tetris survival with SRS.

\section{Rotation Systems}\label{sec:rs}

\subsection{Super Rotation System (SRS)}

For our hardness results involving tetrominoes, we assume the standard \defn{Super Rotation System (SRS)} \cite{SRS_TetrisWiki}.

Each piece has a defined \defn{rotation center},
as indicated by dots in Figure \ref{allpieces},
except for $\II$ and $\OO$, whose rotation centers are the centers of the $4\times 4$ squares in Figure \ref{allpieces}. When unobstructed, all non-$\OO$ tetrominoes will rotate purely about the rotation center (note that $\OO$ pieces cannot rotate). The key feature about SRS is \defn{kicking}: if a tetromino is obstructed when a rotation is attempted, the game will attempt to ``kick'' the tetromino into one of four alternate positions, each tested sequentially; if all four positions do not work, then the rotation will fail. See Figure \ref{SRSrotate} for an example of this kicking process. The full data for wall kicks can be found in Tables \ref{tab:kickdataMain} and \ref{tab:kickdataI}, and at \cite{SRS_TetrisWiki}; a visualization of the locations of the tests in relation to the $\II$ piece can be found in Figure~\ref{fig:allrots}. Of note is that SRS wall kicks are vertically symmetric for all pieces or pairs of pieces (i.e., $\JJ\leftrightarrow \LL$ and $\SS\leftrightarrow \ZZ$) except for the $\II$ piece, so all rotations can be mirrored.

\begin{figure}[ht]
    \centering
    \includegraphics[width=160pt]{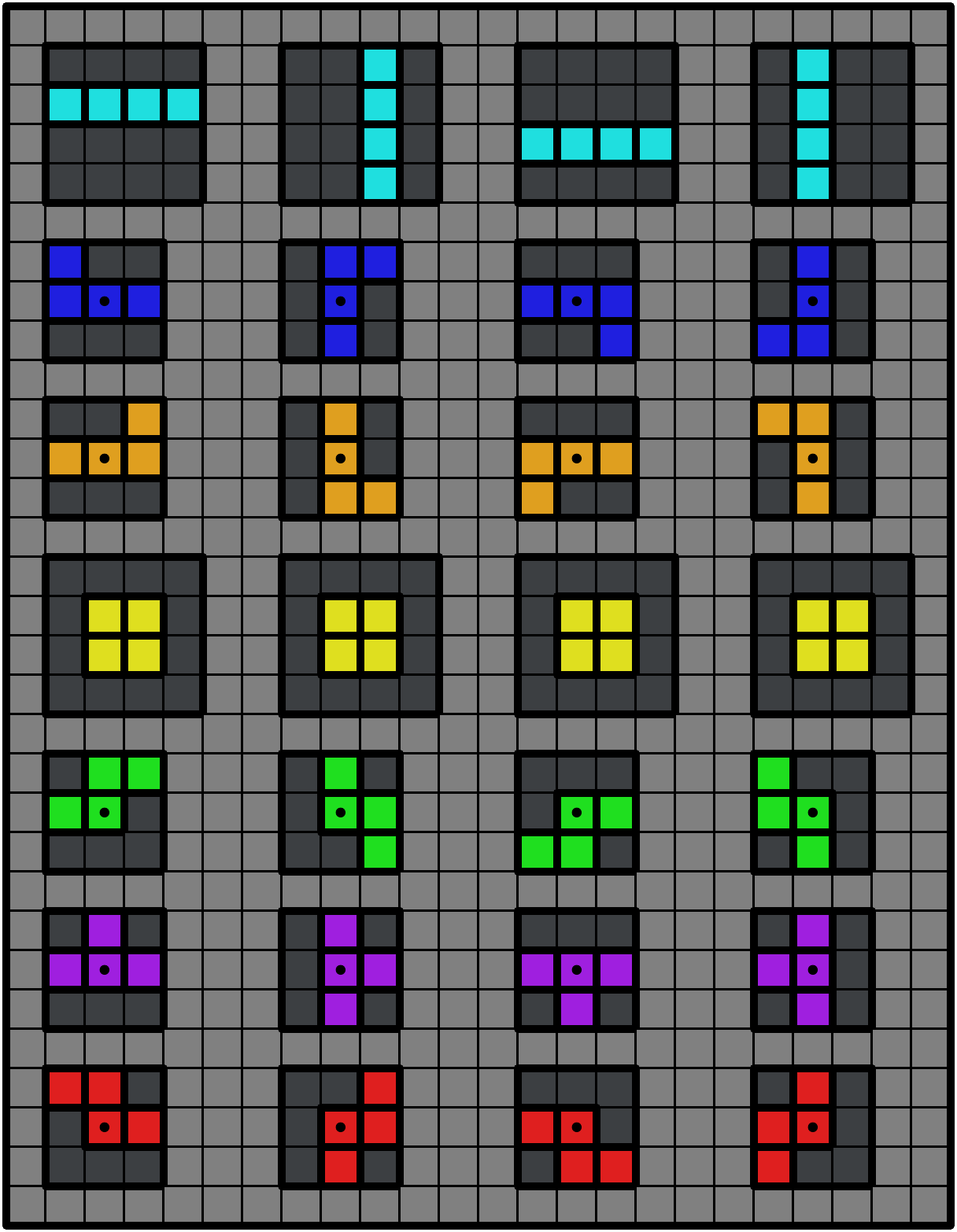}
    \caption{All tetromino pieces, in order from top to bottom: $\II$, $\JJ$, $\LL$, $\OO$, $\SS$, $\TT$, $\ZZ$. The first column is the default orientation of a piece upon spawning in; each column to the right indicates a $90^\circ$ rotation clockwise about the rotation center of the piece.}
    \label{allpieces}
\end{figure}

\begin{figure}[ht]
  \centering
  \begin{subfigure}[b]{0.3\textwidth}
    \centering
    \includegraphics[width=60pt]{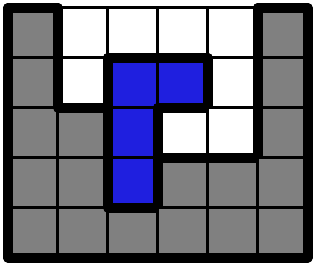}
    \caption{}
  \end{subfigure}
  \begin{subfigure}[b]{0.3\textwidth}
    \centering
    \includegraphics[width=60pt]{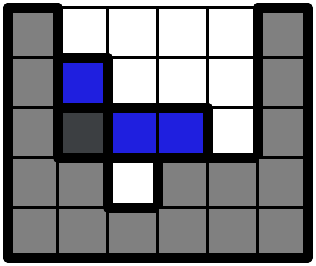}
    \caption{}
  \end{subfigure}
  \begin{subfigure}[b]{0.3\textwidth}
    \centering
    \includegraphics[width=60pt]{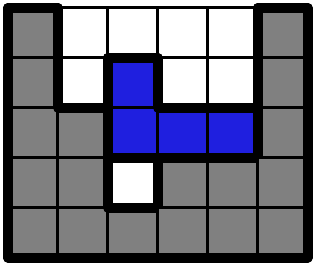}
    \caption{}
  \end{subfigure}
  \caption{An example of the SRS kick system. Suppose the $\JJ$ piece in (a) is being rotated $90^\circ$ counterclockwise. Test 1 (which is $(0, 0)$) would fail, due to the dark gray square shown in (b). Test 2 (which is $(+1, 0)$) would succeed, as shown in (c), and so the $\JJ$ piece would rotate to the position in (c).}
  \label{SRSrotate}
\end{figure}

\begin{table}[ht]
    \centering
    \footnotesize
    \begin{tabular}{c|c|c|c|c|c|}
         & Test 1 & Test 2 & Test 3 & Test 4 & Test 5 \\\hline
         $0\to R$ & $(0, 0)$ & $(-1, 0)$ & $(-1, +1)$ & $(0, -2)$ & $(-1, -2)$ \\\hline
         $R\to 0$ & $(0, 0)$ & $(+1, 0)$ & $(+1, -1)$ & $(0, +2)$ & $(+1, +2)$ \\\hline
         $R\to 2$ & $(0, 0)$ & $(+1, 0)$ & $(+1, -1)$ & $(0, +2)$ & $(+1, +2)$ \\\hline
         $2\to R$ & $(0, 0)$ & $(-1, 0)$ & $(-1, +1)$ & $(0, -2)$ & $(-1, -2)$ \\\hline
         $2\to L$ & $(0, 0)$ & $(+1, 0)$ & $(+1, +1)$ & $(0, -2)$ & $(+1, -2)$ \\\hline
         $L\to 2$ & $(0, 0)$ & $(-1, 0)$ & $(-1, -1)$ & $(0, +2)$ & $(-1, +2)$ \\\hline
         $L\to 0$ & $(0, 0)$ & $(-1, 0)$ & $(-1, -1)$ & $(0, +2)$ & $(-1, +2)$ \\\hline
         $0\to L$ & $(0, 0)$ & $(+1, 0)$ & $(+1, +1)$ & $(0, -2)$ & $(+1, -2)$ \\\hline
    \end{tabular}
    \caption{Kick data for $\JJ$, $\LL$, $\SS$, $\TT$, and $\ZZ$ pieces. $0$ indicates the default orientation, and $R$, $2$, and $L$ indicate the orientation reached from a $90^\circ$, $180^\circ$, and $270^\circ$ rotation clockwise (respectively) from the default orientation. An ordered pair $(a, b)$ denotes a translation of the center by $a$ units in the $x$ direction and $b$ units in the $y$ direction. Positive $x$ direction is rightwards, and positive $y$ direction is upward.}
    \label{tab:kickdataMain}
\end{table}

\begin{table}[ht]
    \centering
    \footnotesize
    \begin{tabular}{c|c|c|c|c|c|}
         & Test 1 & Test 2 & Test 3 & Test 4 & Test 5 \\\hline
         $0\to R$ & $(0, 0)$ & $(-2, 0)$ & $(+1, 0)$ & $(-2, -1)$ & $(+1, +2)$ \\\hline
         $R\to 0$ & $(0, 0)$ & $(+2, 0)$ & $(-1, 0)$ & $(+2, +1)$ & $(-1, -2)$ \\\hline
         $R\to 2$ & $(0, 0)$ & $(-1, 0)$ & $(+2, 0)$ & $(-1, +2)$ & $(+2, -1)$ \\\hline
         $2\to R$ & $(0, 0)$ & $(+1, 0)$ & $(-2, 0)$ & $(+1, -2)$ & $(-2, +1)$ \\\hline
         $2\to L$ & $(0, 0)$ & $(+2, 0)$ & $(-1, 0)$ & $(+2, +1)$ & $(-1, -2)$ \\\hline
         $L\to 2$ & $(0, 0)$ & $(-2, 0)$ & $(+1, 0)$ & $(-2, -1)$ & $(+1, +2)$ \\\hline
         $L\to 0$ & $(0, 0)$ & $(+1, 0)$ & $(-2, 0)$ & $(+1, -2)$ & $(-2, +1)$ \\\hline
         $0\to L$ & $(0, 0)$ & $(-1, 0)$ & $(+2, 0)$ & $(-1, +2)$ & $(+2, -1)$ \\\hline
    \end{tabular}
    \caption{Kick data for $\II$ pieces, with same notation as Table \ref{tab:kickdataMain}.}
    \label{tab:kickdataI}
\end{table}

This system of kicking tetrominoes during rotations allows for moves which are often called \defn{twists} or \defn{spins}. All the spins that we utilize are detailed in Appendices~\ref{sec:i_maneuvers}, \ref{sec:j_maneuvers}, \ref{sec:t_maneuvers}, and \ref{sec:s_maneuvers}.

\subsection{Falling Rotation Model}
\label{sec:falling}

For Tetris with dominoes, we assume a rotation model where pieces
get monotonically lower when they rotate.
Figure~\ref{fallrotate} defines this \defn{falling rotation model} precisely.
In general, whenever we rotate a domino,
one of its lowest pixels remains in place,
and the other pixel moves to a neighboring not-higher cell,
assuming it is unoccupied.
(If the cell is occupied, the rotation fails and the piece does not rotate.)
There are two choices in this definition,
which correspond to clockwise and counterclockwise rotations:
when the domino was horizontal, we can choose either the left or right pixel to remain in place;
when the domino was vertical, we can choose whether the neighboring cell is left or right.

\begin{figure}[ht]
    \centering
    \begin{overpic}[width=160pt]{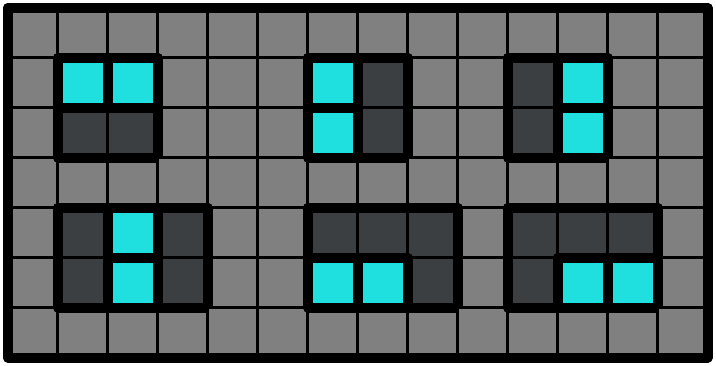}
      \put(36.5,39){\makebox(0,0){$\boldsymbol{\rightarrow}$}}
      \put(36.5,18){\makebox(0,0){$\boldsymbol{\rightarrow}$}}
    \end{overpic}
    \caption{Falling rotation model for dominoes: when the domino on the left attempts to rotate clockwise or counterclockwise, it becomes the respective domino on the right, if that position is unobstructed.}
    \label{fallrotate}
\end{figure}


\section{Tetris with Dominoes and $1 \times k$}\label{sec:dominoes}

Demaine et al.~\cite{TotalTetris_JIP} showed NP-hardness of Tetris clearing with unrotatable dominoes (where rotation is forbidden). They left open the complexity of survival with unrotatable dominoes, and both clearing and survival with rotatable dominoes. In this section, we give a polynomial-time algorithm for survival with rotatable dominoes, as well as a polynomial-time algorithm that determines whether a board is eventually clearable with infinitely/sufficiently many pieces.

We also give a strategy that always survives with only $1\times k$ pieces (for any fixed $k$) if the $k-1$ top rows start empty, and with this strategy, we can decide in polynomial time whether it is possible to eventually fully clear the board given enough such pieces.
We will use this result with $k=2$ as a subroutine for our domino algorithms,
but $k=4$ is also an interesting contrast to our $\II$ hardness result.

\subsection{Survival and Clearing with $1\times k$ Pieces}
\label{sec:1timesk}

\begin{theorem}
\label{thm:1timesk}
    If the board has height at least $2k-1$, and the top $k-1$ rows are initially empty, then it is possible to survive an infinite sequence of $1\times k$ pieces. Furthermore, there is a polynomial-time algorithm to determine whether the board can eventually be fully cleared.
\end{theorem}
\begin{proof}
Our strategy proceeds in two phases. First, we eliminate all holes, where a \defn{hole} is an empty cell that is below a filled cell. Second, we place horizontal pieces to make all columns have the same number of pieces modulo $k$.

For the first phase,
we will show how to eliminate the highest hole without creating any new holes.
Then, by induction on the number of holes, we can remove all holes.
Refer to Figure~\ref{fig:hole-clear} for an example.

\begin{figure}[!ht]
  \centering
  \begin{subfigure}[b]{0.225\textwidth}
    \centering
    \includegraphics[width=80pt]{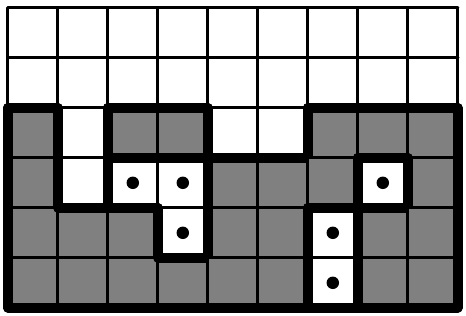}
    \caption{}
  \end{subfigure}
  \begin{subfigure}[b]{0.225\textwidth}
    \centering
    \includegraphics[width=80pt]{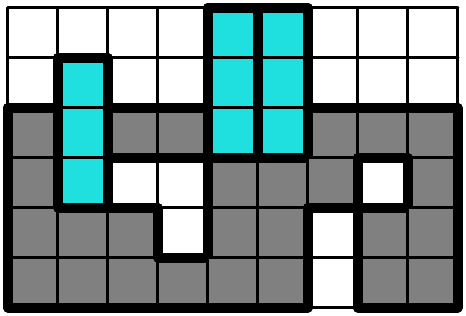}
    \caption{}
  \end{subfigure}
  \begin{subfigure}[b]{0.225\textwidth}
    \centering
    \includegraphics[width=80pt]{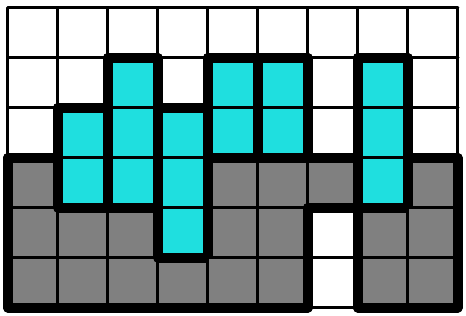}
    \caption{}
  \end{subfigure}
  \begin{subfigure}[b]{0.225\textwidth}
    \centering
    \includegraphics[width=80pt]{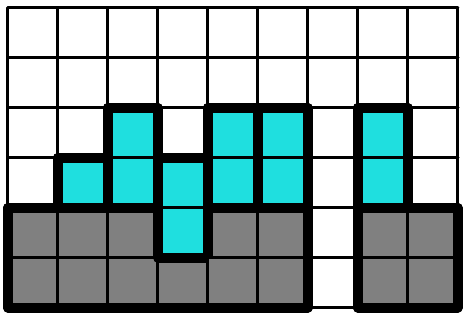}
    \caption{}
  \end{subfigure}
  \caption{Procedure for clearing holes for $k=3$. The dots in (a) show which empty squares are holes.}
  \label{fig:hole-clear}
\end{figure}

Suppose that the highest hole $h$ is in some row $r$ and some column $c$ (if there are multiple holes in row $r$, pick one arbitrarily). By the definition of a hole (and because this is the highest hole), there is a filled cell directly above the hole. Call the row of this cell $r'$. We will give a procedure to clear row $r'$, thus reducing the number of filled cells above our hole.

To clear row $r'$, we find each of the columns that have an empty cell in row $r'$. Because $r'$ is above the highest hole, all of these empty cells are not holes, and thus have only empty cells above them. In each of these columns, we repeatedly place vertical pieces until the cell in $r'$ in that column is filled. Now we will show that placing these pieces is always possible and will never cause you to lose by making too tall of a column. Initially, there were $k-1$ empty rows, so the highest possible hole must have at least $k$ rows above it (because there is at least one occupied cell above the hole, and then $k-1$ empty cells above that). Then row $r'$ must have at least $k-1$ rows above it. Because the final piece that covers row $r'$ has at most $k-1$ cells above $r'$, it will fit in the board without going over the top and making you lose.

By repeatedly applying this procedure, all of the filled cells in $c$ above the hole $h$ will be cleared, and then $h$ will no longer be a hole. This process never creates new holes because it only ever places vertical pieces.

Once there are no holes, it is easy to survive forever: there must be some empty cell on the bottom row (or it would have already cleared), and because it is not a hole, we can put a vertical piece in that column.
By repeatedly placing pieces in the lowest available location, we can also make the height of the highest filled cell be at most $k-1$, giving us at least $k$ empty rows at the top because the total height of the board is at least $2k-1$.

Now we give an algorithm that either clears the board from this position, or proves that the board was impossible to clear to begin with. In the cleared state, every column has the same number (namely, 0) of filled cells modulo $k$. Placing a vertical piece does not change the number of filled cells modulo $k$ in its column. Clearing a row changes all columns' counts by 1. Thus, if it is possible to clear the board, we must place horizontal pieces to make all of the columns' counts equal.

We determine whether this is possible with a greedy strategy. For this algorithm, we will say ``place a piece in column $c$'' to mean to place a horizontal piece with its leftmost cell in column $c$. First, we guess the number of horizontal pieces to place in the leftmost column. For each subsequent column, we are forced to place an amount of horizontal pieces in it so that the number of filled cells in that column is equal to the previous column (mod $k$). If this results in all of the columns having an equal number of cells mod $k$, we win. Because there are only $k$ possible choices for the amount of pieces to place in the first column, we can try all possible values, and if none result in all of the columns having an equal number of cells mod $k$, we conclude that clearing is impossible.

During this process, we may create some holes. After each horizontal piece, we run the hole eliminating algorithm (phase one) again. Because the hole eliminating algorithm ends with at least $k$ empty rows on top, after adding each horizontal piece there are at least $k-1$ empty rows on top, which allows the hole eliminating algorithm to run without losing.
\end{proof}

\subsection{Algorithm for Clearing with Dominoes}

In this subsection, we answer the open questions from \cite{TotalTetris_JIP} about surviving and clearing Tetris with only rotatable dominoes. In \cite{TotalTetris_JIP}, they note that clearing a single line suffices to survive forever. We give an algorithm to decide whether it is possible to ever clear a line with an infinite sequence of domino pieces, thus deciding survivability. After clearing a line, using the algorithm from Section~\ref{sec:1timesk} gives a solution to whether fully clearing the board is possible as well.

We will use the \defn{falling rotation model} from Section~\ref{sec:falling}. The allowed rotations are any rotation that overlaps the previous position in exactly one of the two cells, and the other cell is in a lower row than it previously was.

We give an algorithm to determine whether a particular row $r$ is ever clearable. To determine whether any row is clearable, we can simply run this algorithm separately for each row $r$.

We call a cell $(r,c)$ \defn{reachable} if it is possible from the initial configuration to move a domino to a position with one of the two cells of the domino in $(r,c)$. Similarly, a pair of empty adjacent cells is a reachable domino location if it is possible from the initial configuration to move a domino to that position. Note that we only require the domino to be able to pass through the location; it is still considered reachable even if the domino cannot stop there because there is nothing supporting it.

Call a configuration of dominoes \defn{supported} if every piece has a filled cell or domino directly below it. Thus, in a supported configuration, for each domino, if it was the last one to be placed it would not fall. First, we show that as long as every domino's location is individually reachable (ignoring other dominoes in the configuration), it is possible to sequence the placements of dominoes such that we can build the given configuration. Call a domino location \defn{blocking} another domino location if the second domino's position is not reachable after the first one is placed. Additionally, we say a domino location \defn{supports} another domino location if the second one contains at least one cell directly above a cell in the first domino.

First we construct an ordering on all of the empty reachable cells, which we will call the \defn{path order}.
Refer to Figure~\ref{fig:pathorder} for an example.
We will label each cell with its position in this list. We proceed row by row, with the cells in higher rows always before lower ones. Within each row, we start with all of the cells which are directly below a reachable cell in the above row. (For the first row, this is all of the empty cells). Then we proceed outward along this row from those cells, appending each reachable to the path order cell the first time we come to it.

\begin{figure}[!ht]
    \centering
    \begin{overpic}[width=0.6\textwidth]{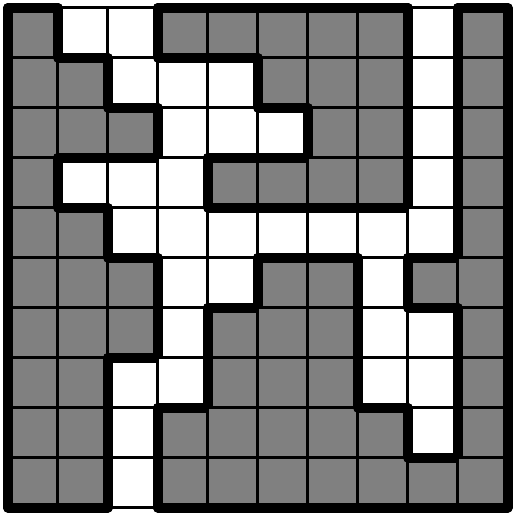}
      \put(15.5,92.5){1}
      \put(25,92.5){2}
      \put(83,92.5){3}
      \put(25,82.5){4}
      \put(34.5,82.5){5}
      \put(44,82.5){6}
      \put(83,82.5){7}
      \put(34.5,73){8}
      \put(44,73){9}
      \put(52.5,73){10}
      \put(82,73){11}
      \put(14.5,63){14}
      \put(24,63){13}
      \put(33.5,63){12}
      \put(82,63){15}
      \put(24,53.5){16}
      \put(33.5,53.5){17}
      \put(43,53.5){19}
      \put(53,53.5){20}
      \put(62.5,53.5){21}
      \put(72.5,53.5){22}
      \put(82,53.5){18}
      \put(33.5,44){23}
      \put(43,44){24}
      \put(72.5,44){25}
      \put(33.5,34){26}
      \put(72.5,34){27}
      \put(82,34){28}
      \put(24,24.5){32}
      \put(33.5,24.5){29}
      \put(72.5,24.5){30}
      \put(82,24.5){31}
      \put(24,15){33}
      \put(82,15){34}
      \put(24,5.5){35}
    \end{overpic}
    \caption{Path order of empty cells in an example board state. Here, a bigger number appears later in the path order.}
    \label{fig:pathorder}
\end{figure}
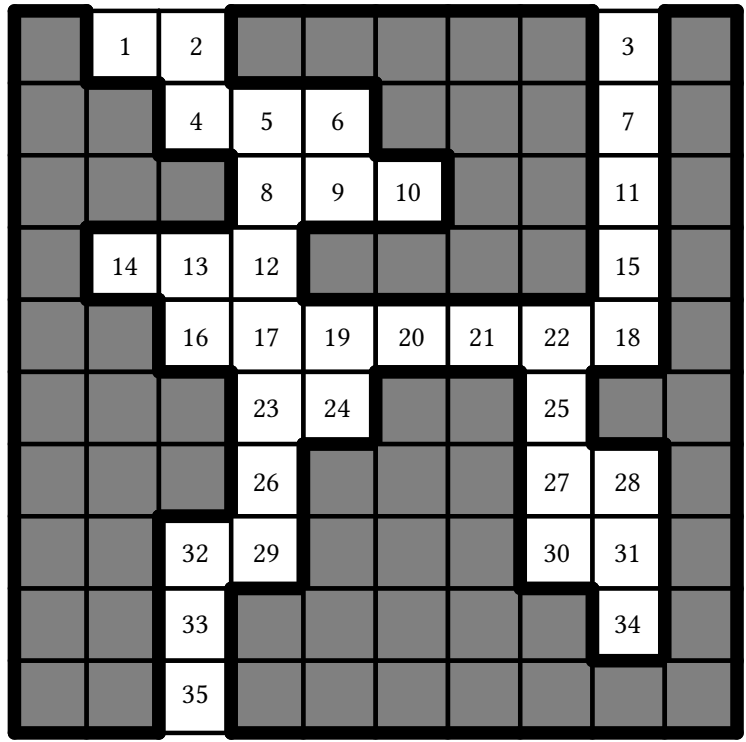

\begin{lemma}
\label{domino-paths}
    A domino can move along any path of empty cells that is monotonically decreasing in $y$ value that starts from an empty top position.
\end{lemma}
\begin{proof}
    Start the domino at the empty top position at the start of the path. The domino will occupy two cells at a time. To reach the next cell, if it is in a line with the current domino position, we can simply move the domino one space left, right, or down as appropriate. Otherwise, with the falling rotation system there is always a rotation that moves the domino to cover the next cell and uncover the earlier of the two cells it occupies.
\end{proof}

\begin{lemma}
\label{lem:domino-path-order}
    For every reachable domino location, there is a path a domino can take to reach it that is monotonically increasing in path order.
\end{lemma}
\begin{proof}
    We will compute this path in reverse order. Start with the domino's final location, ordering the two cells by their path order. Repeatedly prepend to our path any adjacent cell which is earlier in path order from the last cell we added. The way our path order was constructed, this will always result in building a monotone path from the top to the domino's location, which by Lemma~\ref{domino-paths} is traversable by a domino.
\end{proof}

\begin{lemma}
\label{lem:domino-placeable}
    Every supported configuration of dominoes is placeable.
\end{lemma}
\begin{proof}
Label each domino in our configuration with the position of the earlier of its two cells in the path order. We will place the dominoes in this order, with the largest labeled domino going first. When each domino is placed, every cell with a smaller path order than the label of that domino is still unoccupied. Thus, by Lemma~\ref{lem:domino-path-order}, there is a path this domino can take to reach its location. Furthermore, any domino which supported this domino must have already been placed, because labels of lower dominoes are always larger than labels of higher dominoes.
\end{proof}

Now that we know that any configuration is placeable, we describe how to find a configuration that will result in a line clear if one exists.
Call a column $c$ \defn{good} if the following are true:
\begin{itemize}
    \item $(r,c)$ is reachable; and
    \item either $(r+1,c)$ is empty, reachable, and $r$ is not the topmost row, or the nearest filled cell in column $c$ below $r$ is an even distance below $r$.
\end{itemize}

In other words, a good column is one where it is possible to stack vertical dominoes until $(r,c)$ is occupied. A \defn{bad} column is any column that is not good.

A horizontal domino \defn{fixes} a bad column $c$ if the following are true:
\begin{itemize}
    \item the domino's location is reachable;
    \item it is an even distance below $r$; and
    \item there are no filled cells between it and $r$ in column $c$.
\end{itemize}

In other words, a horizontal domino fixes a column $c$ if we can now place a (possibly empty) stack of vertical dominoes on top of it that results in $(r,c)$ being occupied. Note that it is often the case that a single horizontal domino will fix both columns it is in.

\begin{lemma}\label{lem:domino-condition}
    $r$ is clearable if and only if there is a non-overlapping set of horizontal dominoes that fix every bad column.
\end{lemma}
\begin{proof}
    In order to clear row $r$, every bad column must have some horizontal domino in it, since the definition of bad prevents filling it entirely with vertical dominoes. Furthermore, that domino must fix the column, because otherwise the empty space above the horizontal domino would still make the column bad. Thus, $r$ is clearable only if there is a non-overlapping set of horizontal dominoes that fix every bad column.

    Now we show that such a set of horizontal dominoes suffices. If any are not supported, add a stack of horizontal dominoes below them until they become supported. Now add a stack of vertical dominoes in every column that has $r$ empty to fill $r$. By the definition of horizontal dominoes fixing columns, this will successfully fill every cell in row $r$ in the previously bad columns, and by the definition of good this will fill every cell in row $r$ in the previously good columns.

    This is a supported configuration, so by Lemma~\ref{lem:domino-placeable}, there is a sequence that they can be placed in.
\end{proof}

\begin{figure}[!ht]
    \centering
    \begin{subfigure}[b]{0.45\textwidth}
        \centering
        \includegraphics[width=160pt]{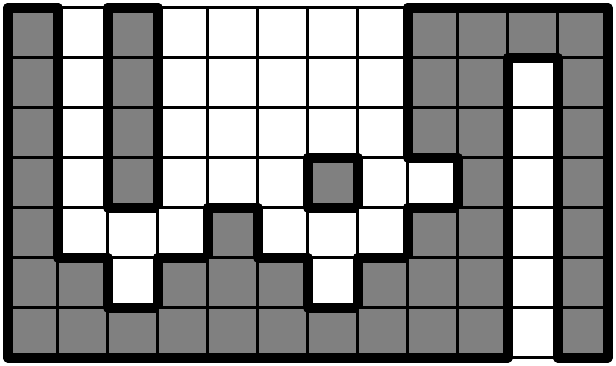}
        \caption{}
    \end{subfigure}
    \begin{subfigure}[b]{0.45\textwidth}
        \centering
        \includegraphics[width=160pt]{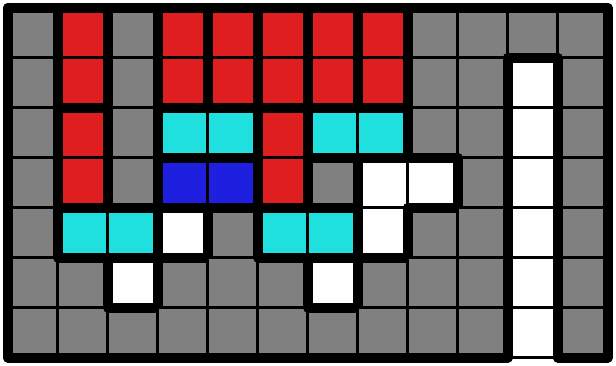}
        \caption{}
    \end{subfigure}
    \caption{A possible board state, along with a set of horizontal dominoes (cyan) which fixes all bad columns, additional horizontal dominoes (dark blue) so that the horizontal dominoes are supported, and vertical dominoes (red) that aid in clearing the top row.}
  \label{fig:horizontal-domino-fix}
\end{figure}

Now, we give an algorithm to find the set of horizontal dominoes required by Lemma~\ref{lem:domino-condition}. Since each domino can only overlap two columns, we can use a dynamic program with a 2 column wide window from left to right, where we maintain all of the ways we can place horizontal dominoes fixing the last two columns as we sweep from left to right. Since we only need at most 1 horizontal domino to fix each column, there are at most $O(n^2)$ possible ways to fix the previous 2 columns, so our dynamic program takes at most $O(n^3)$ time.

\begin{theorem}
    There is a polynomial-time algorithm that determines whether a Tetris position is survivable and whether it is eventually clearable with an infinite sequence of dominoes.
\end{theorem}
\begin{proof}
    The above algorithm finds a set of dominoes that clear a row if one exists. By Lemma~\ref{lem:domino-placeable}, Lemma~\ref{lem:domino-condition}, and Theorem~\ref{thm:1timesk}, this gives a strategy which survives forever if possible, and eventually clears the board if possible.
\end{proof}
\begin{corollary}
\label{cor:survive-k-dominoes}
    There is a polynomial-time algorithm that determines whether a Tetris position is survivable with a sequence of $k$ dominoes.
\end{corollary}
\begin{proof}
    If the position is survivable with infinitely many dominoes, clearly it is also survivable for any finite $k$. The only case where we cannot survive forever is when no row is clearable, as Theorem~\ref{thm:1timesk} shows and \cite{TotalTetris_JIP} notes. In this case, we compute a packing of dominoes into reachable locations. This is easy since it is a maximum matching problem on a bipartite graph, and by Lemma~\ref{lem:domino-placeable} any such configuration is always possible to place. 
\end{proof}

\section{SAT and Graph Orientation}\label{sec:satgo}

Here, we discuss all of the SAT and Graph Orientation problems required for all of our reductions to Tetris clearing with one tetromino piece type.

Our reductions to $\II$-tris will be from or involve the following problems:

\begin{problem}[\textbf{1-in-3SAT}]
    Given a conjunctive normal form (CNF) formula $\phi$ consisting of $m$ clauses with $n$ variables such that each clause $C_j$ contains three literals (i.e., variables or their negations), determine whether there exists a satisfying truth assignment to the variables such that each clause in $\phi$ has exactly one true literal.
\end{problem}

\begin{problem}[\textbf{Graph Orientation}]
    Given an undirected 3-regular graph $G = (V, E)$ such that $V$ is partitioned into three (disjoint) node-subsets $V_\ell$ (\emph{literal nodes}), $V_C$ (\emph{clause nodes}), and $V_{\overline{C}}$ (\emph{negated-clause nodes}), determine whether there exists an orientation of $G$ (i.e., an assignment of directions to all edges in $G$) such that:
    \begin{enumerate}
        \item every literal node in $V_\ell$ has in-degree either zero or three,
        \item every clause node in $V_C$ has in-degree one, and
        \item every negated-clause node in $V_{\overline{C}}$ has in-degree two.
    \end{enumerate}
\end{problem}

Both problems are known to be NP-complete; 1-in-3SAT was shown to be NP-complete in \cite{schaefer1978complexity}, and Graph Orientation was shown to be NP-complete in \cite{horiyama2017complexity}.

\vspace{5mm}

Our reductions to $\JJ$-tris, $\LL$-tris, and $\TT$-tris will be from the following problems:

\begin{problem}[\textbf{Planar Biconnected $\{\{0, 4\}\text{-in-4},\text{1-in-4}\}$ Graph Orientation}]
    Given an undirected $4$-regular planar \emph{biconnected} graph $G = (V, E)$ such that $V$ is partitioned into two (disjoint) node-subsets $V_0$ and $V_1$, determine whether there exists an orientation of $G$ (i.e., an assignment of directions to all edges in $G$) such that every node $v\in V_0$ has indegree either $0$ or $4$ and every node $v\in V_1$ has indegree exactly $1$.
\end{problem}

\begin{problem}[\textbf{Planar Biconnected $\{\{0, 4\}\text{-in-4},\text{3-in-4}\}$ Graph Orientation}]
    Given an undirected $4$-regular \emph{biconnected} planar graph $G = (V, E)$ such that $V$ is partitioned into two (disjoint) node-subsets $V_0$ and $V_3$, determine whether there exists an orientation of $G$ (i.e., an assignment of directions to all edges in $G$) such that every node $v\in V_0$ has indegree either $0$ or $4$ and every node $v\in V_3$ has indegree exactly $3$.
\end{problem}

The versions of these problems without the biconnectivity condition were shown to be NP-complete in \cite{planar-graph-orientation}. We prove that adding the biconnectivity condition does not change the hardness of these problems:

\begin{theorem}
    Planar Biconnected $\{\{0, 4\}\text{-in-}4,1\text{-in-}4\}$ Graph Orientation and Planar Biconnected \allowbreak $\{\{0, 4\}\text{-in-4},\allowbreak 3\text{-in-}4\}$ Graph Orientation are NP-hard.
\end{theorem}

\begin{proof}
    We focus on Planar Biconnected $\{\{0, 4\}\text{-in-4},\text{1-in-4}\}$ Graph Orientation; the proof for the other problem is similar. We reduce from Planar $\{\{0, 4\}\text{-in-4},\text{1-in-4}\}$ Graph Orientation (without the biconnectivity condition) \cite{planar-graph-orientation}. Assume that there are no isolated vertices (as we can remove those vertices).
    Given an instance of Planar $\{\{0, 4\}\text{-in-4},\text{1-in-4}\}$ Graph Orientation, we choose a planar embedding such that all biconnected components share a common ``outside'' face (i.e., no biconnected component is nested inside another biconnected component).
    Then, we do the following:

    \begin{itemize}
        \item First, for some pair of connected components, we choose an edge from each connected component that lies on the ``outside'' face and replace the pair of edges with the gadget shown in Figure~\ref{subfig:biconn-gadget}. The gadget consists of three new vertices, one splitting one of the edges into two and another one splitting the other edge into two, and four new edges connecting the three vertices. We repeatedly do this until the graph is connected.
        \item Next, we do the same for some pair of \emph{adjacent} biconnected components (i.e., biconnected components sharing a vertex): choose an edge from each connected component that lies on the ``outside'' face and replace the pair of edges with the gadget shown in Figure~\ref{subfig:biconn-gadget}. We repeatedly do this until the graph is biconnected.
    \end{itemize}

    We show that this reduction works. First, because the edges the gadget replaces are on the ``outside'' face, the additional edges and vertices can be embedded on the plane while preserving planarity. In addition, each time a gadget replaces a pair of edges in the first or second phase, the number of connected components or biconnected components decreases, respectively (as it connects two separate components or introduces an alternate connection between biconnected components), so the process is finite and will result in a biconnected graph after a finite number of iterations. Thus, the resulting Graph Orientation instance is planar and biconnected.

    It thus suffices to show that the gadget in Figure~\ref{subfig:biconn-gadget} simulates a pair of distinct edges. In particular, the $\{0, 4\}$-in-4 vertex must be oriented as in Figure~\ref{subfig:biconn-gadget-forced}, as the other orientation violates both 1-in-4 vertices. Then, with the fixed orientation in Figure~\ref{subfig:biconn-gadget-forced}, both 1-in-4 vertices require one more incoming edge and one more outgoing edge, thus becoming a 1-in-2 vertex, which simulates an edge between the two (other) vertices it is connected to. Thus, the gadget simulates a pair of distinct edges, so the gadget can replace a pair of distinct edges without affecting whether or not the Graph Orientation instance has a solution.

    Thus, we have shown a way to reduce Planar $\{\{0, 4\}\text{-in-4},\text{1-in-4}\}$ Graph Orientation to Planar Biconnected $\{\{0, 4\}\text{-in-4},\text{1-in-4}\}$ Graph Orientation, so the latter is NP-hard, as desired. Similarly, Planar Biconnected $\{\{0, 4\}\text{-in-4},\text{3-in-4}\}$ Graph Orientation is also NP-hard.
\end{proof}

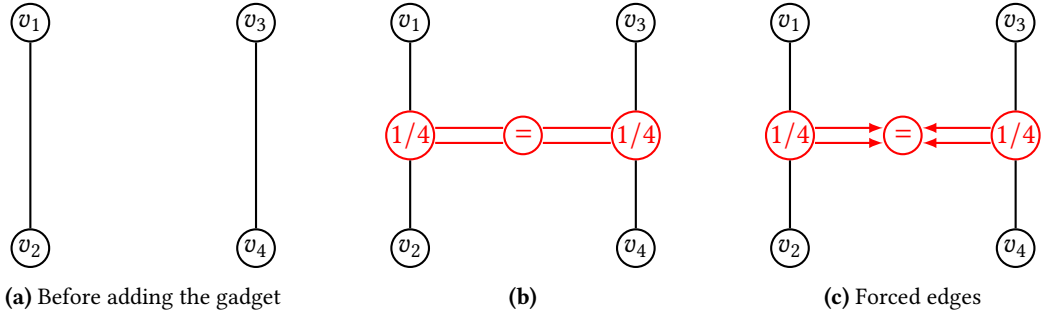
\begin{figure}[!ht]
\centering
\begin{subfigure}[b]{0.3\textwidth}
    \centering
    \begin{tikzpicture}[
node/.style={circle, draw, thick, circle, draw, minimum size=5mm, inner sep=0pt},
edge/.style={-latex, thick}
]
\begin{scope}
    \node[node] (A) at (-1.5,1.5) {$v_1$};
    \node[node] (C) at (-1.5,-1.5) {$v_2$};
    \node[node] (E) at (1.5,1.5) {$v_3$};
    \node[node] (G) at (1.5,-1.5) {$v_4$};
    \draw[thick] (A) -- (C);
    \draw[thick] (E) -- (G);
\end{scope}

  \end{tikzpicture}
  \subcaption{Before adding the gadget}
  \label{subfig:biconn-gadget-pre}
  \end{subfigure}
\begin{subfigure}[b]{0.3\textwidth}
    \centering
    \begin{tikzpicture}[
node/.style={circle, draw, thick, circle, draw, minimum size=5mm, inner sep=0pt},
edge/.style={-latex, thick}
]
\begin{scope}
    \node[node] (A) at (-1.5,1.5) {$v_1$};
    \node[node, red] (B) at (-1.5,0) {$1/4$};
    \node[node] (C) at (-1.5,-1.5) {$v_2$};
    \node[node, red] (D) at (0,0) {$=$};
    \node[node] (E) at (1.5,1.5) {$v_3$};
    \node[node, red] (F) at (1.5,0) {$1/4$};
    \node[node] (G) at (1.5,-1.5) {$v_4$};
    \draw[red, thick] ([yshift=0.1cm]B.east) -- ([yshift=0.1cm]D.west);
    \draw[red, thick] ([yshift=-0.1cm]B.east) -- ([yshift=-0.1cm]D.west);
    \draw[red, thick] ([yshift=0.1cm]F.west) -- ([yshift=0.1cm]D.east);
    \draw[red, thick] ([yshift=-0.1cm]F.west) -- ([yshift=-0.1cm]D.east);
    \draw[thick] (A) -- (B);
    \draw[thick] (B) -- (C);
    \draw[thick] (E) -- (F);
    \draw[thick] (F) -- (G);
\end{scope}

  \end{tikzpicture}
  \subcaption{}
  \label{subfig:biconn-gadget}
  \end{subfigure}
  \begin{subfigure}[b]{0.3\textwidth}
    \centering
\begin{tikzpicture}[
node/.style={circle, draw, thick, circle, draw, minimum size=5mm, inner sep=0pt},
edge/.style={-latex, thick}
]
\begin{scope}
    \node[node] (A) at (-1.5,1.5) {$v_1$};
    \node[node, red] (B) at (-1.5,0) {$1/4$};
    \node[node] (C) at (-1.5,-1.5) {$v_2$};
    \node[node, red] (D) at (0,0) {$=$};
    \node[node] (E) at (1.5,1.5) {$v_3$};
    \node[node, red] (F) at (1.5,0) {$1/4$};
    \node[node] (G) at (1.5,-1.5) {$v_4$};
    \draw[edge, red] ([yshift=0.1cm]B.east) -- ([yshift=0.1cm]D.west);
    \draw[edge, red] ([yshift=-0.1cm]B.east) -- ([yshift=-0.1cm]D.west);
    \draw[edge, red] ([yshift=0.1cm]F.west) -- ([yshift=0.1cm]D.east);
    \draw[edge, red] ([yshift=-0.1cm]F.west) -- ([yshift=-0.1cm]D.east);
    \draw[thick] (A) -- (B);
    \draw[thick] (B) -- (C);
    \draw[thick] (E) -- (F);
    \draw[thick] (F) -- (G);
\end{scope}
\end{tikzpicture}
  \subcaption{Forced edges}
  \label{subfig:biconn-gadget-forced}
  \end{subfigure}
\caption{Gadget for connecting biconnected components in $\{\{0, 4\}\text{-in-4},\text{1-in-4}\}$ Graph Orientation. The $1/4$ denote 1-in-4 vertices, and the $=$ denotes a $\{0, 4\}\text{-in-4}$ vertex. The gadget replaces two edges $(v_1, v_2)$ and $(v_3, v_4)$ in two different biconnected components (where the $v_i$ may not necessarily be distinct).}
\label{fig:biconn-connect}
\end{figure}

For $\JJ$-tris and $\LL$-tris, we will use Planar Biconnected $\{\{0, 4\}\text{-in-4},\text{3-in-4}\}$ Graph Orientation, and for $\TT$-tris, we will use Planar Biconnected $\{\{0, 4\}\text{-in-4},\text{1-in-4}\}$ Graph Orientation.

\vspace{5mm}

Lastly, our reductions to $\SS$-tris and $\ZZ$-tris will be from the following problem:

\begin{problem}[\textbf{Planar Monotone Rectilinear 3SAT}]
    Given a conjunctive normal form (CNF) formula $\phi$ consisting of $m$ clauses with $n$ variables that can be arranged as follows on the plane:
    \begin{enumerate}
        \item Every variable is a horizontal segment on the $x$-axis,
        \item Every clause contains either all positive or all negative literals and is a horizontal segment off the $x$-axis with 3 vertical connections to the variables,
        \item Clauses with all positive literals are above the $x$-axis,
        \item Clauses with all negative literals are below the $x$-axis,
        \item Segments do not cross besides at connections.
    \end{enumerate}
    Determine whether there exists a satisfying truth assignment to the variables such that each clause in $\phi$ has at least one true literal.
\end{problem}

This problem was shown to be NP-complete in \cite{de2010optimal}.

\section{$\II$-tris Clearing}\label{sec:itrisclearing}

\begin{theorem}\label{thm:itrisclear}
    Tetris clearing with SRS is NP-hard even if the type of pieces in the sequence given to the player is restricted to just $\II$ pieces.
\end{theorem}

Here, we will reduce from 1-in-3SAT, passing through Graph Orientation. Given a 1-in-3SAT instance with $n$ variables and $m$ clauses, we first construct a corresponding Graph Orientation instance on a grid using the reduction given by Horiyama et al.\ in \cite[Theorem~2.3]{horiyama2017complexity}; an example is given in Figure~\ref{fig:1in3go}. Notice that the instance consists of two parts:
\begin{itemize}
    \item \textbf{The upper part} – consisting of $n$ variable cycles, vertical wires emanating from the literals, and crossovers created from the intersections of the variable cycles and vertical wires,
    \item \textbf{The lower part} – consisting of $m$ copies of the clause/negated-clause construction in \cite[Figure~5]{horiyama2017complexity} (see Figure~\ref{fig:cnc}), each copy taking up 8 columns horizontally and no two copies using a common column.
\end{itemize}
In addition, we place a red dot at every lattice point that is on a vertex or edge of the graph. The red dots will indicate where we will put our main gadgets in the reduction to Tetris clearing.

\begin{figure}[!ht]
    \centering
    \includegraphics[width=240pt]{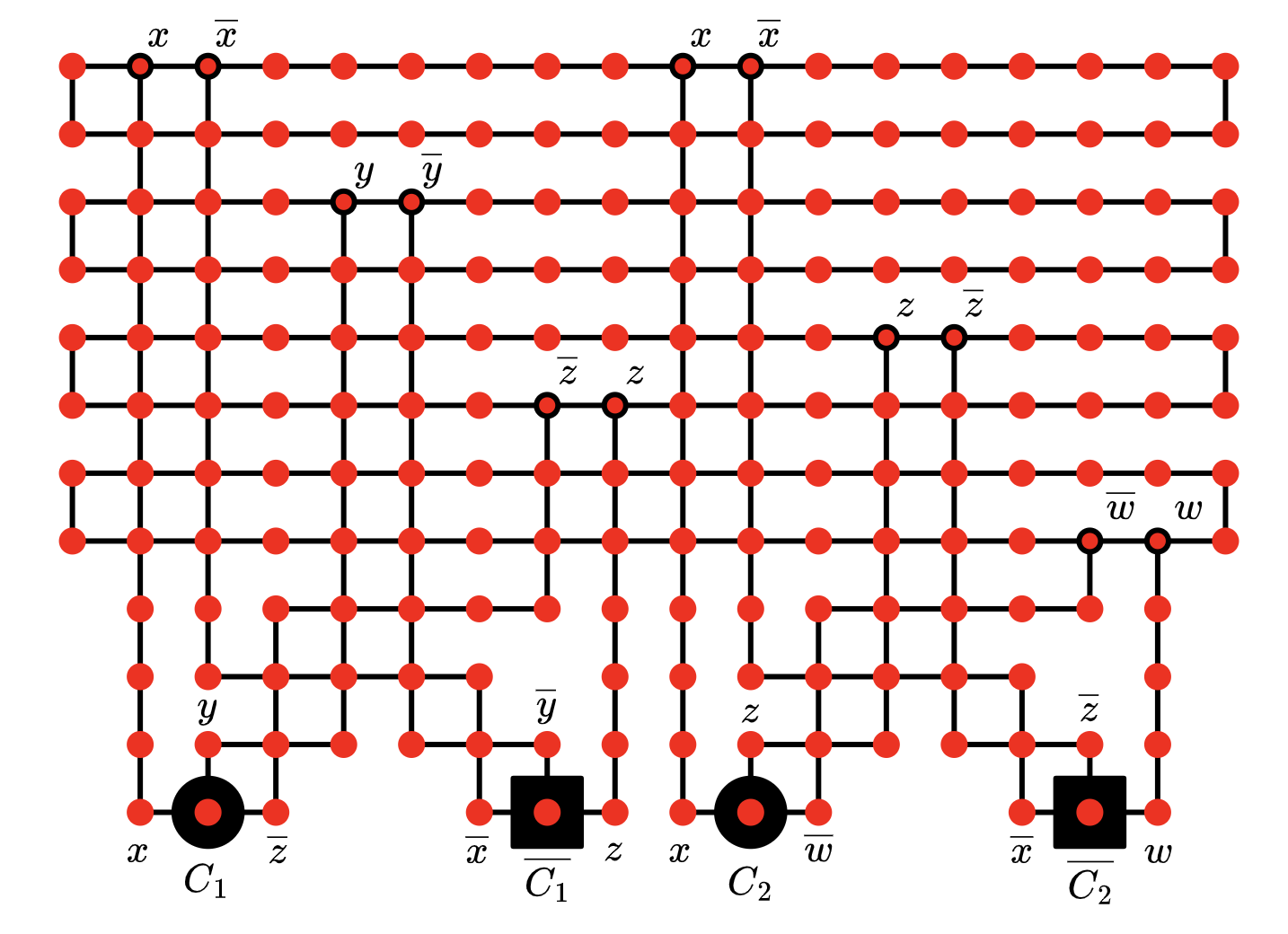}
    \caption{An example of the Graph Orientation instance, using the 1-in-3SAT instance $\phi = (x\vee y\vee \overline{z})\wedge (x\vee z\vee \overline{w})$, including red dots at every lattice point on the graph.}
    \label{fig:1in3go}
\end{figure}

\begin{figure}[!ht]
    \centering
    \includegraphics[width=120pt]{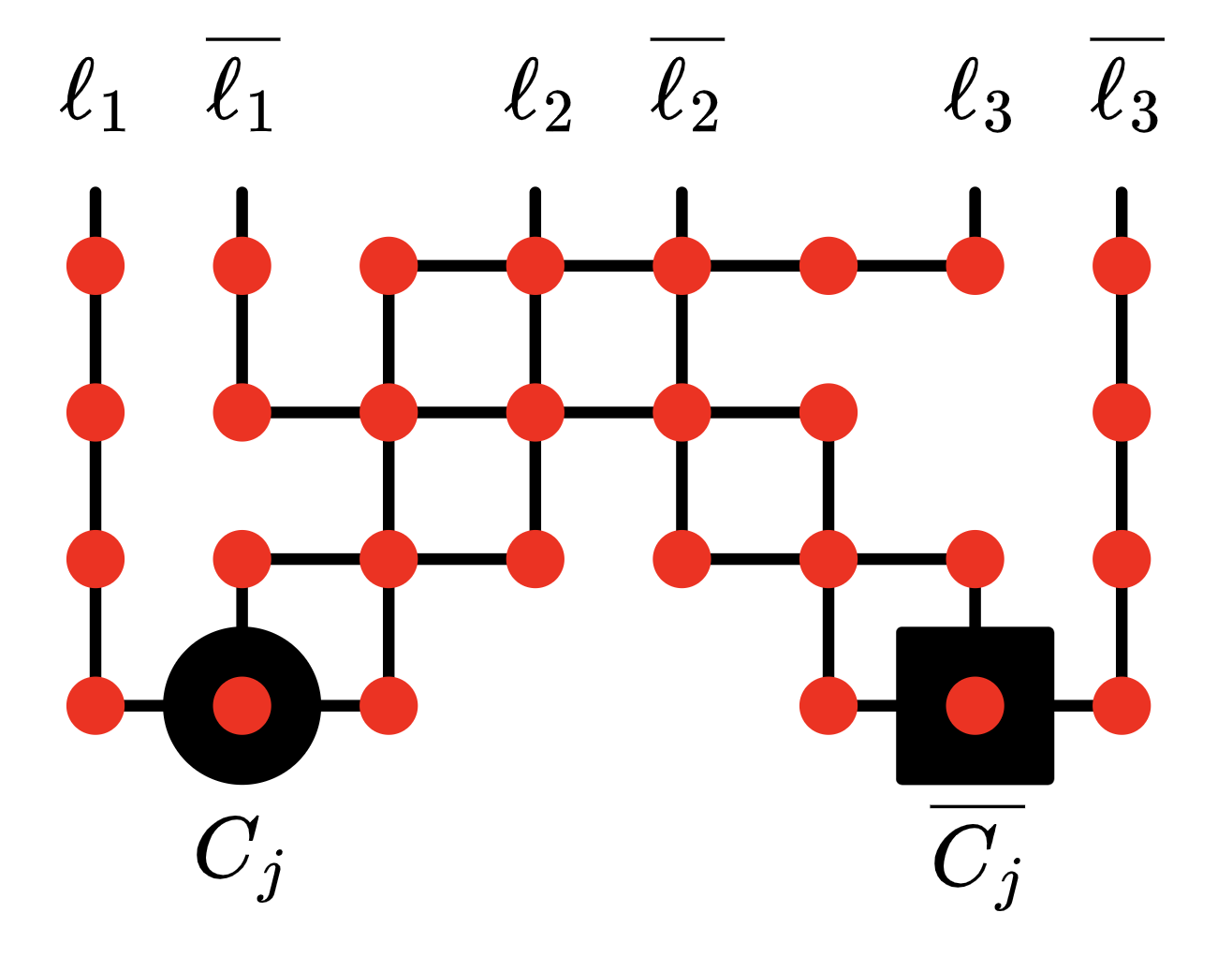}
    \caption{The clause/negated-clause construction.}
    \label{fig:cnc}
\end{figure}

\subsection{General Structure}\label{subsec:i_genstructure}
    
Now we construct our Tetris clearing instance, using the constructed Graph Orientation instance. Suppose the Graph Orientation instance requires a grid with $R$ rows and $C$ columns. We first describe the general structure of the construction, as shown in Figure~\ref{fig:i_structure}. The board has $60R+10$ rows and $48C+16$ columns. The main part of the construction, corresponding to the Graph Orientation instance, is indicated in green and lies between rows $9$ and $60R+8$, inclusive, and between columns $14$ and $48C+13$, inclusive. There is a path of empty squares to the main part of the construction, which is located mostly within the leftmost 12 columns or the topmost 8 rows of the board, and which we will refer to as the ``long tunnel''. There are 4 additional empty squares on the second to rightmost column of the board, two at the two topmost rows of the board and two at the two bottom-most rows of the board. We give the player just enough pieces so that clearing the board requires use of all pieces; in other words, if there are $M$ empty squares in the construction, we give the player $M/4$ $\II$ pieces.

\begin{figure}[!ht]
    \centering
    \includegraphics[width=320pt]{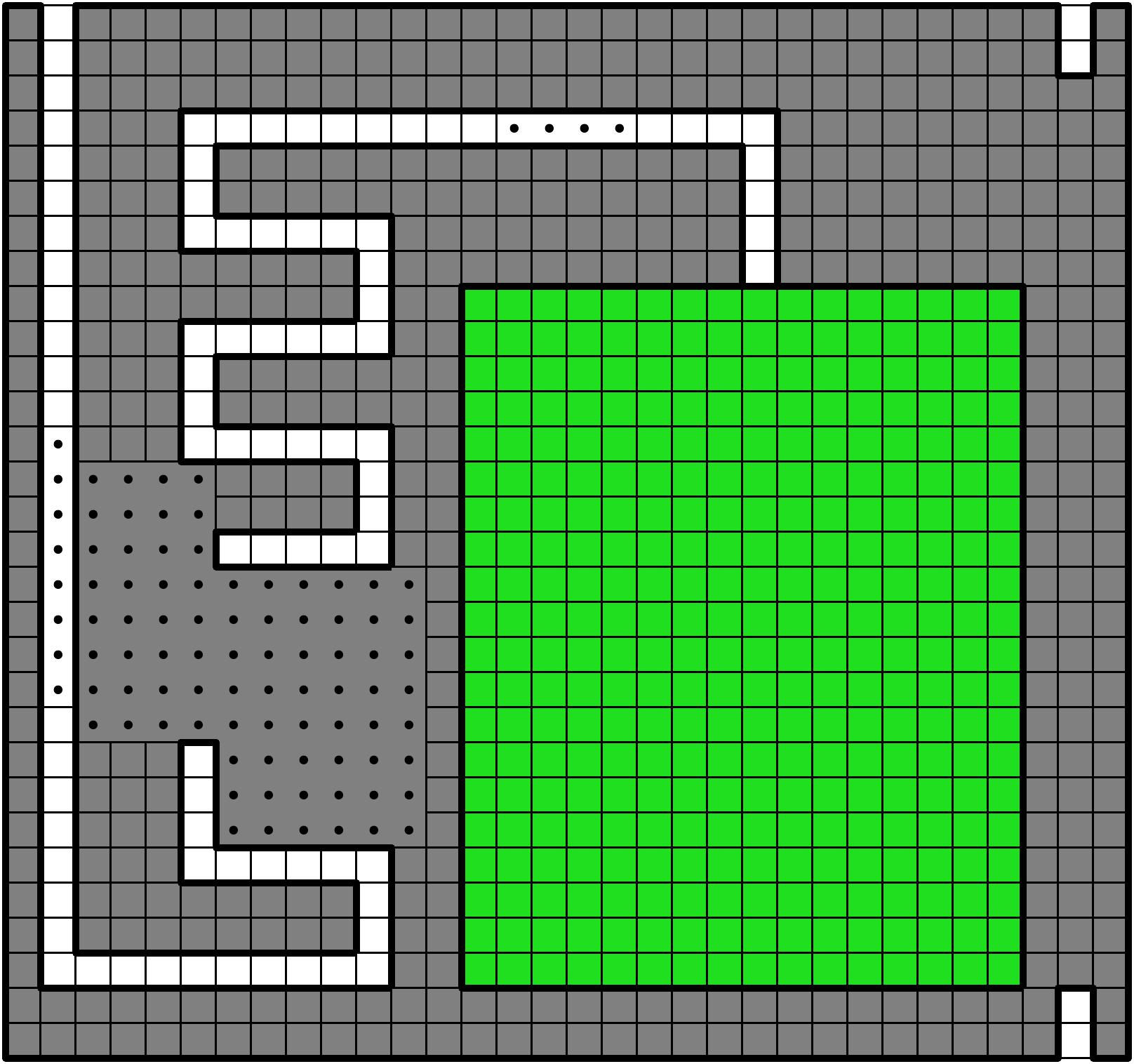}
    \caption{The general structure of the construction for Tetris clearing using only $\II$ pieces. The main part of the construction is indicated in green.}
    \label{fig:i_structure}
\end{figure}

We make the following observations:
    
\begin{itemize}
    \item The 4 empty squares in the second-to-rightmost column must be the last squares to be filled in, after the middle rows are cleared.
    \item There is exactly one way to fill in/tile the long tunnel with $\II$ pieces, as shown in Figure~\ref{fig:i_structure_fill}, due to the turns in the long tunnel, and it is possible (under SRS) to get $\II$ pieces through the long hallway into the main part of the construction (see the maneuver in Appendix~\ref{subsec:i_maneuver1}).
    \item If any $\II$ piece locks in place in the long tunnel before the main part of the construction is completely filled, then the long tunnel becomes blocked, and we cannot clear the board anymore – no more $\II$ pieces can be rotated into the main part of the construction even after some lines clear due to the two empty squares in the top-right corner and the $\II$ piece blocking the long tunnel. Thus, no $\II$ piece can be placed in the long tunnel before we are done filling in the main part of the construction.
    \item Before we place any $\II$ pieces in the long tunnel, the long tunnel also prevents line clears in the rows corresponding to the main part of the construction.
\end{itemize}

\begin{figure}[!ht]
    \centering
    \includegraphics[width=320pt]{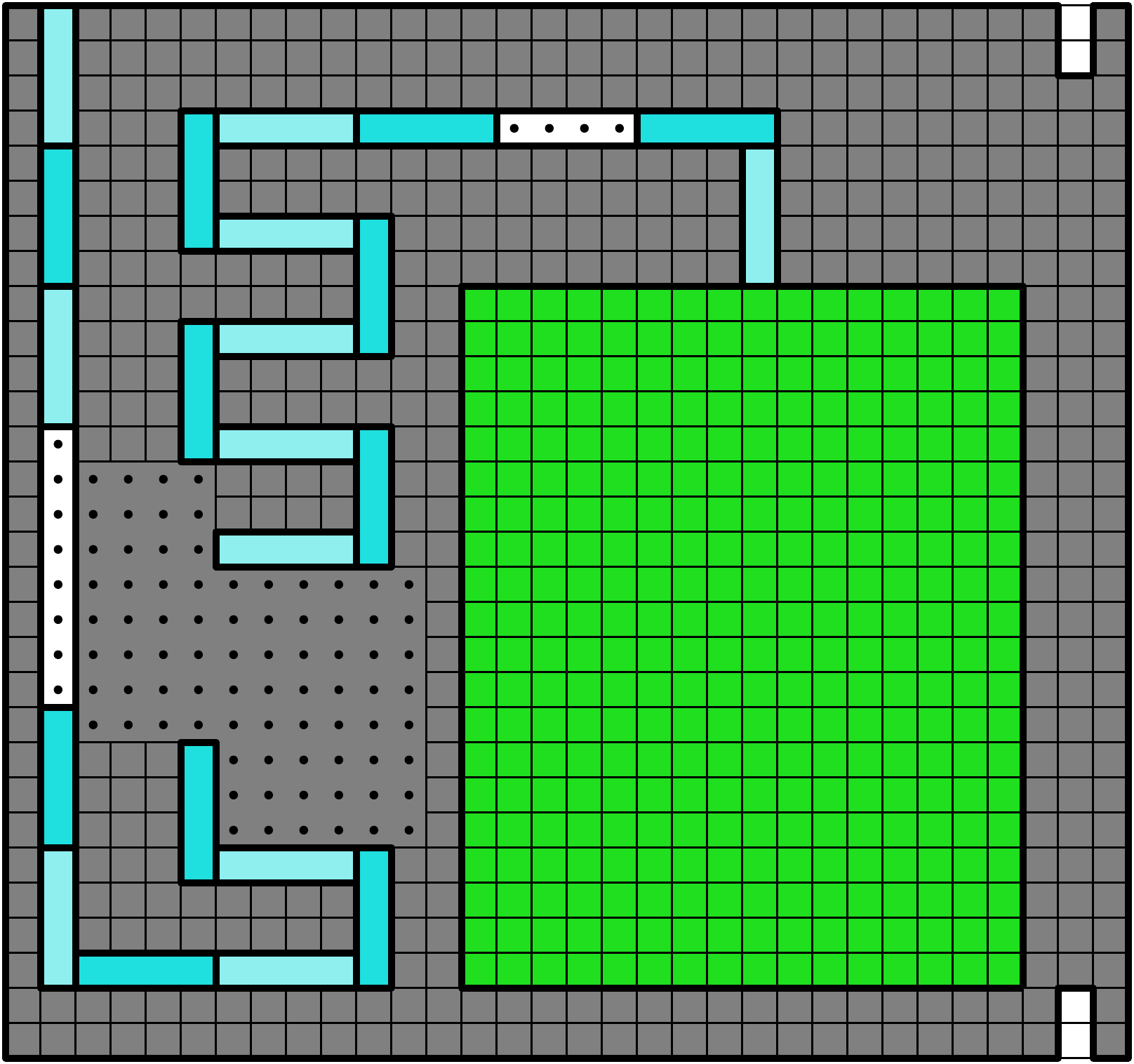}
    \caption{The only possible tiling of the long tunnel with $\II$ pieces.}
    \label{fig:i_structure_fill}
\end{figure}

Thus, we need to be able to tile the main part of the construction with $\II$ pieces exactly, with the additional restriction that we need to be able to maneuver all the pieces into place through Tetris and SRS rules.

\subsection{Gadgets}

Now, we introduce the gadgets used in the main part of the construction. Roughly speaking, the gadgets are similar to the gadgets used in the 2BAR-Tiling reduction in \cite{horiyama2017complexity}, with slight modifications to allow for the player to be able to maneuver $\II$ pieces through the gadgets. This means that we have corner, line, duplicator, clause, negated-clause, and crossover gadgets.

Here, we will call the entrances into the gadgets \textbf{ports} (so for example, the horizontal line gadget has a port on the left and a port on the right); we will call the ports \emph{up, left, right, down} (in the case of the duplicator, both of the lower ports are down ports). The ports will be indicated as an empty (white) square with a dot in the center. Note that these ports also serve as ``extra'' squares which help denote the ``orientation'' around the gadget.

We will also indicate specific squares in the diagrams for each of the gadgets as ``centers'', to help with the overall description of the main part of the construction; each gadget will have one center, except for the duplicator gadget, which will have two. Of note is that all left and right ports are on the same row as a center and all up and down ports are on the same column as a center.

In general, we want the gadgets to have the following properties:

\begin{itemize}
    \item Any traversal of the following form is possible, assuming both ports exist: up $\to$ $\{$left, right, down$\}$, left $\leftrightarrow$ right (i.e., both left to right and right to left), $\{$left, right$\}$ $\to$ down.
    \item The possible ways to tile the gadgets roughly match the possible ways to tile the corresponding gadgets used in the 2BAR-Tiling reduction in \cite{horiyama2017complexity}.
    \item For each possible tiling, each individual gadget can be filled with $\II$ pieces under SRS, where if a gadget has an up port, the $\II$ pieces come from the up port, and otherwise the $\II$ pieces come from a left or right port. (In other words, the $\II$ pieces are able to move to where they need to go in the tiling, particularly as the gadget is partially filled with $\II$ pieces.)
\end{itemize}

The main trick we use is the fact that $\II$ pieces can \emph{turn corners} – the fourth and fifth tests for rotation of $\II$ pieces (as shown in Figure~\ref{fig:allrots}) can be thought of as rotations about one of the two end squares of an $\II$ piece, so assuming the squares around an active $\II$ piece are filled such that the first three tests fail, an active $\II$ piece can rotate $90^\circ$ about either of its two end squares (like multiple rotations in the maneuver in Appendix~\ref{subsec:i_maneuver1}).

\subsubsection{Corner and Line gadgets}

The corner gadget is fairly straightforward; two orientations are shown, along with center squares, in Figure~\ref{fig:corner} (the other two orientations can be found via reflection across a vertical line). Call a corner with a down port a \emph{lower corner (LC)} and a corner with an up port an \emph{upper corner (UC)}.

\begin{figure}[!ht]
  \centering
  \begin{subfigure}[b]{0.49\textwidth}
    \centering
    \includegraphics[width=160pt]{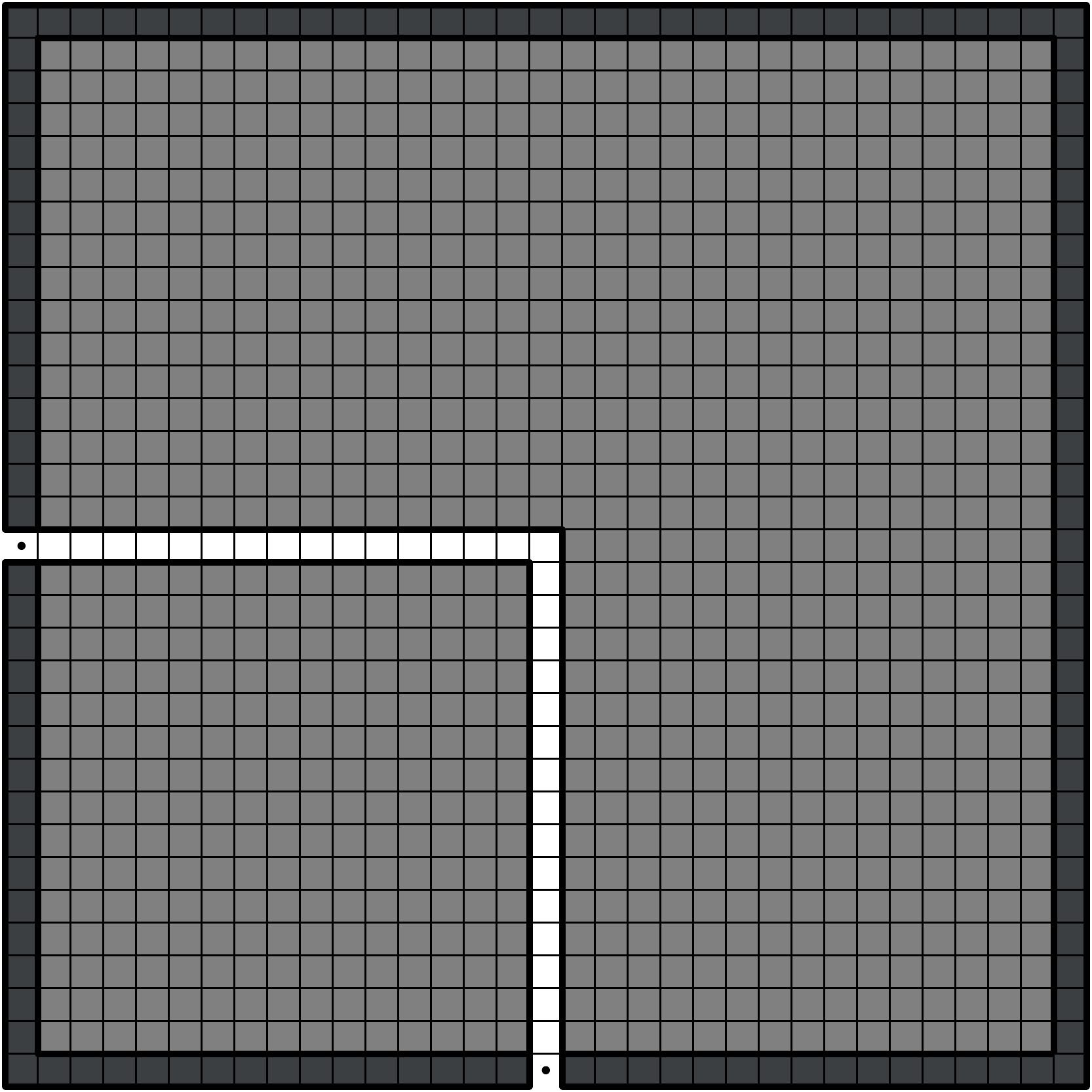}
    \caption{Corner with down port (LC)}
  \end{subfigure}
  \begin{subfigure}[b]{0.49\textwidth}
    \centering
    \includegraphics[width=160pt]{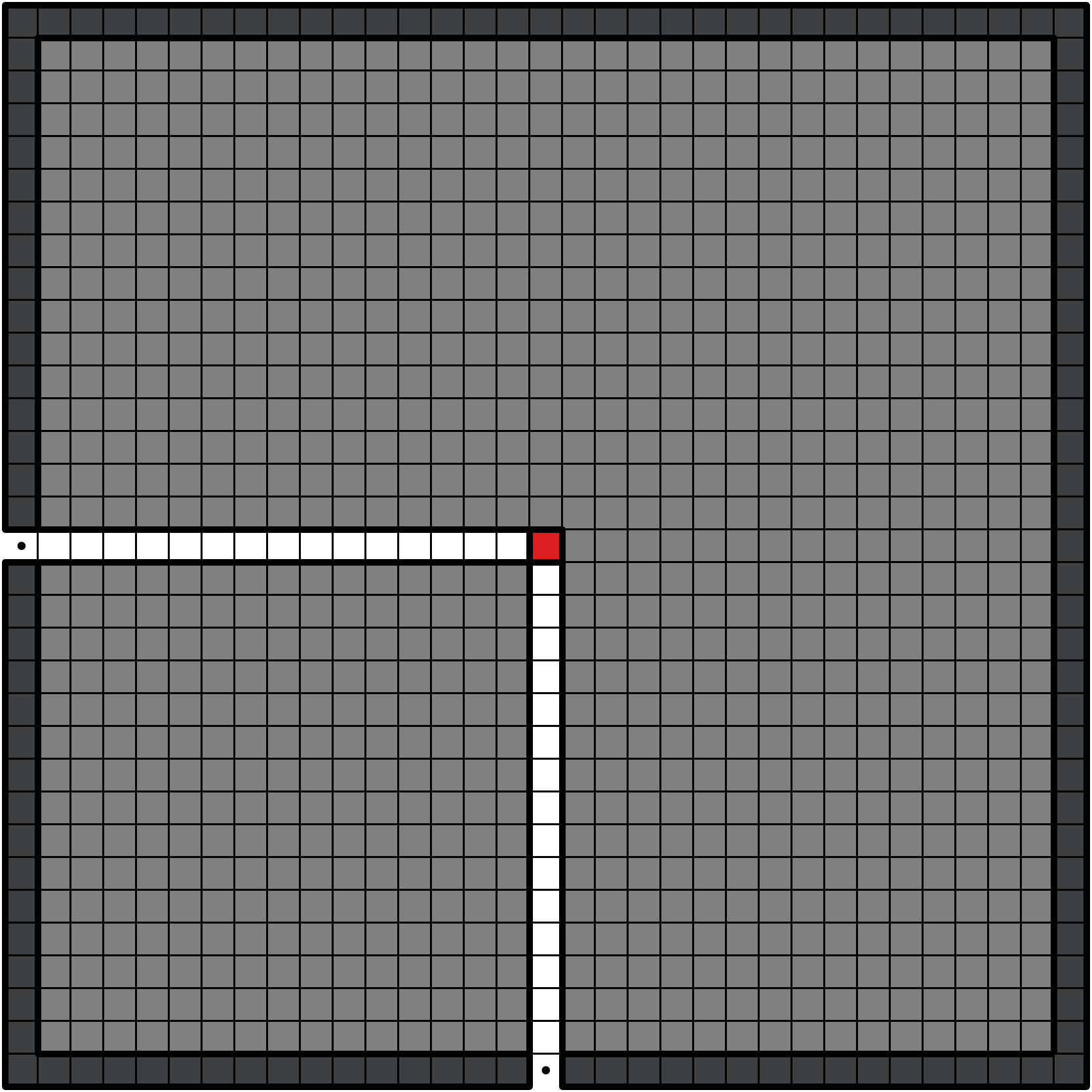}
    \caption{Center square of LC gadget}
  \end{subfigure}
  \begin{subfigure}[b]{0.49\textwidth}
    \centering
    \includegraphics[width=160pt]{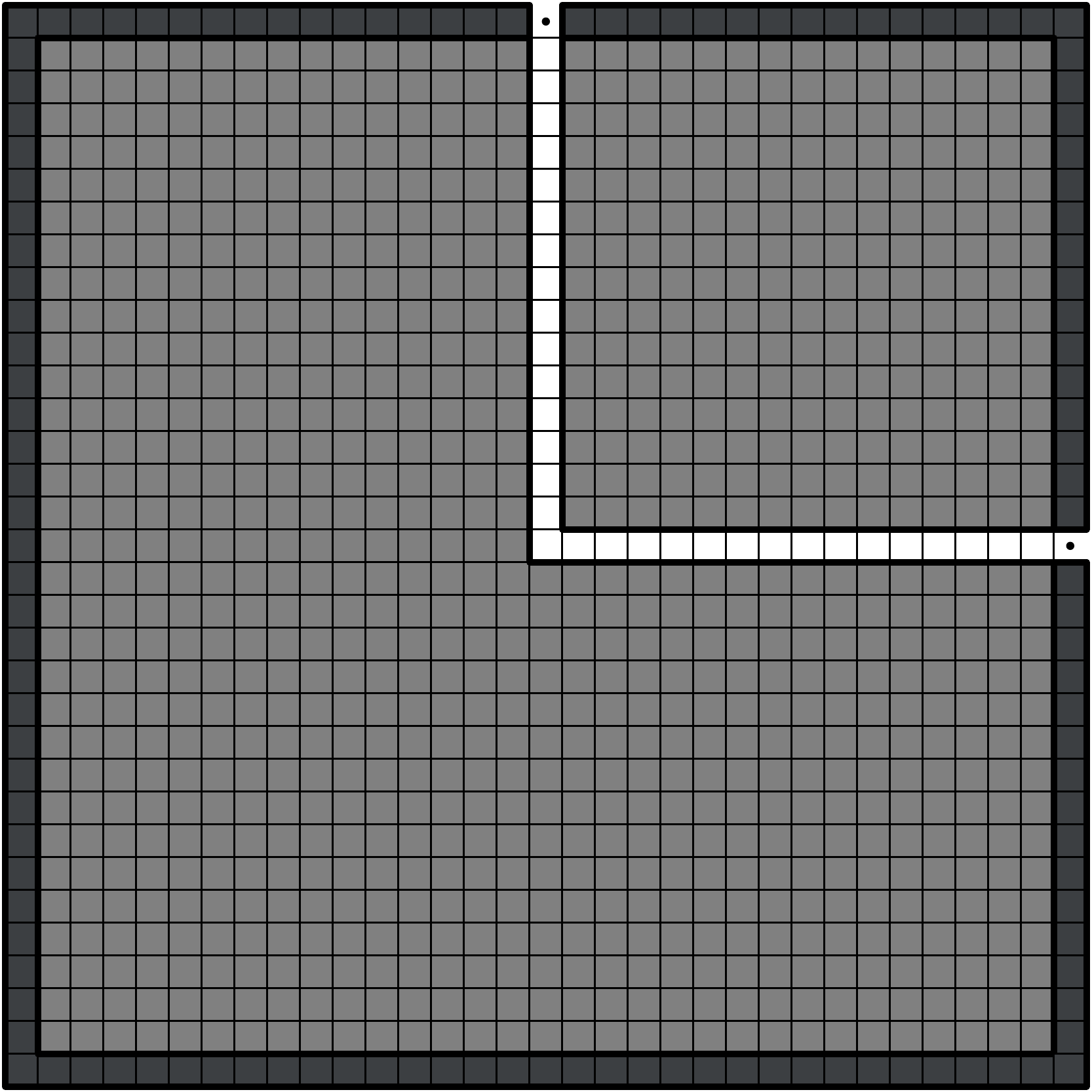}
    \caption{Corner with up port (UC)}
  \end{subfigure}
  \begin{subfigure}[b]{0.49\textwidth}
    \centering
    \includegraphics[width=160pt]{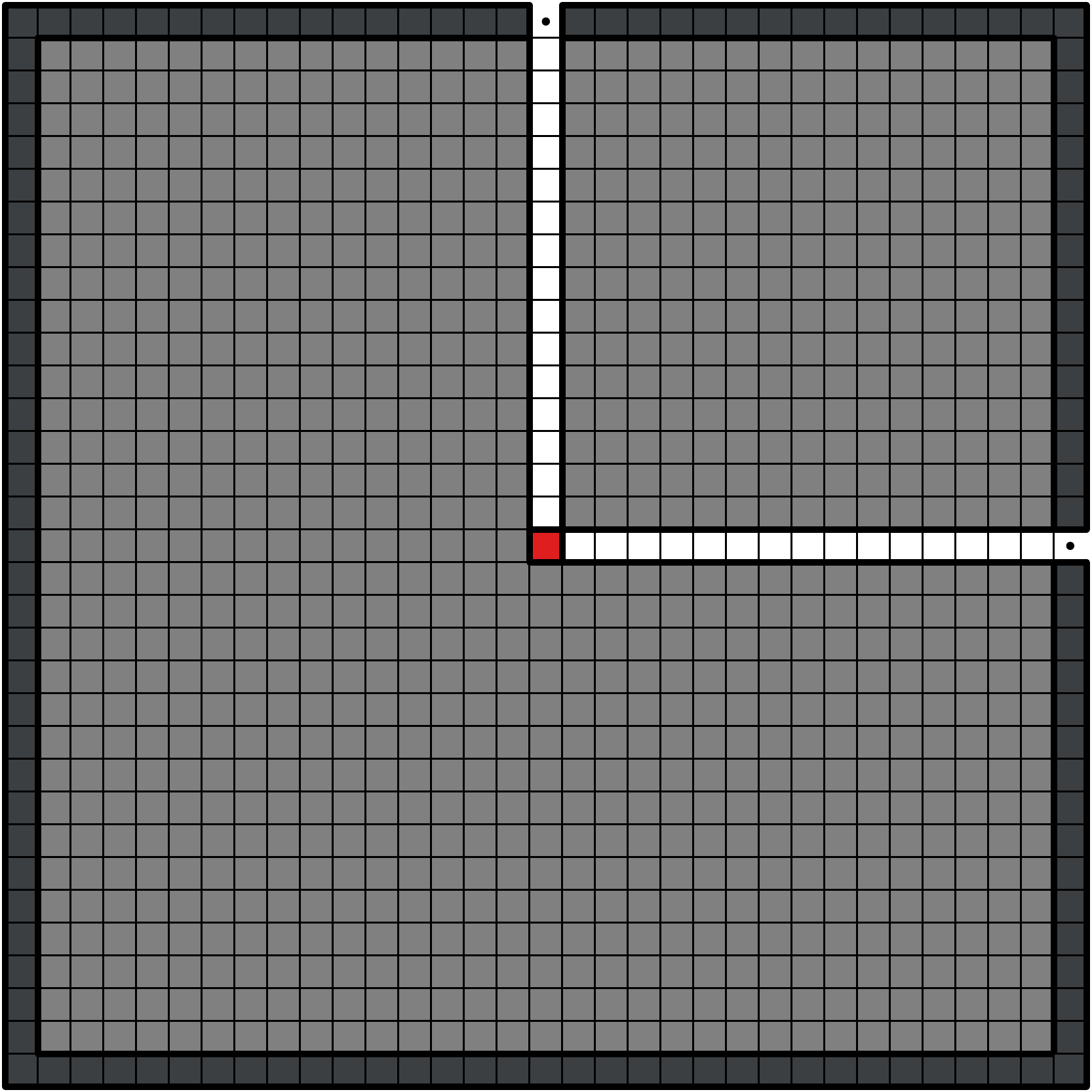}
    \caption{Center square of UC gadget}
  \end{subfigure}
  \caption{Corner gadgets.}
  \label{fig:corner}
\end{figure}

As $\II$ pieces can turn corners, the traversal from the non-down port to the down port is possible in an LC gadget, and the traversal from the up port to the non-up port is possible in a UC gadget. The only possible tilings, along with an order in which the $\II$ pieces can be placed to ensure that the $\II$ pieces tile the gadget correctly, are shown in Figure~\ref{fig:i_corner_tiling}.


Because $\II$ pieces need to enter the main part of the gadget, we create a special \emph{entry corner (EC)} gadget, as shown in Figure~\ref{fig:entrycorner}. One can verify that the upper part of the gadget is forced to be tiled a specific way as shown in Figure~\ref{fig:ec_forced}, while the rest of the gadget functions as a normal LC gadget. As $\II$ pieces can turn corners, the traversal from the upper part of the gadget to either of the two ports is possible. The only possible tilings, along with an order in which the $\II$ pieces can be placed to ensure that the $\II$ pieces tile the gadget correctly, are shown in Figure~\ref{fig:i_entrycorner_tilings}.

\begin{figure}[!ht]
  \centering
  \begin{subfigure}[b]{0.325\textwidth}
    \centering
    \includegraphics[width=140pt]{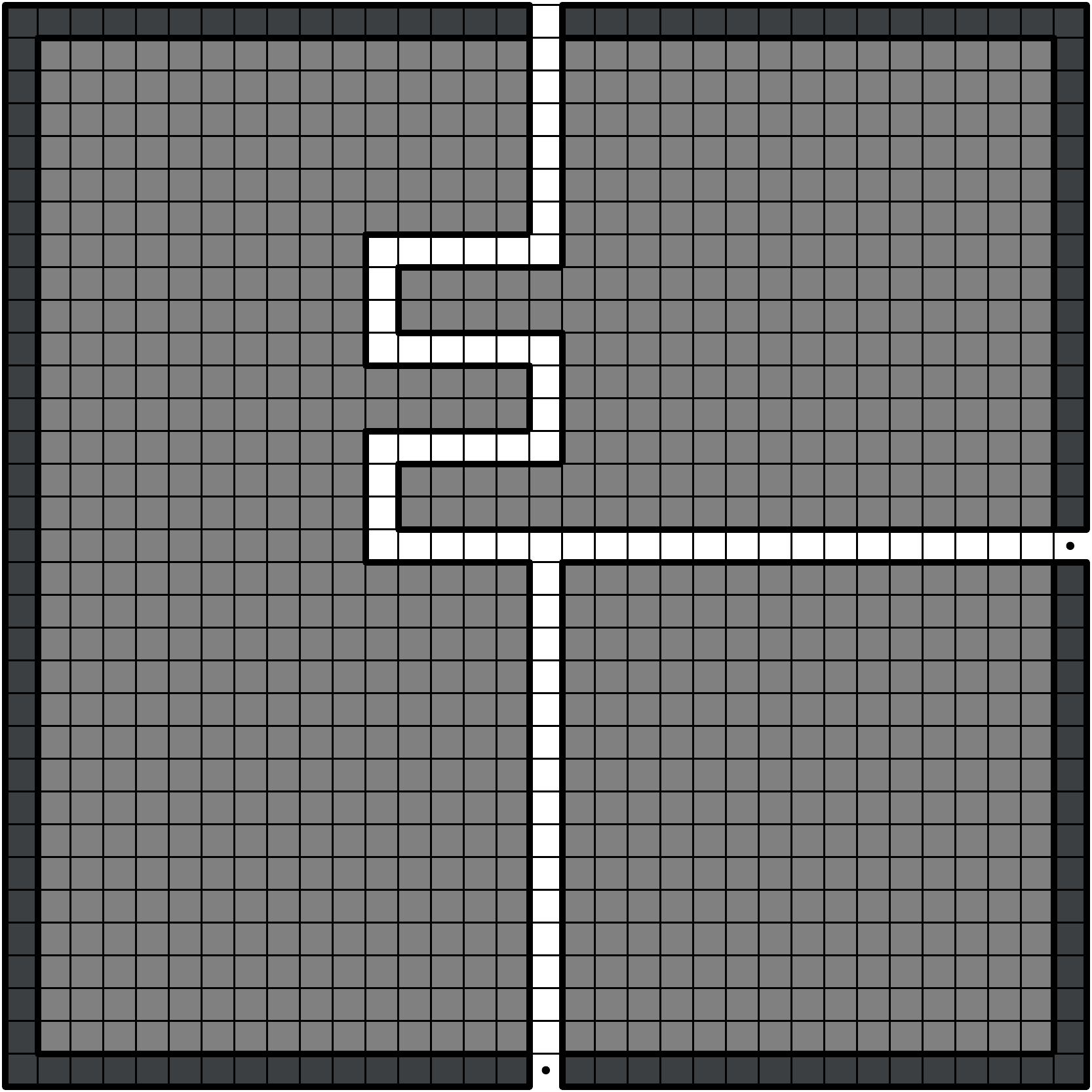}
    \caption{}
  \end{subfigure}
  \begin{subfigure}[b]{0.325\textwidth}
    \centering
    \includegraphics[width=140pt]{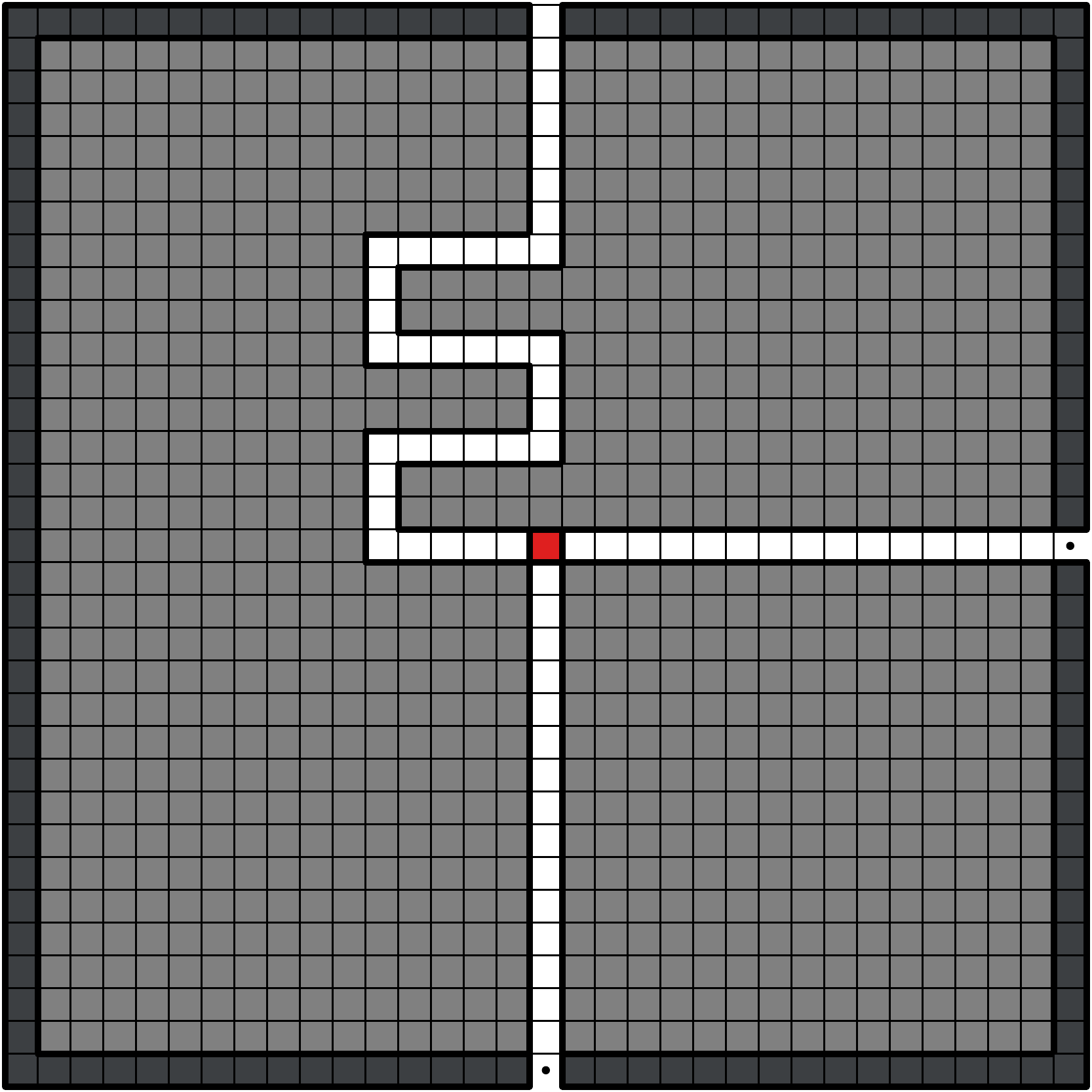}
    \caption{Center square of gadget}
  \end{subfigure}
  \begin{subfigure}[b]{0.325\textwidth}
    \centering
    \includegraphics[width=140pt]{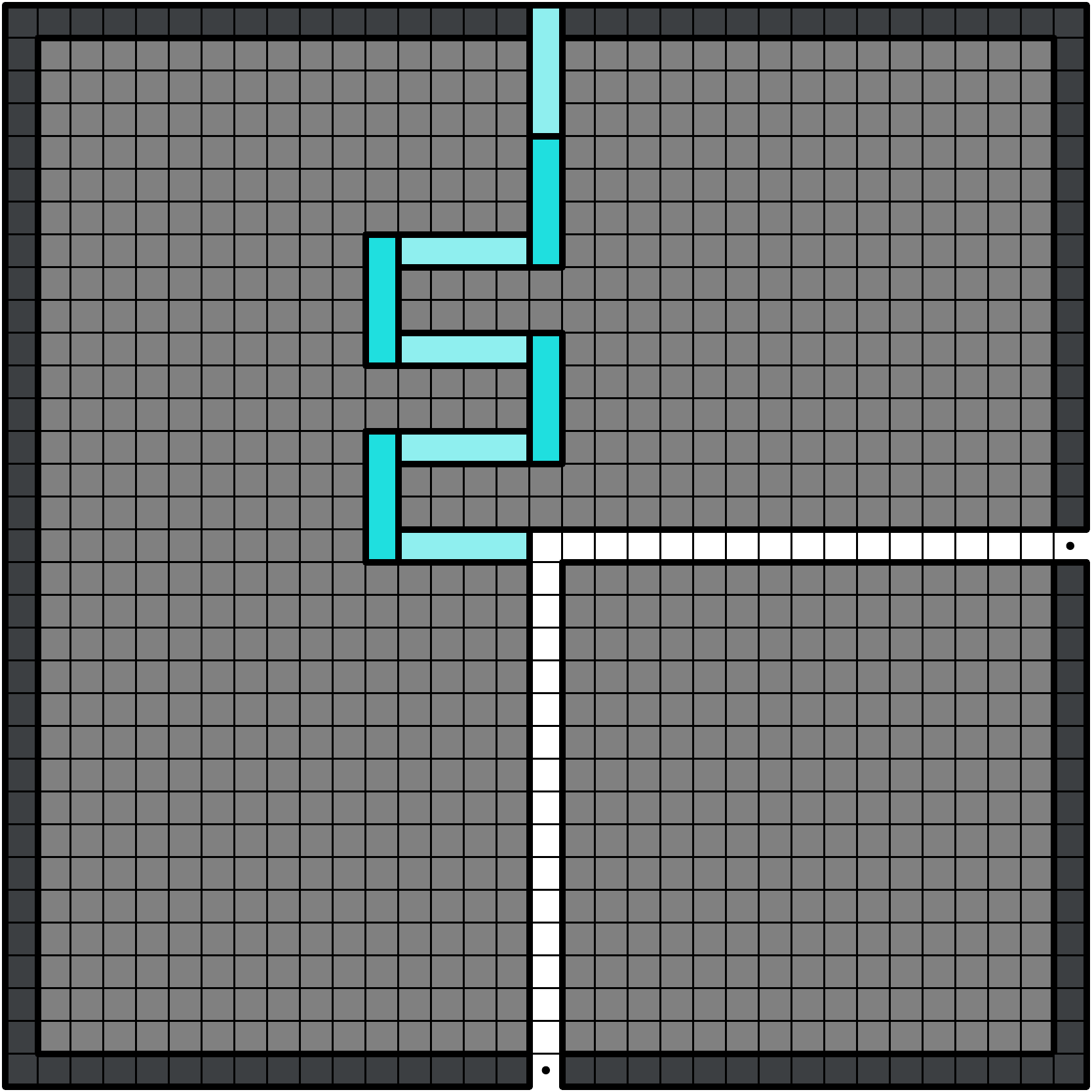}
    \caption{Forced placements}
    \label{fig:ec_forced}
  \end{subfigure}
  \caption{The entry corner (EC) gadget.}
  \label{fig:entrycorner}
\end{figure}


Like the normal corner gadgets, the vertical line (VL) gadget is exactly the vertical line gadget in \cite{horiyama2017complexity}, as shown in Figure~\ref{fig:vertline}. As $\II$ pieces can turn corners, the traversal from the up port to the down port is possible. The only possible tilings, along with an order in which the $\II$ pieces can be placed to ensure that the $\II$ pieces tile the gadget correctly, are shown in Figure~\ref{fig:i_vertline_tilings}.

\begin{figure}[!ht]
  \centering
  \begin{subfigure}[b]{0.49\textwidth}
    \centering
    \includegraphics[width=160pt]{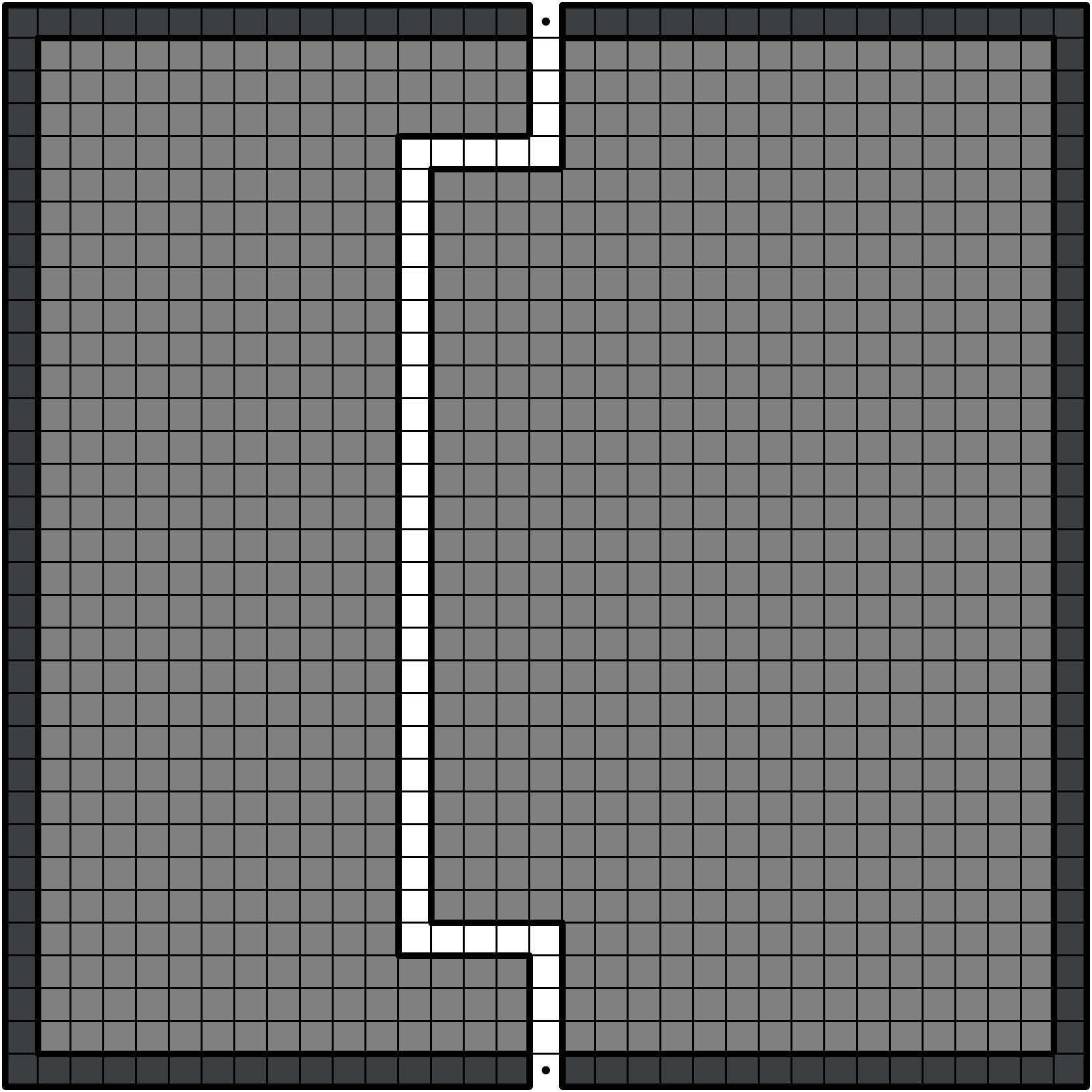}
    \caption{}
  \end{subfigure}
  \begin{subfigure}[b]{0.49\textwidth}
    \centering
    \includegraphics[width=160pt]{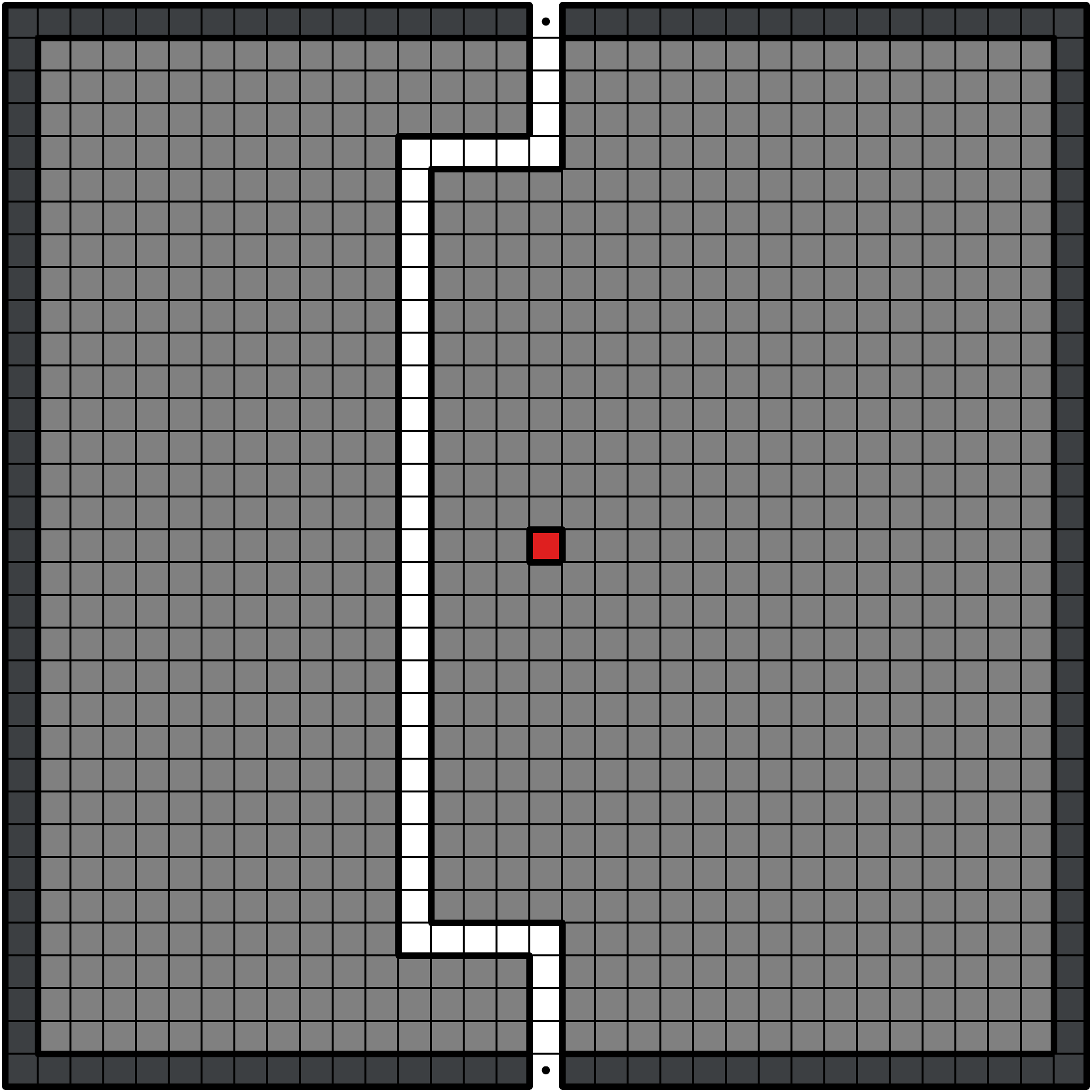}
    \caption{Center square of gadget}
  \end{subfigure}
  \caption{The vertical line (VL) gadget.}
  \label{fig:vertline}
\end{figure}


The horizontal line (HL) gadget, on the other hand, is similar to the horizontal line gadget in \cite{horiyama2017complexity}, but with some additional squares to allow for an easier traversal from the left port to the right port and vice versa (as the player just needs to slide the horizontally-oriented $\II$-piece from left to right or vice versa). The HL gadget is shown in Figure~\ref{fig:horline}. One can verify that the upper part of the center rectangle is forced to be tiled a specific way as shown in Figure~\ref{fig:hl_forced}. The only possible tilings, along with orders in which the $\II$ pieces can be placed to ensure that the $\II$ pieces tile the gadget correctly, are shown in Figure~\ref{fig:i_horline_tilings}. Note that each tiling has two orderings, one for if the pieces come from the left port and one for if the pieces come from the right port. The tilings are also possible because $\II$ pieces can turn corners as necessary.

\begin{figure}[!ht]
  \centering
  \begin{subfigure}[b]{0.325\textwidth}
    \centering
    \includegraphics[width=140pt]{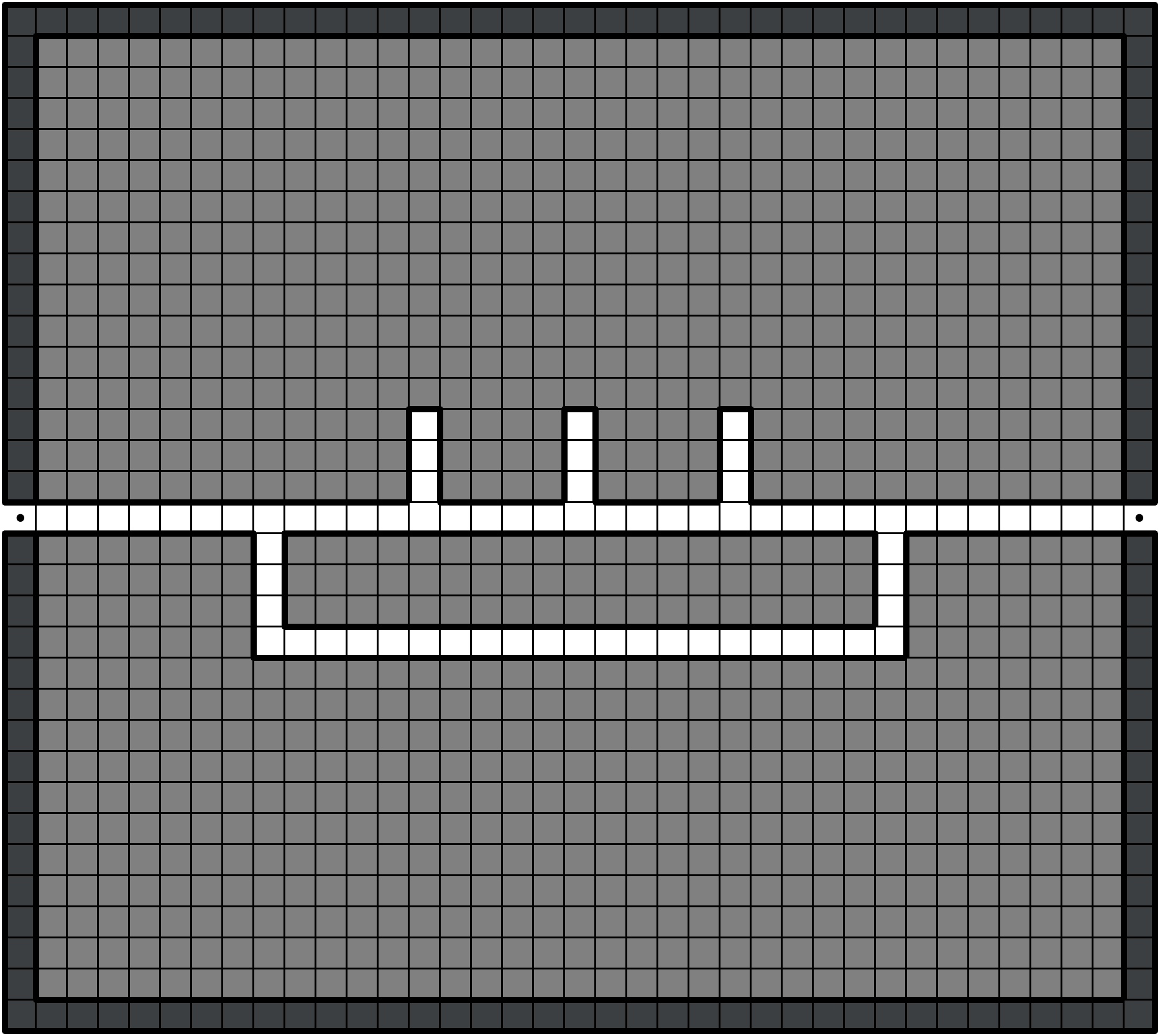}
    \caption{}
  \end{subfigure}
  \begin{subfigure}[b]{0.325\textwidth}
    \centering
    \includegraphics[width=140pt]{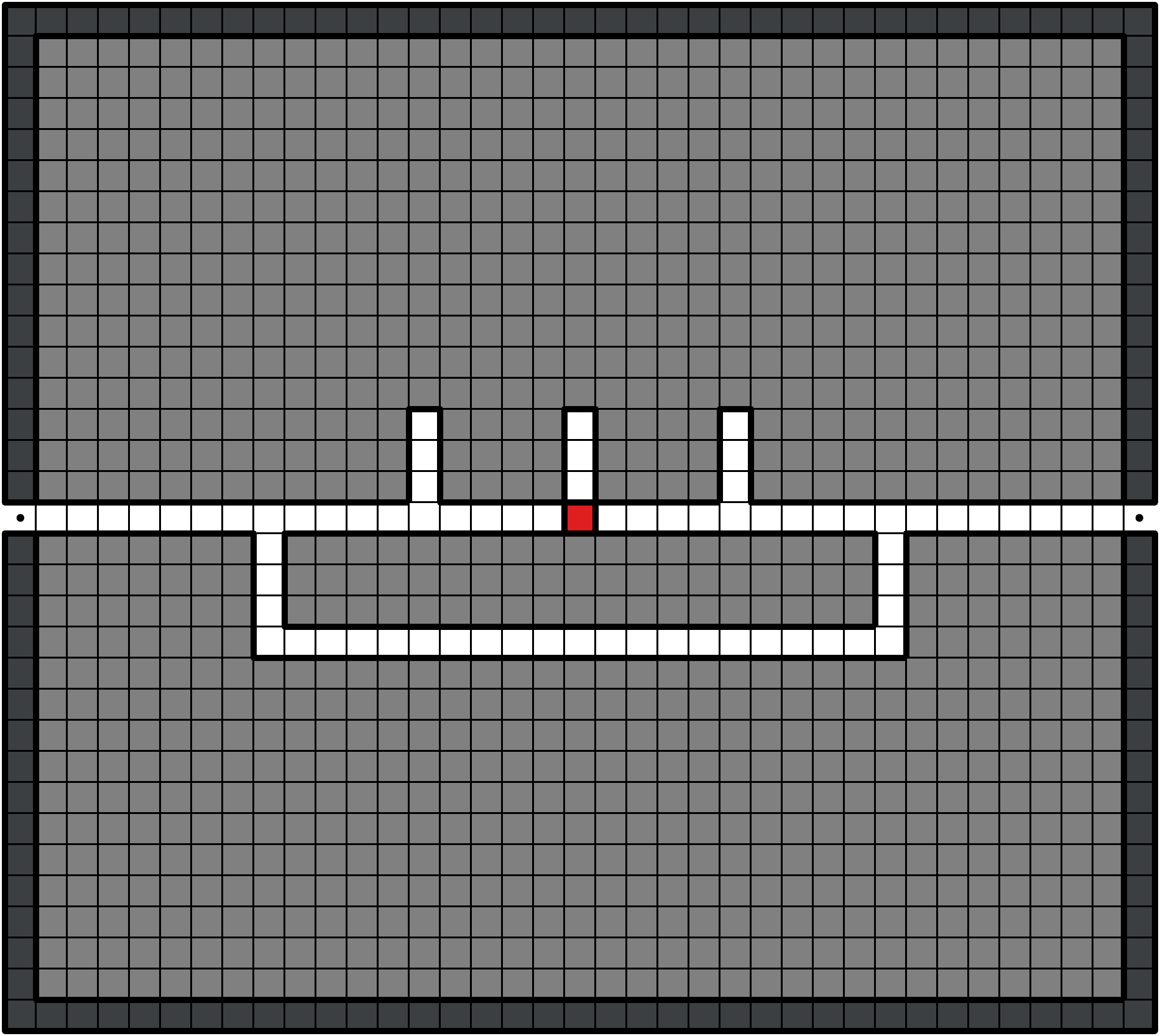}
    \caption{Center square of gadget}
  \end{subfigure}
  \begin{subfigure}[b]{0.325\textwidth}
    \centering
    \includegraphics[width=140pt]{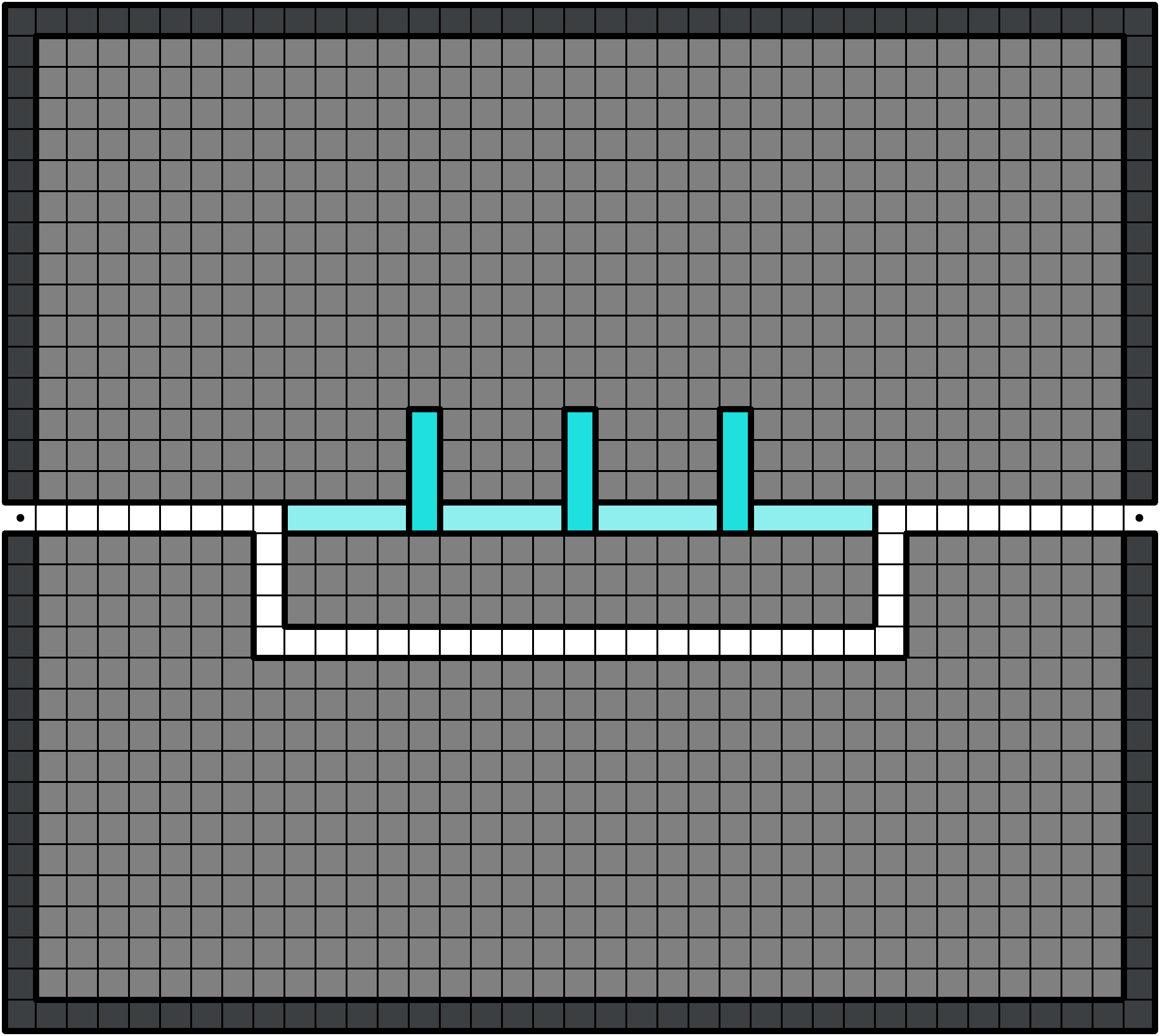}
    \caption{Forced placements}
    \label{fig:hl_forced}
  \end{subfigure}
  \caption{The horizontal line (HL) gadget.}
  \label{fig:horline}
\end{figure}


\subsubsection{Duplicator gadget}

The duplicator gadget, like the HL gadget, is similar to the duplicator gadget in \cite{horiyama2017complexity}, but with some slight alterations near the left and right ports and additional squares to allow for an easier traversal from the left port to the right port and vice versa (as the player just needs to slide the horizontally-oriented $\II$-piece from left to right or vice versa). The duplicator gadget is shown in Figure~\ref{fig:duplicator}. One can verify that the ``bridge'' between the left port and right port areas is forced to be tiled a specific way as shown in Figure~\ref{fig:duplicator_forced}.

As $\II$ pieces can turn corners, the traversals from either the left port or the right port to either of the down ports are possible, as are the traversals between the left and right ports. The only possible tilings, along with orders in which the $\II$ pieces can be placed to ensure that the $\II$ pieces tile the gadget correctly, are shown in Figure~\ref{fig:i_duplicator_tilings}. Note that each tiling has two orderings, one for if the pieces come from the left port and one for if the pieces come from the right port. The tilings are possible because $\II$ pieces can turn corners as necessary. The rotations into piece numbers $48$ and $60$ in all tilings require that the $\II$ piece be in the default orientation before the rotation; however, if the $\II$ piece were in the opposite ($180^\circ$) orientation, the player can still rotate the $\II$ piece into the default orientation in the area slightly to the right of the left port or the area slightly to the left of the right port via two counterclockwise rotations that pivot about the same end square of the $\II$ piece (see Appendix~\ref{subsec:i_maneuver2} for figures detailing this maneuver).

\begin{figure}[!ht]
  \centering
  \begin{subfigure}[b]{0.49\textwidth}
    \centering
    \includegraphics[width=200pt]{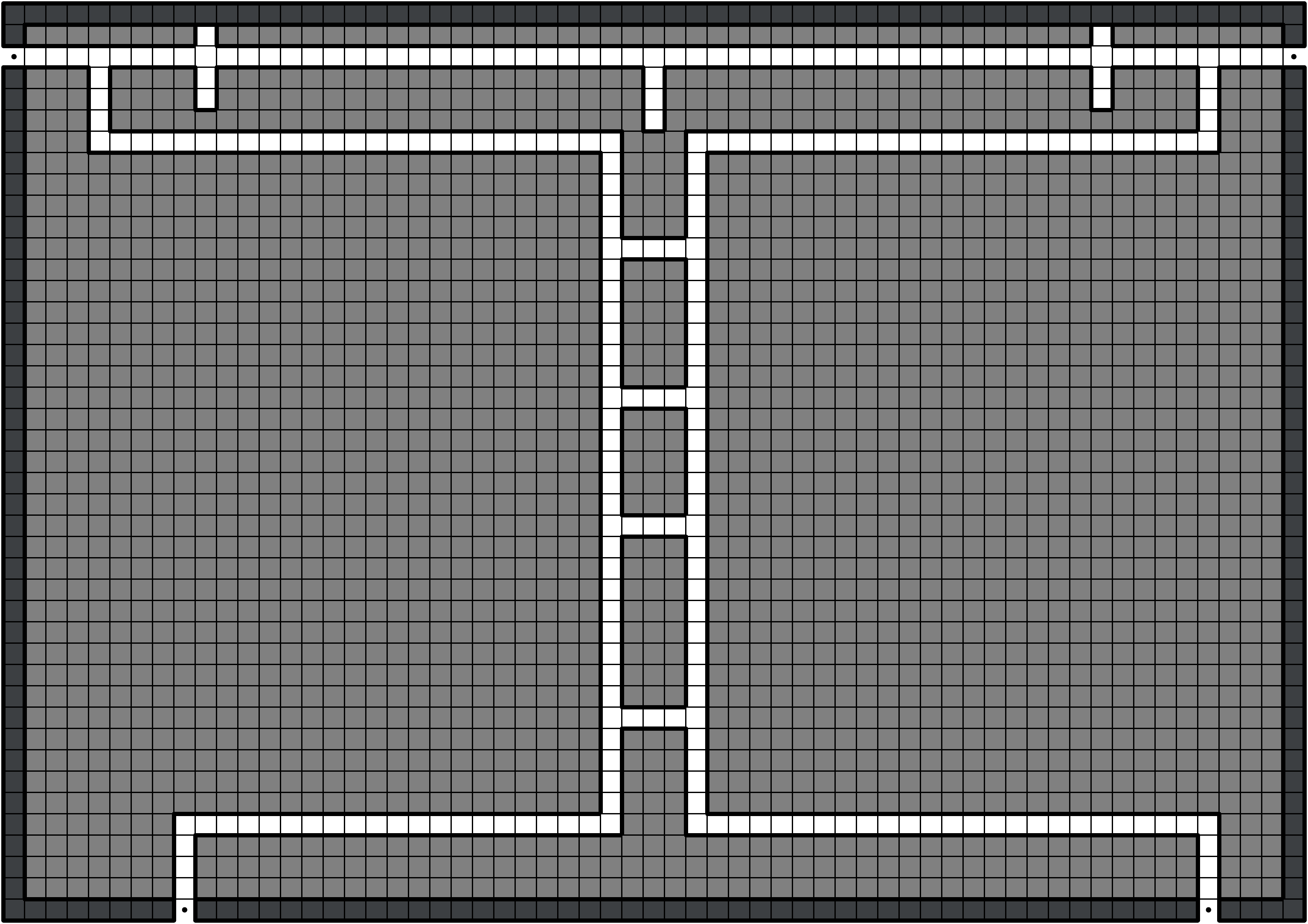}
    \caption{}
  \end{subfigure}
  \begin{subfigure}[b]{0.49\textwidth}
    \centering
    \includegraphics[width=200pt]{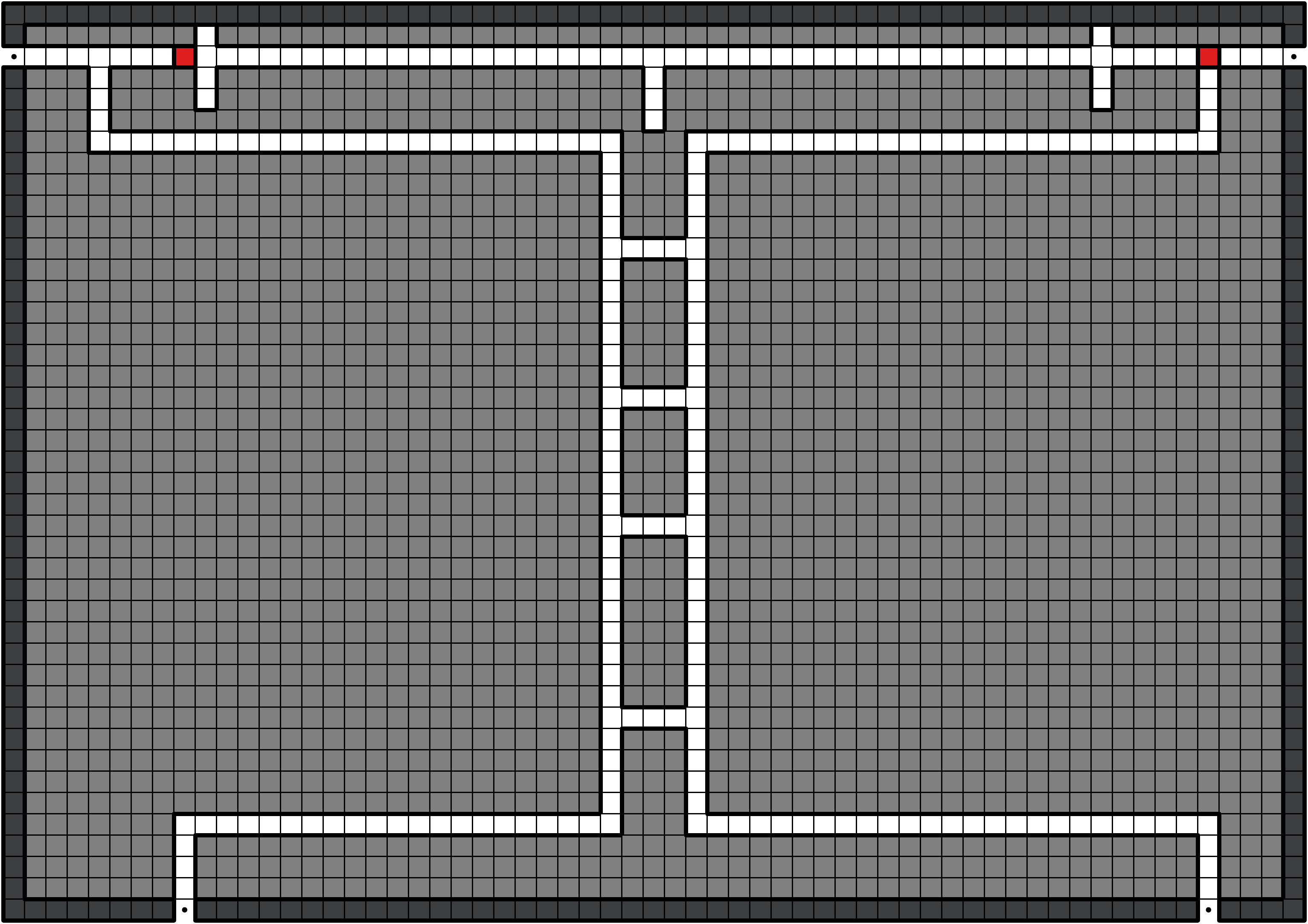}
    \caption{Center squares of gadget}
  \end{subfigure}
  \begin{subfigure}[b]{0.6\textwidth}
    \centering
    \includegraphics[width=250pt]{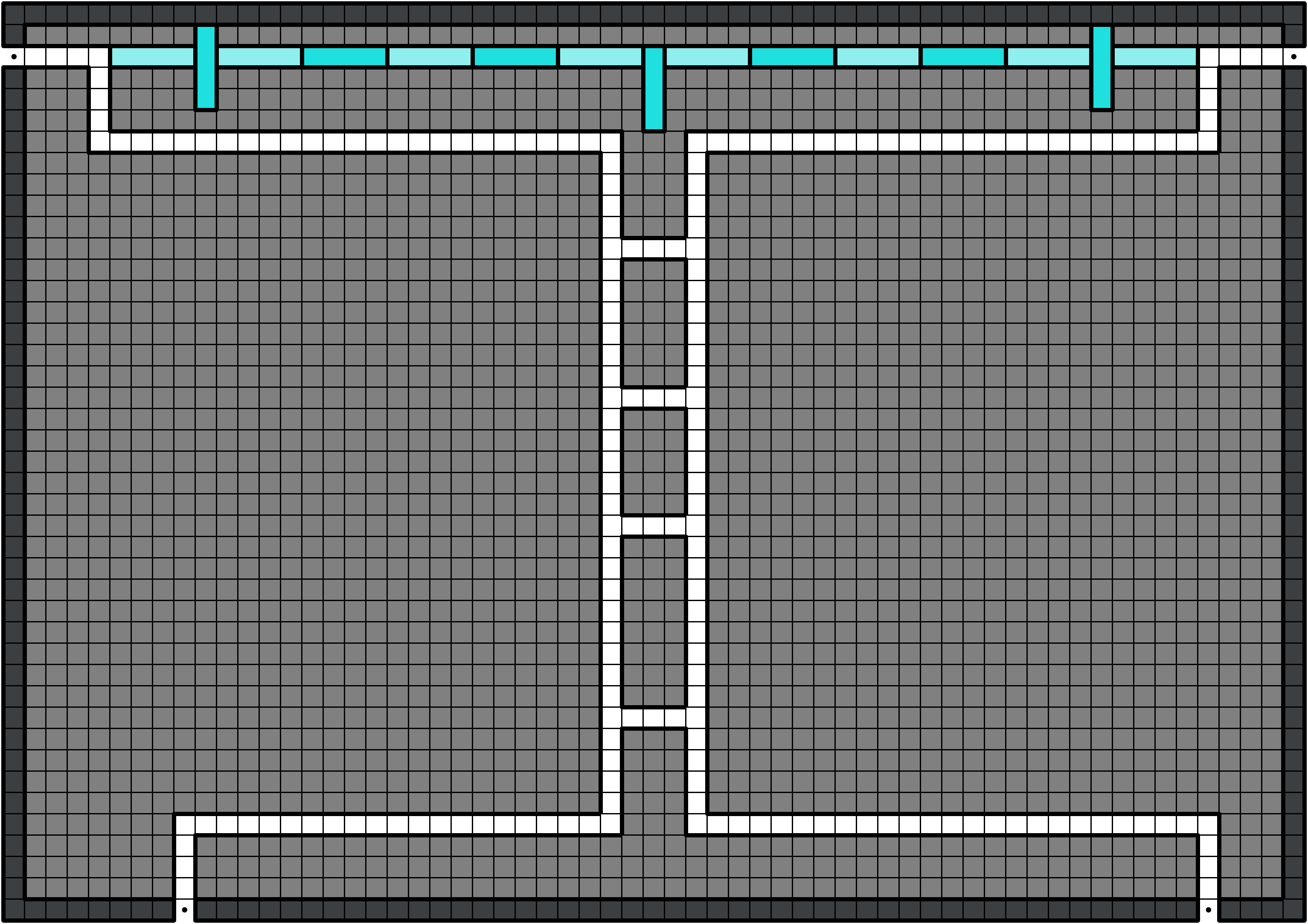}
    \caption{Forced placements}
    \label{fig:duplicator_forced}
  \end{subfigure}
  \caption{The duplicator gadget.}
  \label{fig:duplicator}
\end{figure}


\subsubsection{Clause and Negated-Clause gadgets}

The clause gadget, with similar alterations and additional squares as the other gadgets, is shown in Figure~\ref{fig:clause}. One can verify that the ``bridge'' between the left port and right port areas is forced to be tiled a specific way as shown in Figure~\ref{fig:clause_forced}.

As $\II$ pieces can turn corners, the traversals from the up port to either of the other ports are possible, as are the traversals between the left and right ports. The only possible tilings, along with orders in which the $\II$ pieces can be placed to ensure that the $\II$ pieces tile the gadget correctly, are shown in Figure~\ref{fig:i_clause_tilings}.

\begin{figure}[!ht]
  \centering
  \begin{subfigure}[b]{0.325\textwidth}
    \centering
    \includegraphics[width=140pt]{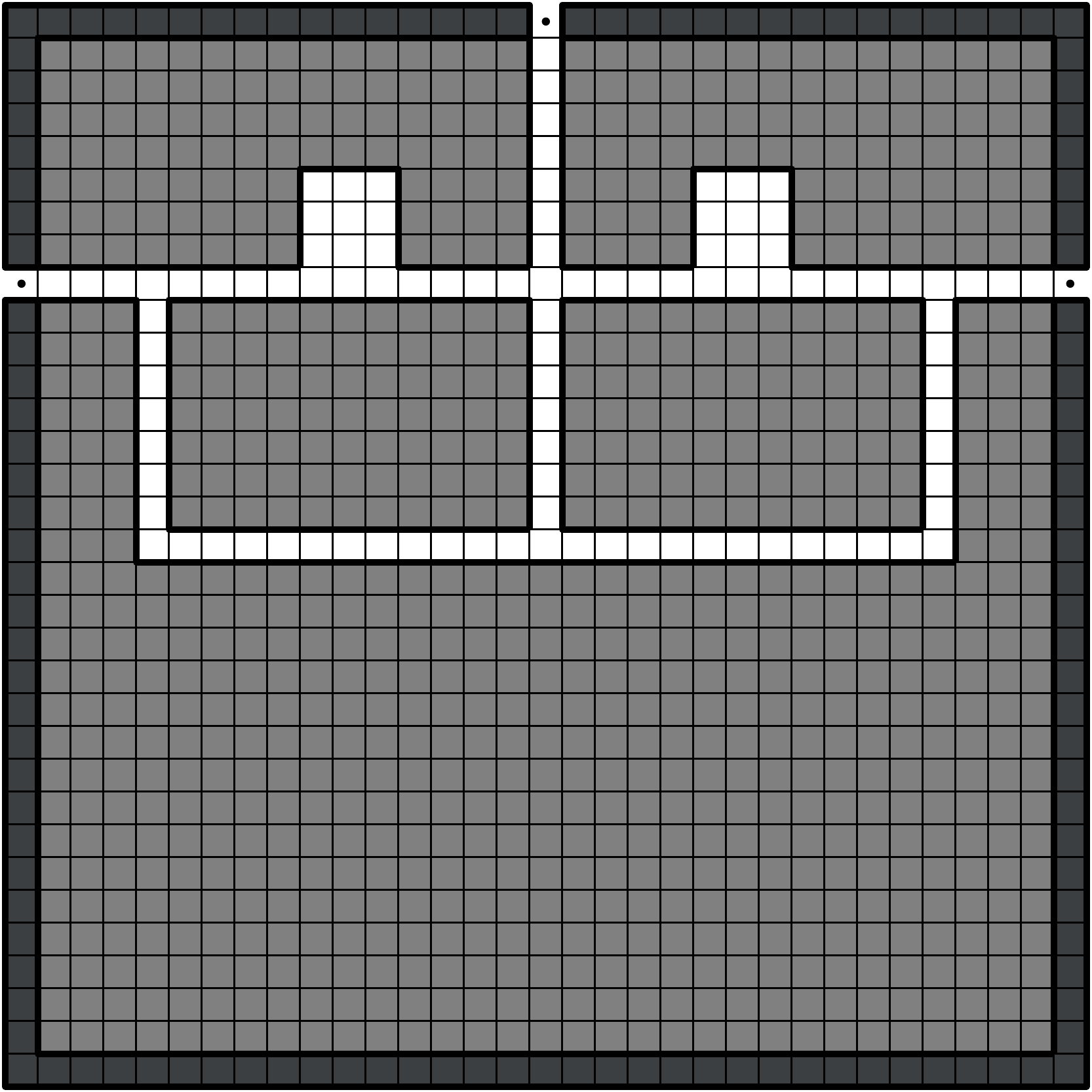}
    \caption{}
  \end{subfigure}
  \begin{subfigure}[b]{0.325\textwidth}
    \centering
    \includegraphics[width=140pt]{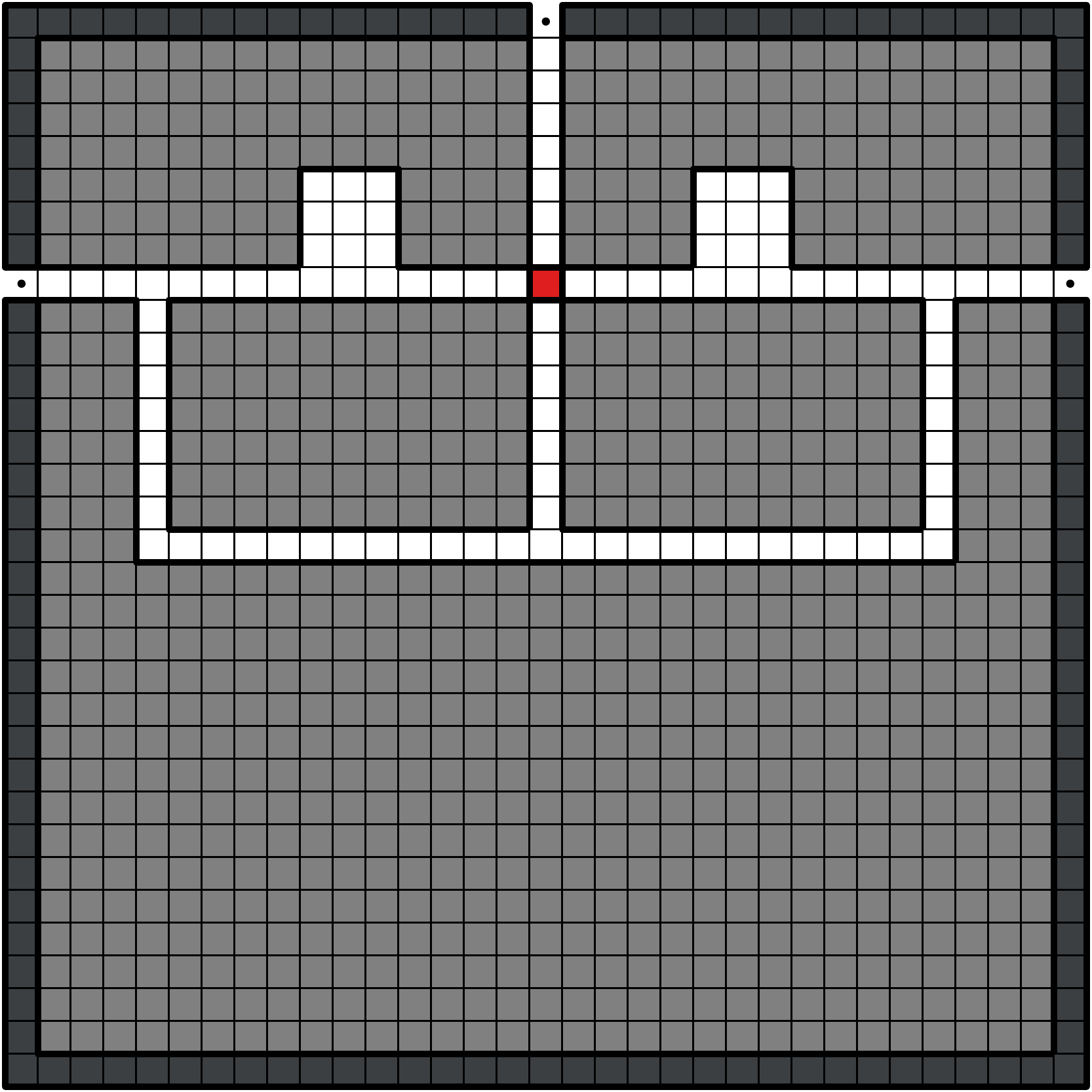}
    \caption{Center square of gadget}
  \end{subfigure}
  \begin{subfigure}[b]{0.325\textwidth}
    \centering
    \includegraphics[width=140pt]{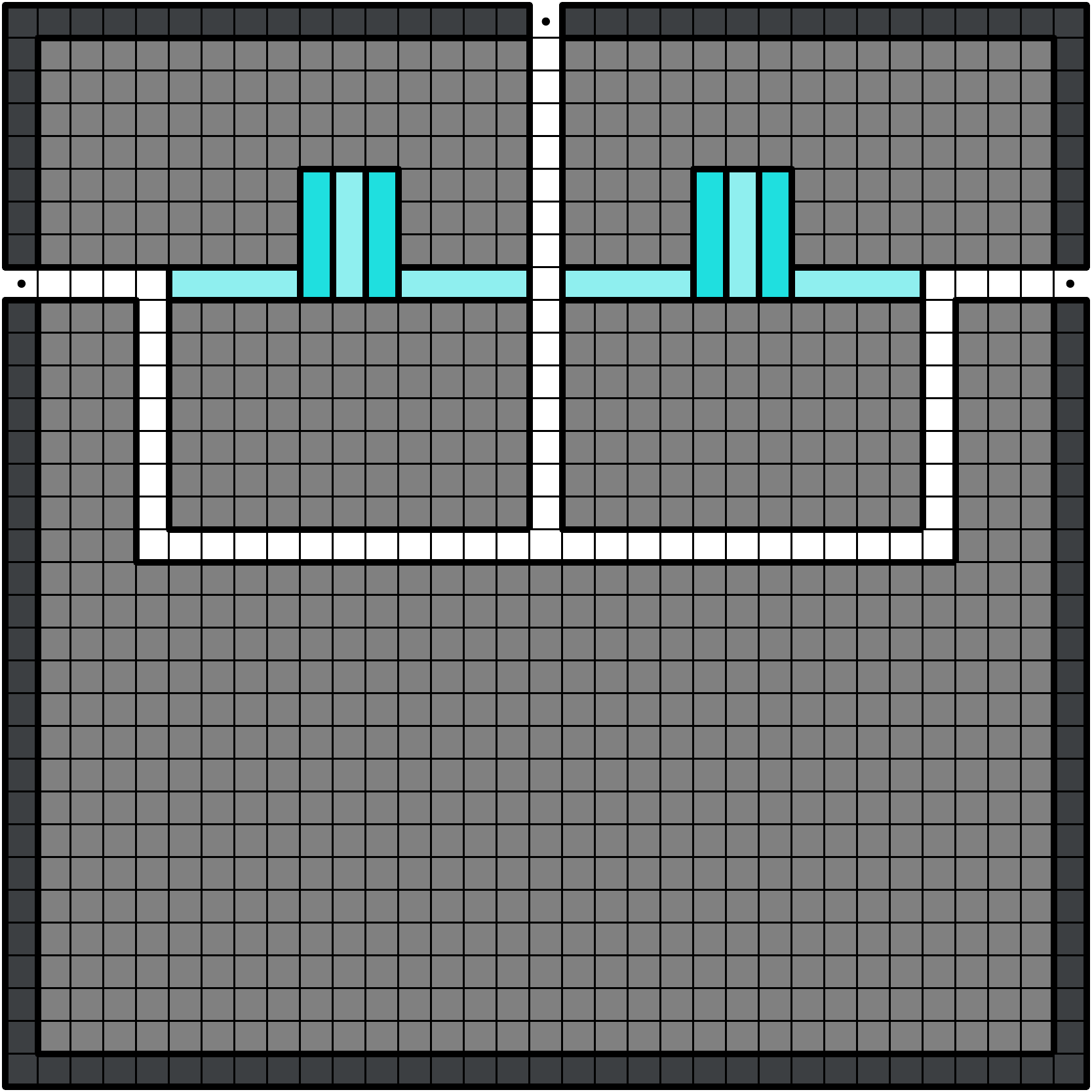}
    \caption{Forced placements}
    \label{fig:clause_forced}
  \end{subfigure}
  \caption{The clause gadget.}
  \label{fig:clause}
\end{figure}


The negated-clause gadget has an additional alteration in the middle of the gadget to allow for $\II$ pieces to rotate to their desired locations when the gadget is being filled. The negated-clause gadget is shown in Figure~\ref{fig:negclause}. One can verify that the ``bridge'' between the left port and right port areas, along with the additional structure in the middle of the gadget, is forced to be tiled a specific way as shown in Figure~\ref{fig:negclause_forced}.

As $\II$ pieces can turn corners, the traversals from the up port to either of the other ports are possible, as are the traversals between the left and right ports. The only possible tilings, along with orders in which the $\II$ pieces can be placed to ensure that the $\II$ pieces tile the gadget correctly, are shown in Figure~\ref{fig:i_negclause_tilings}. Of note is piece number 15 in the first and second tilings – the maneuver to get the $\II$ piece to that location requires careful attention to the orientation of the piece as it traverses through the middle structure, but it is possible and detailed in Appendix~\ref{subsec:i_maneuver3}.

\begin{figure}[!ht]
  \centering
  \begin{subfigure}[b]{0.325\textwidth}
    \centering
    \includegraphics[width=140pt]{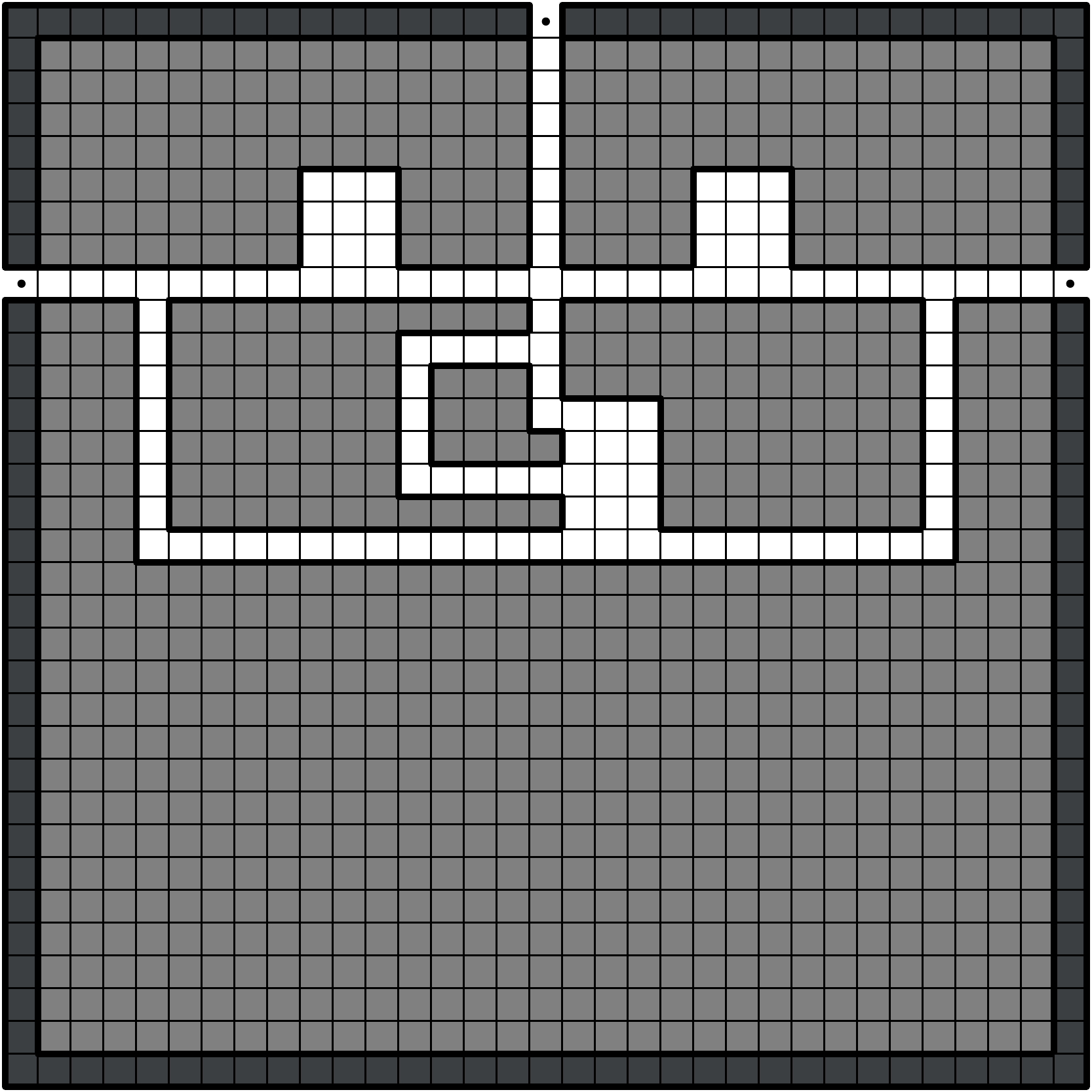}
    \caption{}
  \end{subfigure}
  \begin{subfigure}[b]{0.325\textwidth}
    \centering
    \includegraphics[width=140pt]{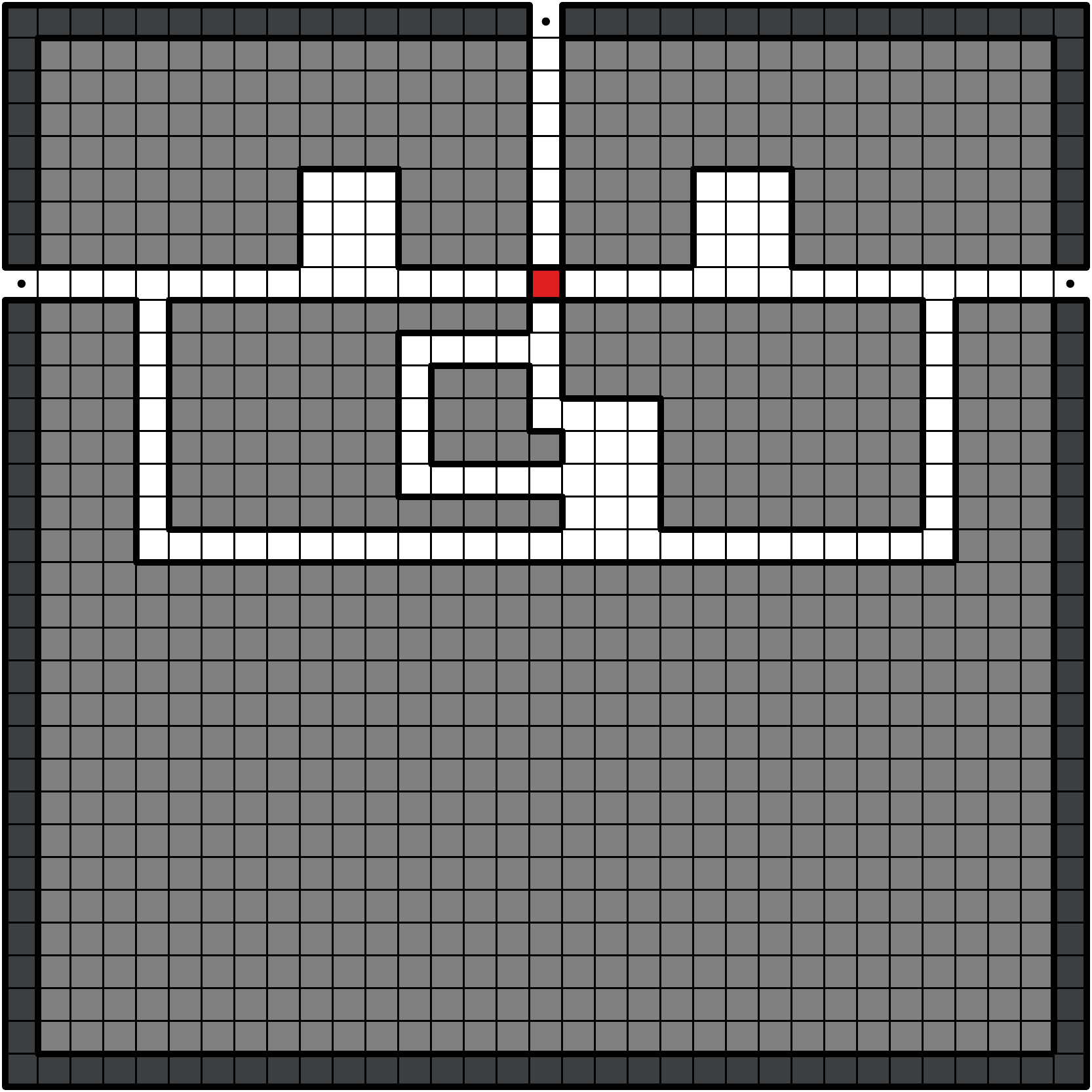}
    \caption{Center square of gadget}
  \end{subfigure}
  \begin{subfigure}[b]{0.325\textwidth}
    \centering
    \includegraphics[width=140pt]{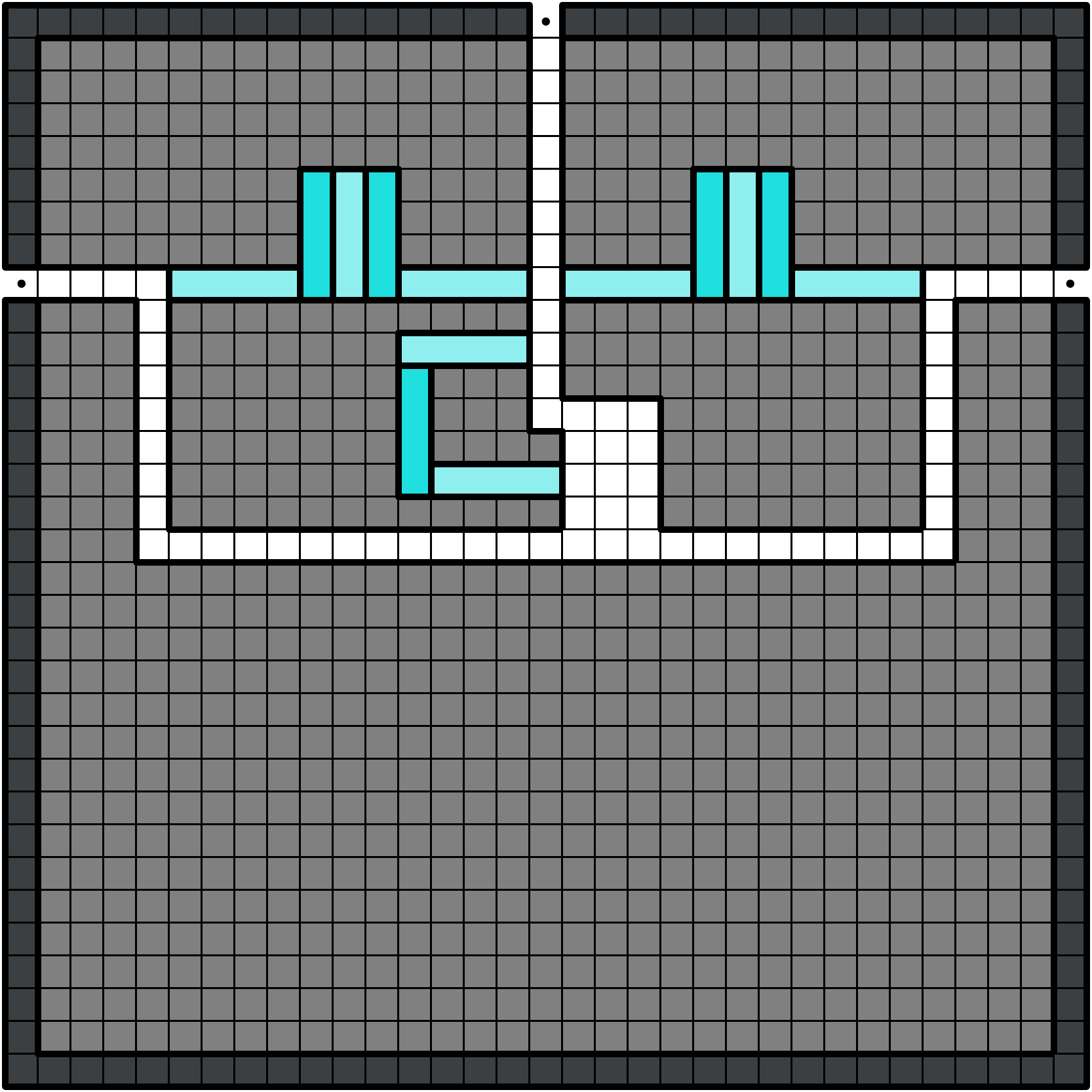}
    \caption{Forced placements}
    \label{fig:negclause_forced}
  \end{subfigure}
  \caption{The negated-clause gadget.}
  \label{fig:negclause}
\end{figure}


\subsubsection{Crossover gadget}

The crossover gadget, shown in Figure~\ref{fig:cross} is the most complex gadget and requires the most detail. Compared to the crossover gadget in \cite{horiyama2017complexity}, there are some slight alterations near the left and right ports and additional squares to allow for an easier traversal from the left port to the right port and vice versa. We will use the word \emph{chambers} to denote the three middle $9\times 5$ rectangles with only a $1\times 7$ rectangle filled in. One can verify that the ``bridge'' between the left port and right port areas is forced to be tiled a specific way as shown in Figure~\ref{fig:cross_forced}.

In terms of traversals, the traversals from the up port to the left port, from the up port to the right port, and between the left and right ports are easily seen to be possible. However, the crossover gadget is special in that we want to place some initial $\II$ pieces either on the left or on the right, depending on the desired tiling and as shown in Figure~\ref{fig:cross_init}. This is to help facilitate any traversal to the down port by ensuring that all wall kick tests before the desired test fail when we perform the necessary rotations. With the piece placements in Figure~\ref{fig:cross_init}, all traversals from any non-down port to the down port are possible (see Appendix~\ref{subsec:i_maneuver4} for an example maneuver to get $\II$ pieces through the chambers in the middle of the gadget); furthermore, the piece placements and ordering in Figure~\ref{fig:cross_init} are possible in the original state of the gadget.

The only possible tilings, along with orders in which the $\II$ pieces can be placed to ensure that the $\II$ pieces tile the gadget correctly, are shown in Figure~\ref{fig:i_cross_tilings}.

\begin{figure}[!ht]
  \centering
  \begin{subfigure}[b]{0.325\textwidth}
    \centering
    \includegraphics[width=140pt]{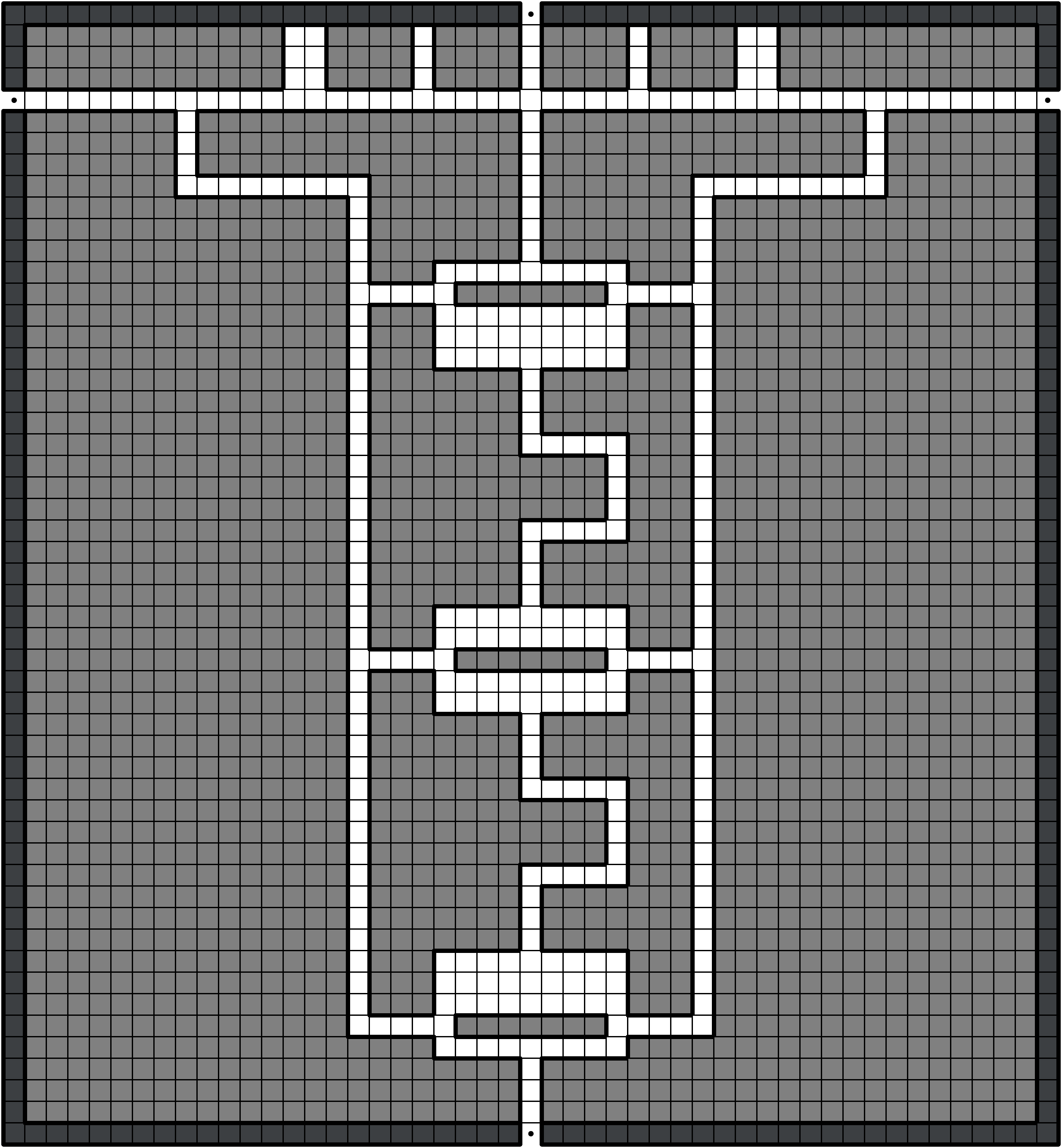}
    \caption{}
  \end{subfigure}
  \begin{subfigure}[b]{0.325\textwidth}
    \centering
    \includegraphics[width=140pt]{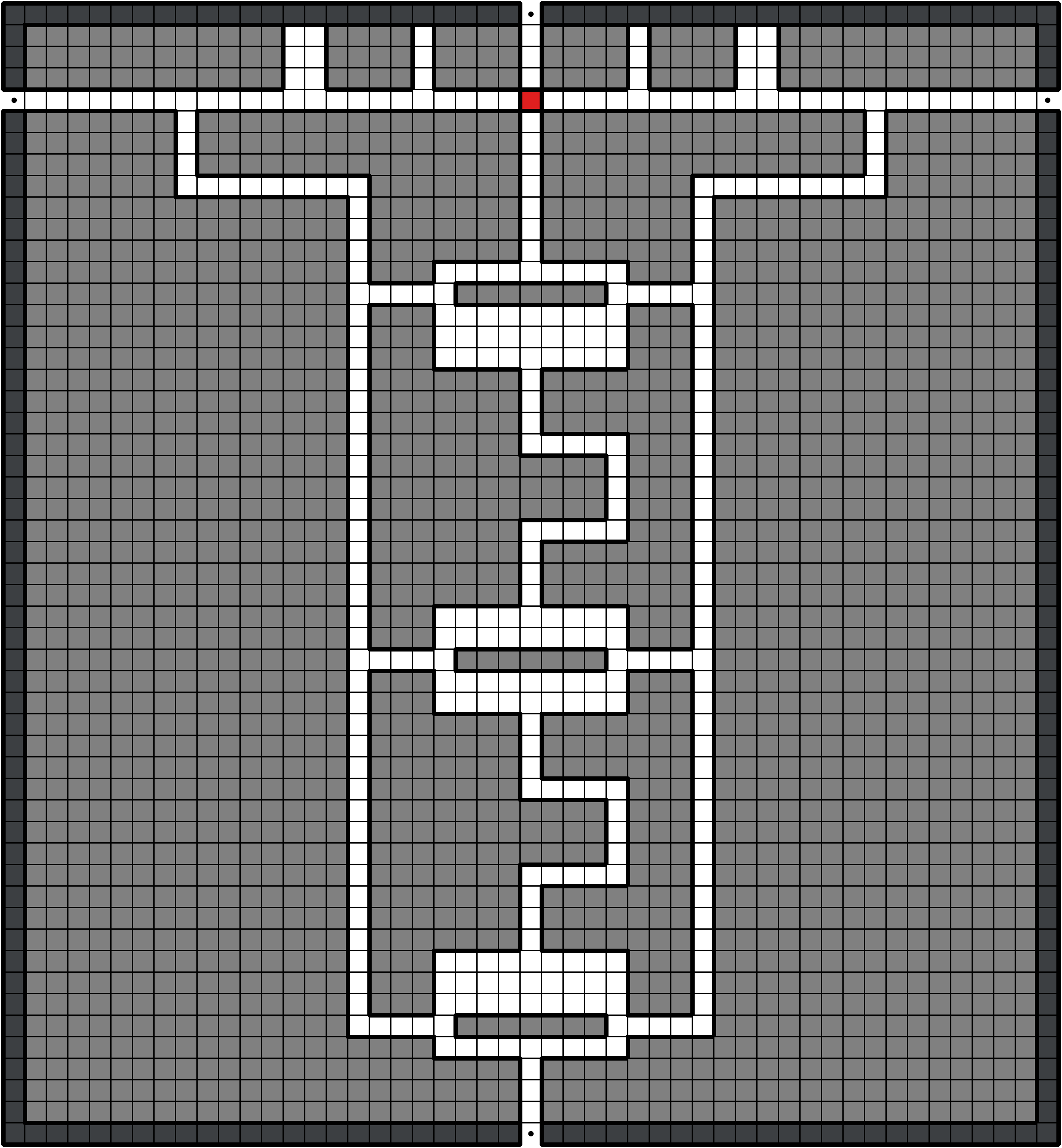}
    \caption{Center square of gadget}
  \end{subfigure}
  \begin{subfigure}[b]{0.325\textwidth}
    \centering
    \includegraphics[width=140pt]{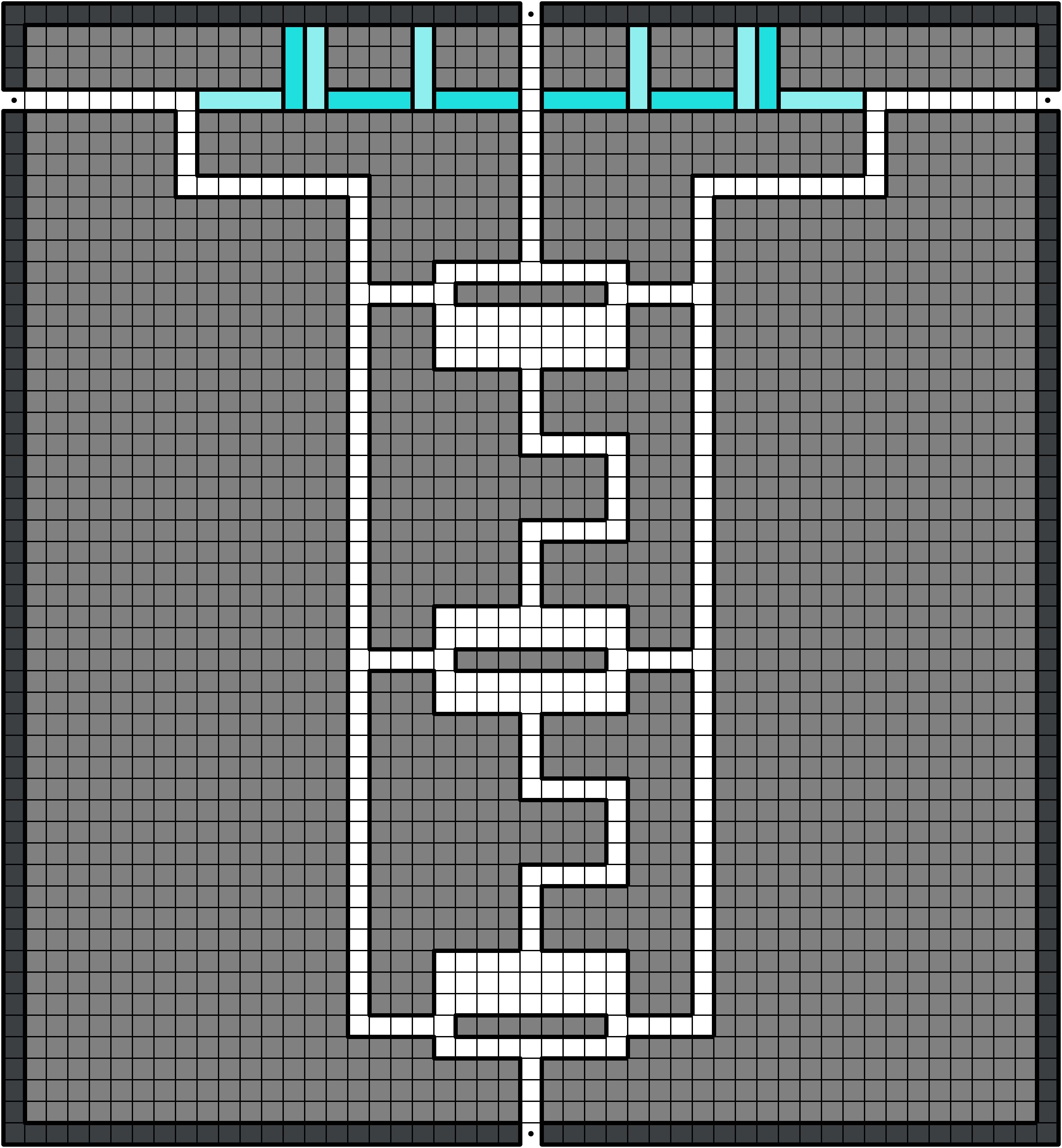}
    \caption{Forced placements}
    \label{fig:cross_forced}
  \end{subfigure}
  \caption{The crossover gadget.}
  \label{fig:cross}
\end{figure}

\begin{figure}[!ht]
  \centering
  \begin{subfigure}[b]{0.49\textwidth}
    \centering
    \includegraphics[width=220pt]{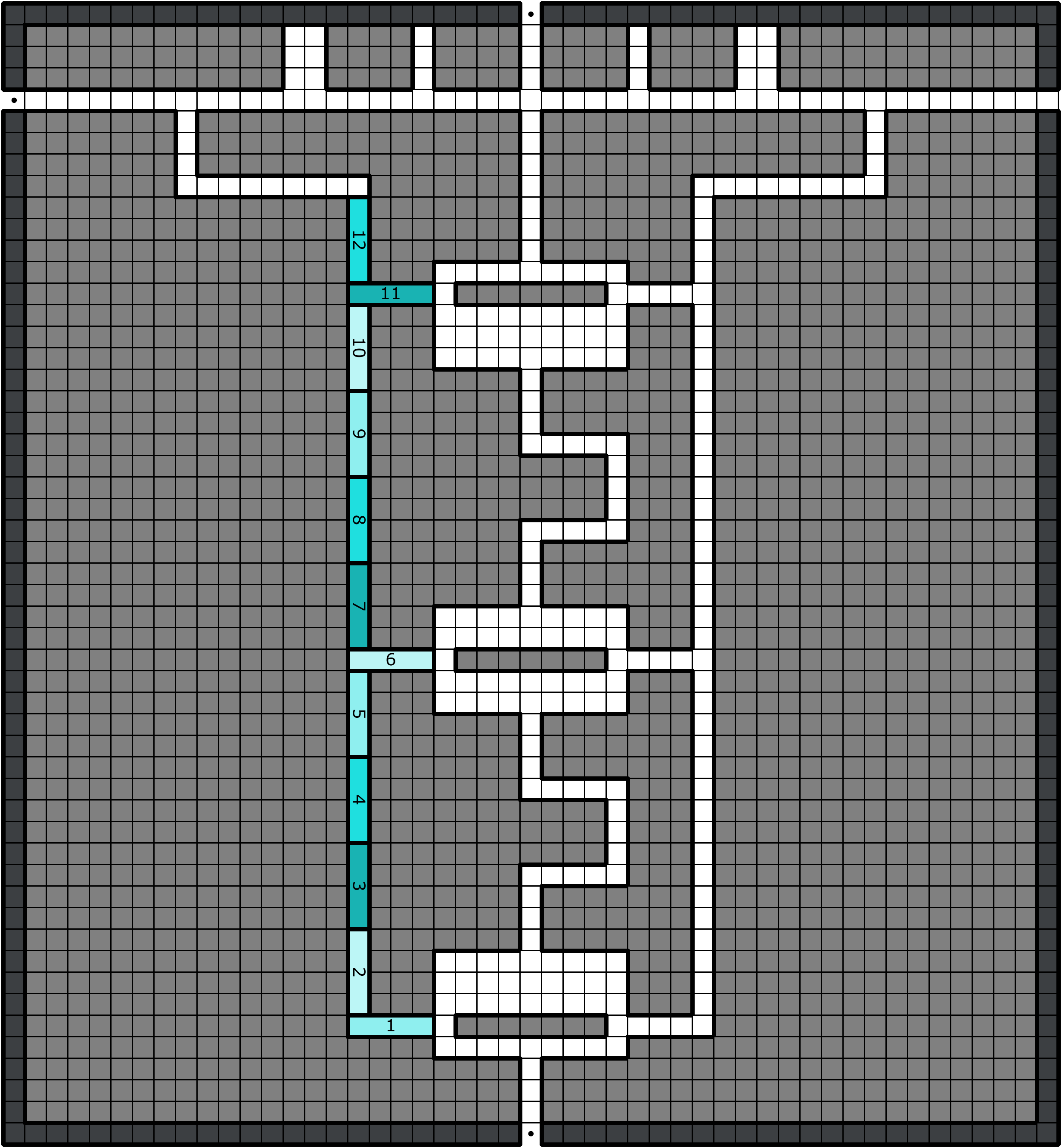}
    \caption{Left version}
  \end{subfigure}
  \begin{subfigure}[b]{0.49\textwidth}
    \centering
    \includegraphics[width=220pt]{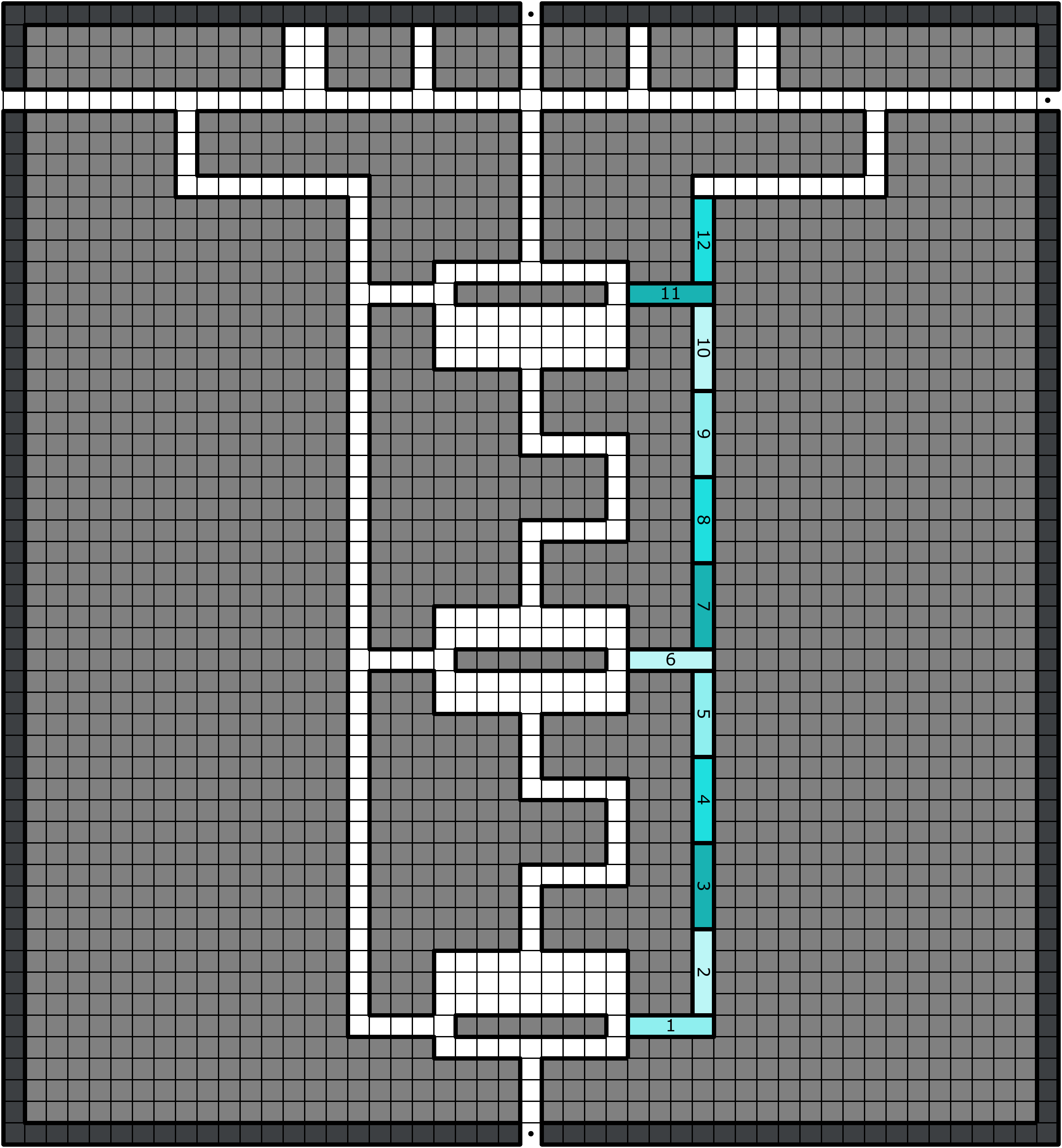}
    \caption{Right version}
  \end{subfigure}
  \caption{Initial piece placements for the crossover gadget to make the down port accessible.}
  \label{fig:cross_init}
\end{figure}


\subsection{Putting it all together}

Now we describe the main part of the construction. Take the Graph Orientation instance and embed it on the Tetris board such that horizontally adjacent red dots are on the same row, vertically adjacent red dots are on the same column, and each red dot has coordinates of the form $(60i+38, 48j+37)$ for integers $i$ and $j$ (here, $(1, 1)$ denotes the top-left square of the board). At each pair of horizontally adjacent red dots corresponding to a pair of literals, place a duplicator gadget such that the center squares of the duplicator gadget lie exactly on the red dots. At the topmost leftmost red dot, place an entry corner gadget such that the center square lies exactly on the red dot. Otherwise, at each red dot, place the corresponding gadget (corner, vertical line, horizontal line, clause, negated-clause, crossover) such that the center square lies exactly on the red dot.

In addition, use horizontal wires (i.e., a $1\times n$ rectangle) or vertical wires (i.e., an $n\times 1$ rectangle) to hook up adjacent gadgets. Due to mod $4$ constraints, all of the wires will have length equal to one less than a multiple of $4$, so that for each wire, exactly one of the two ports at the ends of the wire must be covered by an $\II$ piece used in the tiling of the wire, which helps preserve/transmit information about the orientations of the gadgets.

\subsection{Correctness}

As we have now described all of the parts of the construction, we now prove correctness of the construction; in other words, the whole construction can be cleared with $\II$ pieces under SRS if and only if the 1-in-3SAT instance has a solution. From the discussion in Subsection~\ref{subsec:i_genstructure}, the whole construction can be cleared with $\II$ pieces under SRS if and only if the main part of the construction can be filled with $\II$ pieces under SRS.

First, if the main part of the construction can be filled with $\II$ pieces under SRS, then we necessarily get a tiling of the main part of the construction with $\II$ pieces. From there, we can read off the orientations of the gadgets to get the orientation of the edges in the Graph Orientation instance, which gives us a solution to the 1-in-3SAT instance.

Now, suppose we have a solution to the 1-in-3SAT instance. Using the solution, we can quickly (i.e., in polynomial time) get an orientation of the edges in the Graph Orientation instance, and thus the orientations of the gadgets and wires connecting the gadgets. From here, we want to show that we can maneuver and place the $\II$ pieces into their desired locations in the main part of the construction.

The key idea is to use the structure of the Graph Orientation instance in \cite[Theorem 2.3]{horiyama2017complexity}. We start with filling the gadgets corresponding to the \emph{lower part} of the construction, as defined at the beginning of Section~\ref{sec:itrisclearing}. Because the lower part consists of $m$ copies of the clause/negated-clause construction as indicated in Figure~\ref{fig:cnc}, we can fill the gadgets here in the same manner, as shown in Figure~\ref{fig:cncplusorder}:

\begin{enumerate}
    \item Fill in the gadgets along the path from the leftmost literal $\ell_1$ to the clause gadget $C_j$, along with the corner gadget to the right of the clause gadget $C_j$, from the bottom-most rightmost gadget to the topmost leftmost gadget on the path.
    \item Fill in the gadgets along the path from the rightmost literal $\overline{\ell_3}$ to the negated-clause gadget $\overline{C_j}$, along with the corner gadget to the left of the negated-clause gadget $\overline{C_j}$, from the bottom-most leftmost gadget to the topmost rightmost gadget on the path.
    \item Fill in the remaining gadgets along the second-to-bottom row, starting from the leftmost and rightmost remaining gadgets and moving towards the middle columns.
    \item Fill in the remaining gadgets along the third-to-bottom row from right to left.
    \item Fill in the remaining gadgets from left to right.
\end{enumerate}

\begin{figure}[!ht]
    \centering
    \includegraphics[width=280pt]{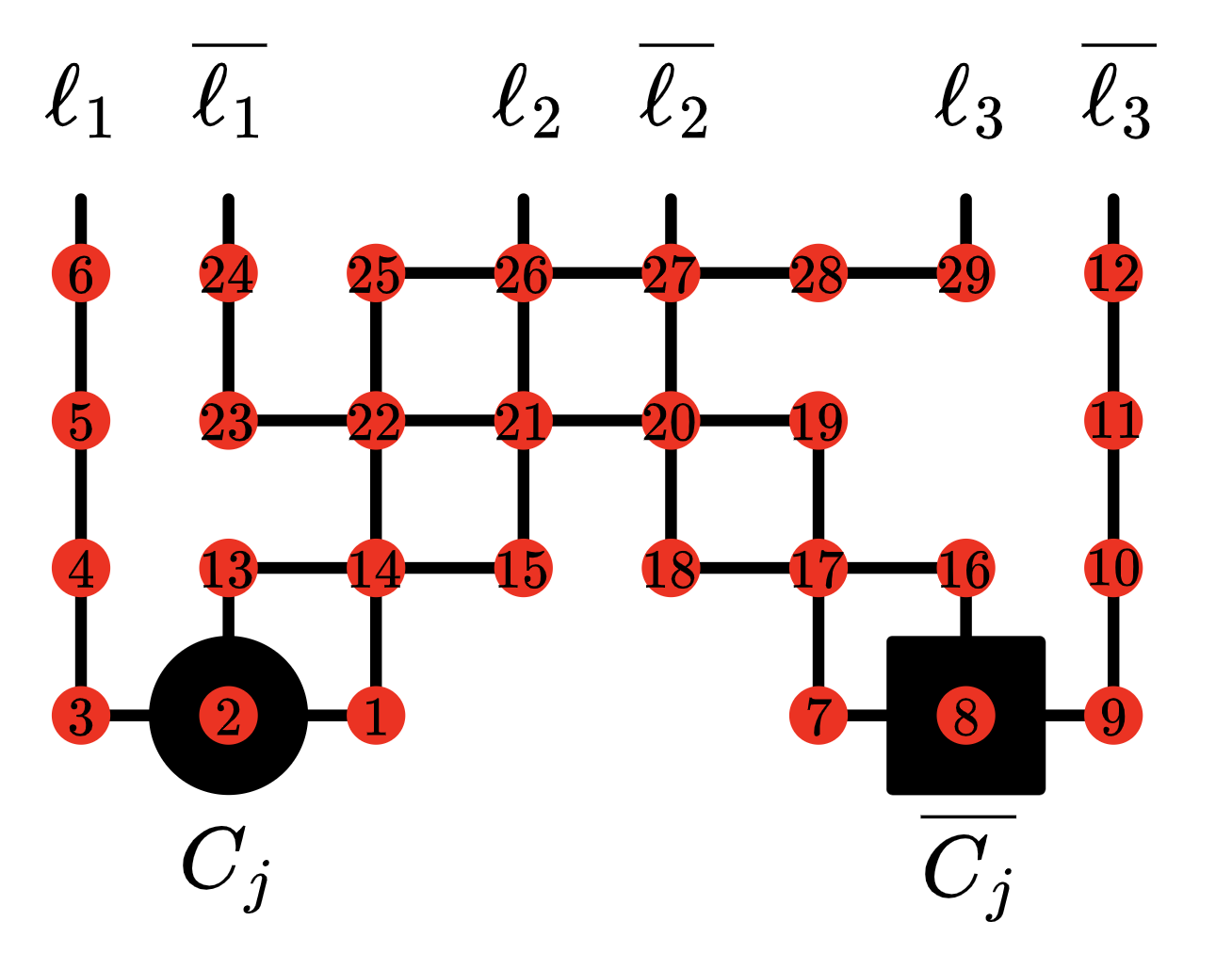}
    \caption{The order in which to fill the gadgets in each clause/negated-clause construction. Smaller numbers are filled first (so $1$ is filled first and $29$ is filled last).}
    \label{fig:cncplusorder}
\end{figure}

All gadgets with an up port will be filled with pieces moving through the up port, and all other gadgets (without an up port) will be filled with pieces moving through the left or right port that remains unblocked (i.e., still has an unblocked path from a literal to that gadget).

Once the gadgets corresponding to the lower part of the construction have been filled, we can then fill the gadgets corresponding to the upper part of the construction. To do so, we first fill all the gadgets to the left of the leftmost vertical wire containing a literal \emph{except for} the EC gadget, and then fill the rest of the gadgets, starting from the bottom-most row, and going from right to left within each row of gadgets. Note that the only gadgets with up ports in the upper part of the construction are crossovers and UC gadgets, and it is easy to maneuver the relevant $\II$ pieces to the up ports of these gadgets in the overall construction; otherwise, no other gadget in the upper part of the construction has any restrictions on where the pieces need to come from while the gadget is being filled.

While we are filling the gadgets, we also fill the wires between gadgets – whenever we are about to fill a gadget, we check all adjacent gadgets, and if any of them have been filled, we fill the wire connecting the current gadget to the already-filled gadget right before filling the current gadget.

In particular, this method of placing the $\II$ pieces in the overall construction ensures that none of the gadgets or wires become prematurely blocked or inaccessible, and that all of the gadgets and wires will be properly filled, if we are given a valid orientation of the gadgets. Therefore, given a valid solution to the 1-in-3SAT instance, we are able to determine a way to fill the main part of the construction with $\II$ pieces under SRS.

Therefore, the main part of the construction can be filled with $\II$ pieces under SRS if and only if the 1-in-3SAT instance has a solution, meaning that our construction works as intended. Thus, we have shown that Tetris clearing with SRS is NP-hard even if the player is only given $\II$ pieces, as desired.

\subsection{Aside: Bags}

Now that we have NP-hardness for Tetris clearing even if the player is only given $\II$ pieces, we can actually prove NP-hardness for Tetris clearing under the 7-bag randomizer:

\begin{theorem}
    Tetris clearing with SRS is NP-hard even if the piece sequence must be able to be generated from the 7-bag randomizer.
\end{theorem}

Here, we combine the reduction in Theorem~\ref{thm:itrisclear} with the bottle blocking mechanisms in the reductions to two piece types in \cite{mithg2024tetris}. More specifically, let $N+1$ be the number of $\II$ pieces required to clear the reduction structure in Theorem~\ref{thm:itrisclear} (including the four empty squares in the second-to-rightmost column). Then, we modify the top two rows of the structure and add an additional $2N+4$ lines on top of the structure as shown in Figure~\ref{fig:bag_addition} (the figure shows the case where $N$ is even; if $N$ is odd, then the left tunnel has slightly different offsets in the top lines). The top $2N$ lines consist of roughly the same $2$-line structure repeated $N$ times. The input sequence consists of $N+5$ bags, or a total of $7N+35$ pieces. The order of pieces within each bag does not matter here.

\begin{figure}[!ht]
    \centering
    \includegraphics[width=320pt]{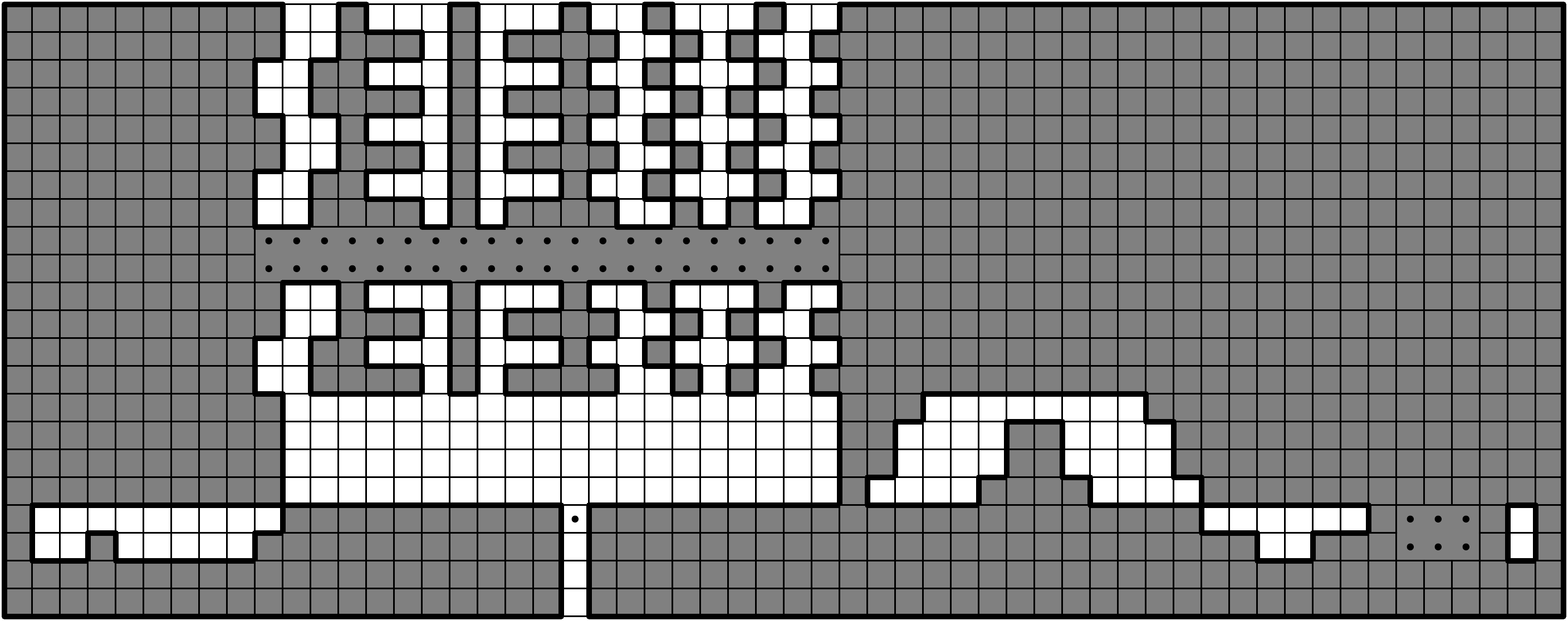}
    \caption{The additional structure for the reduction to Tetris clearing under the 7-bag randomizer.}
    \label{fig:bag_addition}
\end{figure}

For the first $N$ bags, the non-$\II$ pieces must fill in their respective tunnels in the topmost two rows. In particular, because the $\OO$ piece placement is forced (otherwise part of the $\OO$ piece will stick out over the topmost line), the other non-$\II$ pieces must be placed such that the topmost two rows are cleared out so that the next $\OO$ piece can be properly placed. This forces the placements shown in Figure~\ref{fig:bag_addition_fill1}.

\begin{figure}[!ht]
    \centering
    \includegraphics[width=320pt]{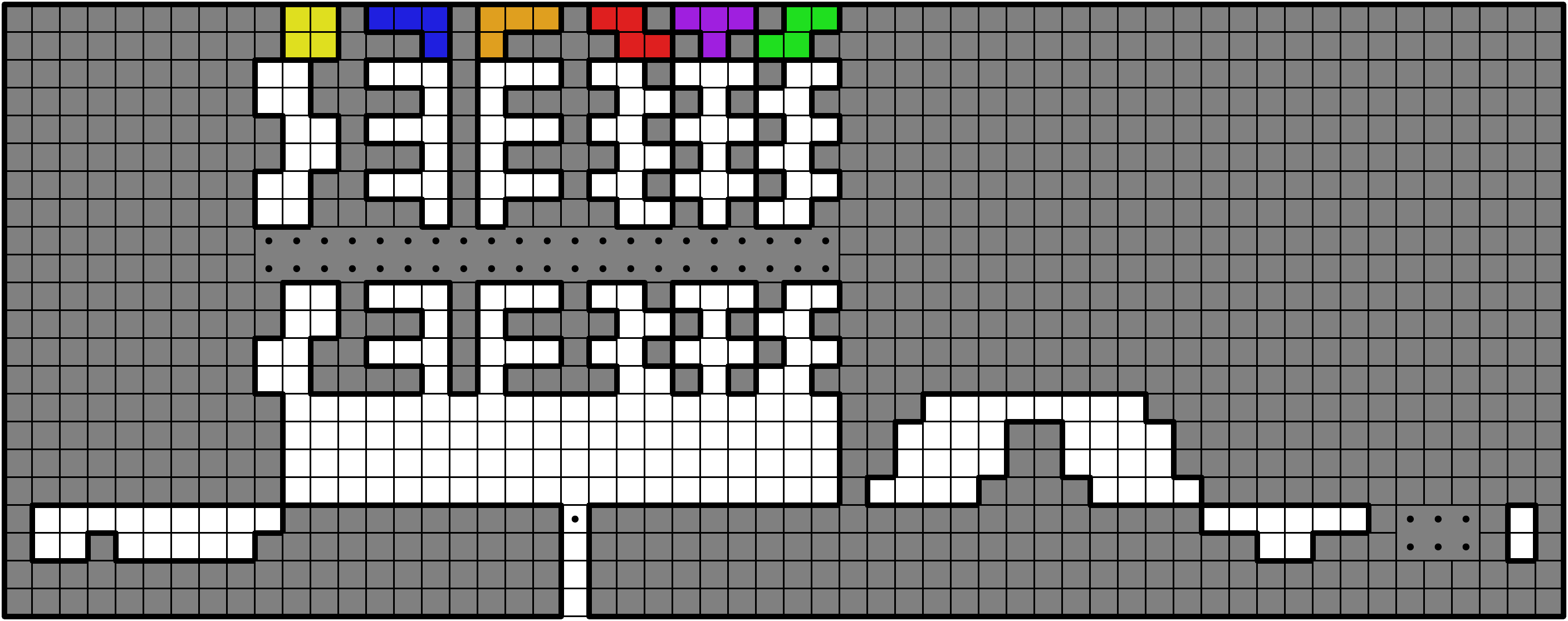}
    \caption{Placements of non-$\II$ pieces in the first $N$ bags.}
    \label{fig:bag_addition_fill1}
\end{figure}

Now, for the first $N$ bags, if the $\II$ piece is first or last in the bag, then all the tunnels will be open (as the non-$\II$ pieces will either have not been placed or have cleared out the topmost lines), so we can move the $\II$ piece through any of the tunnels and into the entrance of the main reduction structure. Otherwise, between $1$ and $5$ distinct non-$\II$ pieces will have been placed in the topmost two rows, meaning that at least one of the tunnels is still open, so we can move the $\II$ piece through that tunnel and into the entrance of the main reduction structure. In any case, after each bag and regardless of the ordering of the pieces within the bag, the topmost $2$ lines get cleared and exactly one additional $\II$ piece gets put into the main reduction structure.

The last $5$ bags are used for cleanup of the remaining lines; one possible way to clean up the remaining lines is shown in Figure~\ref{fig:bag_addition_fill2}.

\begin{figure}[!ht]
    \centering
    \begin{subfigure}[b]{0.98\textwidth}
        \centering
        \includegraphics[width=320pt]{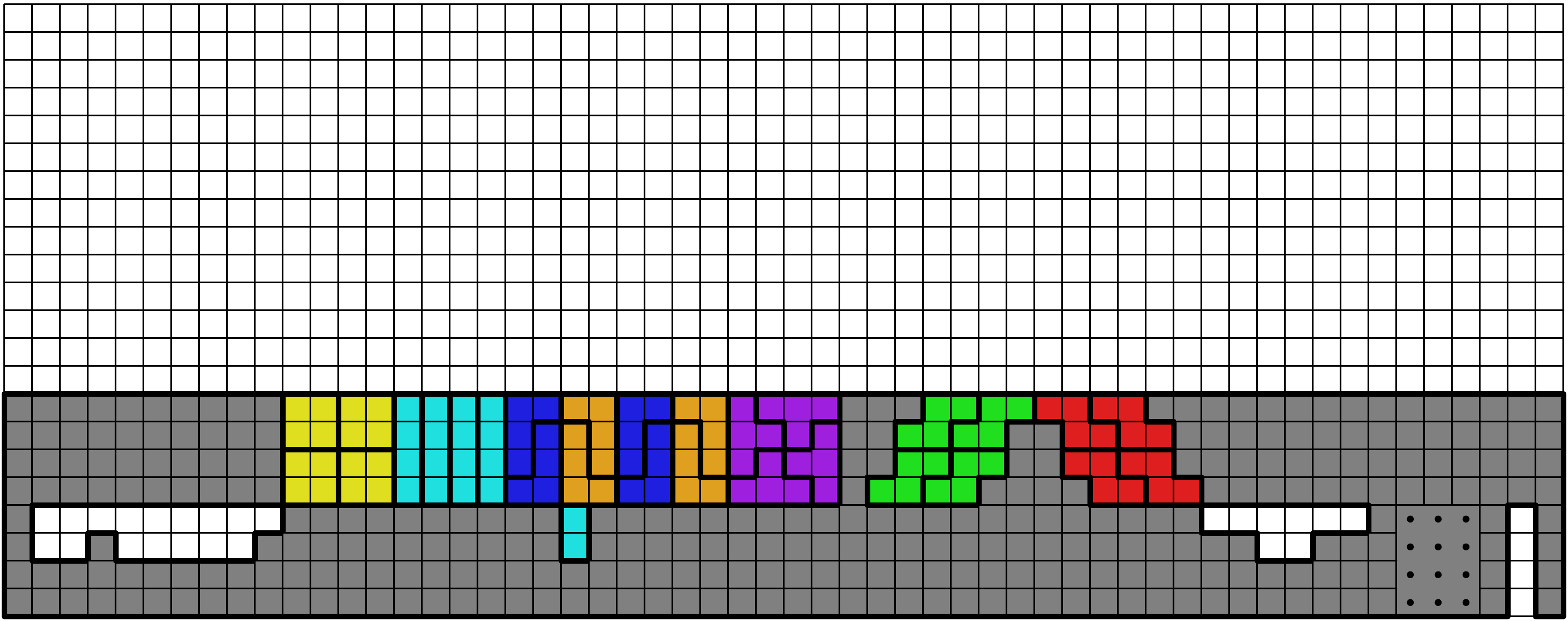}
        \caption{First 4 bags}
    \end{subfigure}
    \begin{subfigure}[b]{0.98\textwidth}
        \centering
        \includegraphics[width=320pt]{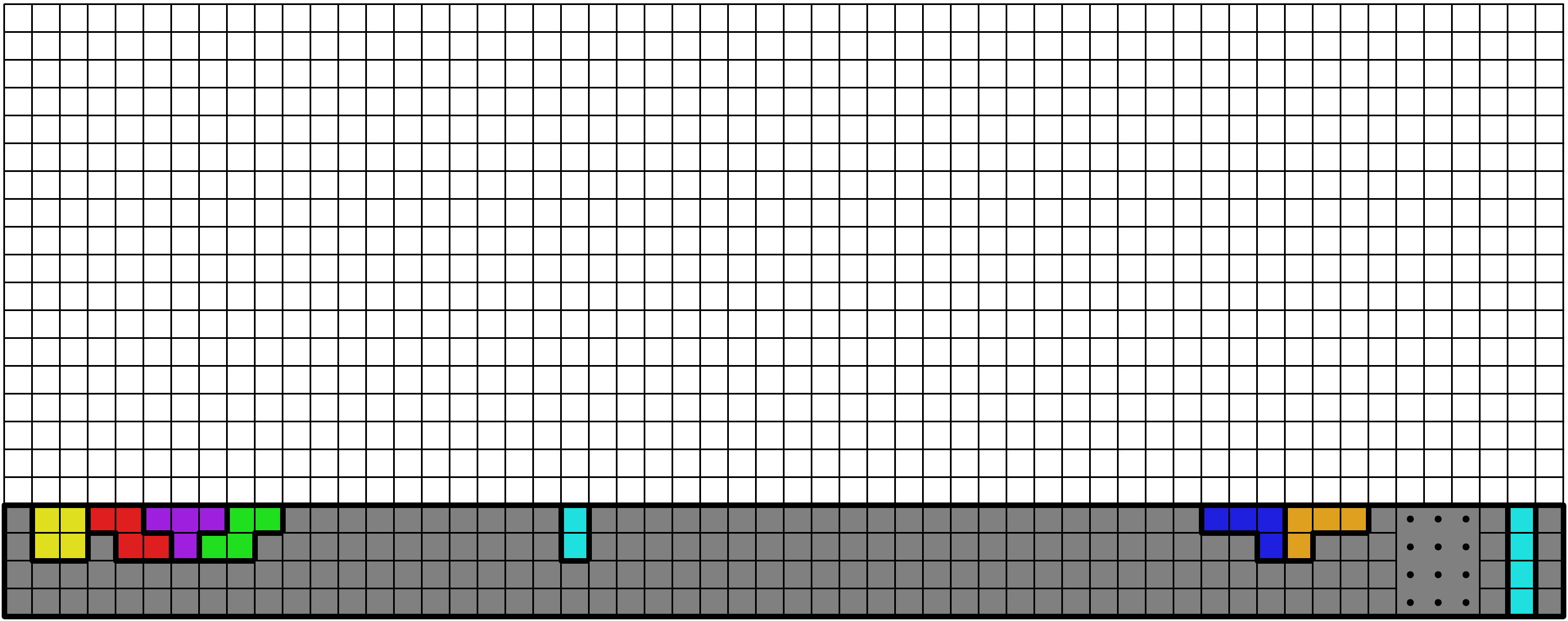}
        \caption{Last bag}
    \end{subfigure}
    \caption{Cleaning up the remaining lines.}
    \label{fig:bag_addition_fill2}
\end{figure}

The additional structure does not significantly change the argument for correctness. The clearing strategy is still required to be of the form ``Fill the embedded Graph Orientation structure, then fill in the long tunnel, then clear the top two and bottom two rows around the main reduction structure'', but the additional structure now requires regular line-clearing at the top of the overall structure. Failing to fill the embedded Graph Orientation structure still causes the last step (of clearing the top two and bottom two rows around the main reduction structure) to fail due to intermediate lines in the main reduction structure not being cleared. In addition, due to the width of the entrance into the long tunnel, only $\II$ pieces can go into the main reduction structure; all other pieces must be placed within the top $2N+4$ lines. Thus, we get NP-hardness of Tetris clearing even if the piece sequence must be able to be generated from the 7-bag randomizer.

We make some additional remarks regarding this reduction:

\begin{itemize}
    \item This works if the player is allowed to \emph{hold} pieces, or even \emph{permute} the order of the pieces in the input sequence to their liking, even if the final ordering is not one that can be generated by the 7-bag randomizer. Having the ability to delay specific pieces (or even choose the order in which pieces are placed) does not affect the clearability of the board as only $\II$ pieces can go into the main reduction structure, which is what matters for clearing the whole board.
    \item This structure can be easily modified to prove NP-hardness for \defn{$7k$-bag randomizers} for any positive integer $k\geq 1$. First, we can make small adjustments to the main reduction structure such that $N/k$ is an integer, and give the player $N/k+5$ bags. Then, we make $k$ copies of the additional structure, which now takes up $2N/k+4$ lines, and place them in parallel in the top $2N/k+4$ lines, but only have one copy leading to the main reduction structure. One can still argue that, for each of the first $N/k$ bags, exactly $k$ additional $\II$ pieces get put into the main reduction structure, though the argument is slightly trickier (in particular, the general idea is to prioritize placing non-$\II$ pieces into the copies of the additional structure not leading into the main reduction structure).
\end{itemize}

\section{$\JJ$-tris/$\LL$-tris Clearing}\label{sec:jtrisclearing}

\begin{theorem}
    Tetris clearing with SRS is NP-hard even if the type of pieces in the sequence given to the player is restricted to just $\JJ$ pieces, or the type of pieces in the sequence given to the player is restricted to just $\LL$ pieces.
\end{theorem}

We focus on $\JJ$ pieces; the $\LL$ pieces case is symmetric.

Here, we will reduce from Planar Biconnected $\{\{0, 4\}\text{-in-4},\text{3-in-4}\}$ Graph Orientation. We first describe the general structure of the construction, which is similar to the general structure for the $\II$-tris reduction and is shown in Figure~\ref{fig:j_structure}. The main part of the construction, corresponding to the Planar Biconnected $\{\{0, 4\}\text{-in-4},\text{3-in-4}\}$ Graph Orientation instance, is indicated in green. There is a path of empty squares to the main part of the construction, which is located mostly within the leftmost 12 columns or the topmost 8 rows of the board, and which we will refer to as the ``long tunnel''; we will also refer to the vertical structure between columns $8$ and $11$ (from the left) as the ``climbable vertical structure''. There are 4 additional empty squares in the rightmost three columns of the board. We give the player just enough pieces so that clearing the board requires use of all pieces; in other words, if there are $M$ empty squares in the construction, we give the player $M/4$ $\JJ$ pieces.

\begin{figure}[!ht]
    \centering
    \includegraphics[width=320pt]{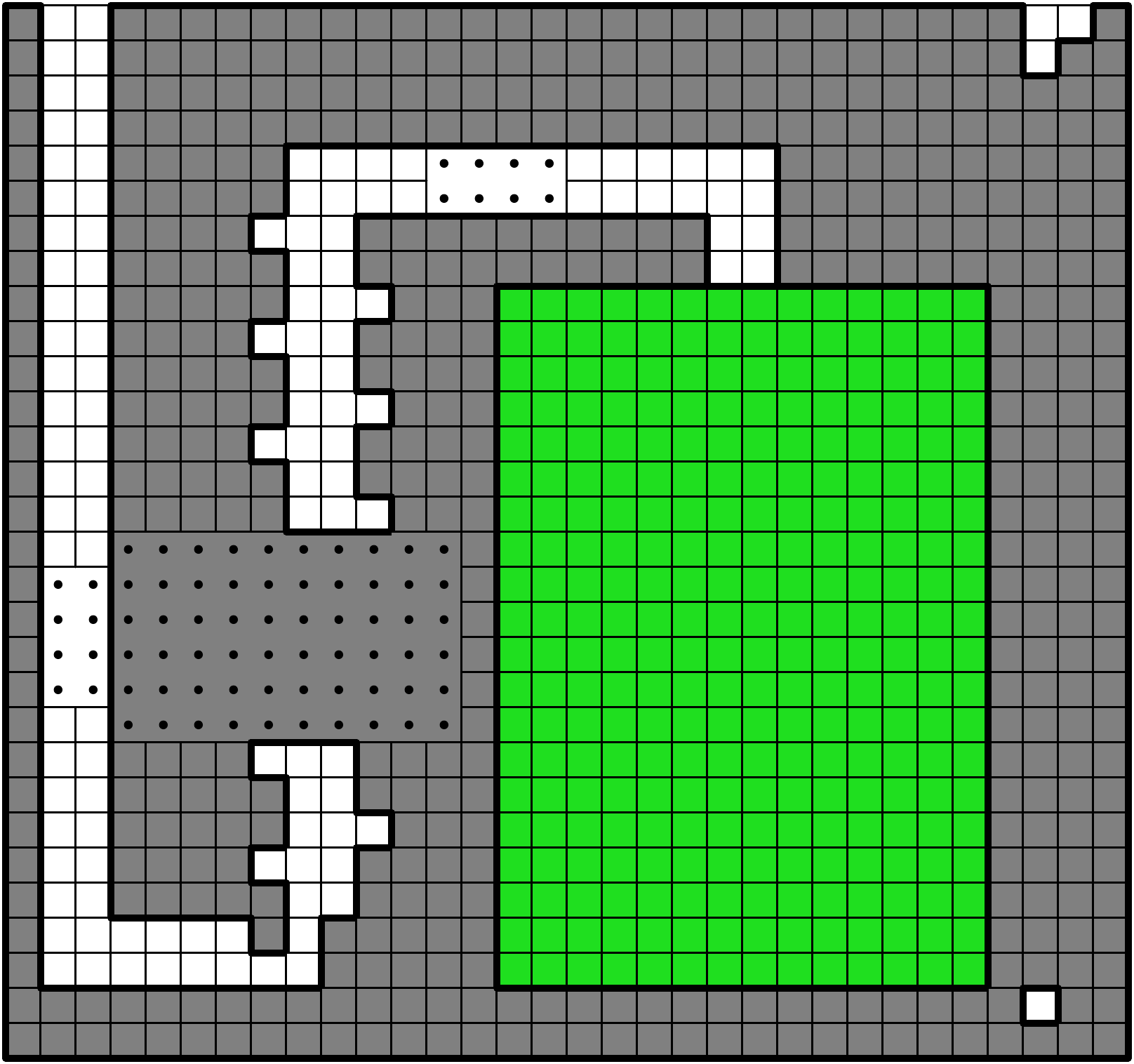}
    \caption{The general structure of the construction for Tetris clearing using only $\JJ$ pieces. The main part of the construction is indicated in green.}
    \label{fig:j_structure}
\end{figure}

Similar to the $\II$-tris reduction, we note that the 4 empty squares in the three rightmost columns must be the last squares to be filled in, after the middle rows are cleared. In addition, there is exactly one way to fill in/tile the long tunnel with $\JJ$ pieces, as shown in Figures~\ref{fig:j_structure_fill} and \ref{fig:j_cvs_fill}, and it is possible (under SRS) to get $\JJ$ pieces through the long tunnel (even the climbable vertical structure) into the main part of the construction (see the maneuver in Appendix~\ref{subsec:j_maneuver1}). Locking any $\JJ$ piece in place in the long tunnel before the main part of the construction is completely filled blocks off the main part of the construction, causing the board to become unclearable. Thus, no $\JJ$ piece can be placed in the long tunnel before we are done filling in the main part of the construction. Lastly, before we place any $\JJ$ pieces in the long tunnel, the long tunnel also prevents line clears in the rows corresponding to the main part of the construction.

\begin{figure}[!ht]
    \centering
    \includegraphics[width=320pt]{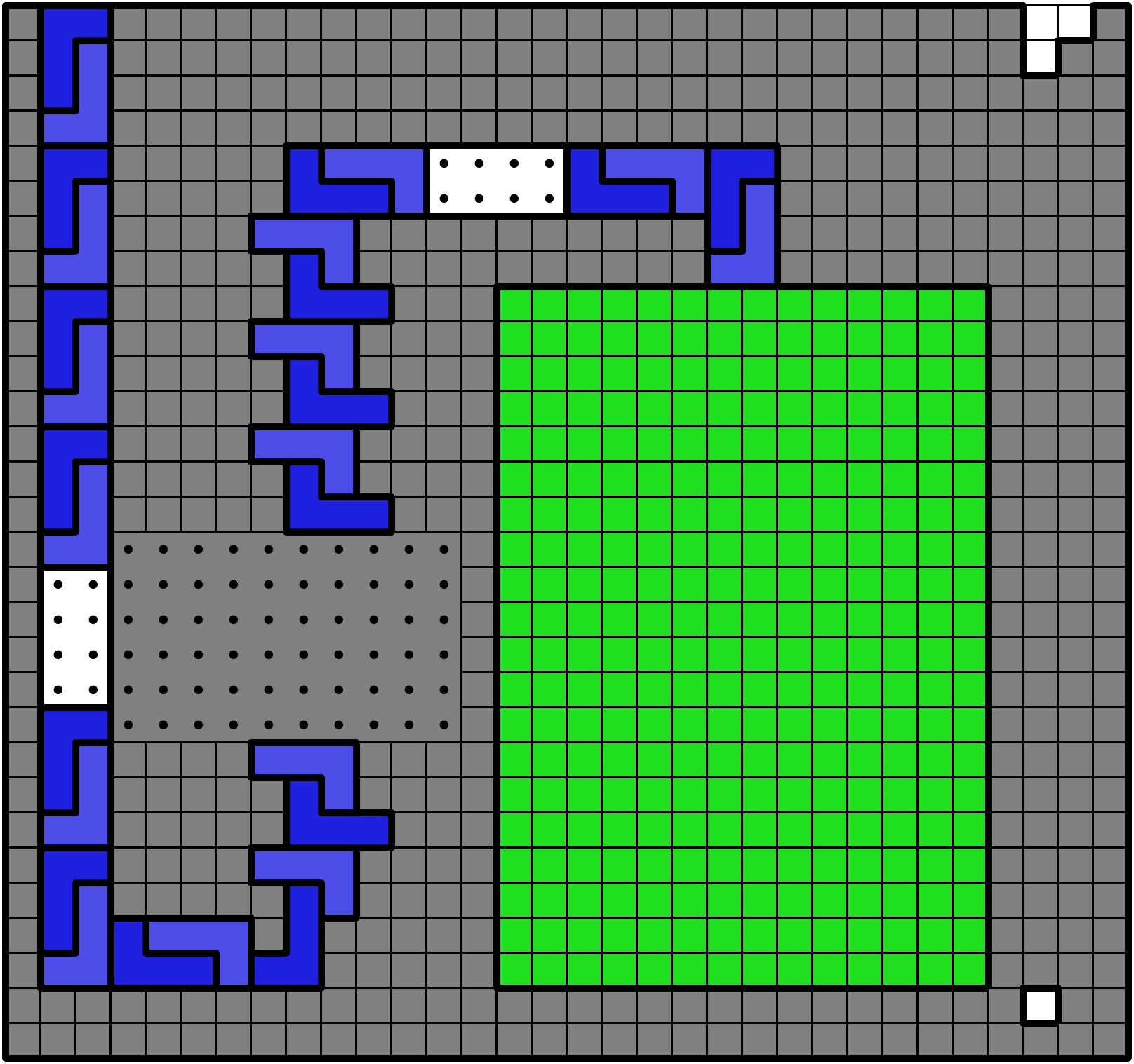}
    \caption{The only possible tiling of the climbable vertical structure with $\JJ$ pieces.}
    \label{fig:j_structure_fill}
\end{figure}

\begin{figure}[!ht]
    \centering
    \begin{subfigure}[b]{0.45\textwidth}
        \centering
        \includegraphics[width=150pt]{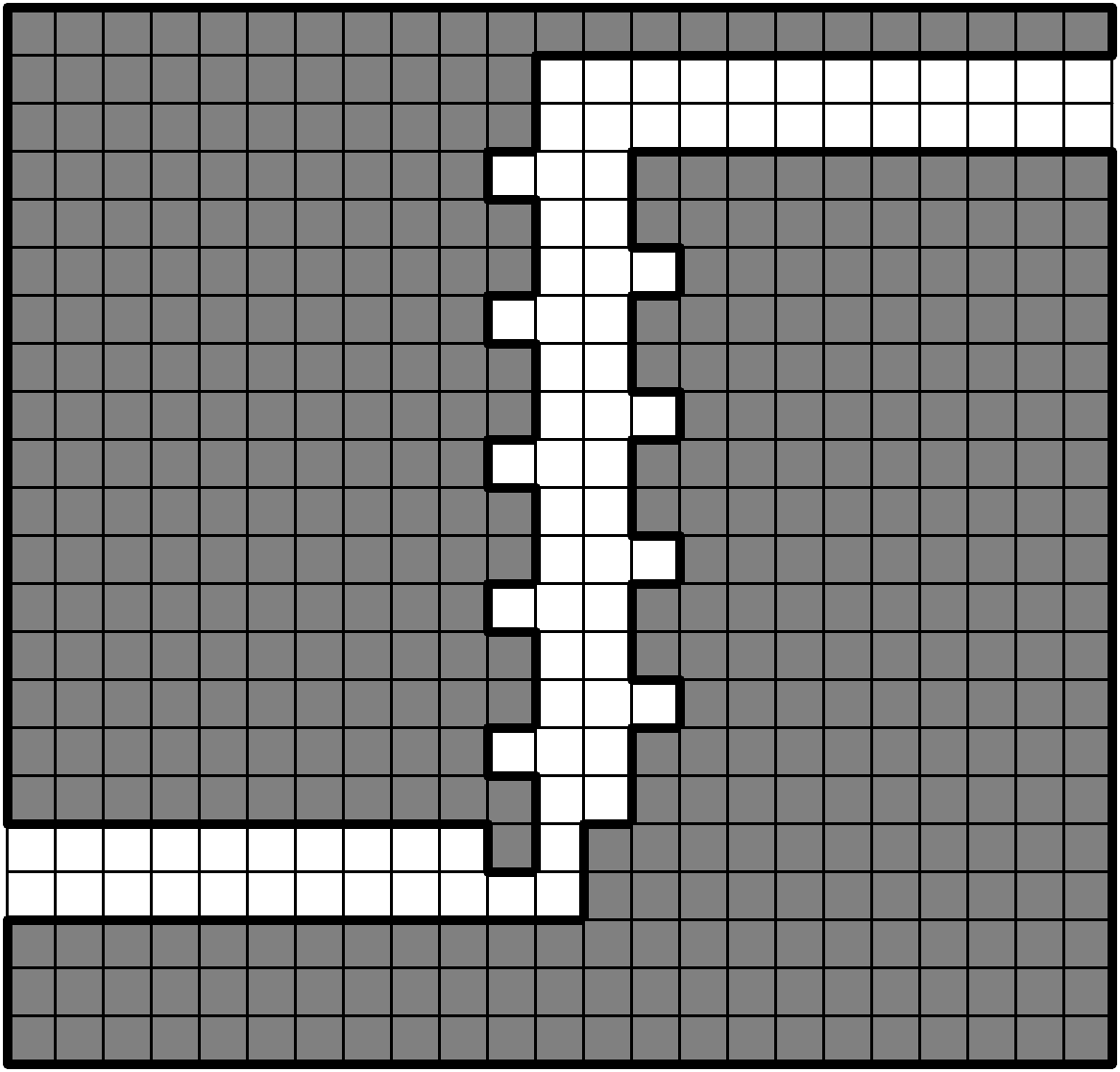}
        \caption{Unfilled}
    \end{subfigure}
    \begin{subfigure}[b]{0.45\textwidth}
        \centering
        \includegraphics[width=150pt]{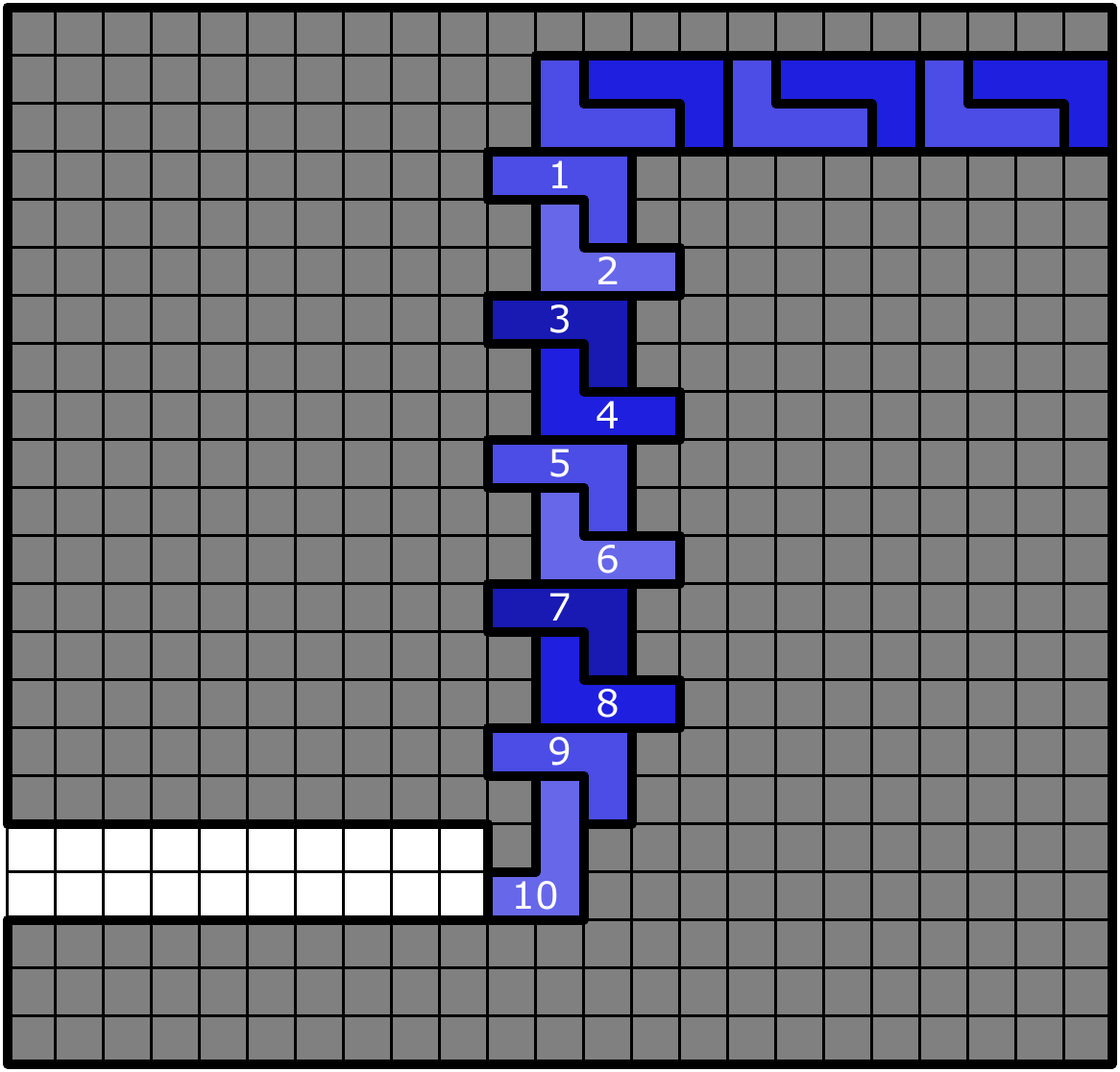}
        \caption{Filled}
    \end{subfigure}
    \caption{A closer look at a part of the climbable vertical structure, including the piece ordering of one of the segments of the climbable vertical structure.}
    \label{fig:j_cvs_fill}
\end{figure}

Thus, we need to be able to tile the main part of the construction with $\JJ$ pieces exactly, with the additional restriction that we need to be able to maneuver all the pieces into place through Tetris and SRS rules.

\subsection{Gadgets}

Now, we introduce the gadgets used in the main part of the construction. Like with the $\II$-tris gadgets, we will call the entrances \textbf{ports}, which will be indicated as empty squares with dots in their centers. We will need $0$-or-$4$ (corresponding to vertices that must have indegree either $0$ or $4$), $3$-in-$4$ (corresponding to vertices that must have indegree $3$), horizontal line and U-turn gadgets (corresponding to edges between vertices), along with a special gadget called a parity fixer that helps ``align'' the coordinates of the ports between gadgets modulo $4$ in case the coordinates of the ports modulo $4$ are not the same.

Like with the $\II$-tris gadgets, we want the gadgets to have the following properties:

\begin{itemize}
    \item Any traversal from a port with a higher $y$-coordinate to a lower $y$-coordinate, or between ports with the same $y$-coordinate, is possible.
    \item The possible ways to tile the gadgets either correspond to a setting of the edges around a vertex (where an edge is pointed into a vertex if the port is filled by a piece in the tiling and out of a vertex if the port is not filled by a piece in the tiling), in the case of the $0$-or-$4$ gadgets and the $3$-in-$4$ gadget, or the setting of an edge, in the case of the other gadgets.
    \item For each possible tiling, each individual gadget can be filled with $\JJ$ pieces under SRS, where if a gadget has an up port, the $\JJ$ pieces come from the up port, and otherwise the $\JJ$ pieces come from a left or right port. (In other words, the $\JJ$ pieces are able to move to where they need to go in the tiling, particularly as the gadget is partially filled with $\JJ$ pieces.)
\end{itemize}

\subsubsection{Horizontal Line and U-Turn gadgets}

The horizontal line (HL) and the two versions of the U-turn (UL, UR) gadgets (one left-down-right version and one right-down-left version), along with their ports, are shown in Figure~\ref{fig:j_hl_and_u}. The HL gadget and the vertical part of the UL and UR gadget can be extended in length by multiples of $4$ as necessary. None of the gadgets change the coordinates of the ports modulo $4$.

Traversals between ports are clearly possible in the HL gadget assuming the $\JJ$ piece is in the default or $180^\circ$ orientation (where the piece is two squares tall), and through some rotations and one wall kick in the bottom corner of either U-turn gadget, any $\JJ$ piece can traverse from the up port to the down port while ending up in either the default orientation or the $180^\circ$ orientation at the down port. The possible tilings, along with an order in which the $\JJ$ pieces can be placed to ensure that the $\JJ$ pieces tile the gadget correctly, are shown in Figures~\ref{fig:j_horline_tilings} and \ref{fig:j_uturn_tilings}.

\begin{figure}[!ht]
    \centering
    \begin{subfigure}[b]{0.32\textwidth}
        \centering
        \includegraphics[width=120pt]{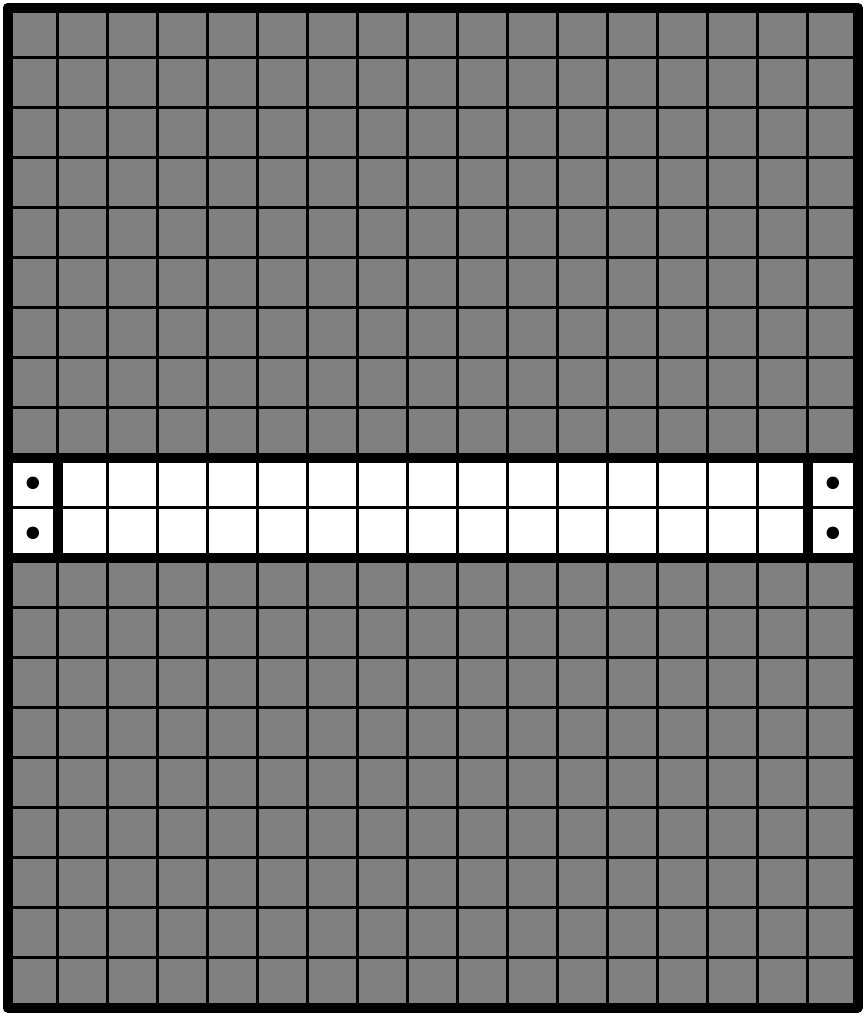}
        \caption{Horizontal Line (HL)}
    \end{subfigure}
    \begin{subfigure}[b]{0.32\textwidth}
        \centering
        \includegraphics[width=120pt]{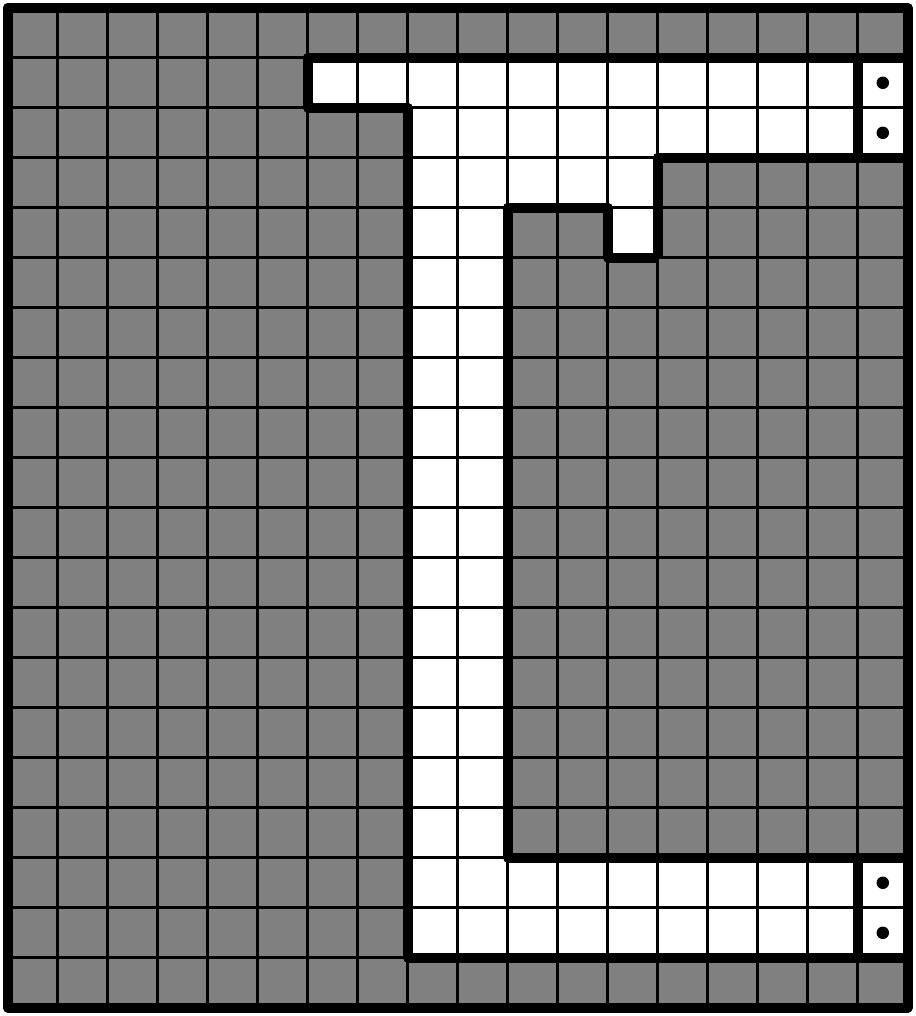}
        \caption{U-Turn (UL; left-down-right)}
    \end{subfigure}
    \begin{subfigure}[b]{0.32\textwidth}
        \centering
        \includegraphics[width=120pt]{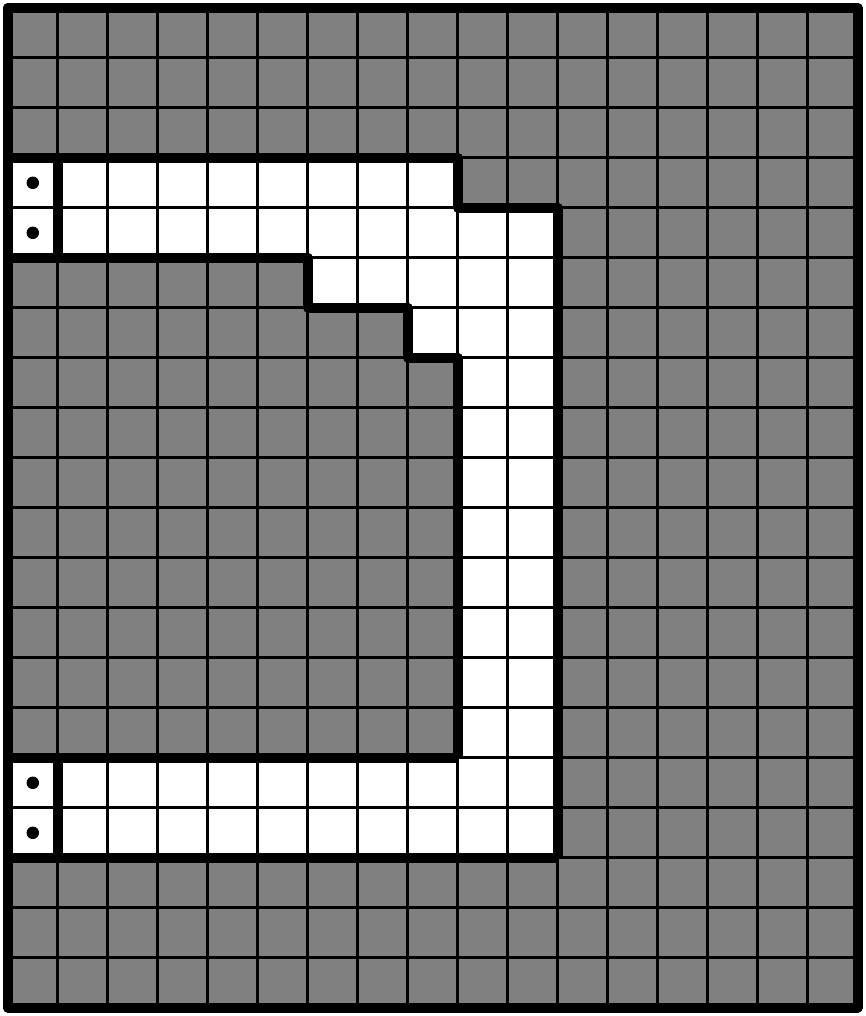}
        \caption{U-Turn (UR; right-down-left)}
    \end{subfigure}
    \caption{}
    \label{fig:j_hl_and_u}
\end{figure}

Because $\JJ$ pieces need to enter the main part of the gadget, we create a special \emph{entry corner (EC)} gadget, as shown in Figure~\ref{fig:j_entrycorner}. Due to the structure of the long tunnel, the upper-left part of the gadget must be tiled a specific way so that the majority of the gadget functions as a normal UL gadget. Using the necessary rotations and wall kicks, the traversal from the upper part of the gadget to either of the two ports is possible with the $\JJ$ piece ending up in either the default orientation or the $180^\circ$ orientation. The only possible tilings, along with an order in which the $\JJ$ pieces can be placed to ensure that the $\JJ$ pieces tile the gadget correctly, are shown in Figure~\ref{fig:j_entrycorner_tilings}.

\begin{figure}[!ht]
    \centering
    \includegraphics[width=150pt]{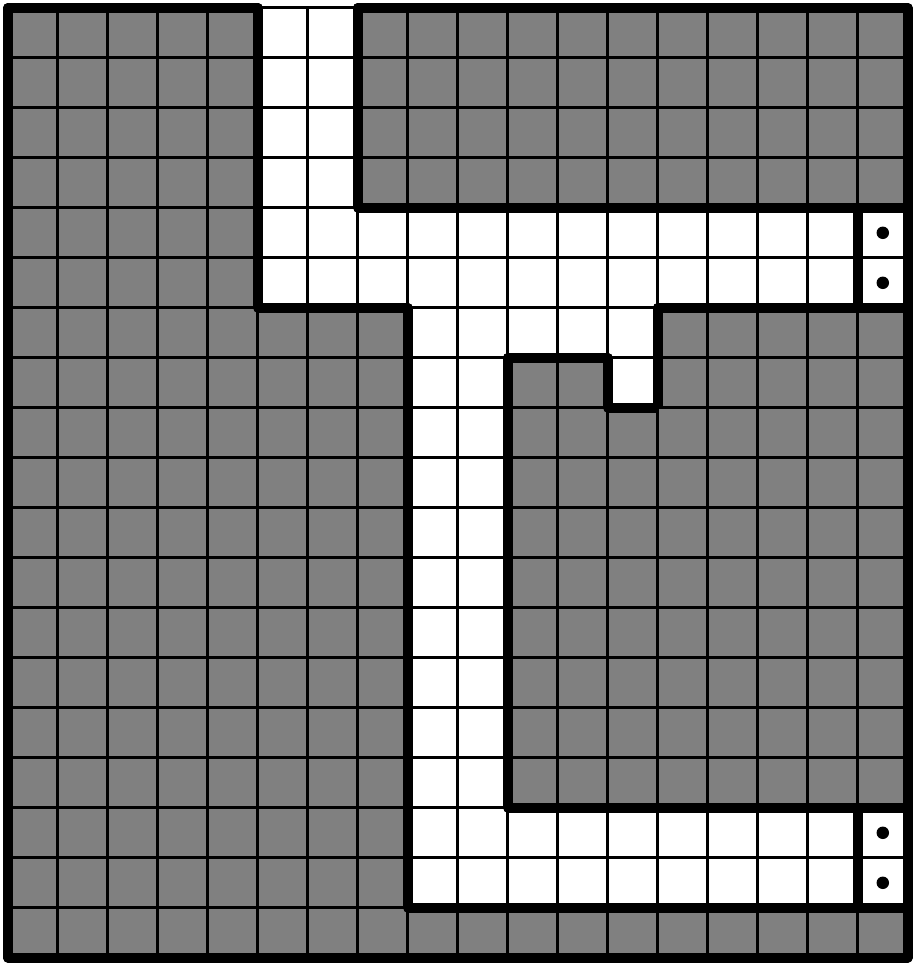}
    \caption{The entry corner (EC) gadget.}
    \label{fig:j_entrycorner}
\end{figure}

\subsubsection{$0$-or-$4$ and $3$-in-$4$ gadgets}

The $0$-or-$4$ and $3$-in-$4$ gadgets, along with their ports, are shown in Figure~\ref{fig:j_03mod4}. The ports in the gadgets do not agree modulo $4$, hence the requirement of the parity fixer gadgets.

Traversals from the upper ports to the lower ports, along with the filling of the gadgets, are possible through rotations and some wall kicks around the bends and the center area of the gadgets.

The only possible tilings, along with an order in which the $\JJ$ pieces can be placed to ensure that the $\JJ$ pieces tile the gadget correctly, are shown in Figures~\ref{fig:j_0mod4_tilings} and \ref{fig:j_3mod4_tilings}.

\begin{figure}[!ht]
    \centering
    \begin{subfigure}[b]{0.49\textwidth}
        \centering
        \includegraphics[width=150pt]{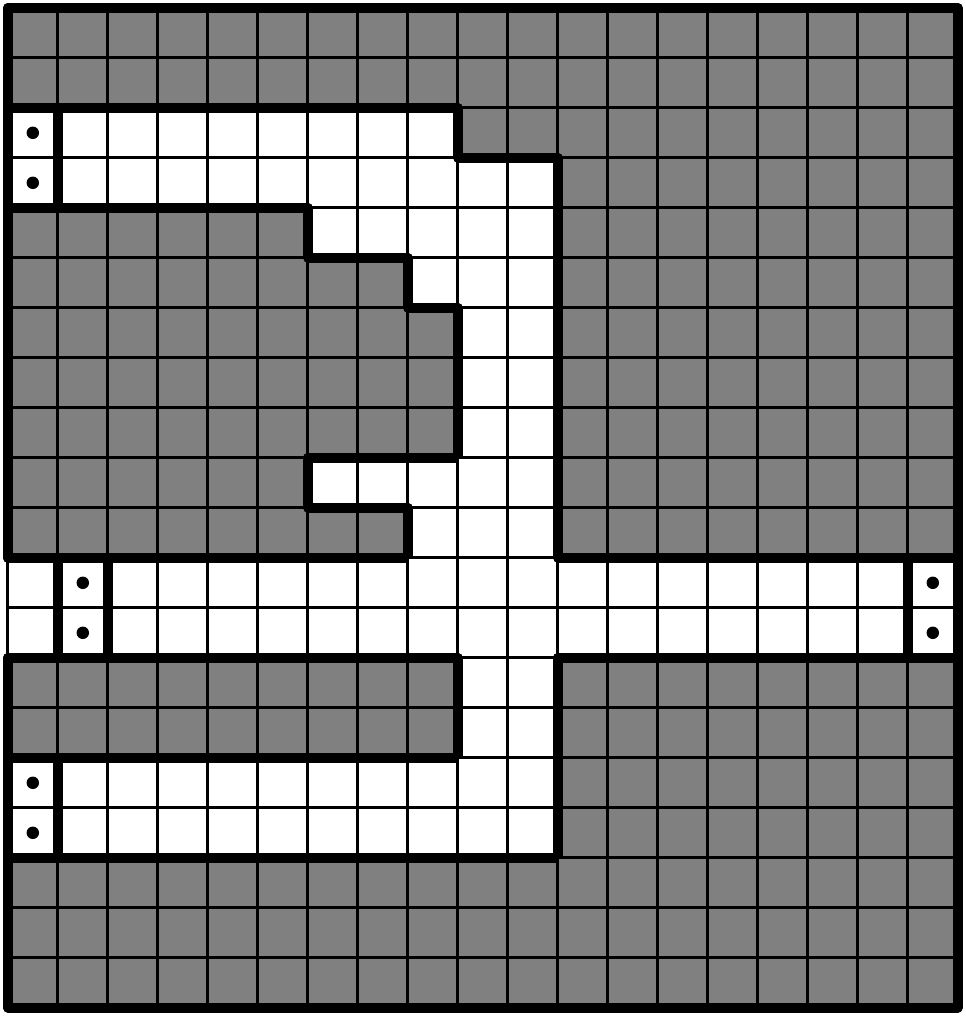}
        \caption{$0$-or-$4$}
    \end{subfigure}
    \begin{subfigure}[b]{0.49\textwidth}
        \centering
        \includegraphics[width=100pt]{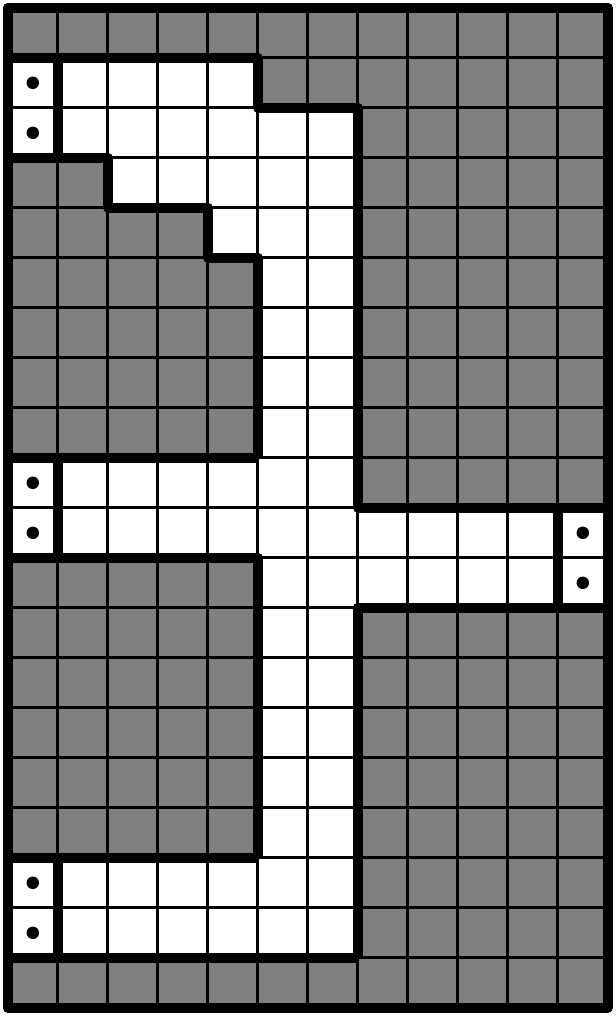}
        \caption{$3$-in-$4$}
    \end{subfigure}
    \caption{The $0$-or-$4$ and $3$-in-$4$ gadgets.}
    \label{fig:j_03mod4}
\end{figure}

\subsubsection{Parity-Fixer gadgets}

The two parity-fixer (PF) gadgets, along with their ports, are shown in Figure~\ref{fig:j_parityfixer}. The upper parity fixer shifts domino ports by $(+1, +2)$ modulo $4$, and the lower parity fixer shifts domino ports by $(-1, +1)$ modulo $4$ (here, $+x$ is to the right and $+y$ is upwards). Together, these parity fixers cover all possible $(x, y)$ combinations modulo $4$.

Traversals from the upper ports to the lower ports, along with the filling of the gadgets, are possible through rotations and some wall kicks around the bends of the gadgets.

Because these gadgets will be attached to the other gadgets, which have specific constraints on how they are tiled, each parity fixer has only two possible tilings, each one corresponding to covering one domino port but not the other. The possible tilings, along with an order in which the $\JJ$ pieces can be placed to ensure that the $\JJ$ pieces tile the gadget correctly, are shown in Figure~\ref{fig:j_parityfixer_tilings}.

\begin{figure}[!ht]
    \centering
    \includegraphics[width=100pt]{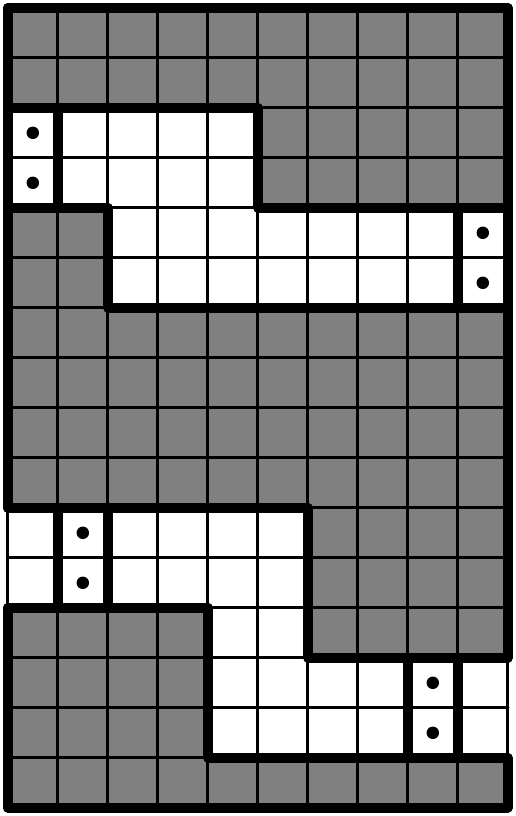}
    \caption{The two parity-fixer (PF) gadgets.}
    \label{fig:j_parityfixer}
\end{figure}

\subsection{Putting it all together}

Now we describe the main part of the construction. Take the Planar Biconnected $\{\{0, 4\}\text{-in-4},\text{3-in-4}\}$ Graph Orientation instance and compute an embedding of the graph on a grid using the linear-time BICONN algorithm in \cite{biedl1998better}. The computed embedding also has some special properties, namely that if we add an artificial ``root vertex'' to the top-left ``bend'' of the topmost edge, then all vertices are on different $y$-coordinates, and every point (vertex or along an edge) can be reached via a ``downward'' (i.e., monotone decreasing $y$-coordinate) path from the root vertex. We will call the latter property \emph{downward reachable}.

From here, we embed the grid graph drawing on the Tetris board with sufficient spacing around vertices, place an entry corner gadget at the root vertex, and place the $0$-or-$4$ and $3$-in-$4$ gadgets at their corresponding vertices (the former for vertices that must have indegree either $0$ or $4$, the latter at vertices that must have indegree $3$). Lastly, use horizontal lines, both U-turns, and both parity fixers as necessary to connect the $0$-or-$4$ and $3$-in-$4$ gadgets (these play the role of edges and transmit information about the orientations of the gadgets).

\subsection{Correctness}

As we have now described all of the parts of the construction, we now prove correctness of the construction; in other words, the whole construction can be cleared with $\JJ$ pieces under SRS if and only if the Planar Biconnected $\{\{0, 4\}\text{-in-4},\text{3-in-4}\}$ Graph Orientation instance has a solution. From the discussion at the beginning of Section~\ref{sec:jtrisclearing}, the whole construction can be cleared with $\JJ$ pieces under SRS if and only if the main part of the construction can be filled with $\JJ$ pieces under SRS.

First, if the main part of the construction can be filled with $\JJ$ pieces under SRS, then we necessarily get a tiling of the main part of the construction with $\JJ$ pieces. From there, we can read off the orientations of the gadgets to get the orientation of the edges, which is a solution to the Planar Biconnected $\{\{0, 4\}\text{-in-4},\text{3-in-4}\}$ Graph Orientation instance.

Now, suppose we have a solution to the Planar Biconnected $\{\{0, 4\}\text{-in-4},\text{3-in-4}\}$ Graph Orientation instance. Using the solution, we can quickly (i.e., in polynomial time) get the orientations of all of the gadgets. From here, we want to show that we can maneuver and place the $\JJ$ pieces into their desired locations in the main part of the construction.

The key idea is to use the special properties of the graph embedding. The general process can be described as follows:

\begin{enumerate}
    \item Start by filling in the horizontal line and U-turn gadgets that are on or touching the bottommost rows of the main part of the construction (see the rightmost figure in Figure 2 of \cite{biedl1998better}).
    \item Then, starting from the lowest $0$-or-$4$ or $3$-in-$4$ gadget and working upwards, fill in any of the horizontal line, U-turn, or parity fixer gadgets corresponding to any edges to a ``lower'' (smaller $y$-coordinate) $0$-or-$4$ or $3$-in-$4$ gadget, and then fill in the $0$-or-$4$ or $3$-in-$4$ gadget itself.
    \item Once all of the $0$-or-$4$ or $3$-in-$4$ gadgets are filled in, only the entry corner and some gadgets (horizontal line, U-turn, parity fixers) around the entry corner remain. Fill in the gadgets around the entry corner before filling in the entry corner itself.
\end{enumerate}

In particular, because all vertices are on different $y$-coordinates, we do not need to break ties when considering which $0$-or-$4$ or $3$-in-$4$ gadget is ``lower'' than the other. In addition, because the embedding is downward reachable, this process is guaranteed to fill in all of the gadgets in the main part of the construction as long as the gadgets are filled in reverse BFS order (from the root vertex). Thus, given a valid solution to the Planar Biconnected $\{\{0, 4\}\text{-in-4},\text{3-in-4}\}$ Graph Orientation instance, we are able to determine a way to fill the main part of the construction with $\JJ$ pieces under SRS.

Therefore, the main part of the construction can be filled with $\JJ$ pieces under SRS if and only if the Planar Biconnected $\{\{0, 4\}\text{-in-4},\text{3-in-4}\}$ Graph Orientation instance has a solution, meaning that our construction works as intended. Thus, we have shown that Tetris clearing with SRS is NP-hard even if the player is only given $\JJ$ pieces, as desired.

By symmetry, Tetris clearing with SRS is NP-hard even if the player is only given $\LL$ pieces.

\section{$\TT$-tris Clearing}\label{sec:ttrisclearing}

\begin{theorem}
    Tetris clearing with SRS is NP-hard even if the type of pieces in the sequence given to the player is restricted to just $\TT$ pieces.
\end{theorem}

Here, we will reduce from Planar Biconnected $\{\{0, 4\}\text{-in-4},\text{1-in-4}\}$ Graph Orientation. We first describe the general structure of the construction, which is similar to the general structure for the $\II$-tris and $\JJ$-tris reduction and is shown in Figure~\ref{fig:t_structure}. The main part of the construction, corresponding to the Planar Biconnected $\{\{0, 4\}\text{-in-4},\text{1-in-4}\}$ Graph Orientation instance, is indicated in green. There is a path of empty squares to the main part of the construction, which is located mostly within the leftmost 23 columns or the topmost 11 rows of the board, and which we will refer to as the ``long tunnel''; we will also refer to the vertical structure between columns $7$ and $22$ (from the left) as the ``climbable vertical structure''. There are 4 additional empty squares in the three rightmost columns of the board. If there are $M$ empty squares in the construction, we give the player $M/4$ $\TT$ pieces.

\begin{figure}[!ht]
    \centering
    \includegraphics[width=320pt]{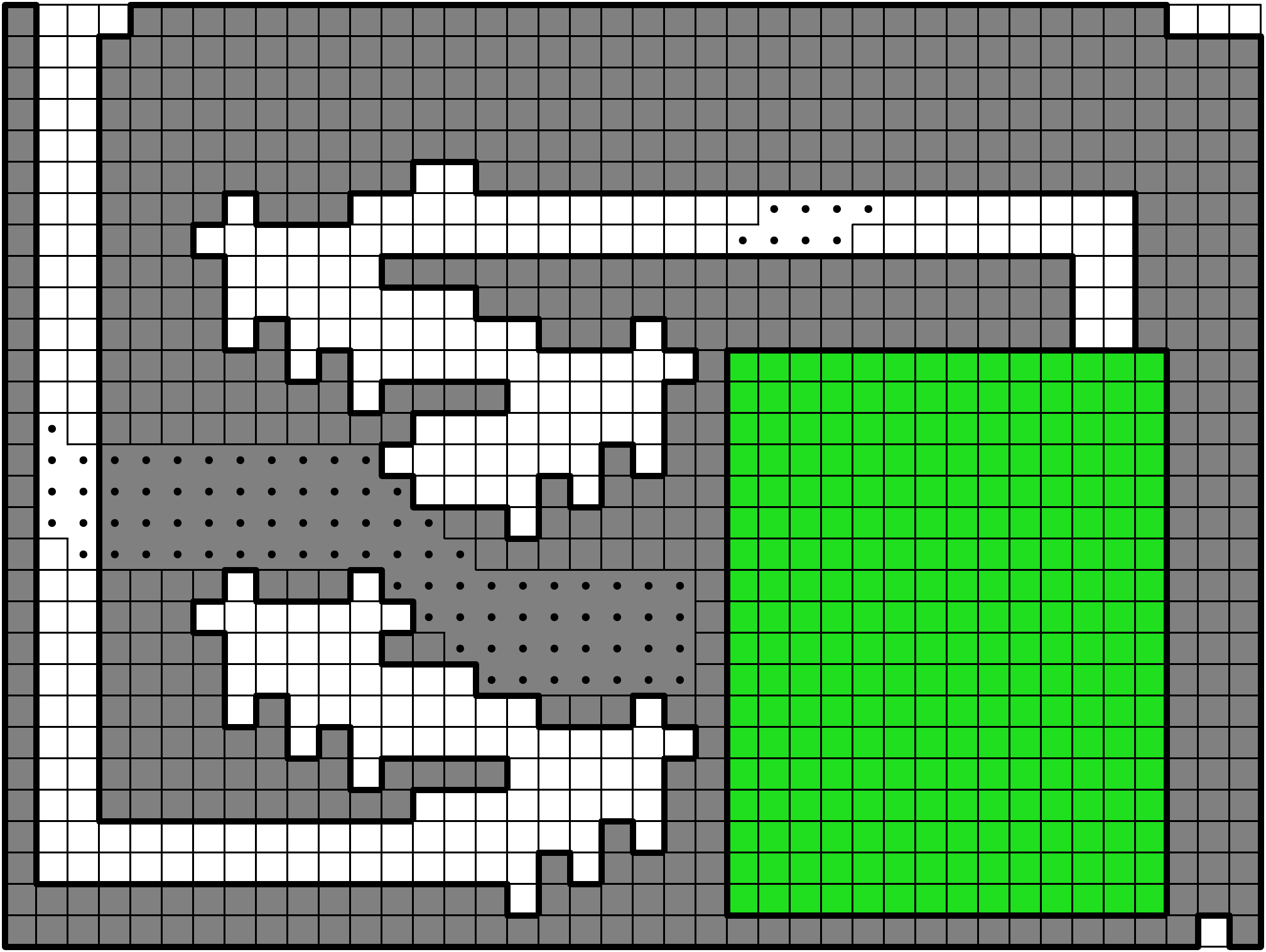}
    \caption{The general structure of the construction for Tetris clearing using only $\TT$ pieces. The main part of the construction is indicated in green.}
    \label{fig:t_structure}
\end{figure}

Similar to the $\JJ$-tris reduction, we note that the 4 empty squares in the three rightmost columns must be the last squares to be filled in, after the middle rows are cleared. In addition, there is exactly one way to fill in/tile the long tunnel with $\TT$ pieces, as shown in Figures~\ref{fig:t_structure_fill} and \ref{fig:t_cvs_fill}, and it is possible (under SRS) to get $\TT$ pieces through the long tunnel (even the climbable vertical structure) into the main part of the construction (see the maneuvers in Appendix~\ref{sec:t_maneuvers}). Locking any $\TT$ piece in place in the long tunnel before the main part of the construction is completely filled blocks off the main part of the construction, causing the board to become unclearable. Thus, no $\TT$ piece can be placed in the long tunnel before we are done filling in the main part of the construction. Lastly, before we place any $\TT$ pieces in the long tunnel, the long tunnel also prevents line clears in the rows corresponding to the main part of the construction.

\begin{figure}[!ht]
    \centering
    \includegraphics[width=320pt]{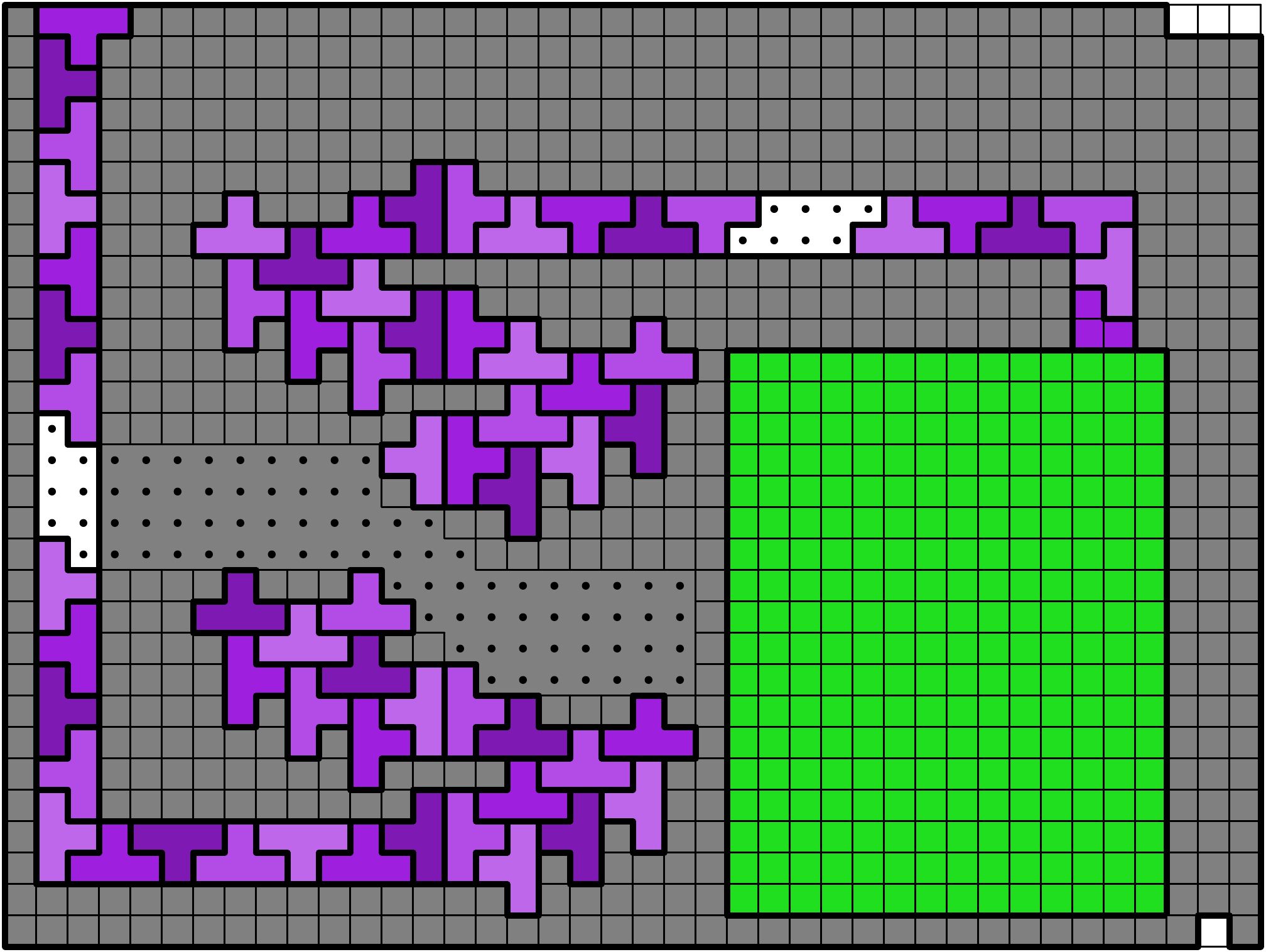}
    \caption{The only possible tiling of the climbable vertical structure with $\TT$ pieces.}
    \label{fig:t_structure_fill}
\end{figure}

\begin{figure}[!ht]
    \centering
    \begin{subfigure}[b]{0.45\textwidth}
        \centering
        \includegraphics[width=150pt]{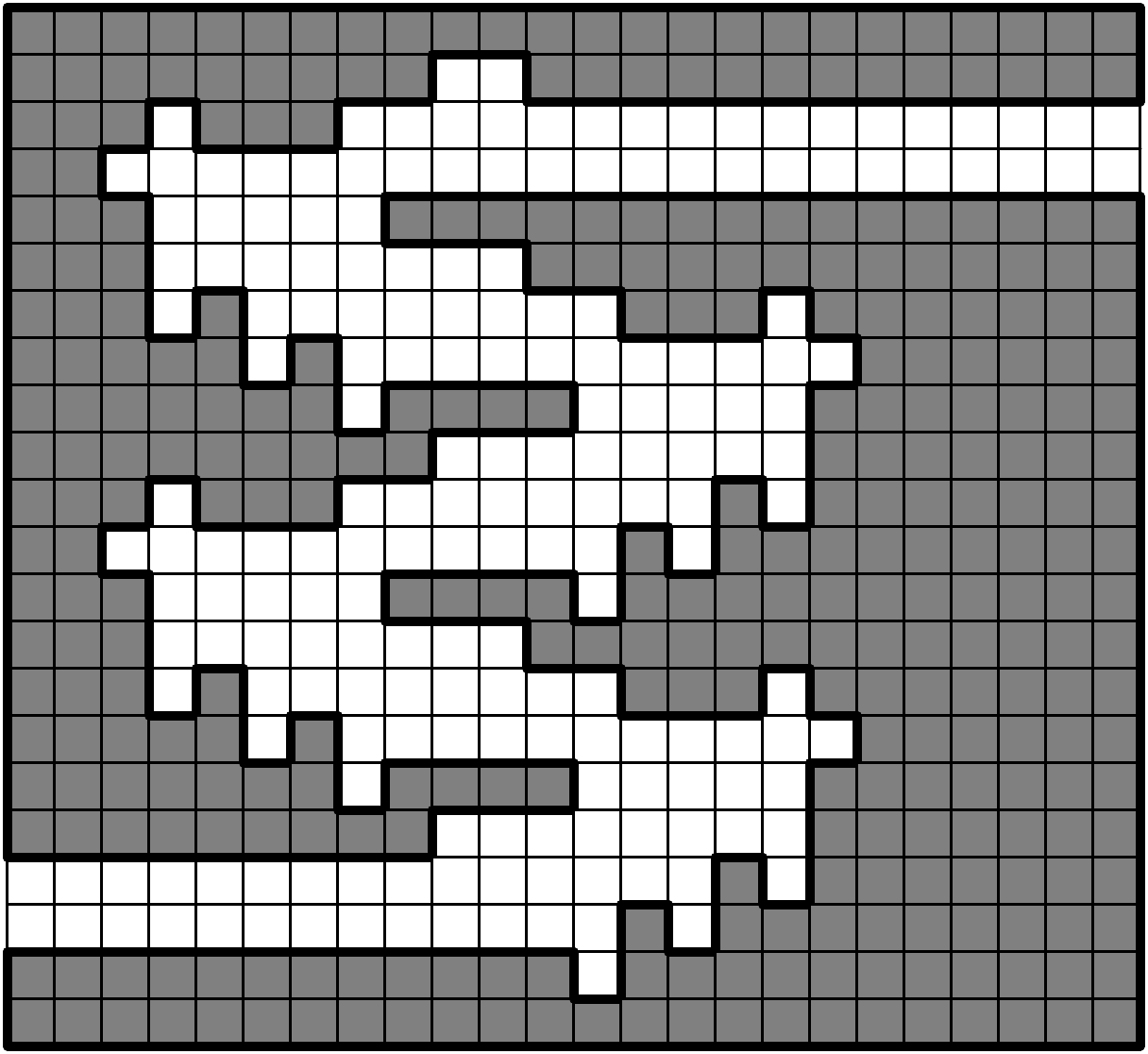}
        \caption{Unfilled}
    \end{subfigure}
    \begin{subfigure}[b]{0.45\textwidth}
        \centering
        \includegraphics[width=150pt]{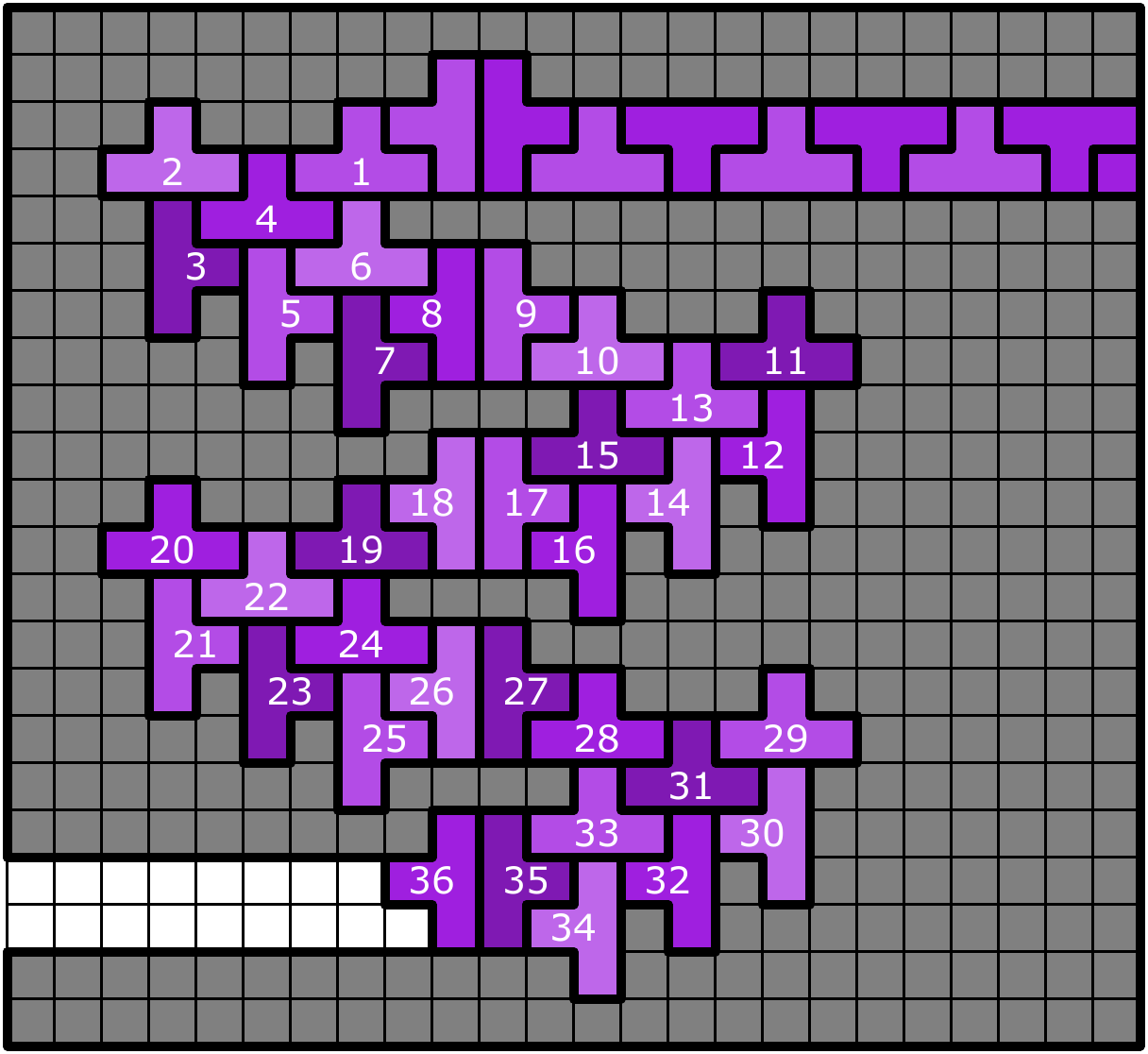}
        \caption{Filled}
    \end{subfigure}
    \caption{A closer look at a part of the climbable vertical structure, including the piece ordering of one of the segments of the climbable vertical structure.}
    \label{fig:t_cvs_fill}
\end{figure}

Thus, we need to be able to tile the main part of the construction with $\TT$ pieces exactly, with the additional restriction that we need to be able to maneuver all the pieces into place through Tetris and SRS rules.

\subsection{Gadgets}

Now, we introduce the gadgets used in the main part of the construction. Like with the $\JJ$-tris gadgets, we will call the entrances \textbf{ports}, which will be indicated by empty squares with dots in their centers. We will need $0$-or-$4$ (corresponding to vertices that must have indegree either $0$ or $4$), $1$-in-$4$ (corresponding to vertices that must have indegree $1$), horizontal line and U-turn gadgets (corresponding to edges between vertices), along with a special gadget called a parity fixer that helps ``align'' the coordinates of the ports between gadgets modulo $4$ in case the coordinates of the ports modulo $4$ are not the same.

Like with the $\JJ$-tris gadgets, we want the gadgets to have the following properties:

\begin{itemize}
    \item Any traversal from a port with a higher $y$-coordinate to a lower $y$-coordinate, or between ports with the same $y$-coordinate, is possible.
    \item The possible ways to tile the gadgets either correspond to a setting of the edges around a vertex (where an edge is pointed into a vertex if the port is filled by a piece in the tiling and out of a vertex if the port is not filled by a piece in the tiling), in the case of the $0$-or-$4$ gadgets and the $1$-in-$4$ gadget, or the setting of an edge, in the case of the other gadgets.
    \item For each possible tiling, each individual gadget can be filled with $\TT$ pieces under SRS, where if a gadget has an up port, the $\TT$ pieces come from the up port, and otherwise the $\TT$ pieces come from a left or right port. (In other words, the $\TT$ pieces are able to move to where they need to go in the tiling, particularly as the gadget is partially filled with $\TT$ pieces.)
\end{itemize}

\subsubsection{Horizontal Line and U-Turn gadgets}

The horizontal line (HL) and U-turn gadgets, along with their ports, are shown in Figure~\ref{fig:t_hl_and_u} (only the right-down-left version is shown; the left-down-right version is symmetric). The HL gadget and the vertical part of the U-turn gadget can be extended in length by multiples of $4$ as necessary. The HL gadget does not change the coordinates of the ports modulo $4$; however, the ports in the U-turn gadget are offset by $(+2, +2)$ modulo $4$.

Traversals between ports are clearly possible in the HL gadget assuming the $\TT$ piece is in the default or $180^\circ$ orientation (where the piece is two squares tall), and through some rotations (and some wall kick during the filling process), any $\TT$ piece can traverse from the up port to the down port while ending up in either the default orientation or the $180^\circ$ orientation at the down port. The possible tilings, along with an order in which the $\TT$ pieces can be placed to ensure that the $\TT$ pieces tile the gadget correctly, are shown in Figures~\ref{fig:t_horline_tilings} and \ref{fig:t_uturn_tilings}.

\begin{figure}[!ht]
    \centering
    \begin{subfigure}[b]{0.45\textwidth}
        \centering
        \includegraphics[width=150pt]{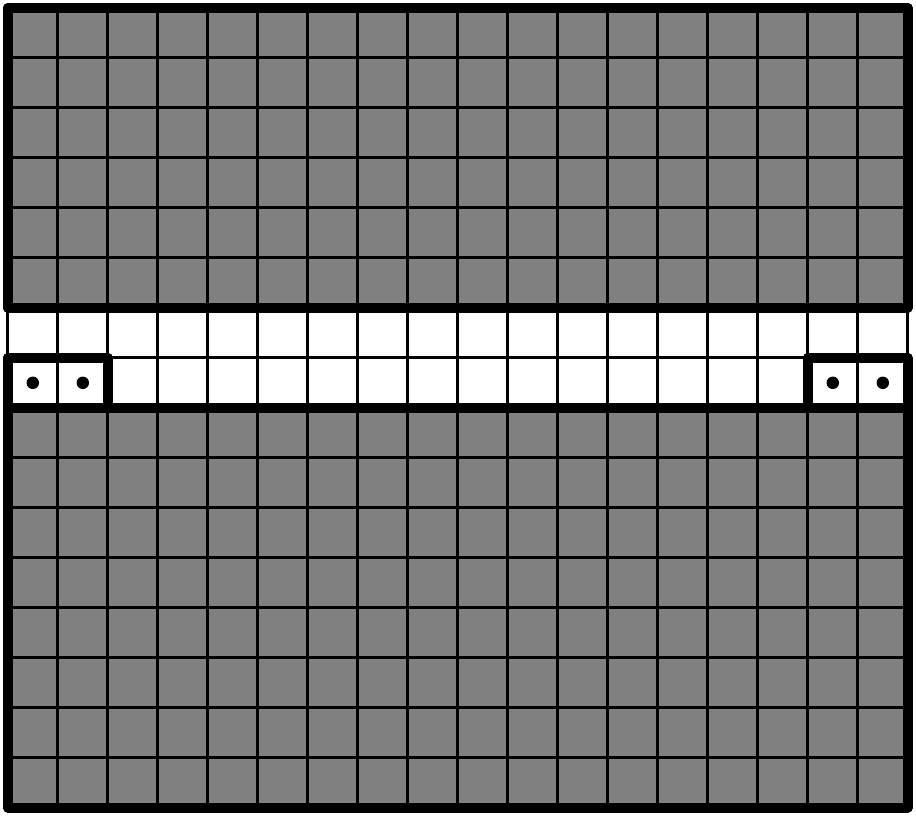}
        \caption{Horizontal Line (HL)}
    \end{subfigure}
    \begin{subfigure}[b]{0.45\textwidth}
        \centering
        \includegraphics[width=150pt]{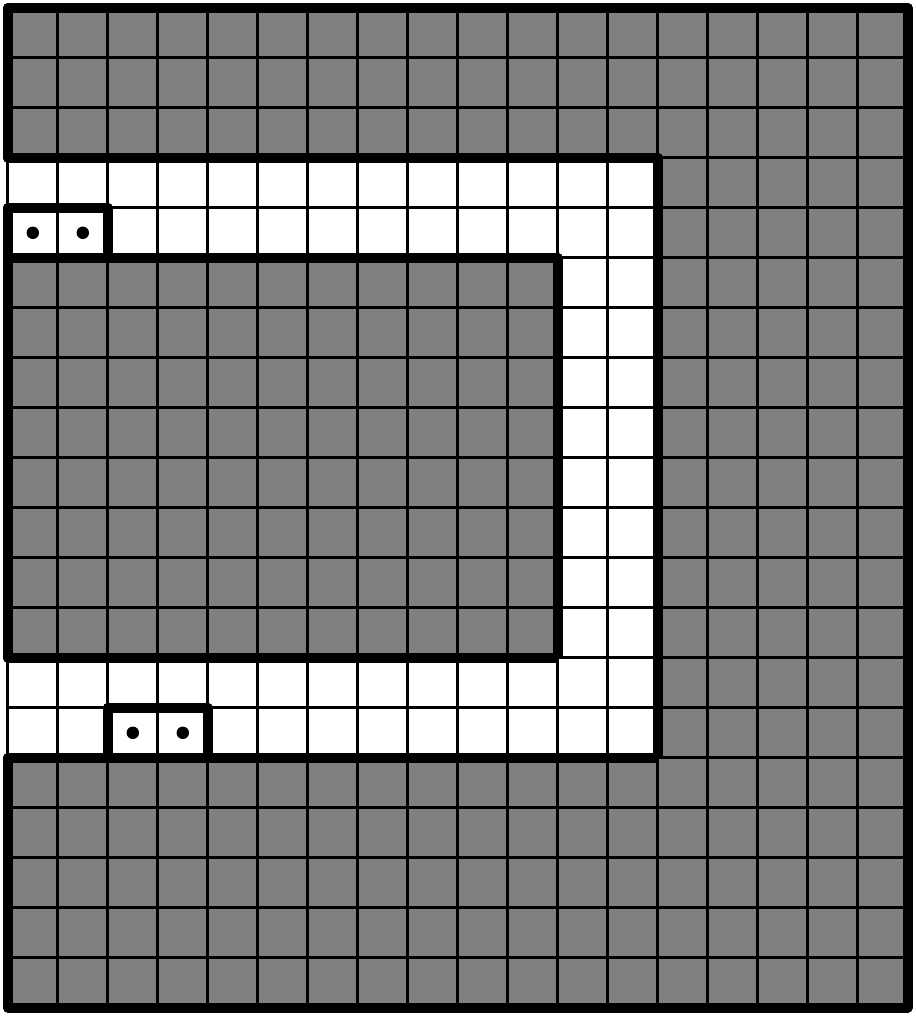}
        \caption{U-Turn}
    \end{subfigure}
    \caption{}
    \label{fig:t_hl_and_u}
\end{figure}

Because $\TT$ pieces need to enter the main part of the gadget, we create a special \emph{entry corner (EC)} gadget, as shown in Figure~\ref{fig:t_entrycorner}. Due to the structure of the long tunnel, the upper-right part of the gadget must be tiled a specific way so that the majority of the gadget functions as a normal U-turn gadget. Using the necessary rotations and wall kicks, the traversal from the upper part of the gadget to either of the two ports is possible with the $\TT$ piece ending up in either the default orientation or the $180^\circ$ orientation. The only possible tilings, along with an order in which the $\TT$ pieces can be placed to ensure that the $\TT$ pieces tile the gadget correctly, are shown in Figure~\ref{fig:t_entrycorner_tilings}.

\begin{figure}[!ht]
    \centering
    \includegraphics[width=150pt]{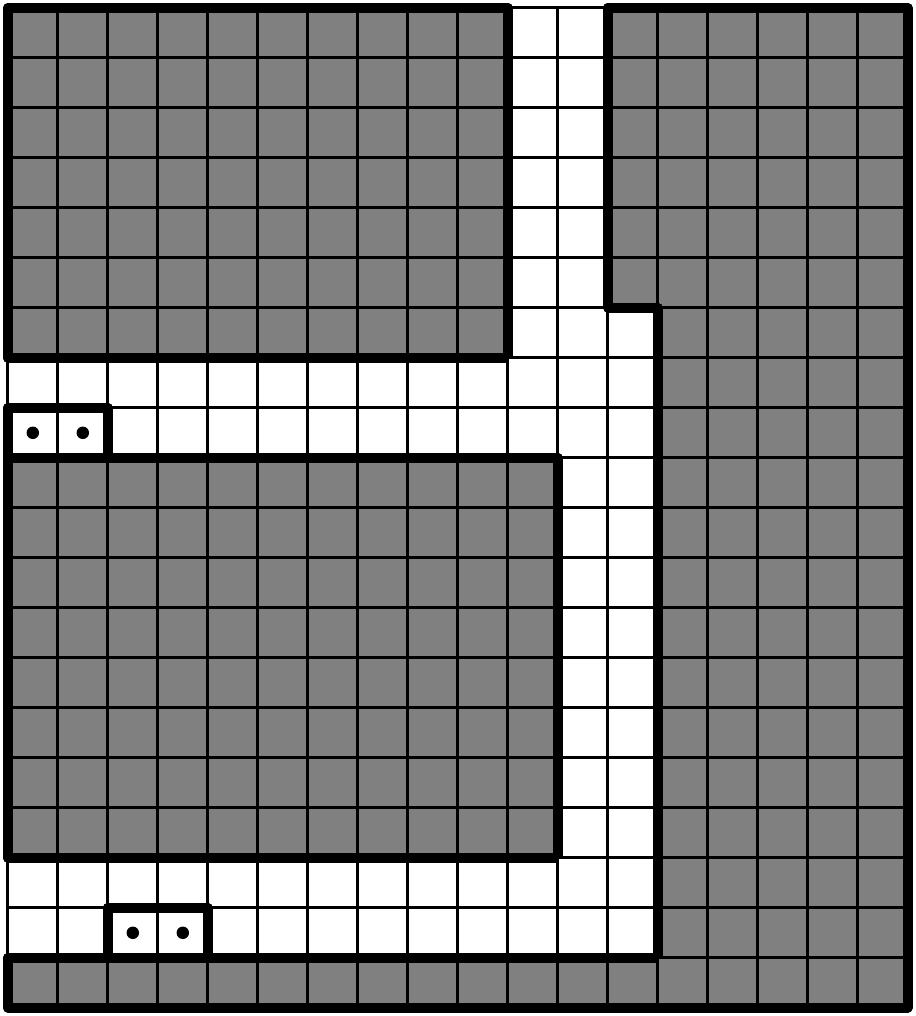}
    \caption{The entry corner (EC) gadget.}
    \label{fig:t_entrycorner}
\end{figure}

\subsubsection{$0$-or-$4$ and $1$-in-$4$ gadgets}

The $0$-or-$4$ and $1$-in-$4$ gadgets, along with their ports, are shown in Figure~\ref{fig:t_01mod4}. The ports in the gadgets do not agree modulo $4$, hence the requirement of the parity fixer gadgets.

Traversals from the upper ports to the lower ports, along with the filling of the gadgets, are possible through rotations and some wall kicks around the bends and the center area of the gadgets.

The only possible tilings, along with an order in which the $\TT$ pieces can be placed to ensure that the $\TT$ pieces tile the gadget correctly, are shown in Figures~\ref{fig:t_0mod4_tilings} and \ref{fig:t_1mod4_tilings}.

\begin{figure}[!ht]
    \centering
    \begin{subfigure}[b]{0.49\textwidth}
        \centering
        \includegraphics[width=150pt]{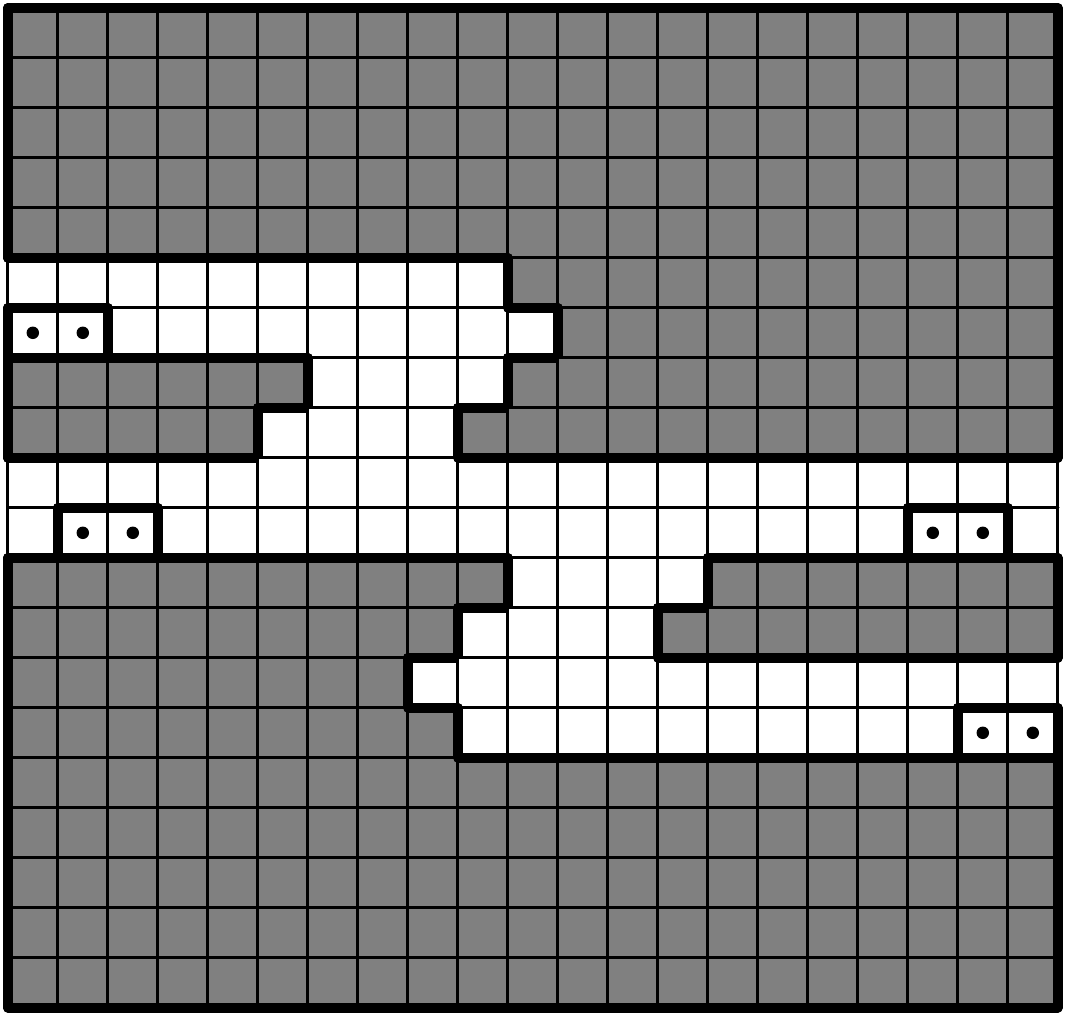}
        \caption{$0$-or-$4$}
    \end{subfigure}
    \begin{subfigure}[b]{0.49\textwidth}
        \centering
        \includegraphics[width=150pt]{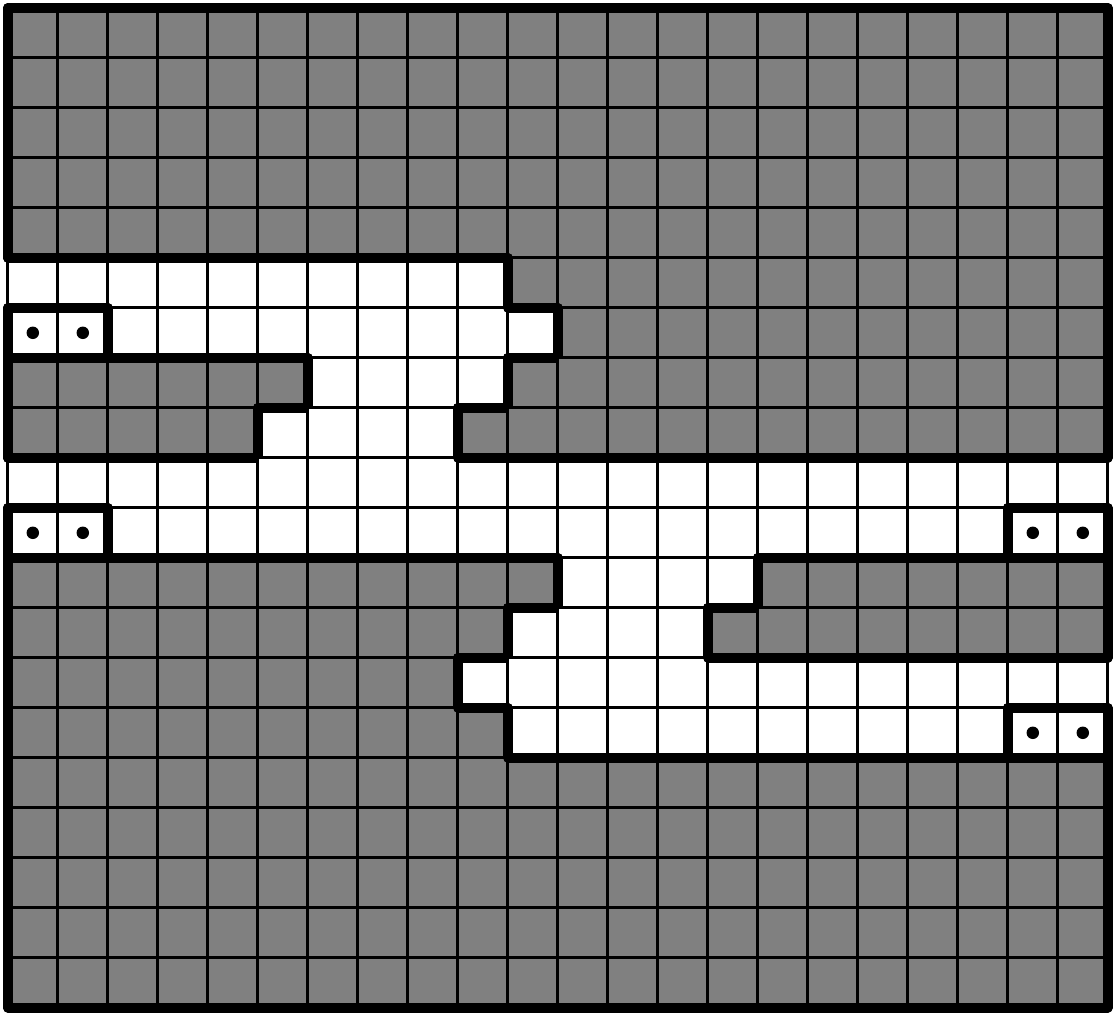}
        \caption{$1$-in-$4$}
    \end{subfigure}
    \caption{The $0$-or-$4$ and $1$-in-$4$ gadgets.}
    \label{fig:t_01mod4}
\end{figure}

\subsubsection{Parity-Fixer gadgets}

The two parity-fixer (PF) gadgets, along with their ports, are shown in Figure~\ref{fig:t_parityfixer}. The first parity fixer shifts domino ports by $(0, +1)$ modulo $4$, and the second parity fixer shifts domino ports by $(-1, +1)$ modulo $4$ (here, $+x$ is to the right and $+y$ is upwards). Together, these parity fixers cover all possible $(x, y)$ combinations modulo $4$.

Traversals from the upper ports to the lower ports, along with the filling of the gadgets, are possible through rotations and some wall kicks around the bends of the gadgets.

Because these gadgets will be attached to the other gadgets, which have specific constraints on how they are tiled, each parity fixer has only two possible tilings, each one corresponding to covering one domino port but not the other. The possible tilings, along with an order in which the $\TT$ pieces can be placed to ensure that the $\TT$ pieces tile the gadget correctly, are shown in Figure~\ref{fig:t_parityfixer_tilings}.

\begin{figure}[!ht]
    \centering
    \begin{subfigure}[b]{0.49\textwidth}
        \centering
        \includegraphics[width=60pt]{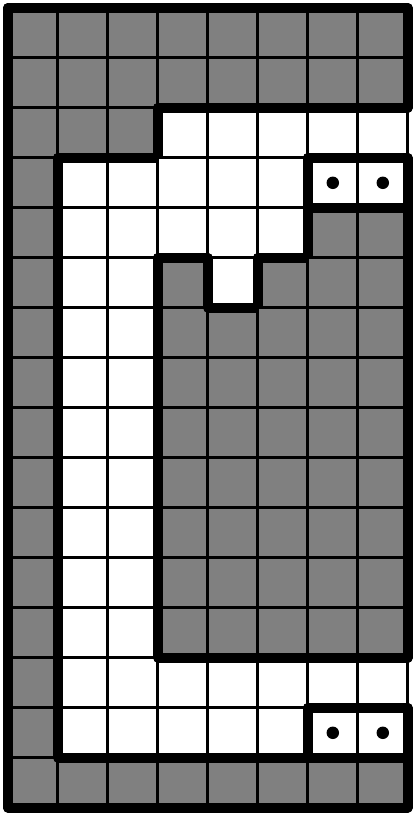}
    \end{subfigure}
    \begin{subfigure}[b]{0.49\textwidth}
        \centering
        \includegraphics[width=150pt]{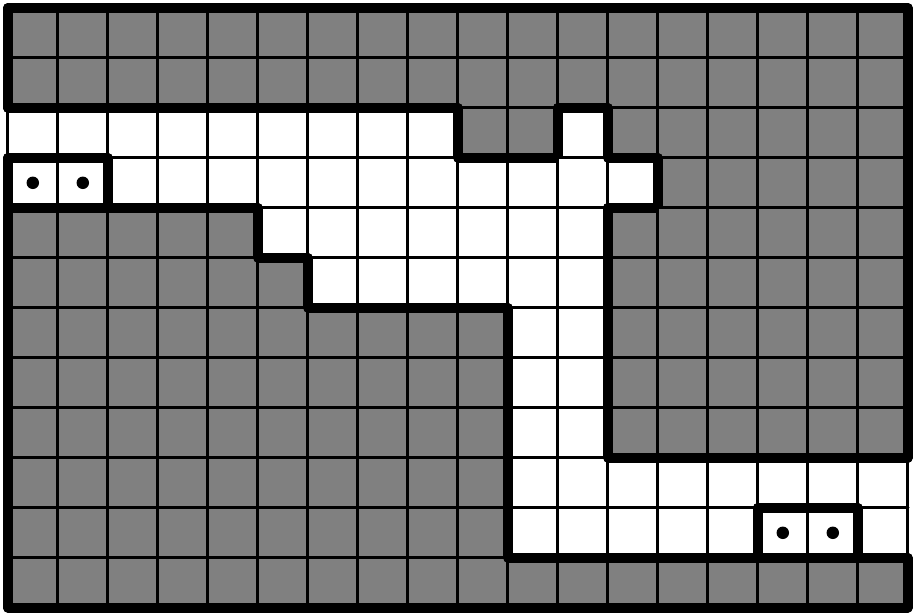}
    \end{subfigure}
    \caption{The two parity-fixer (PF) gadgets.}
    \label{fig:t_parityfixer}
\end{figure}

\subsection{Putting it all together}

Now we describe the main part of the construction, which is similar to the $\JJ$-tris construction. Take the Planar Biconnected $\{\{0, 4\}\text{-in-4},\text{1-in-4}\}$ Graph Orientation instance and compute an embedding of the graph on a grid using the linear-time BICONN algorithm in \cite{biedl1998better}. As before, if we add an artificial ``root vertex'' to the top-right ``bend'' of the topmost edge, then all vertices are on different $y$-coordinates, and the embedding is \emph{downward reachable}.

From here, we embed the grid graph drawing on the Tetris board with sufficient spacing around vertices, place an entry corner gadget at the root vertex, and place the $0$-or-$4$ and $1$-in-$4$ gadgets at their corresponding vertices (the former for vertices that must have indegree either $0$ or $4$, the latter at vertices that must have indegree $1$). Lastly, use horizontal lines, U-turns, and both parity fixers as necessary to connect the $0$-or-$4$ and $1$-in-$4$ gadgets (these play the role of edges and transmit information about the orientations of the gadgets).

\subsection{Correctness}

The correctness argument closely follows the correctness argument for $\JJ$-tris. In particular, if the main part of the construction can be filled with $\TT$ pieces under SRS, then we necessarily get a tiling of the main part of the construction with $\TT$ pieces, which gives us an orientation of the gadgets, which gives a solution to the Planar Biconnected $\{\{0, 4\}\text{-in-4},\text{1-in-4}\}$ Graph Orientation instance. In addition, if we have a solution to the Planar Biconnected $\{\{0, 4\}\text{-in-4},\text{1-in-4}\}$ Graph Orientation instance, then we can quickly (i.e., in polynomial time) get the orientations of all of the gadgets. From there, we can maneuver and place the $\TT$ pieces into their desired locations in the main part of the construction in a similar way to the $\JJ$-tris reduction, and thus obtain a way to fill the main part of the construction with $\TT$ pieces under SRS.

Therefore, the main part of the construction can be filled with $\TT$ pieces under SRS if and only if the Planar Biconnected $\{\{0, 4\}\text{-in-4},\text{1-in-4}\}$ Graph Orientation instance has a solution, meaning that our construction works as intended. Thus, we have shown that Tetris clearing with SRS is NP-hard even if the player is only given $\TT$ pieces, as desired.

\section{$\SS$-tris/$\ZZ$-tris Clearing}\label{sec:strisclearing}

\begin{theorem}
    Tetris clearing with SRS is NP-hard even if the type of pieces in the sequence given to the player is restricted to just $\SS$ pieces, or the type of pieces in the sequence given to the player is restricted to just $\ZZ$ pieces.
\end{theorem}

We focus on $\SS$ pieces; the $\ZZ$ pieces case is symmetric.

Here, we will reduce from Planar Monotone Rectilinear 3SAT. To do so, we first build variable and clause gadgets (along with wires to connect them) and a ``climbable vertical structure''. These gadgets are shown in Figure \ref{fig:strisgadgets}.

\begin{figure}[!ht]
  \centering
  \begin{subfigure}[b]{0.45\textwidth}
    \centering
    \includegraphics[width=200pt]{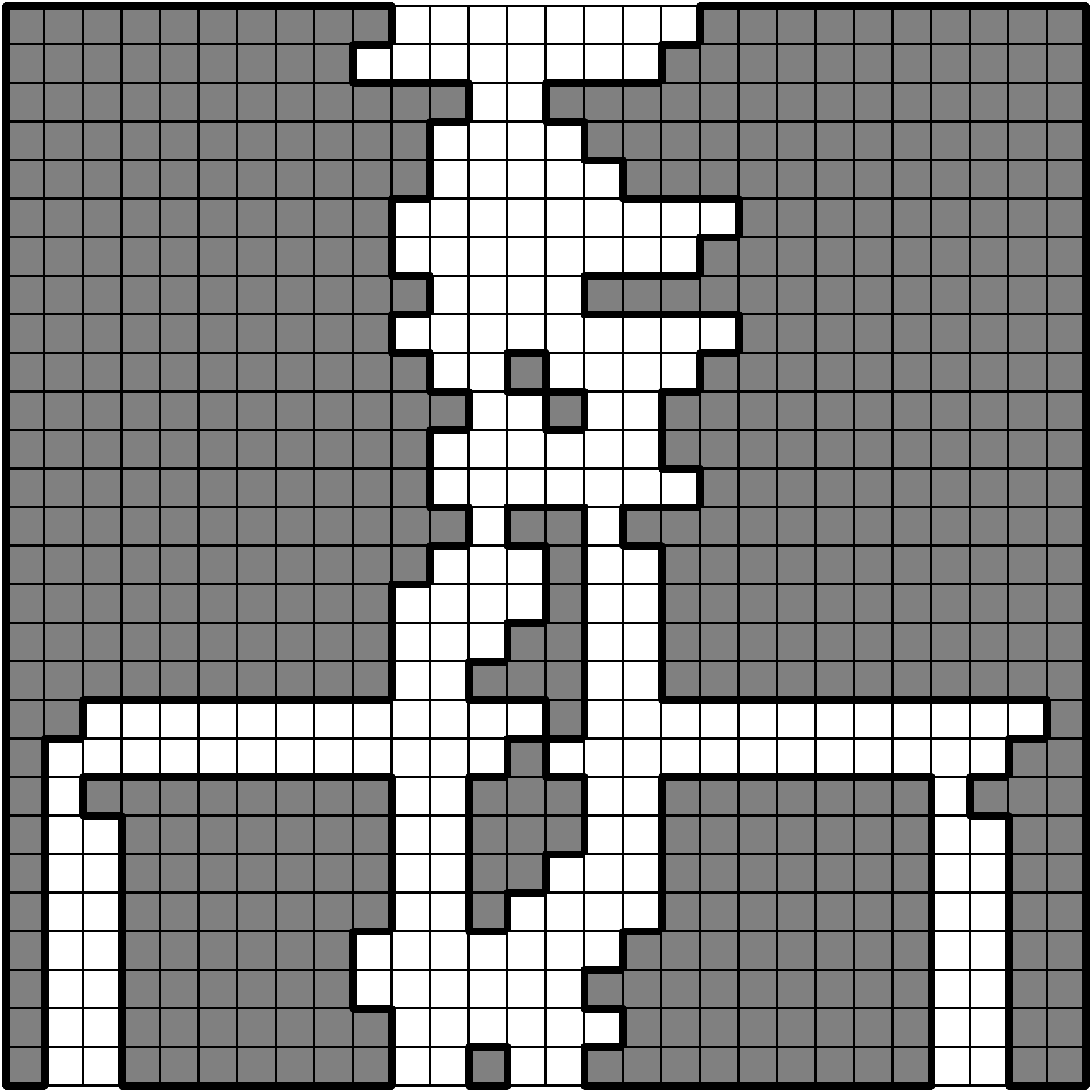}
    \caption{Variable gadget with two wires (one on the ``True'' side, one on the ``False'' side)}
  \end{subfigure}
  \begin{subfigure}[b]{0.45\textwidth}
    \centering
    \includegraphics[width=200pt]{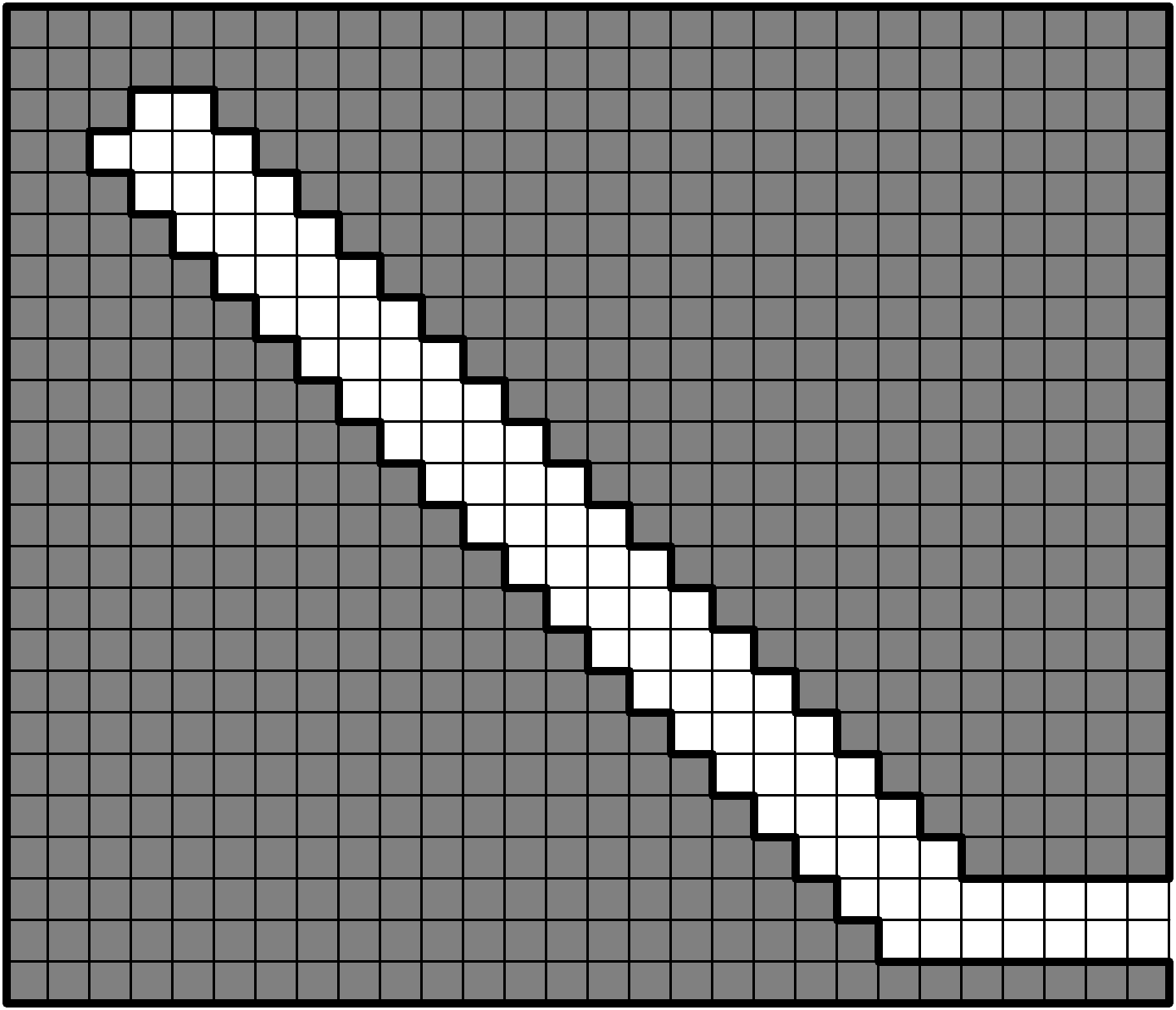}
    \caption{Climbable vertical structure}
  \end{subfigure}
  \begin{subfigure}[b]{0.9\textwidth}
    \centering
    \includegraphics[width=320pt]{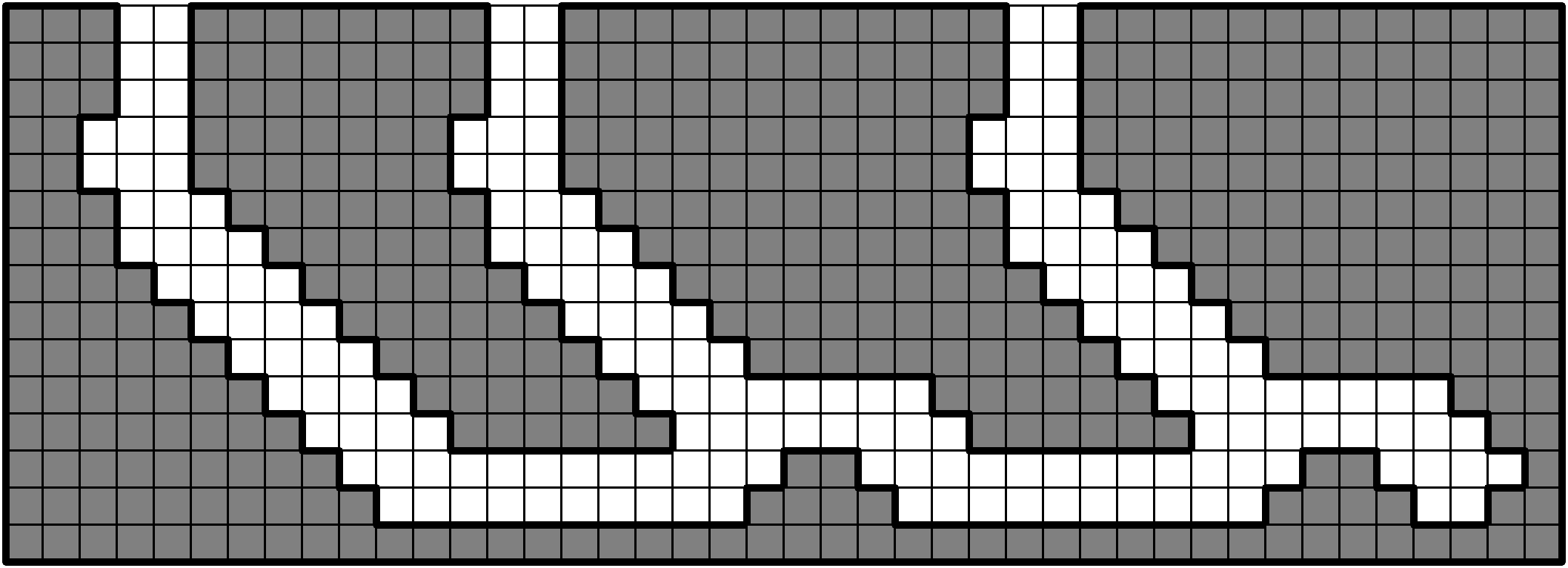}
    \caption{Clause gadget}
  \end{subfigure}
  \caption{Gadgets for $\SS$-tris.}
  \label{fig:strisgadgets}
\end{figure}

\begin{figure}[!ht]
  \centering
  \begin{subfigure}[b]{0.45\textwidth}
    \centering
    \includegraphics[width=200pt]{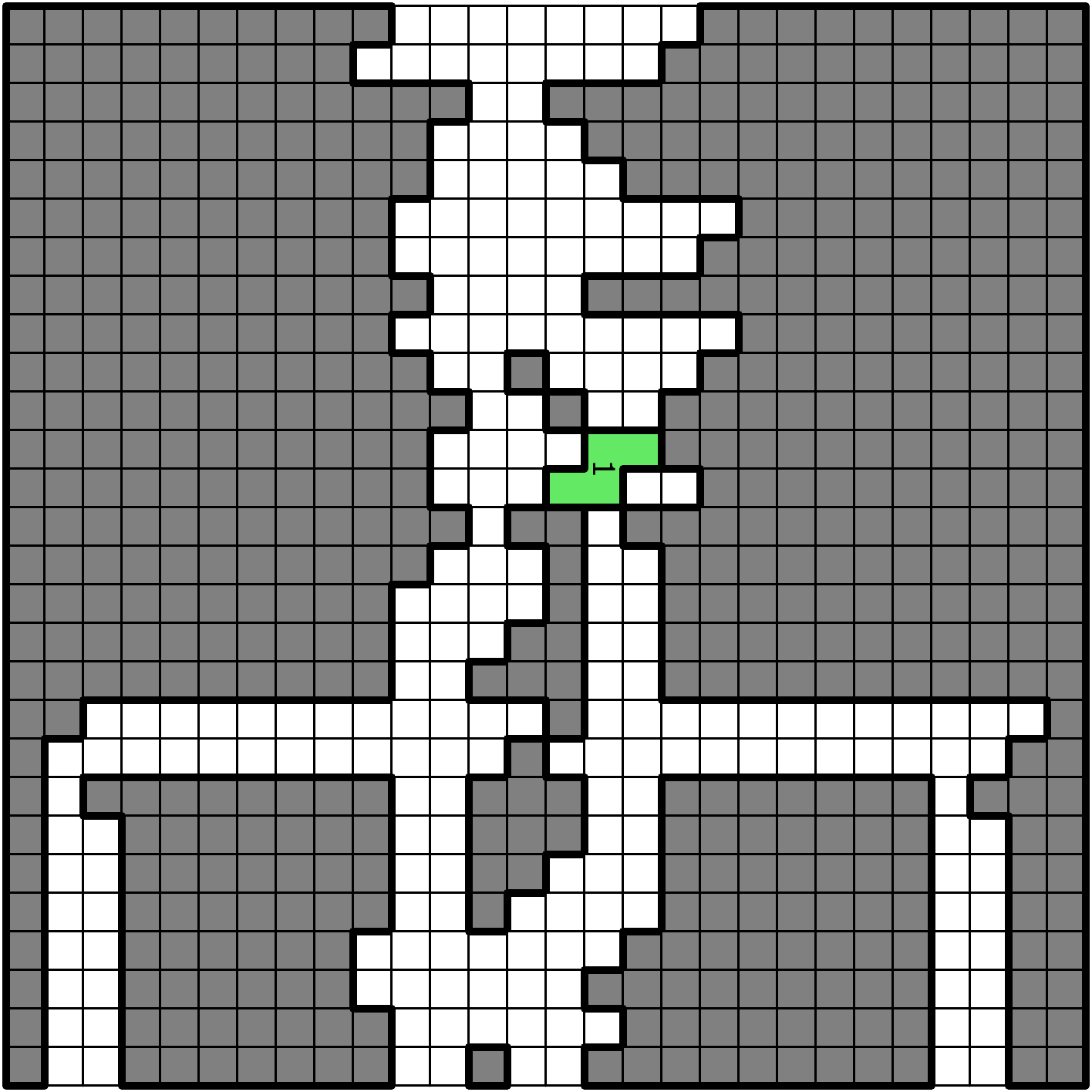}
    \caption{``True'' setting}
  \end{subfigure}
  \begin{subfigure}[b]{0.45\textwidth}
    \centering
    \includegraphics[width=200pt]{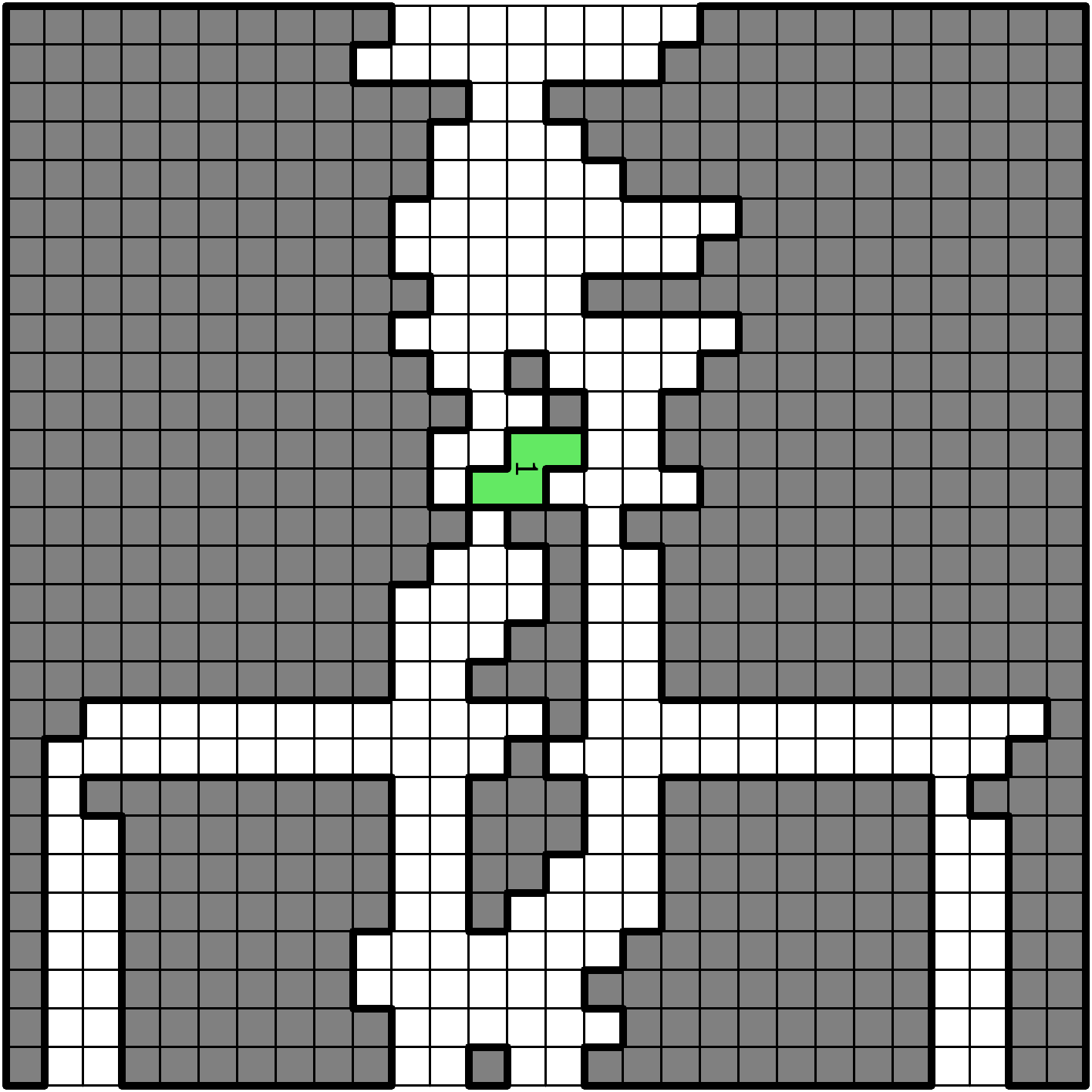}
    \caption{``False'' setting}
  \end{subfigure}
  \caption{$\SS$-tris variable settings.}
  \label{fig:strisvars}
\end{figure}

The idea behind the variable gadget is that there are two settings based on where an initial $\SS$ piece is placed, as shown in Figure \ref{fig:strisvars}, that allow $\SS$ pieces to traverse through one of the tunnels at the expense of the other tunnel being blocked. Before any line clears are possible, there are no possible piece placements that allow both tunnels to be traversable at the same time. Once lines are able to be cleared, the variable gadget can be collapsed into a form that allows for both tunnels to be traversed, as shown in Figure \ref{fig:s_var_filling}. Specific maneuvers described in Appendices \ref{subsec:s_maneuver4}, \ref{subsec:s_maneuver5}, and \ref{subsec:s_maneuver6} are used to fill and/or traverse through the variable gadget.

For the clause gadget, there are three entrances corresponding to the three connections. Once one of the entrances is accessible, most of the clause gadget is fillable other than maybe a few squares, as shown in Figure \ref{fig:s_cl_filling}. The squares that might not get filled do not affect our later arguments and can be filled in at the very end if necessary. However, the squares that can get filled once at least one of the entrances is accessible contribute to some lines getting cleared that help make other parts of the structure become accessible. Specific maneuvers described in Appendices \ref{subsec:s_maneuver1}, \ref{subsec:s_maneuver2}, and \ref{subsec:s_maneuver3} are used to fill and/or traverse through the clause gadget.

The climbable vertical structure gives us a way to prevent line clears from happening before the climbable vertical structure is made accessible, and can be filled from the bottom upwards (see Figure \ref{fig:s_cvs_tiling}). Specific maneuvers described in Appendices \ref{subsec:s_maneuver1} and \ref{subsec:s_maneuver2} are used to fill and/or traverse through the climbable vertical structure.

Now that we have our gadgets, we describe the general structure; an example is given in Figure \ref{fig:strisgenstructure}. Given an instance of Planar Monotone Rectilinear 3SAT, we can take the planar rectilinear arrangement, rotate the entire arrangement 90 degrees with the first variable segments rotating towards the top, rotate each individual clause segment 90 degrees, and rearrange the clause segments and wires such that the arrangement is still planar and rectilinear. Once we have this arrangement, we can insert our variable and clause gadgets where there are variable and clause segments, connect the gadgets with wires, and add some broken-up wires below the last variable gadget that leads to the climbable vertical structure that extends upwards and to the left of the rest of the structure and stops two lines below the top line. We also add one additional copy of the top two lines of the variable gadget above the topmost variable gadget and some empty squares in the rightmost few columns to prevent the top two lines from being cleared prematurely. For the piece sequence, if there are $M$ empty squares in the general structure, we give the player $M/4$ $\SS$ pieces so that no part of any $\SS$ piece can be placed above the topmost row in the general structure if the player wants to clear the board.

\begin{figure}[!ht]
  \centering
  \includegraphics[width=320pt]{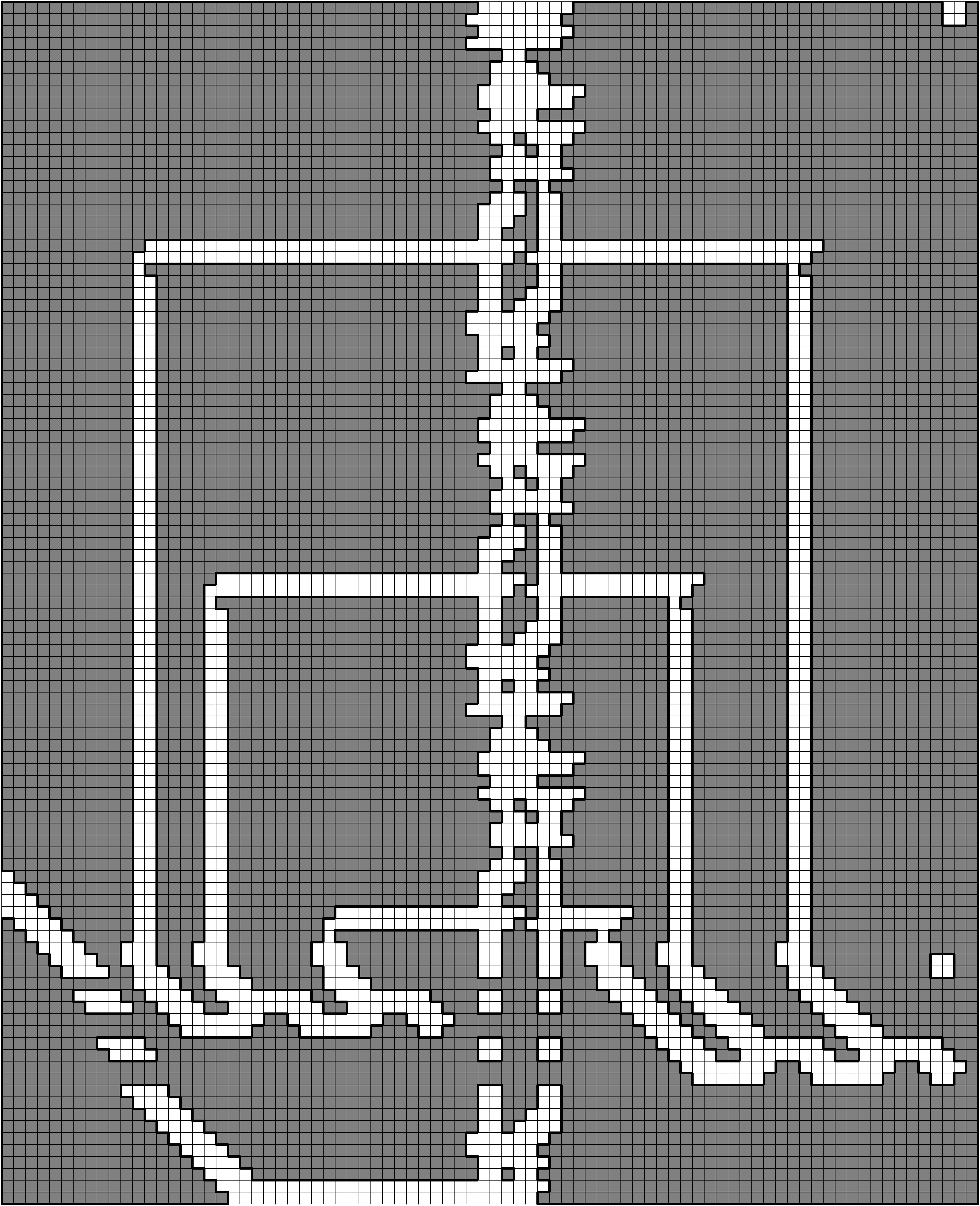}
  \caption{$\SS$-tris reduction general structure.}
  \label{fig:strisgenstructure}
\end{figure}

\subsection{Correctness}

Now we prove correctness of the construction; in other words, the whole construction can be cleared with $\SS$ pieces under SRS if and only if the Planar Monotone Rectilinear SAT instance has a solution.

First, if this construction can be cleared under SRS, then every clause must be accessible under some specific setting of the variable clauses. In particular, based on the way the structure is set up and the fact that no part of any $\SS$ piece can be placed above the topmost row of the structure, the top two lines must be cleared last (as the empty squares in the top two lines and the rightmost few columns line up with the other empty squares in the rightmost few columns), meaning that the only way for the climbable vertical structure to be accessible is if we clear the lines corresponding to the lines the clause gadgets can clear. This means that there must be a specific way we can set the variable clauses that leads to every clause being accessible. In other words, there is an assignment of variables such that each clause has at least one true literal, meaning that the Planar Monotone Rectilinear SAT instance has a solution.

Now, suppose we have a solution to the Planar Monotone Rectilinear SAT instance. Using the solution, we perform the following to clear the construction:

\begin{enumerate}
    \item Set the states of the variable gadgets (True, False) based on the solution.
    \item Fill all the clause gadgets in one of three ways as shown in Figure \ref{fig:s_cl_filling} (this is possible as the solution guarantees that all clause gadgets are accessible), causing some lines to be cleared and the bottom portion of the structure to be accessible.
    \item Fill the climbable vertical structure, now that the bottom portion of the structure is accessible. This allows for the lines containing the variable gadgets to be cleared.
    \item Starting from the topmost variable, fill and clear out the variable gadgets along with the relevant wires to the clause gadgets, as shown in Figure \ref{fig:s_var_filling} (this is now possible as the lines containing the variable gadgets can be cleared).
    \item Clear the 8 squares in the rightmost few columns and the remaining empty squares in the clause gadgets (if necessary).
\end{enumerate}

This method guarantees that we can clear the construction given a valid solution to the Planar Monotone Rectilinear SAT instance. Thus, if the Planar Monotone Rectilinear SAT instance has a solution, then the whole construction can be cleared with $\SS$ pieces under SRS.

As this construction can be created in polynomial time given a Planar Monotone Rectilinear SAT instance, and the whole construction can be cleared with $\SS$ pieces under SRS if and only if the Planar Monotone Rectilinear SAT instance has a solution, we have thus shown that Tetris clearing with SRS is NP-hard even if the player is only given $\SS$ pieces, as desired.

By symmetry, Tetris clearing with SRS is NP-hard even if the player is only given $\ZZ$ pieces.

\section{Survival}\label{sec:survival}

\begin{theorem}
    For any piece type $P\in \{\II, \JJ, \LL, \SS, \TT, \ZZ\}$, Tetris survival with SRS is NP-hard even if the type of pieces in the sequence given to the player is restricted to just $P$.
\end{theorem}

We first discuss the $\II$-tris case. Much of the proof remains the same as the proof for Tetris clearing with SRS and with just $\II$ pieces; however, the main modification we make is to the rightmost two columns of the construction, as shown in Figure~\ref{fig:structure_survival}. Here, we make the rightmost two columns have a checkerboard-like pattern to prevent any $\II$ pieces from being placed in the rightmost two columns (otherwise the player will lose) and to prevent any lines from being cleared. In addition, compared to the Tetris clearing construction, we give the player one less $\II$ piece, so that the player can fill the entire board except for the empty squares in the rightmost two columns and can survive if and only if the corresponding 1-in-3SAT instance has a solution. With this modification, we get NP-hardness for Tetris survival with SRS even if the player is only given $\II$ pieces.

\begin{figure}[!ht]
    \centering
    \includegraphics[width=320pt]{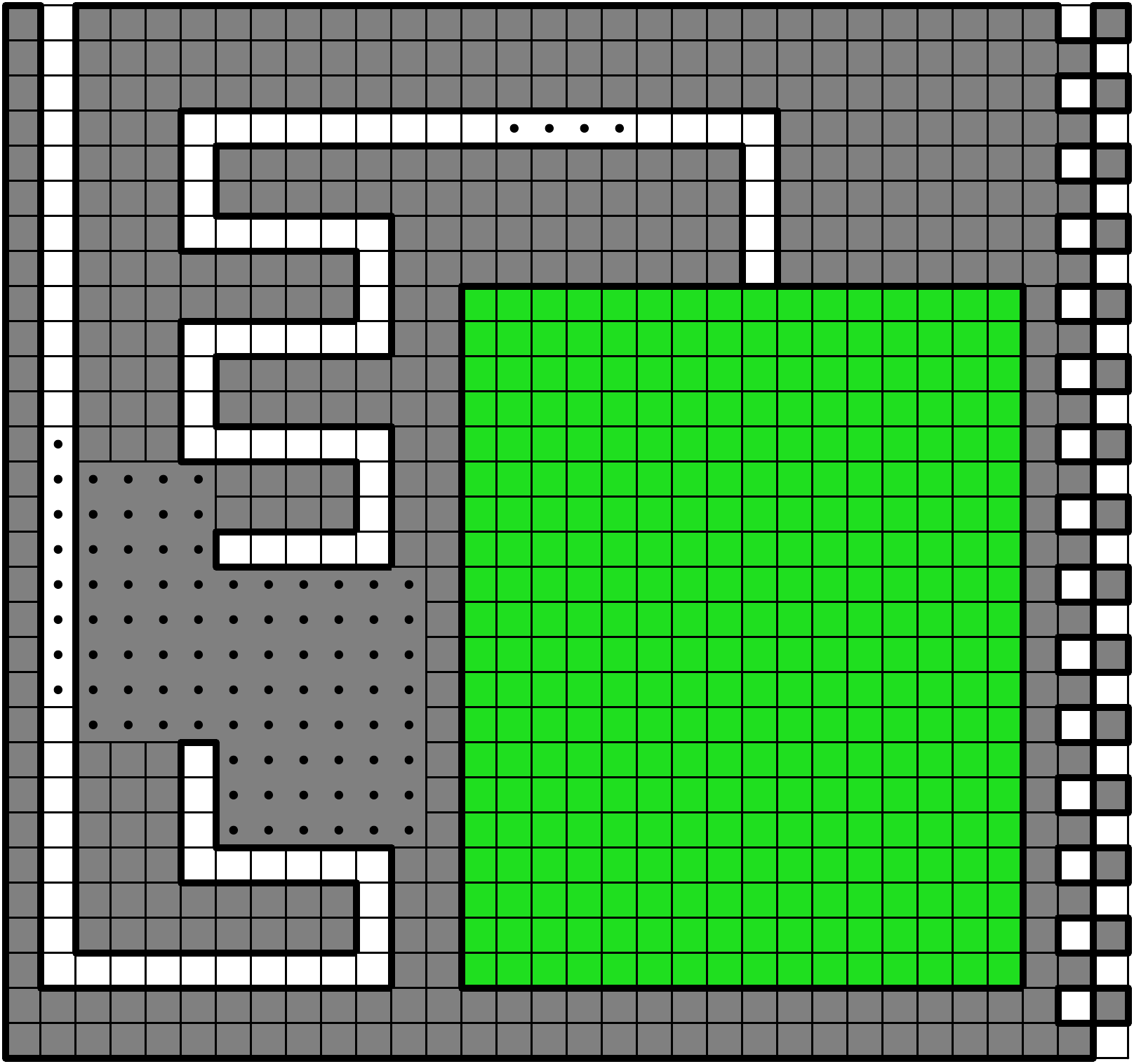}
    \caption{The general structure of the construction for Tetris survival using only $\II$ pieces. The main part of the construction is indicated in green.}
    \label{fig:structure_survival}
\end{figure}

A similar approach (modifying the rightmost two columns) yields NP-hardness for Tetris survival with SRS even if the player is only given pieces of type $P$ for $P\in \{\JJ, \LL, \TT\}$.

For $\SS$, in addition to modifying the rightmost two columns, we need to make additional modifications, as the general structure depends on line clears to access all parts of the structure, which we want to prevent in our reductions to Tetris survival. Here, we remove the pieces of the climbable vertical structure, and modify each clause gadget as shown in Figure~\ref{fig:s_clause_survival} such that each clause gadget is attached to a $(2C_1+2)\times (2C_2+2)$ rectangle of empty squares, where $C_1 = \Theta(C_2)$ and $C_2$ is sufficiently large that the number of empty squares not part of any of these rectangles is smaller than $C_1C_2$ (such $C_1$, $C_2$ exist as the number of such empty squares only grows linearly with $C_1+C_2$), lengthening any ``wires'' as necessary (i.e., any hallways of length or width 2 not part of a variable or clause gadget or the climbable vertical structure) so that the rectangles do not intersect each other or any existing gadgets. For the piece sequence, we give the player $mC_1C_2$ $\SS$ pieces in the input sequence (where $m$ is the number of clauses in the Planar Monotone Rectilinear SAT instance).

For correctness, note that because the number of empty squares not part of any of the rectangles attached to clause gadgets is smaller than $C_1C_2$, the only way for the player to be able to fit $mC_1C_2$ $\SS$ pieces without losing is if all $m$ $(2C_1+2)\times (2C_2+2)$ rectangles attached to clause gadgets contain at least $1$ $\SS$ piece, meaning that all $m$ rectangles, and thus all $m$ clause gadgets, must be accessible. Furthermore, if all $m$ clause gadgets are accessible, then each $(2C_1+2)\times (2C_2+2)$ rectangle can easily fit $C_1C_2$ $\SS$ pieces, meaning that all $mC_1C_2$ $\SS$ pieces can be placed without losing. Thus, we get NP-hardness for Tetris survival with SRS even if the player is only given $\SS$ pieces. We also get NP-hardness for $\ZZ$ by symmetry.

\begin{figure}[!ht]
    \centering
    \includegraphics[width=240pt]{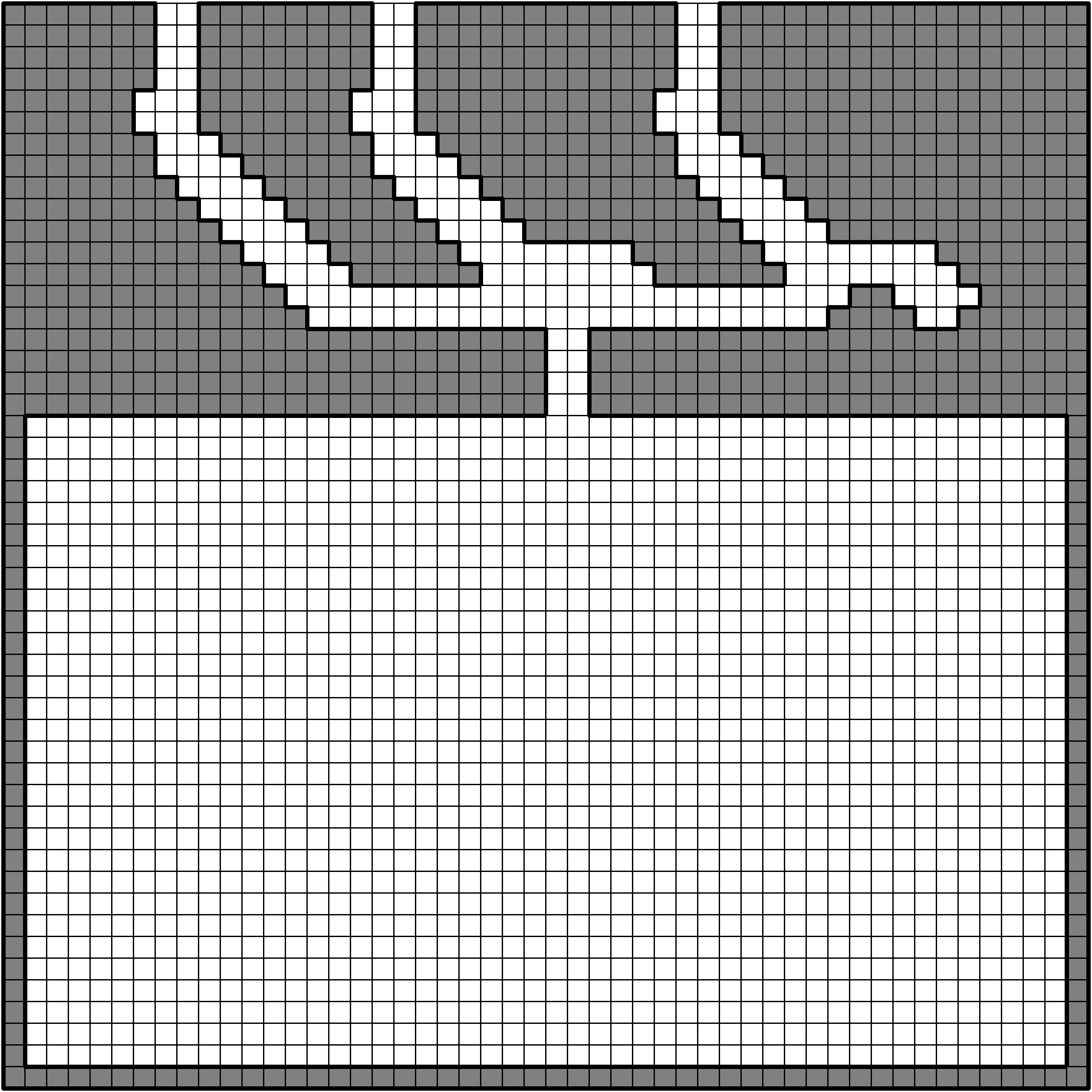}
    \caption{The (modified) clause gadget for Tetris survival with only $\SS$ pieces.}
    \label{fig:s_clause_survival}
\end{figure}

\section{Open Problems}\label{sec:openprobs}

Some of the open problems in \cite{TotalTetris_JIP} and \cite{mithg2024tetris} still remain open. For rotatable dominoes, one problem that remains open is whether Tetris clearing and survival remain solvable in polynomial time if we assume a different rotation model, such as a version of SRS for dominoes. In our polynomial-time algorithms for Tetris clearing and survival with rotatable dominoes, we assume that the rotation center only moves monotonically downwards. While the algorithm might be able to be adapted to different rotation systems that allow for the rotation center to move upwards, the details are much more tricky, and it is also possible that the problem is NP-hard.

For tetromino piece types, the main open problems of interest are the computational complexity of Tetris clearing and survival for $\OO$ pieces. For clearing, while $\OO$ pieces cannot rotate, thus simplifying the problem (as there are no spins involving $\OO$ pieces) and leading us to believe that this is plausibly in P, it is still not straightforward to prove due to the row-clearing mechanic. For survival, we have evidence that suggests that the problem is NP-hard, as $2\times 2$ square packing in orthogonally convex grid polygons is NP-hard \cite{abrahamsen2026hardness}, but there are still a few tricky details to work out, particularly in ensuring that all pieces can be successfully placed under Tetris rules.

Another open problem is whether or not Tetris clearing and survival with SRS and with just one piece type is $\#$P-hard or ASP-complete.\footnote{Recall that an NP search problem is \defn{ASP-complete} \cite{Yato-Seta-2003} if there is a parsimonious reduction from all NP search problems (including a polynomial-time bijection between solutions)} In particular, it is unlikely that our construction is $c$-monious\footnote{A reduction is \defn{$c$-monious} if it blows up the number of solutions by exactly a multiplicative factor of~$c$.} for any $c$ as each gadget type may not have the same number of ways to fill them depending on the tiling. Does there exist a different construction that is $c$-monious for some $c$?

Lastly, most of our results analyze Tetris clearing and survival with a finite input sequence. How does the complexity of Tetris clearing and survival change if the player is instead given an infinite number of pieces of a given type?

\section*{Acknowledgments}

This paper was initiated during open problem solving in the MIT class on Algorithmic Lower Bounds: Fun with Hardness Proofs (6.5440) taught by Erik Demaine in Fall 2023. We thank the other participants of that class --- specifically
Eugeniya Artemova,
Cecilia Chen,
Lily Chung,
Jenny Diomidova,
Holden Hall,
Jonathan Li,
Jayson Lynch,
Frederick Stock, and
Ethan Zhou
--- for helpful discussions and providing an inspiring atmosphere.

\bibliographystyle{alpha}
\bibliography{biblio-tetone}

\appendix
\clearpage

\section{$\II$ Piece Rotations Figure}

\begin{figure}[!ht]
    \centering
    \includegraphics[width=320pt]{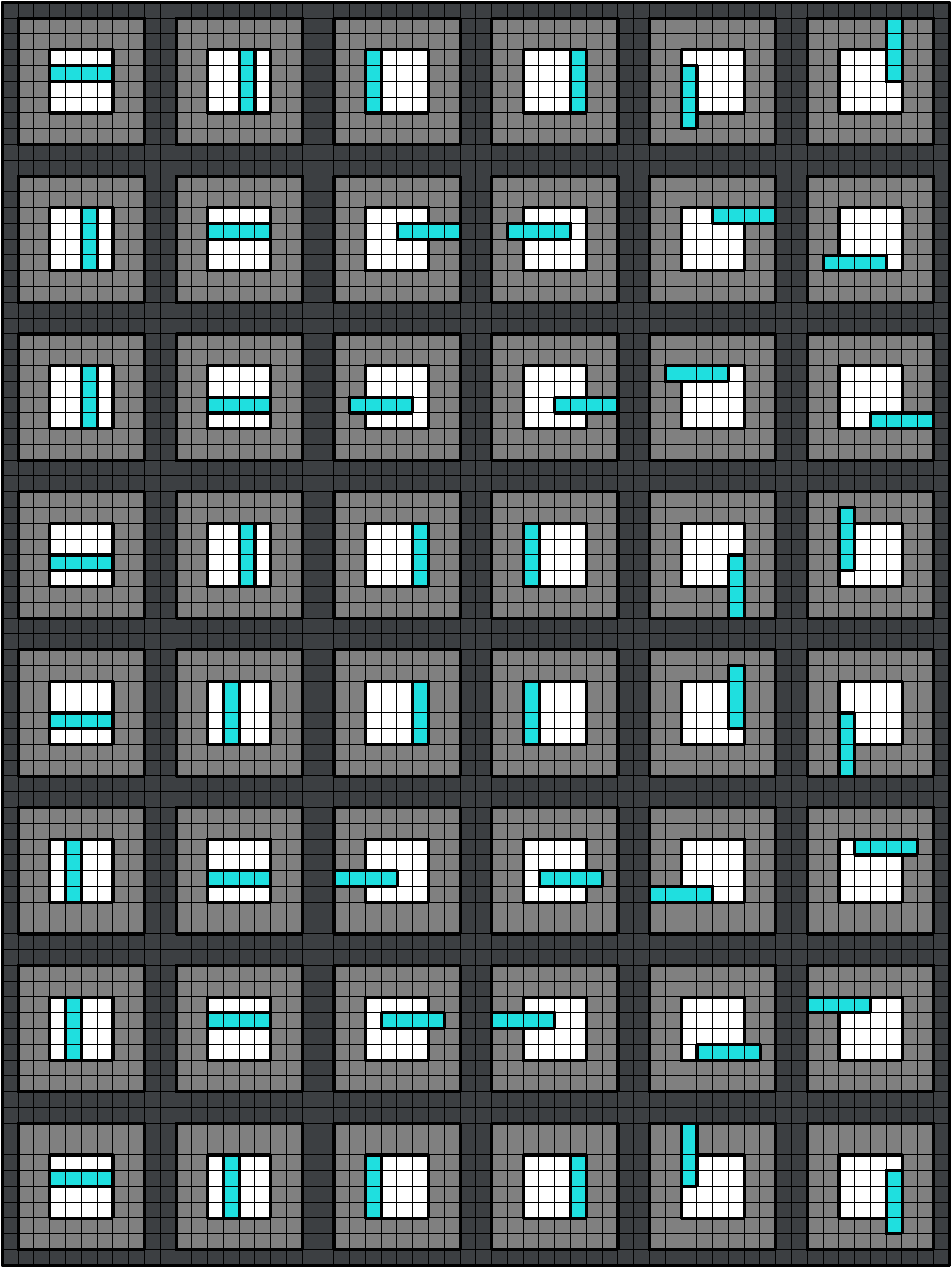}
    \caption{All tests for all rotations of $\II$ pieces. From top to bottom, the rotations are: $0\to R$, $R\to 0$, $R\to 2$, $2\to R$, $2\to L$, $L\to 2$, $L\to 0$, $0\to L$. $0$ indicates the default orientation, and $R$, $2$, and $L$ indicate the orientation reached from a $90^\circ$, $180^\circ$, and $270^\circ$ rotation clockwise (respectively) from the default orientation.}
    \label{fig:allrots}
\end{figure}

\section{Specific $\II$-tris Maneuvers}\label{sec:i_maneuvers}

\subsection{Maneuver 1}\label{subsec:i_maneuver1}

This maneuver, indicated in Figure~\ref{fig:i_maneuver1}, is used to get $\II$ pieces through the long tunnel and into the main part of the construction (through the EC gadget). This maneuver also roughly illustrates how $\II$ pieces can "turn corners". The $\II$ piece can be in either horizontal orientation (the default orientation or the $180^\circ$-rotated orientation) in (a), and in either horizontal orientation in (h), although whichever orientation the $\II$ piece is in in (a) will influence the orientation the $\II$ piece is in in (h).

\begin{figure}[!ht]
  \begin{subfigure}[b]{0.245\textwidth}
    \centering
    \includegraphics[width=80pt]{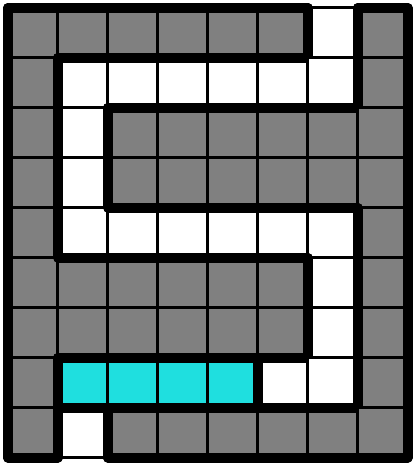}
    \caption{}
  \end{subfigure}
  \begin{subfigure}[b]{0.245\textwidth}
    \centering
    \includegraphics[width=80pt]{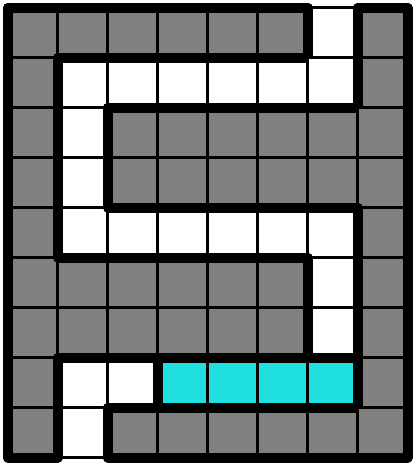}
    \caption{}
  \end{subfigure}
  \begin{subfigure}[b]{0.245\textwidth}
    \centering
    \includegraphics[width=80pt]{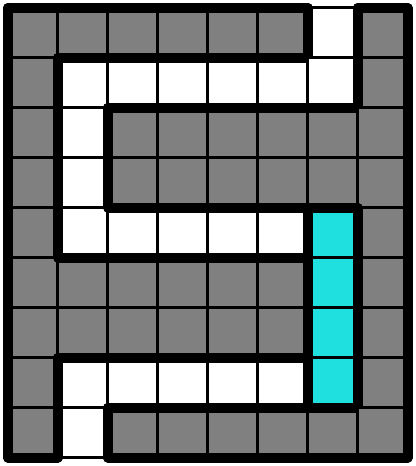}
    \caption{}
  \end{subfigure}
  \begin{subfigure}[b]{0.245\textwidth}
    \centering
    \includegraphics[width=80pt]{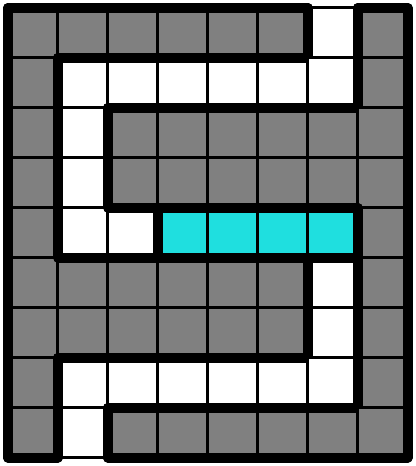}
    \caption{}
  \end{subfigure}
  \begin{subfigure}[b]{0.245\textwidth}
    \centering
    \includegraphics[width=80pt]{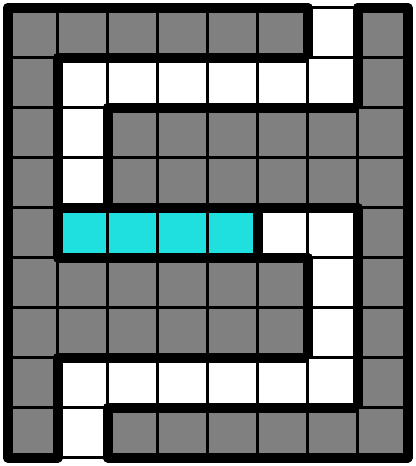}
    \caption{}
  \end{subfigure}
  \begin{subfigure}[b]{0.245\textwidth}
    \centering
    \includegraphics[width=80pt]{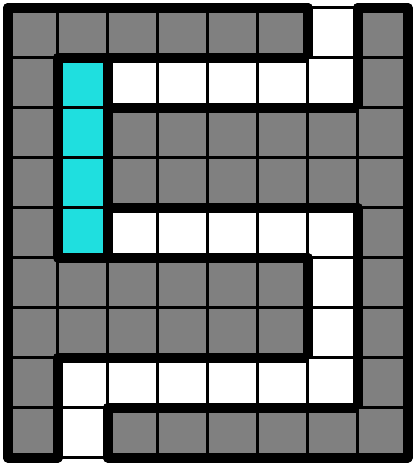}
    \caption{}
  \end{subfigure}
  \begin{subfigure}[b]{0.245\textwidth}
    \centering
    \includegraphics[width=80pt]{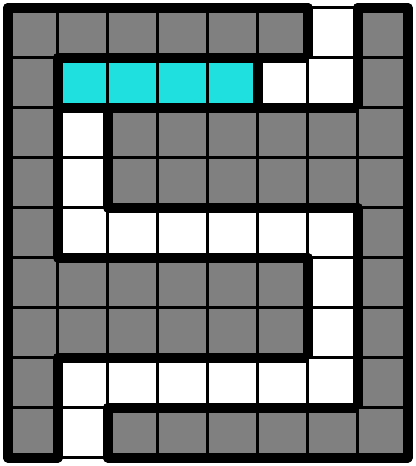}
    \caption{}
  \end{subfigure}
  \begin{subfigure}[b]{0.245\textwidth}
    \centering
    \includegraphics[width=80pt]{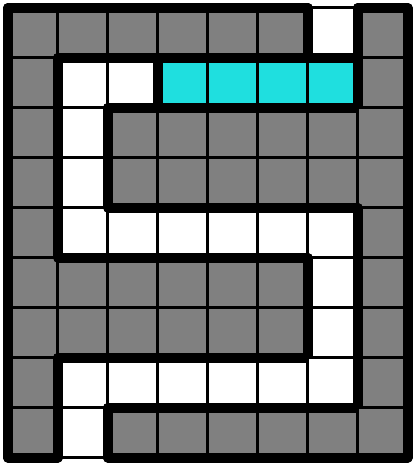}
    \caption{}
  \end{subfigure}
  \caption{Maneuver 1. This is reversible; i.e., both going from (a) to (h) and going from (h) to (a) are possible.}
  \label{fig:i_maneuver1}
\end{figure}

\subsection{Maneuver 2}\label{subsec:i_maneuver2}

This maneuver, indicated in Figure~\ref{fig:i_maneuver2}, or a similar maneuver (reflected about a vertical line) is used in case the $\II$ piece is in the incorrect orientation (i.e., the non-default horizontal orientation) when placing either piece 48 or piece 60 in a duplicator gadget.

\begin{figure}[!ht]
  \begin{subfigure}[b]{0.325\textwidth}
    \centering
    \includegraphics[width=100pt]{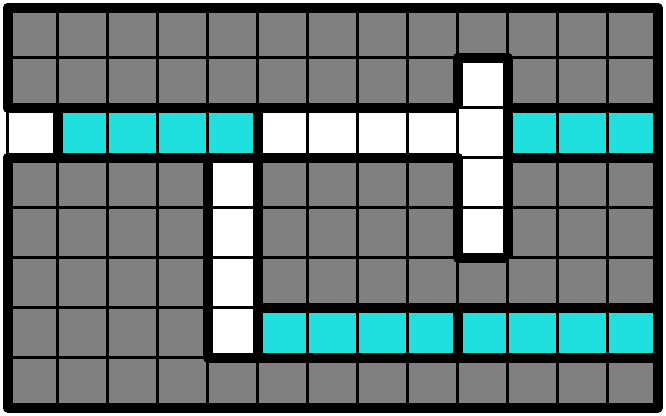}
    \caption{Piece is in $2$ orientation}
  \end{subfigure}
  \begin{subfigure}[b]{0.325\textwidth}
    \centering
    \includegraphics[width=100pt]{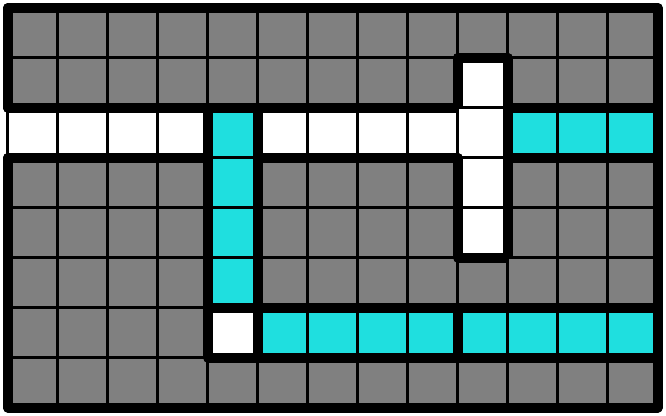}
    \caption{Piece is in $R$ orientation}
  \end{subfigure}
  \begin{subfigure}[b]{0.325\textwidth}
    \centering
    \includegraphics[width=100pt]{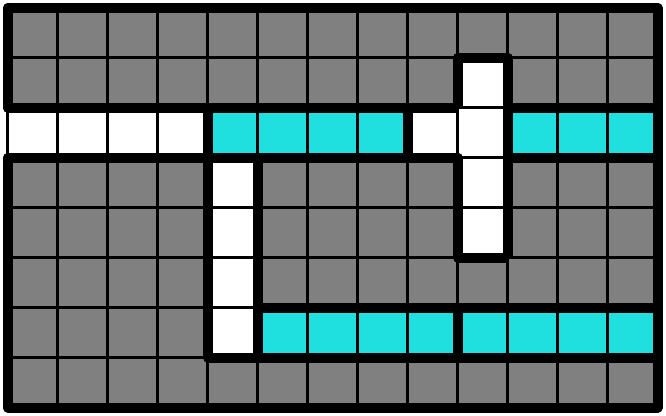}
    \caption{Piece is in $0$ orientation}
  \end{subfigure}
  \begin{subfigure}[b]{0.325\textwidth}
    \centering
    \includegraphics[width=100pt]{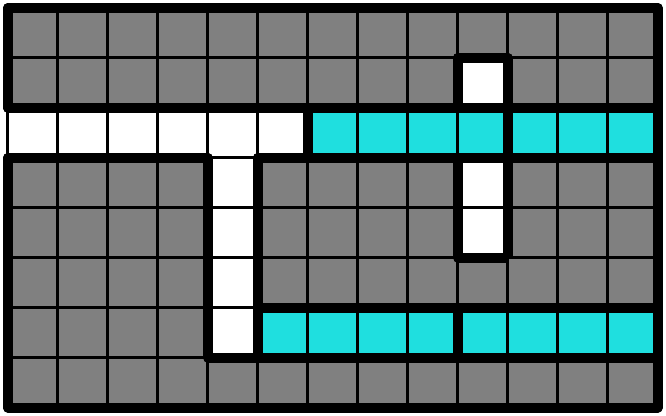}
    \caption{Piece is in $0$ orientation}
  \end{subfigure}
  \begin{subfigure}[b]{0.325\textwidth}
    \centering
    \includegraphics[width=100pt]{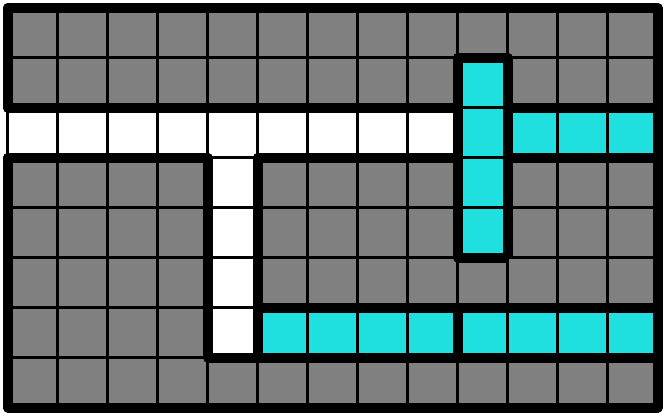}
    \caption{Piece is in $L$ or $R$ orientation}
  \end{subfigure}
  \caption{Maneuver 2. Captions indicate the orientation of the $\II$ piece. $0$ indicates the default orientation, and $R$, $2$, and $L$ indicate the orientation reached from a $90^\circ$, $180^\circ$, and $270^\circ$ rotation clockwise (respectively) from the default orientation.}
  \label{fig:i_maneuver2}
\end{figure}

\subsection{Maneuver 3}\label{subsec:i_maneuver3}

This maneuver, indicated in Figure~\ref{fig:i_maneuver3}, is used to get $\II$ pieces to the location indicated by piece number 15 in the first and second tilings of the negated-clause gadget.

\begin{figure}[!ht]
  \begin{subfigure}[b]{0.245\textwidth}
    \centering
    \includegraphics[width=80pt]{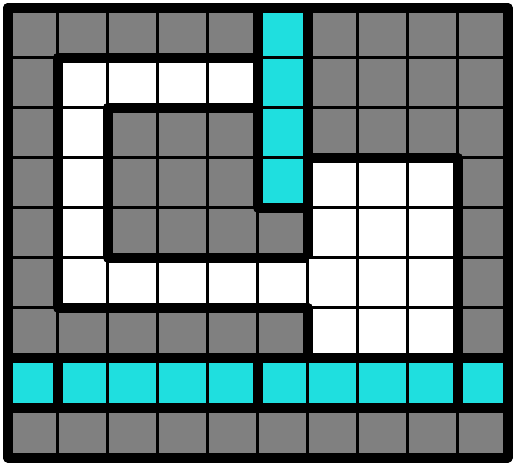}
    \caption{Piece is in $L$ or $R$ orientation}
  \end{subfigure}
  \begin{subfigure}[b]{0.245\textwidth}
    \centering
    \includegraphics[width=80pt]{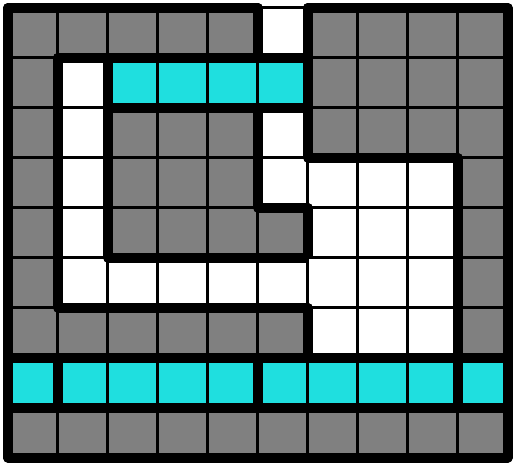}
    \caption{Piece is in $0$ orientation}
  \end{subfigure}
  \begin{subfigure}[b]{0.245\textwidth}
    \centering
    \includegraphics[width=80pt]{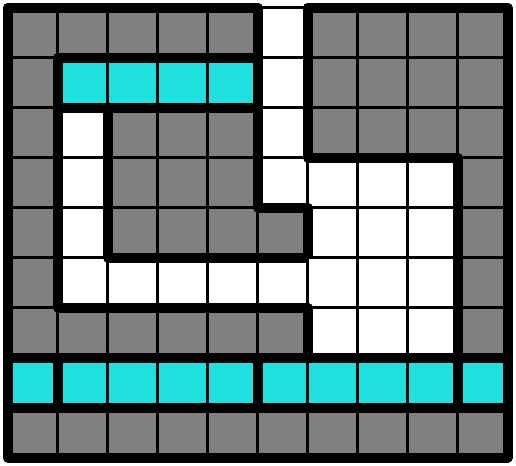}
    \caption{Piece is in $0$ orientation}
  \end{subfigure}
  \begin{subfigure}[b]{0.245\textwidth}
    \centering
    \includegraphics[width=80pt]{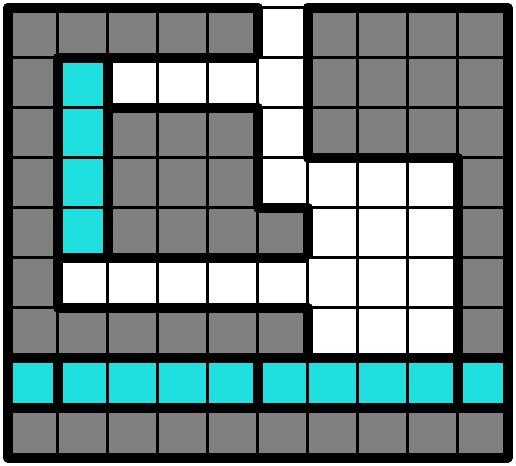}
    \caption{Piece is in $R$ orientation}
  \end{subfigure}
  \begin{subfigure}[b]{0.245\textwidth}
    \centering
    \includegraphics[width=80pt]{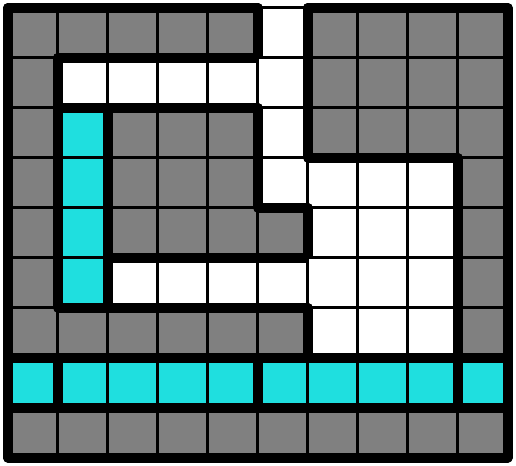}
    \caption{Piece is in $R$ orientation}
  \end{subfigure}
  \begin{subfigure}[b]{0.245\textwidth}
    \centering
    \includegraphics[width=80pt]{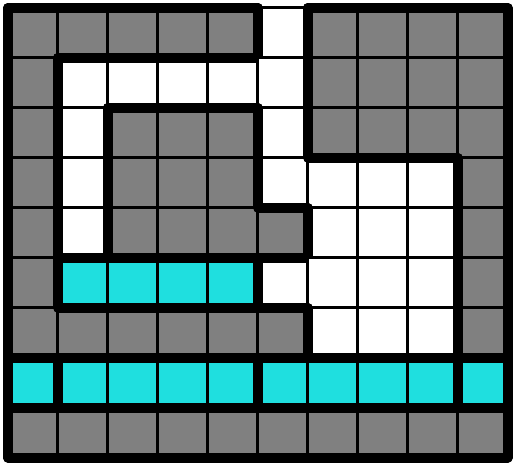}
    \caption{Piece is in $2$ orientation}
  \end{subfigure}
  \begin{subfigure}[b]{0.245\textwidth}
    \centering
    \includegraphics[width=80pt]{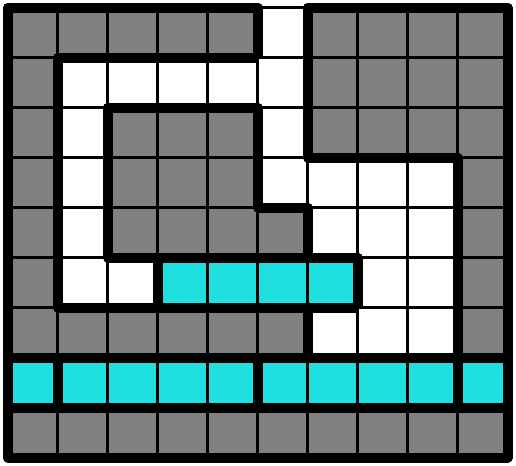}
    \caption{Piece is in $2$ orientation}
  \end{subfigure}
  \begin{subfigure}[b]{0.245\textwidth}
    \centering
    \includegraphics[width=80pt]{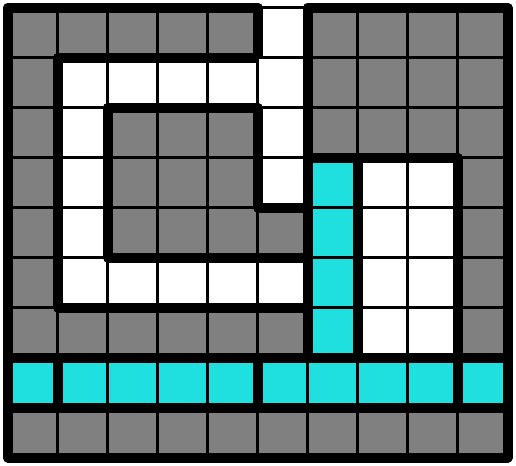}
    \caption{Piece is in $L$ or $R$ orientation}
  \end{subfigure}
  \caption{Maneuver 3. Captions indicate the orientation of the $\II$ piece. $0$ indicates the default orientation, and $R$, $2$, and $L$ indicate the orientation reached from a $90^\circ$, $180^\circ$, and $270^\circ$ rotation clockwise (respectively) from the default orientation.}
  \label{fig:i_maneuver3}
\end{figure}

\subsection{Maneuver 4}\label{subsec:i_maneuver4}

This maneuver, indicated in Figure~\ref{fig:i_maneuver4}, and similar maneuvers are used to get $\II$ pieces through the middle section (particularly the chambers) of a crossover gadget.

The figure shows the maneuver for the middle chamber, although step (e) requires one downward movement for the $\II$ piece. Steps (b) through (d) may be more rigid (either $R\to 0\to R$ or $L\to 2\to L$ is forced) depending on whether the chamber already has other pieces in it, but getting the $\II$ piece to the position indicated in (e) is possible regardless of the orientation ($L$ or $R$) in step (d). For the other chambers, the steps are similar.

\begin{figure}[!ht]
  \begin{subfigure}[b]{0.245\textwidth}
    \centering
    \includegraphics[width=100pt]{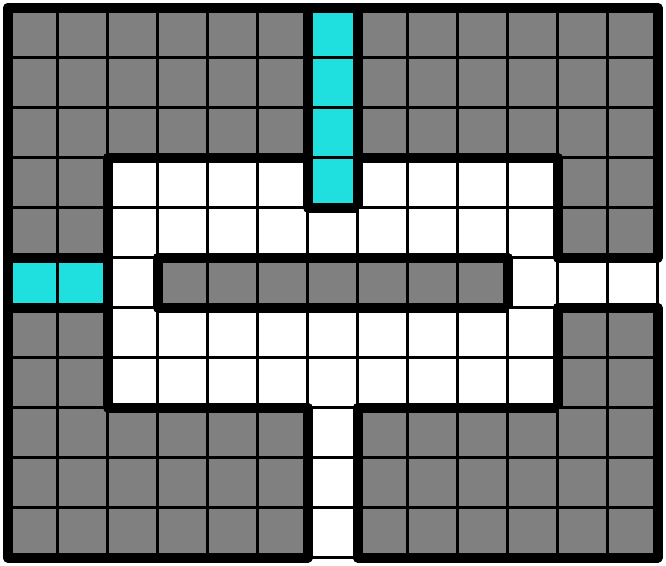}
    \caption{Piece is in $L$ or $R$ orientation}
  \end{subfigure}
  \begin{subfigure}[b]{0.245\textwidth}
    \centering
    \includegraphics[width=100pt]{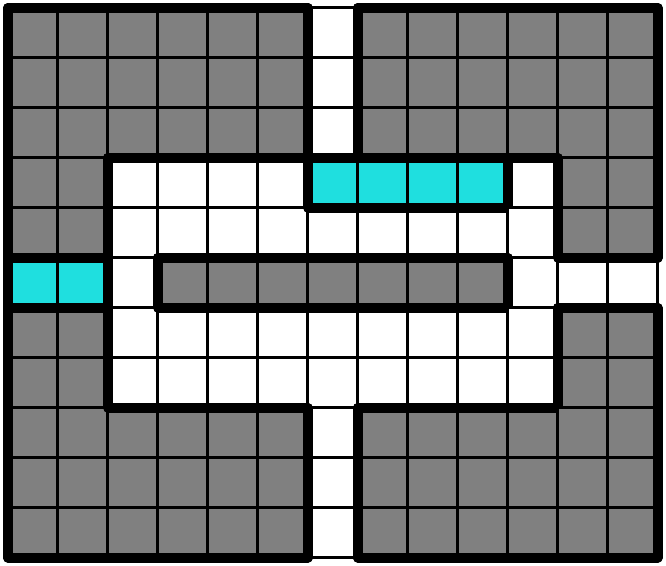}
    \caption{Possible intermediate step}
  \end{subfigure}
  \begin{subfigure}[b]{0.245\textwidth}
    \centering
    \includegraphics[width=100pt]{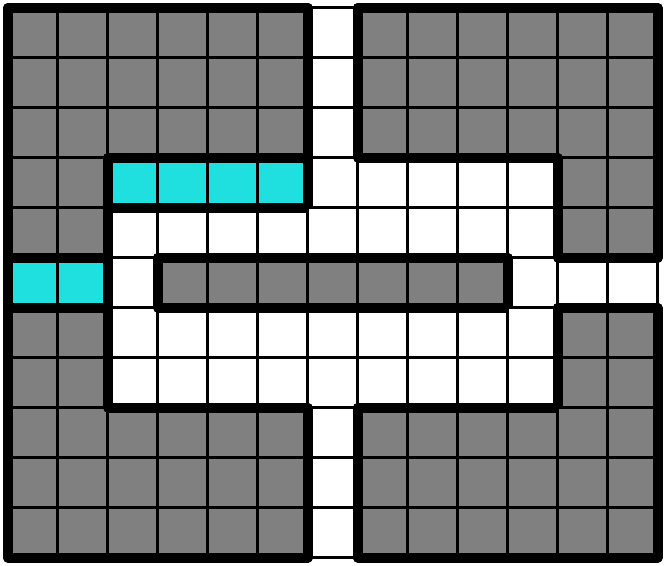}
    \caption{Piece is in $0$ orientation}
  \end{subfigure}
  \begin{subfigure}[b]{0.245\textwidth}
    \centering
    \includegraphics[width=100pt]{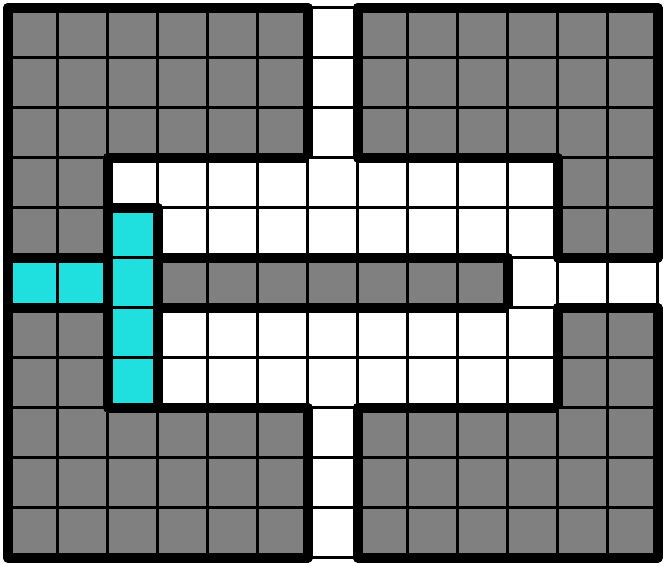}
    \caption{Piece is in $R$ orientation}
  \end{subfigure}
  \begin{subfigure}[b]{0.245\textwidth}
    \centering
    \includegraphics[width=100pt]{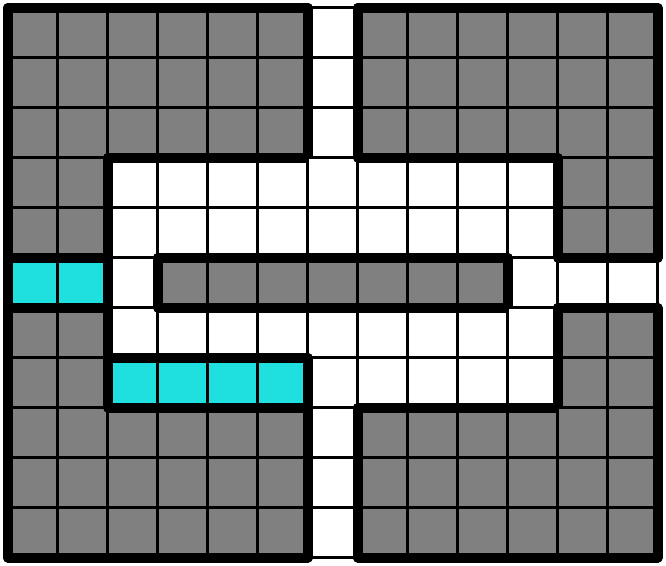}
    \caption{Piece is in $2$ orientation}
  \end{subfigure}
  \begin{subfigure}[b]{0.245\textwidth}
    \centering
    \includegraphics[width=100pt]{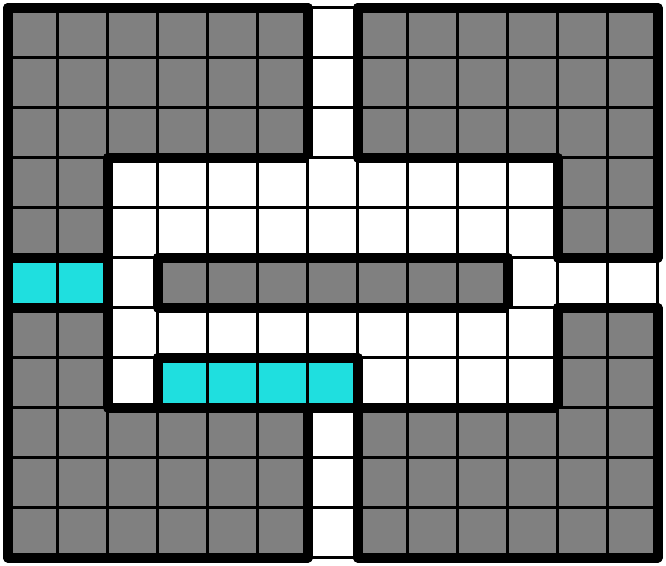}
    \caption{Piece is in $2$ orientation}
  \end{subfigure}
  \begin{subfigure}[b]{0.245\textwidth}
    \centering
    \includegraphics[width=100pt]{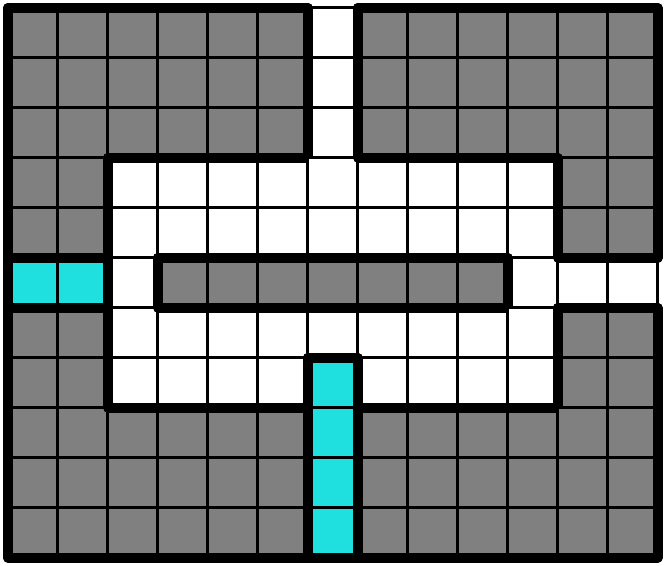}
    \caption{Piece is in $R$ orientation}
  \end{subfigure}
  \caption{Maneuver 4. Captions indicate the orientation of the $\II$ piece. $0$ indicates the default orientation, and $R$, $2$, and $L$ indicate the orientation reached from a $90^\circ$, $180^\circ$, and $270^\circ$ rotation clockwise (respectively) from the default orientation.}
  \label{fig:i_maneuver4}
\end{figure}

\pagebreak

\section{Specific $\JJ$-tris Maneuvers}\label{sec:j_maneuvers}

\subsection{Maneuver 1}\label{subsec:j_maneuver1}

This maneuver, indicated in Figure~\ref{fig:j_maneuver1}, is used to move $\JJ$ pieces upwards through the climbable vertical structure.

\begin{figure}[!ht]
  \centering
  \begin{subfigure}[b]{0.16\textwidth}
    \centering
    \includegraphics[width=70pt]{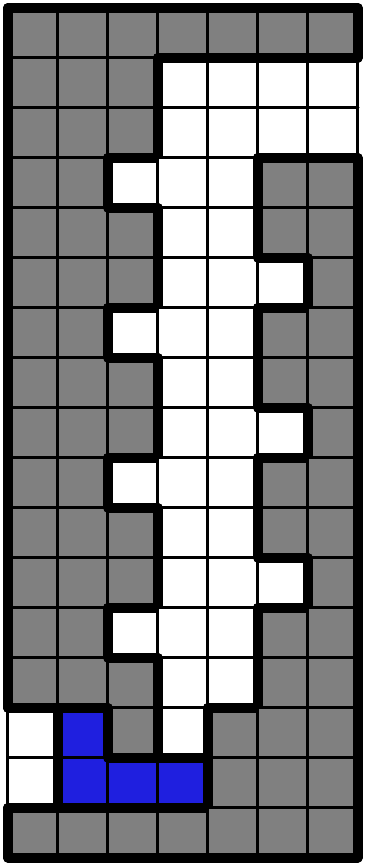}
    \caption{}
  \end{subfigure}
  \begin{subfigure}[b]{0.16\textwidth}
    \centering
    \includegraphics[width=70pt]{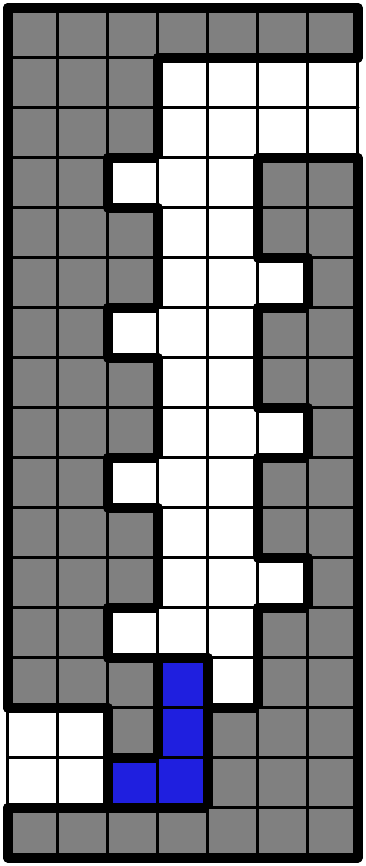}
    \caption{}
  \end{subfigure}
  \begin{subfigure}[b]{0.16\textwidth}
    \centering
    \includegraphics[width=70pt]{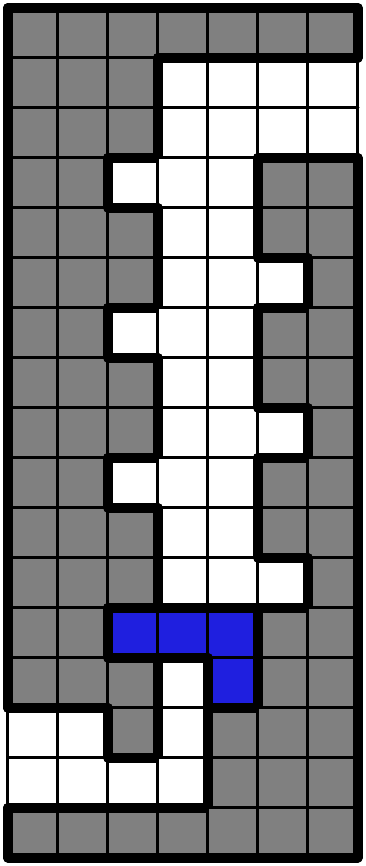}
    \caption{}
  \end{subfigure}
  \begin{subfigure}[b]{0.16\textwidth}
    \centering
    \includegraphics[width=70pt]{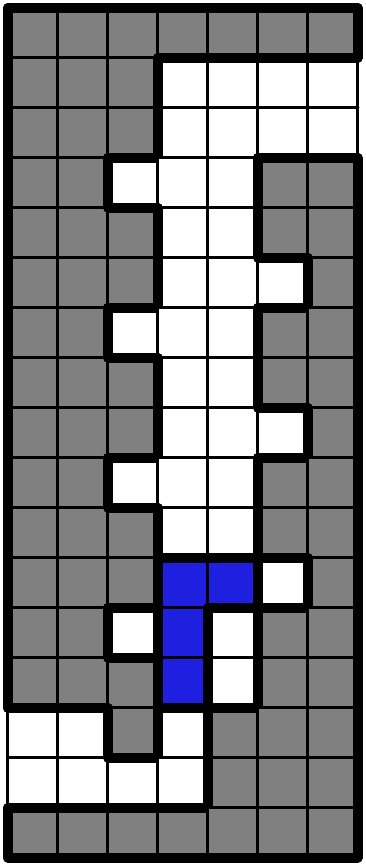}
    \caption{}
  \end{subfigure}
  \begin{subfigure}[b]{0.16\textwidth}
    \centering
    \includegraphics[width=70pt]{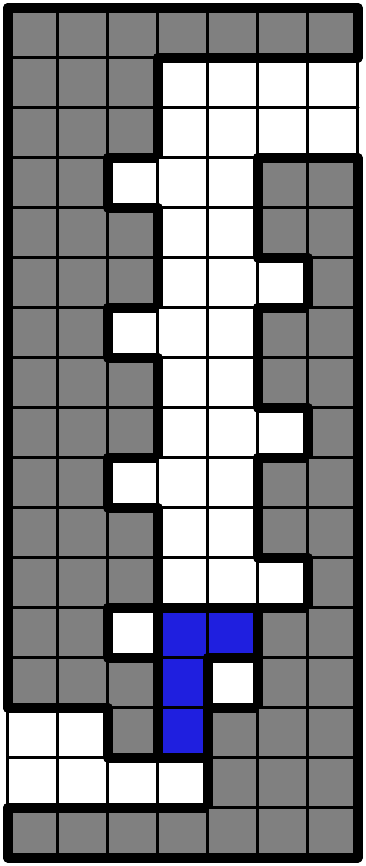}
    \caption{}
  \end{subfigure}
  \begin{subfigure}[b]{0.16\textwidth}
    \centering
    \includegraphics[width=70pt]{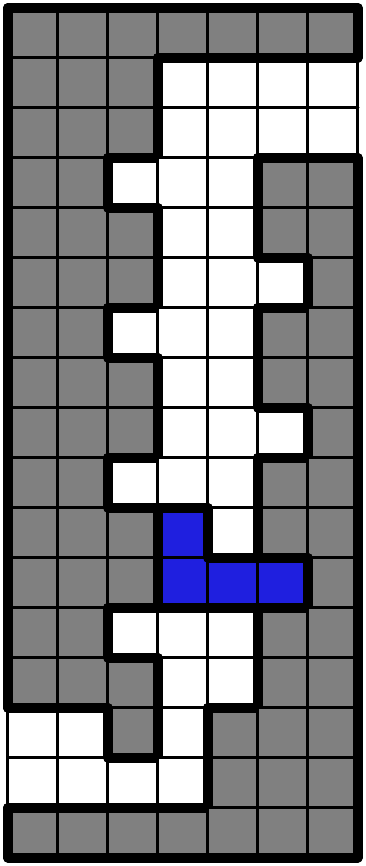}
    \caption{}
  \end{subfigure}
  \begin{subfigure}[b]{0.16\textwidth}
    \centering
    \includegraphics[width=70pt]{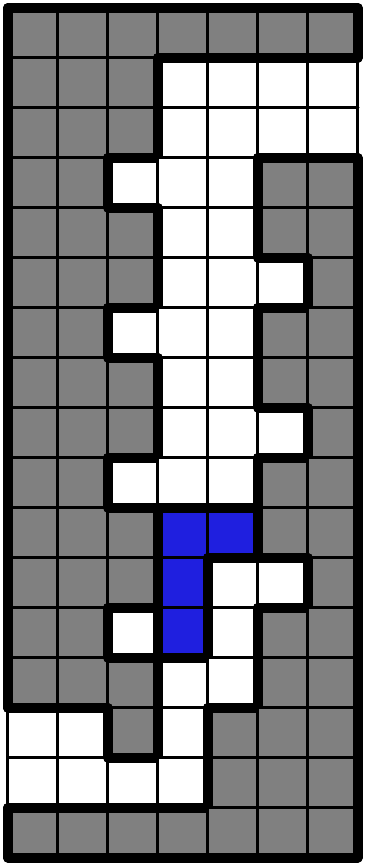}
    \caption{}
  \end{subfigure}
  \begin{subfigure}[b]{0.16\textwidth}
    \centering
    \includegraphics[width=70pt]{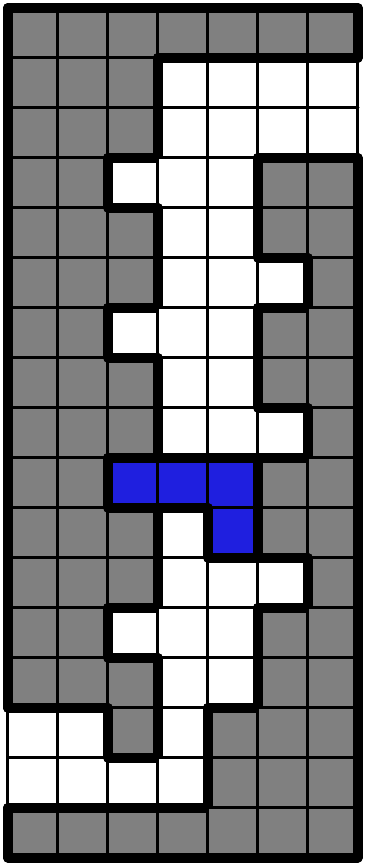}
    \caption{}
  \end{subfigure}
  \begin{subfigure}[b]{0.16\textwidth}
    \centering
    \includegraphics[width=70pt]{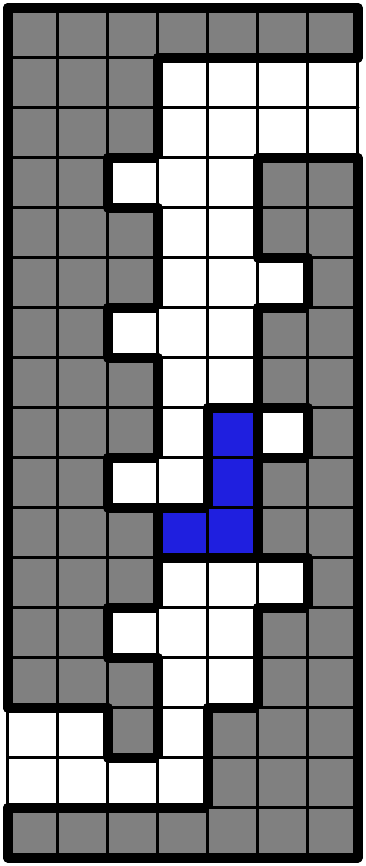}
    \caption{}
  \end{subfigure}
  \begin{subfigure}[b]{0.16\textwidth}
    \centering
    \includegraphics[width=70pt]{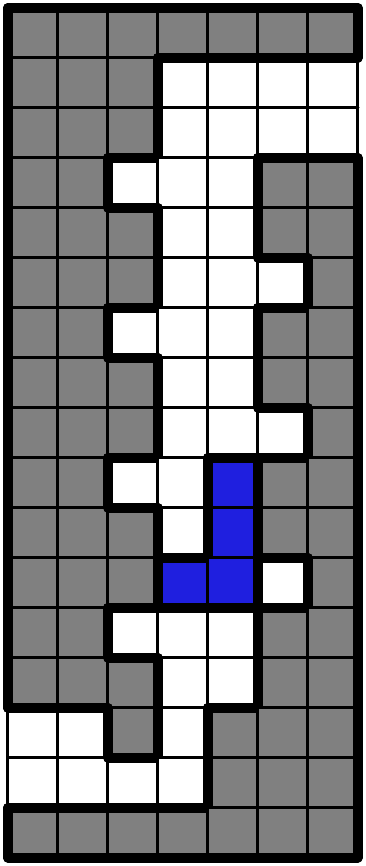}
    \caption{}
  \end{subfigure}
  \begin{subfigure}[b]{0.16\textwidth}
    \centering
    \includegraphics[width=70pt]{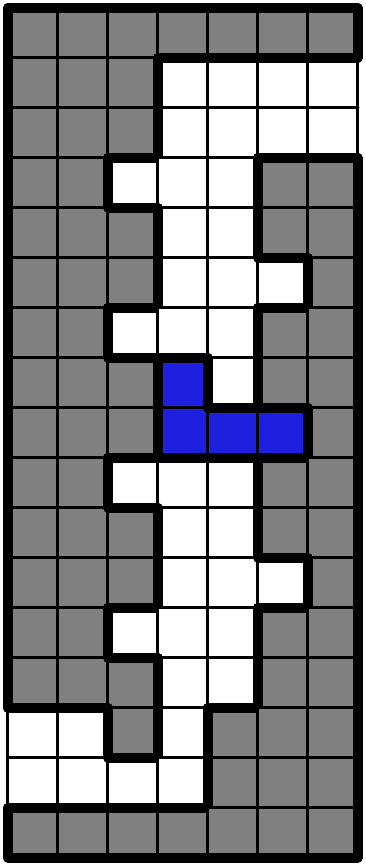}
    \caption{}
  \end{subfigure}
  \caption{Maneuver 1. (a) through (f) are used to get the $\JJ$ piece into the climbable vertical structure, while (f) through (k) are used to climb one segment of the climbable vertical structure.}
  \label{fig:j_maneuver1}
\end{figure}

\subsection{Maneuver 2}\label{subsec:j_maneuver2}

This maneuver, indicated in Figure~\ref{fig:j_maneuver2}, is used to turn corners in some of the gadgets.

\begin{figure}[!ht]
  \begin{subfigure}[b]{0.49\textwidth}
    \centering
    \includegraphics[width=80pt]{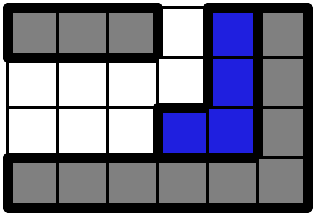}
    \caption{}
  \end{subfigure}
  \begin{subfigure}[b]{0.49\textwidth}
    \centering
    \includegraphics[width=80pt]{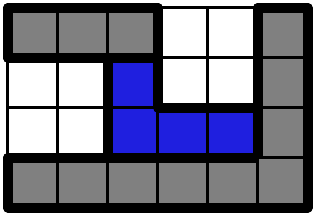}
    \caption{}
  \end{subfigure}
  \caption{Maneuver 2}
  \label{fig:j_maneuver2}
\end{figure}

\subsection{Maneuver 3}\label{subsec:j_maneuver3}

This maneuver, indicated in Figure~\ref{fig:j_maneuver3}, is used to tuck some of the $\JJ$ pieces into their positions in parity fixer gadgets and the entry corner gadget.

\begin{figure}[!ht]
  \begin{subfigure}[b]{0.49\textwidth}
    \centering
    \includegraphics[width=80pt]{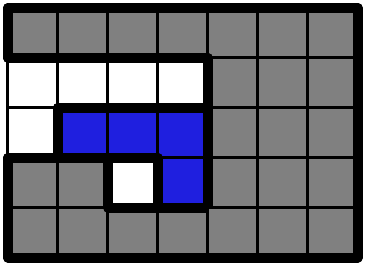}
    \caption{}
  \end{subfigure}
  \begin{subfigure}[b]{0.49\textwidth}
    \centering
    \includegraphics[width=80pt]{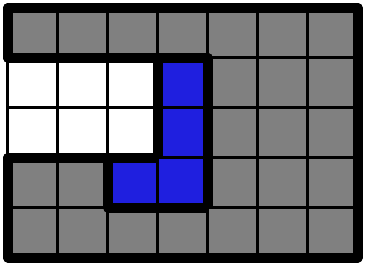}
    \caption{}
  \end{subfigure}
  \caption{Maneuver 3}
  \label{fig:j_maneuver3}
\end{figure}

\pagebreak

\section{Specific $\TT$-tris Maneuvers}\label{sec:t_maneuvers}

We use these maneuvers to push $\TT$ pieces upwards or tuck $\TT$ pieces into specific locations in the gadgets or the climbable vertical structure. These maneuvers are symmetric; we show them for rotations to and from the orientation reached by a $90^\circ$ counterclockwise rotation from the default orientation, and analogous maneuvers exist for the orientation reached by a $90^\circ$ clockwise rotation from the default orientation.

\begin{figure}[!ht]
  \begin{subfigure}[b]{0.49\textwidth}
    \centering
    \includegraphics[width=80pt]{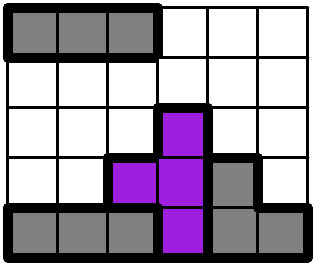}
    \caption{}
  \end{subfigure}
  \begin{subfigure}[b]{0.49\textwidth}
    \centering
    \includegraphics[width=80pt]{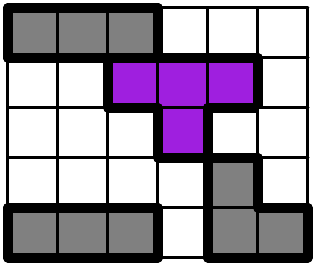}
    \caption{}
  \end{subfigure}
  \caption{Maneuver 1}
  \label{fig:t_maneuver1}
\end{figure}

\begin{figure}[!ht]
  \begin{subfigure}[b]{0.49\textwidth}
    \centering
    \includegraphics[width=80pt]{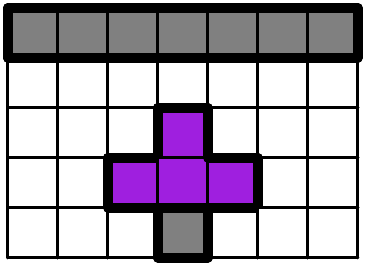}
    \caption{}
  \end{subfigure}
  \begin{subfigure}[b]{0.49\textwidth}
    \centering
    \includegraphics[width=80pt]{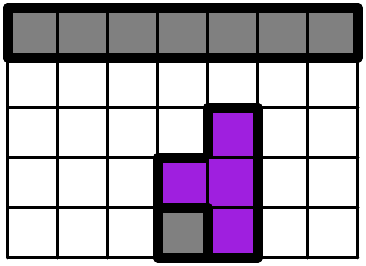}
    \caption{}
  \end{subfigure}
  \caption{Maneuver 2}
  \label{fig:t_maneuver2}
\end{figure}

\begin{figure}[!ht]
  \begin{subfigure}[b]{0.49\textwidth}
    \centering
    \includegraphics[width=80pt]{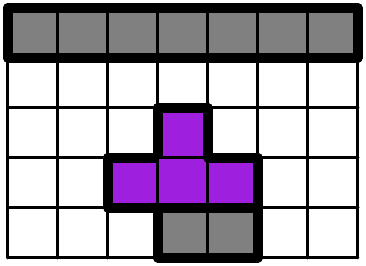}
    \caption{}
  \end{subfigure}
  \begin{subfigure}[b]{0.49\textwidth}
    \centering
    \includegraphics[width=80pt]{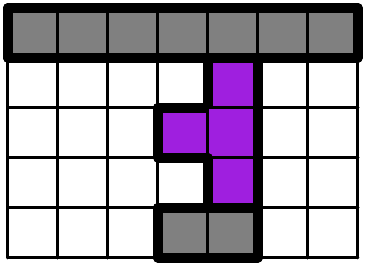}
    \caption{}
  \end{subfigure}
  \caption{Maneuver 3}
  \label{fig:t_maneuver3}
\end{figure}

\pagebreak

\section{Specific $\SS$-tris Maneuvers}\label{sec:s_maneuvers}

The square indicated with a dot is the center of the $\SS$ piece.

\subsection{Maneuver 1}\label{subsec:s_maneuver1}

This maneuver, indicated in Figure~\ref{fig:s_maneuver1}, is used to move $\SS$ pieces upwards through the clause gadgets and through the climbable vertical structure.

\begin{figure}[!ht]
  \begin{subfigure}[b]{0.49\textwidth}
    \centering
    \includegraphics[width=80pt]{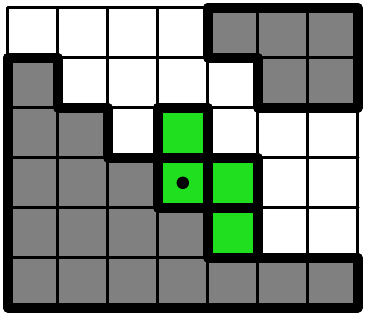}
    \caption{}
  \end{subfigure}
  \begin{subfigure}[b]{0.49\textwidth}
    \centering
    \includegraphics[width=80pt]{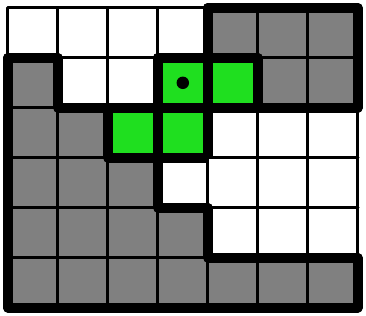}
    \caption{}
  \end{subfigure}
  \caption{Maneuver 1}
  \label{fig:s_maneuver1}
\end{figure}

\subsection{Maneuver 2}\label{subsec:s_maneuver2}

This maneuver, indicated in Figure~\ref{fig:s_maneuver2}, is used to fit some of the $\SS$ pieces into their positions in the clause gadgets and the climbable vertical structure.

\begin{figure}[!ht]
  \begin{subfigure}[b]{0.245\textwidth}
    \centering
    \includegraphics[width=80pt]{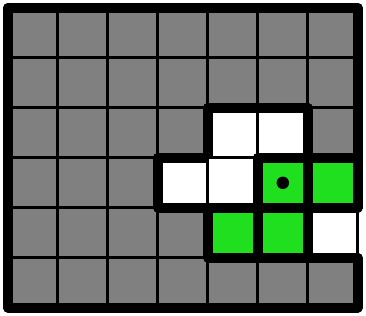}
    \caption{}
  \end{subfigure}
  \begin{subfigure}[b]{0.245\textwidth}
    \centering
    \includegraphics[width=80pt]{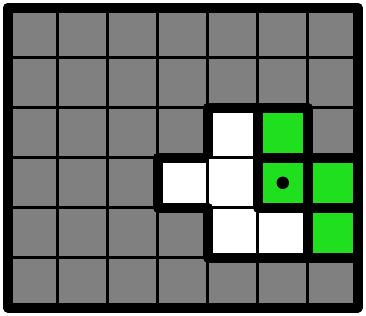}
    \caption{}
  \end{subfigure}
  \begin{subfigure}[b]{0.245\textwidth}
    \centering
    \includegraphics[width=80pt]{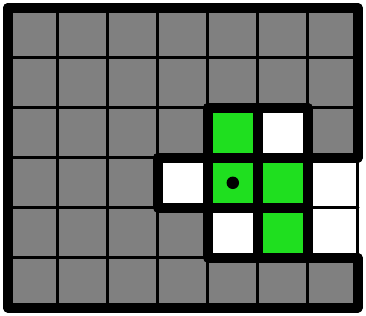}
    \caption{}
  \end{subfigure}
  \begin{subfigure}[b]{0.245\textwidth}
    \centering
    \includegraphics[width=80pt]{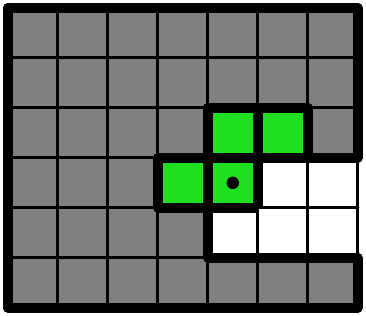}
    \caption{}
  \end{subfigure}
  \caption{Maneuver 2}
  \label{fig:s_maneuver2}
\end{figure}

\subsection{Maneuver 3}\label{subsec:s_maneuver3}

This maneuver, indicated in Figure~\ref{fig:s_maneuver3}, is used to move $\SS$ pieces from left to right in a clause gadget.

\begin{figure}[!ht]
  \begin{subfigure}[b]{0.33\textwidth}
    \centering
    \includegraphics[width=80pt]{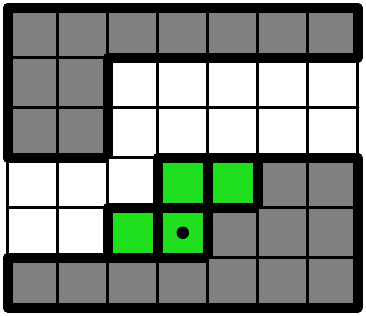}
    \caption{}
  \end{subfigure}
  \begin{subfigure}[b]{0.33\textwidth}
    \centering
    \includegraphics[width=80pt]{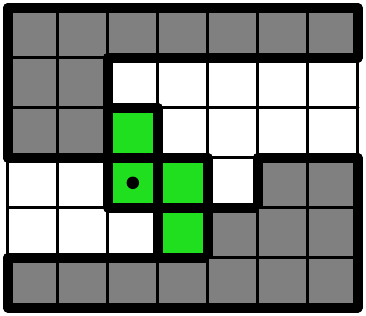}
    \caption{}
  \end{subfigure}
  \begin{subfigure}[b]{0.33\textwidth}
    \centering
    \includegraphics[width=80pt]{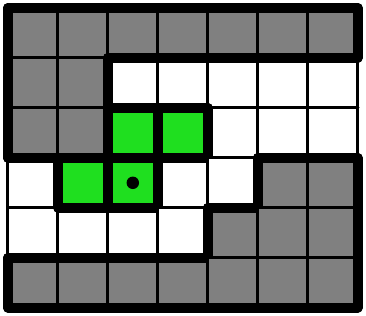}
    \caption{}
  \end{subfigure}
  \caption{Maneuver 3}
  \label{fig:s_maneuver3}
\end{figure}

\subsection{Maneuver 4}\label{subsec:s_maneuver4}

This maneuver, indicated in Figure~\ref{fig:s_maneuver4}, is used to fit $\SS$ pieces through the "False" tunnel of a variable gadget.

\begin{figure}[!ht]
  \begin{subfigure}[b]{0.33\textwidth}
    \centering
    \includegraphics[width=60pt]{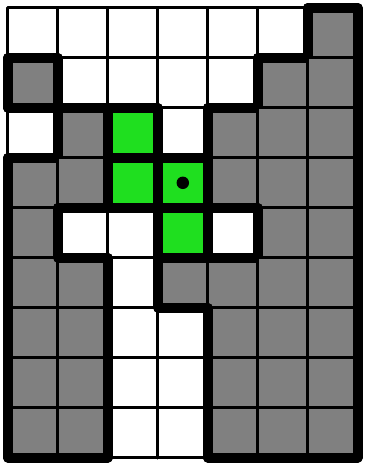}
    \caption{}
  \end{subfigure}
  \begin{subfigure}[b]{0.33\textwidth}
    \centering
    \includegraphics[width=60pt]{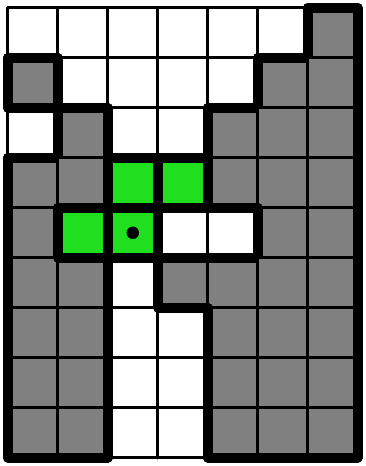}
    \caption{}
  \end{subfigure}
  \begin{subfigure}[b]{0.33\textwidth}
    \centering
    \includegraphics[width=60pt]{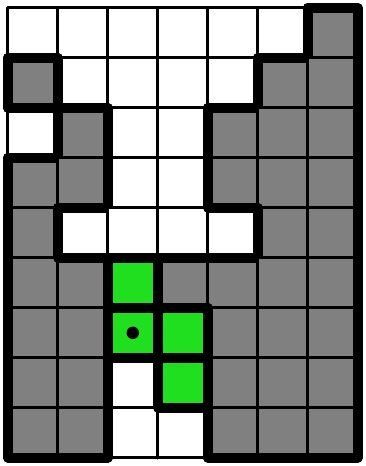}
    \caption{}
  \end{subfigure}
  \caption{Maneuver 4}
  \label{fig:s_maneuver4}
\end{figure}

\subsection{Maneuver 5}\label{subsec:s_maneuver5}

This maneuver, indicated in Figure~\ref{fig:s_maneuver5}, is used to fit $\SS$ piece number $9$ in Figure \ref{fig:s_var_filling}(b).

\begin{figure}[!ht]
  \begin{subfigure}[b]{0.33\textwidth}
    \centering
    \includegraphics[width=80pt]{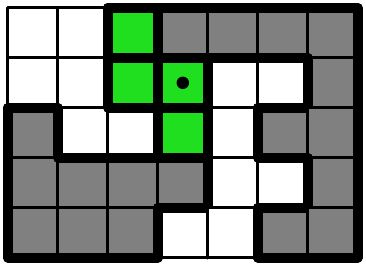}
    \caption{}
  \end{subfigure}
  \begin{subfigure}[b]{0.33\textwidth}
    \centering
    \includegraphics[width=80pt]{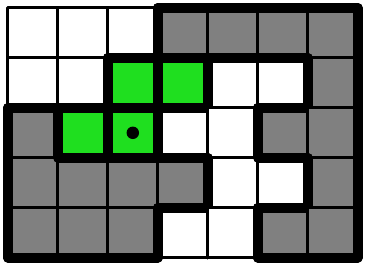}
    \caption{}
  \end{subfigure}
  \begin{subfigure}[b]{0.33\textwidth}
    \centering
    \includegraphics[width=80pt]{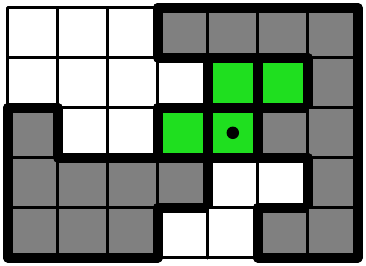}
    \caption{}
  \end{subfigure}
  \begin{subfigure}[b]{0.49\textwidth}
    \centering
    \includegraphics[width=80pt]{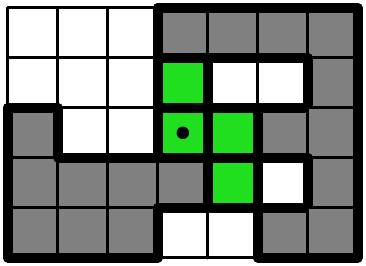}
    \caption{}
  \end{subfigure}
  \begin{subfigure}[b]{0.49\textwidth}
    \centering
    \includegraphics[width=80pt]{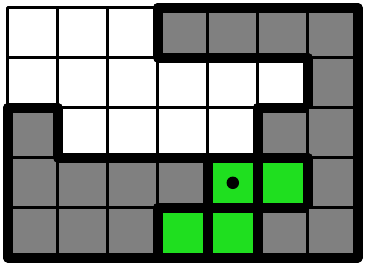}
    \caption{}
  \end{subfigure}
  \caption{Maneuver 5}
  \label{fig:s_maneuver5}
\end{figure}

\subsection{Maneuver 6}\label{subsec:s_maneuver6}

This maneuver, indicated in Figure~\ref{fig:s_maneuver6}, along with a variation of this maneuver (where the initial position of the $\SS$ piece is shifted 1 square to the left) is used to get $\SS$ pieces through certain gaps in variable gadgets.

\begin{figure}[!ht]
  \begin{subfigure}[b]{0.49\textwidth}
    \centering
    \includegraphics[width=80pt]{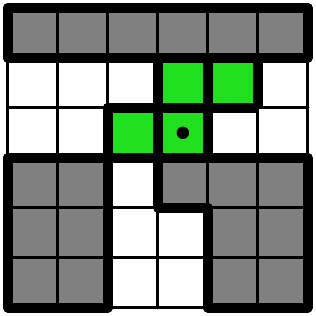}
    \caption{}
  \end{subfigure}
  \begin{subfigure}[b]{0.49\textwidth}
    \centering
    \includegraphics[width=80pt]{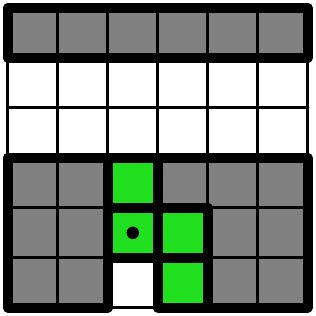}
    \caption{}
  \end{subfigure}
  \caption{Maneuver 6}
  \label{fig:s_maneuver6}
\end{figure}

\pagebreak

\section{$\II$-tris Gadget Tilings}

Numbers in $\II$ pieces denote placement order numbers; smaller numbers are filled first (so $1$ would denote the first $\II$ piece that is placed in a tiling).

\begin{figure}[!ht]
  \centering
  \begin{subfigure}[b]{0.49\textwidth}
    \centering
    \includegraphics[width=220pt]{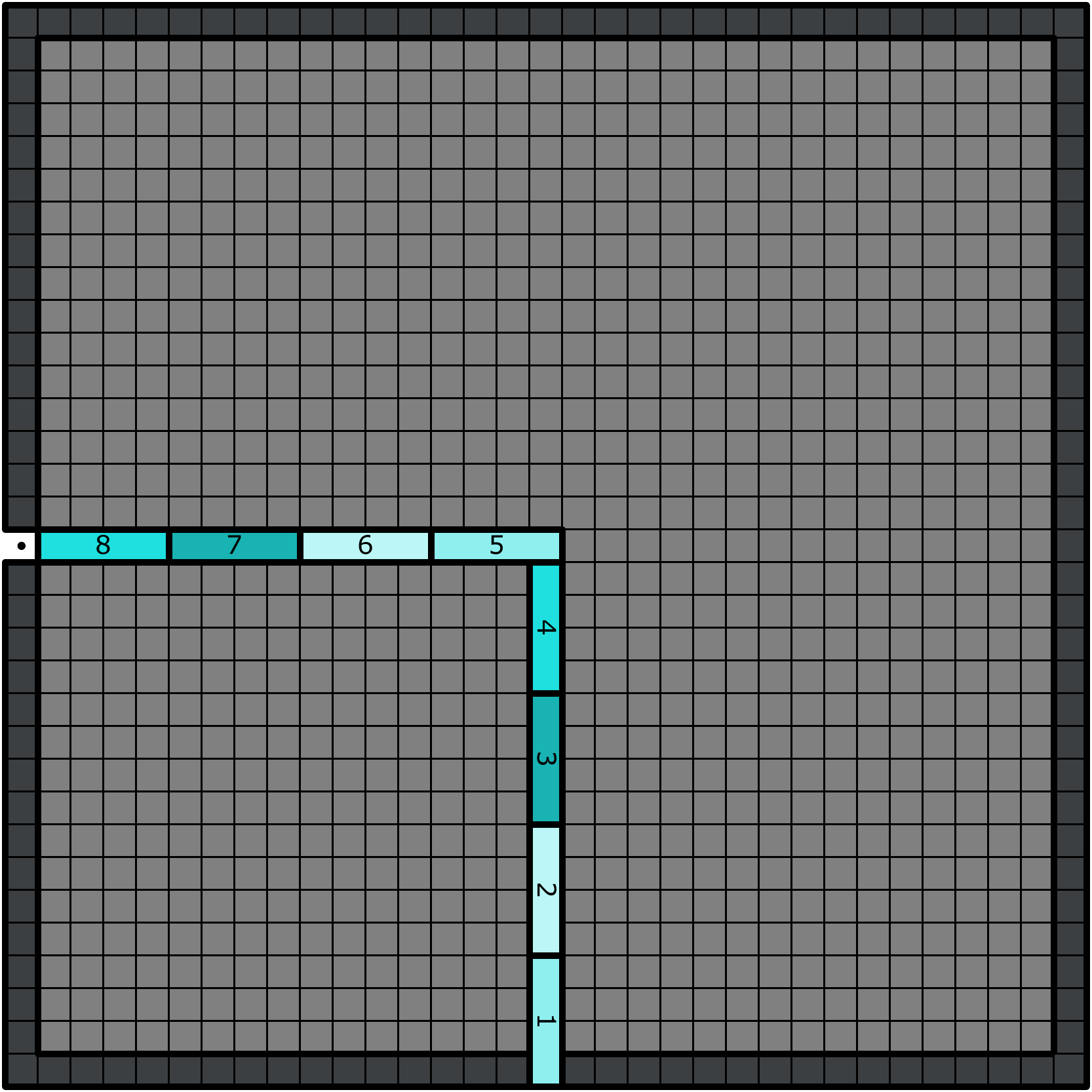}
    \caption{First tiling of an LC gadget}
  \end{subfigure}
  \begin{subfigure}[b]{0.49\textwidth}
    \centering
    \includegraphics[width=220pt]{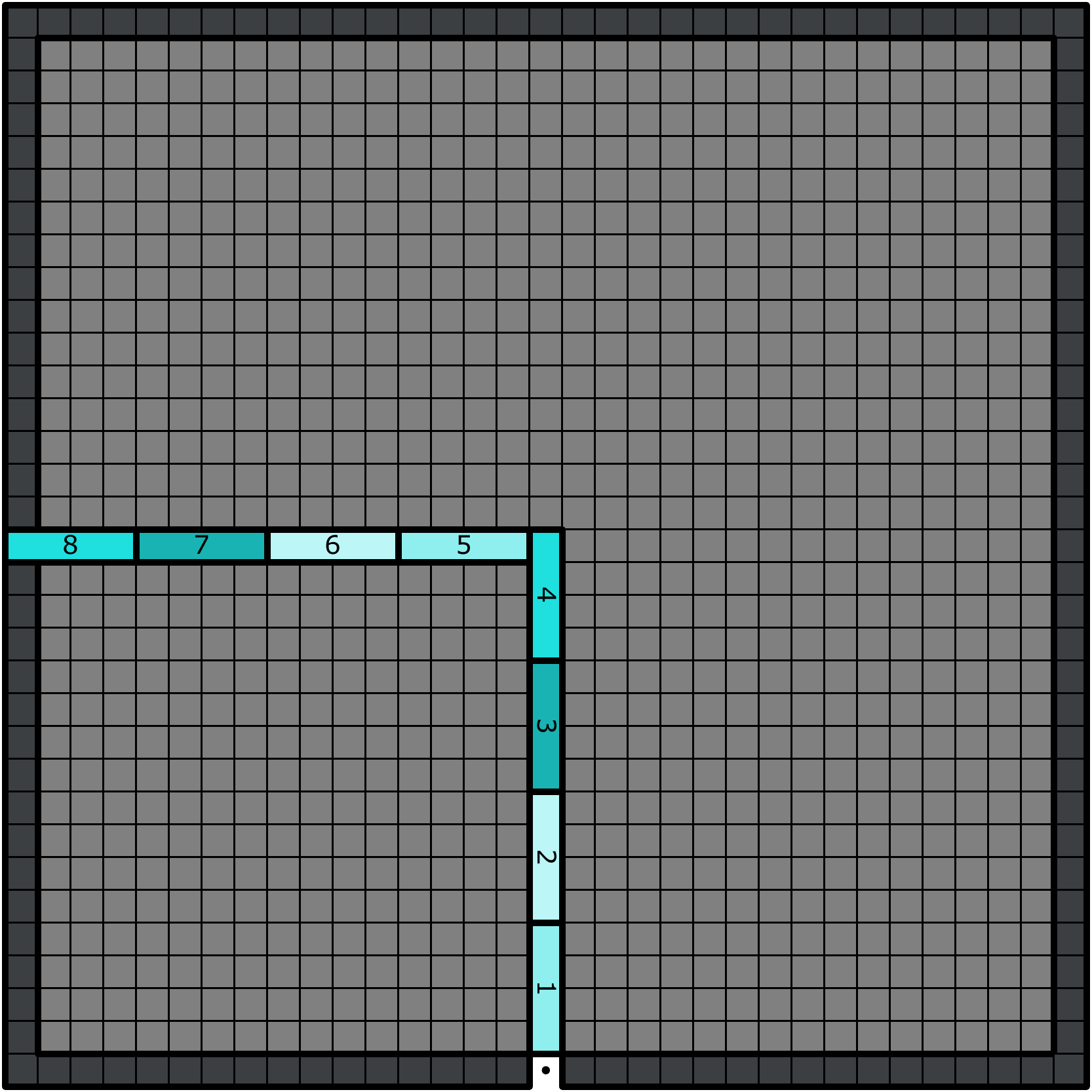}
    \caption{Second tiling of an LC gadget}
  \end{subfigure}
  \begin{subfigure}[b]{0.49\textwidth}
    \centering
    \includegraphics[width=220pt]{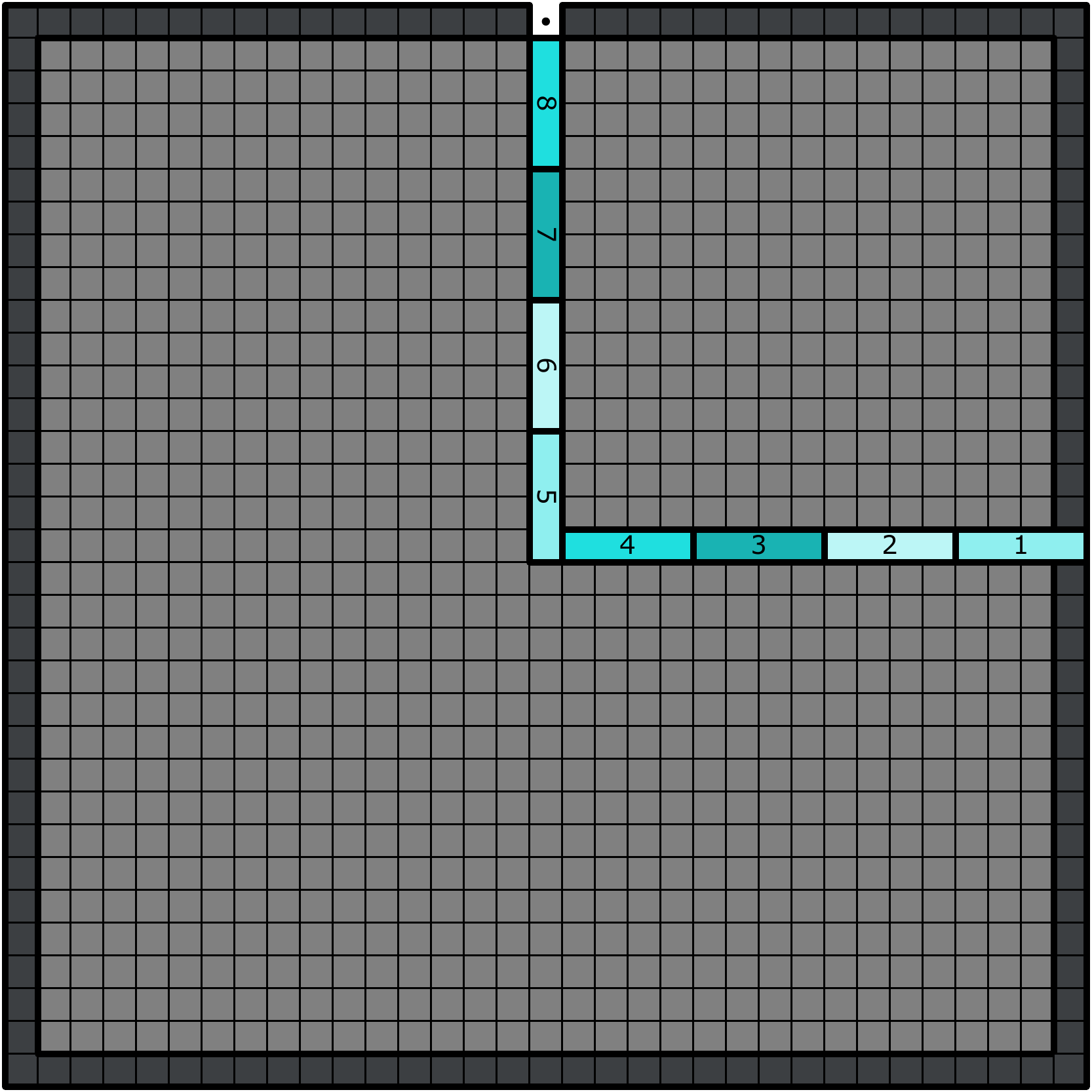}
    \caption{First tiling of a UC gadget}
  \end{subfigure}
  \begin{subfigure}[b]{0.49\textwidth}
    \centering
    \includegraphics[width=220pt]{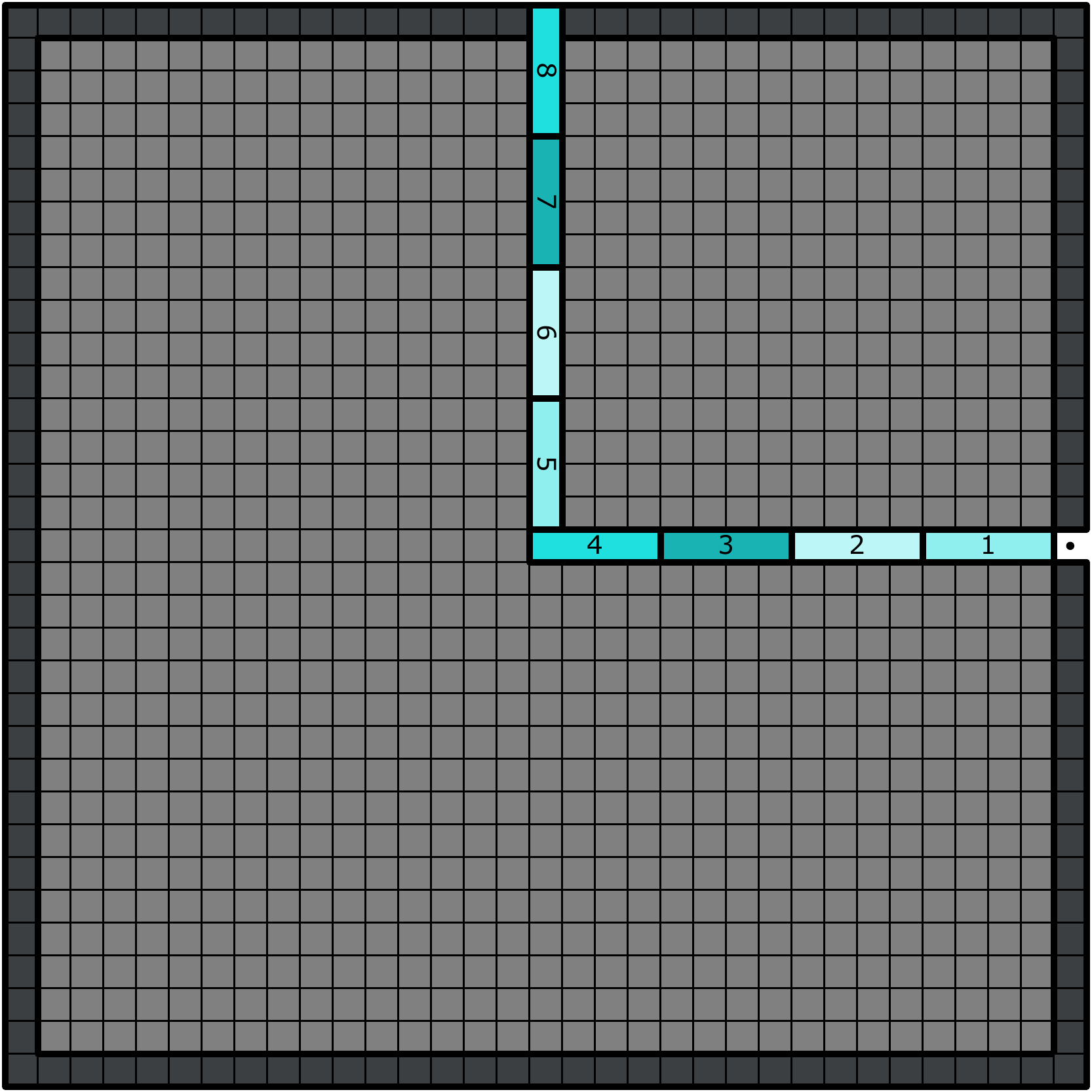}
    \caption{Second tiling of a UC gadget}
  \end{subfigure}
  \caption{Tilings + placement orders for corner gadgets}
  \label{fig:i_corner_tiling}
\end{figure}

\begin{figure}[!ht]
  \begin{subfigure}[b]{0.49\textwidth}
    \centering
    \includegraphics[width=220pt]{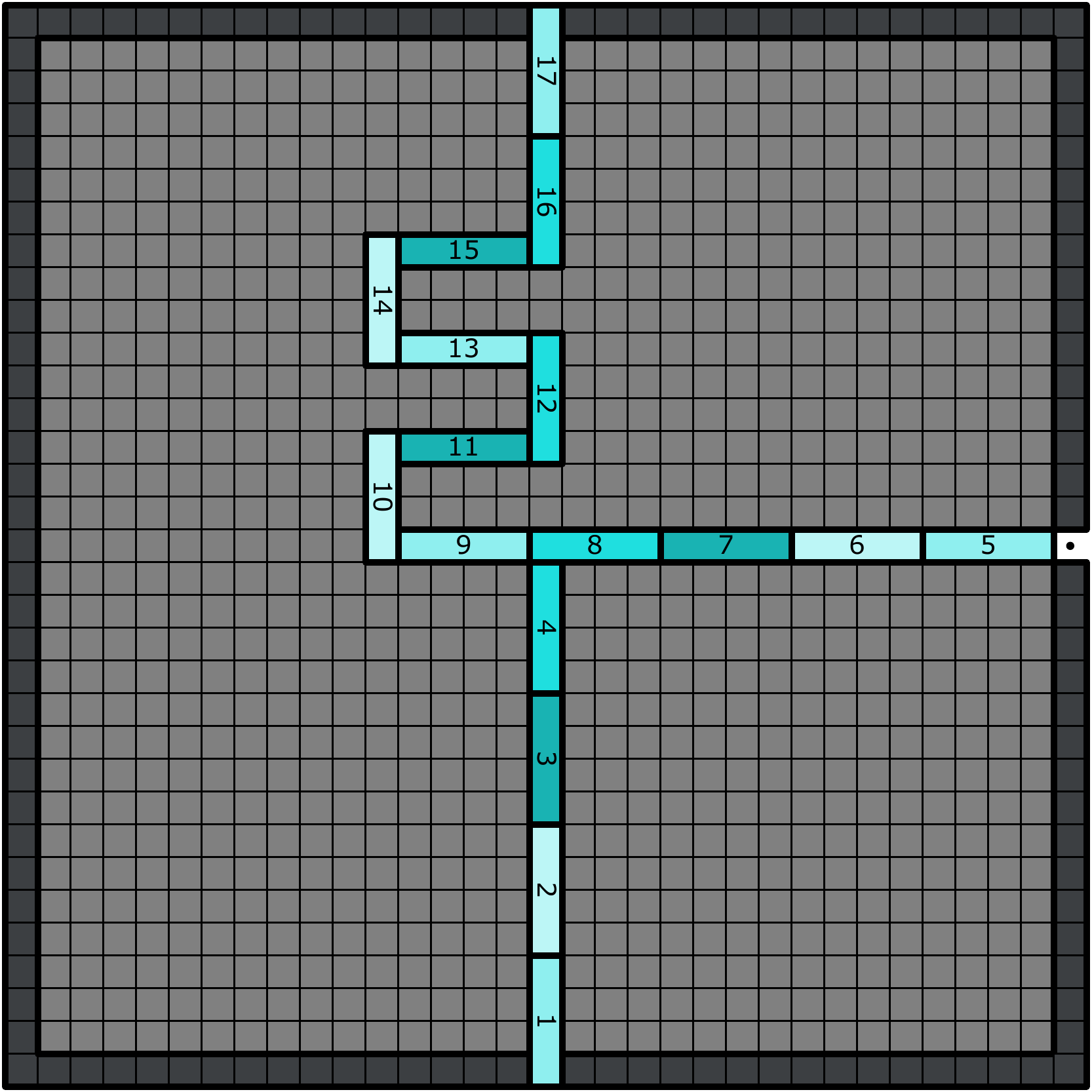}
    \caption{First tiling}
  \end{subfigure}
  \begin{subfigure}[b]{0.49\textwidth}
    \centering
    \includegraphics[width=220pt]{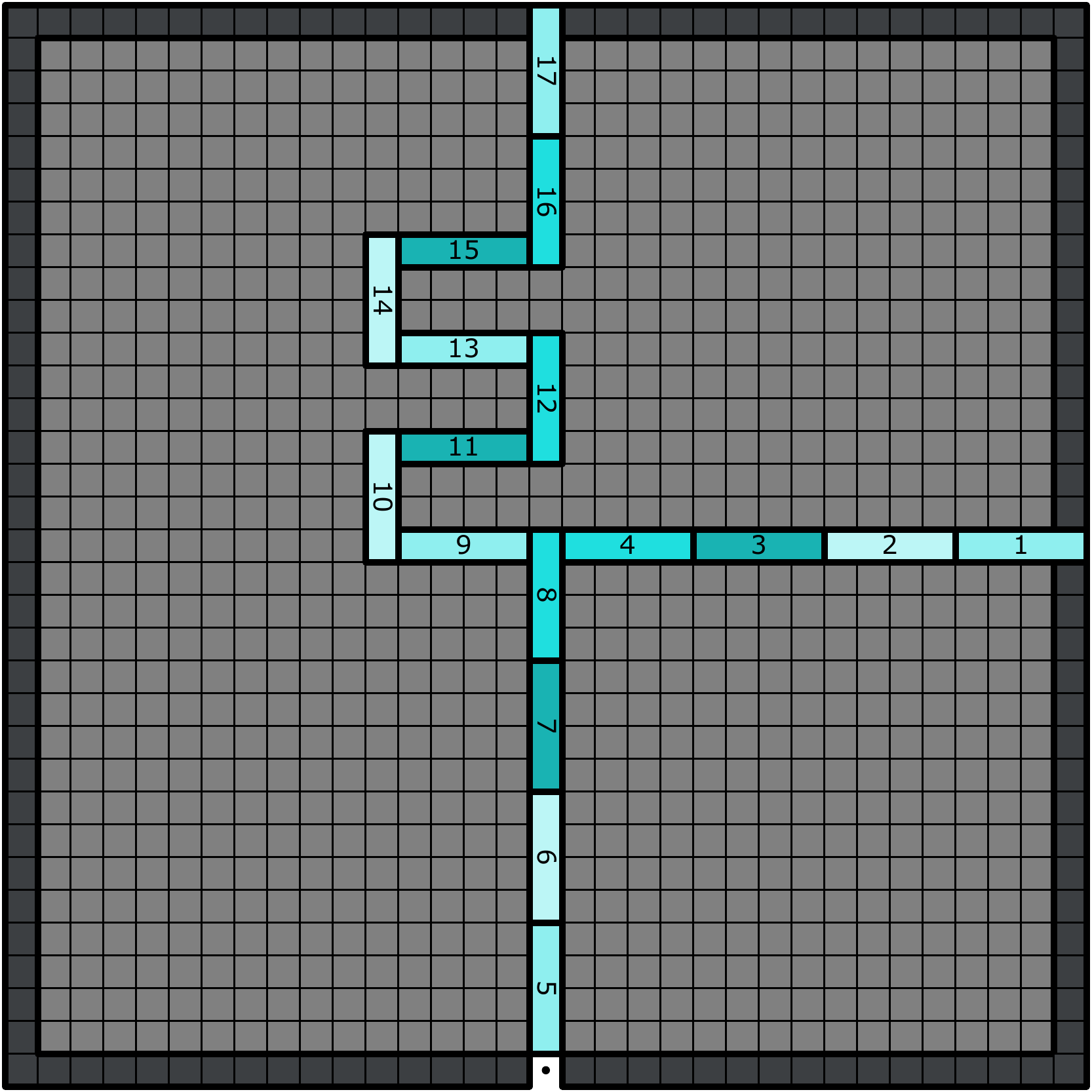}
    \caption{Second tiling}
  \end{subfigure}
  \caption{Tilings + placement orders for the entry corner gadget}
  \label{fig:i_entrycorner_tilings}
\end{figure}

\begin{figure}[!ht]
  \begin{subfigure}[b]{0.49\textwidth}
    \centering
    \includegraphics[width=220pt]{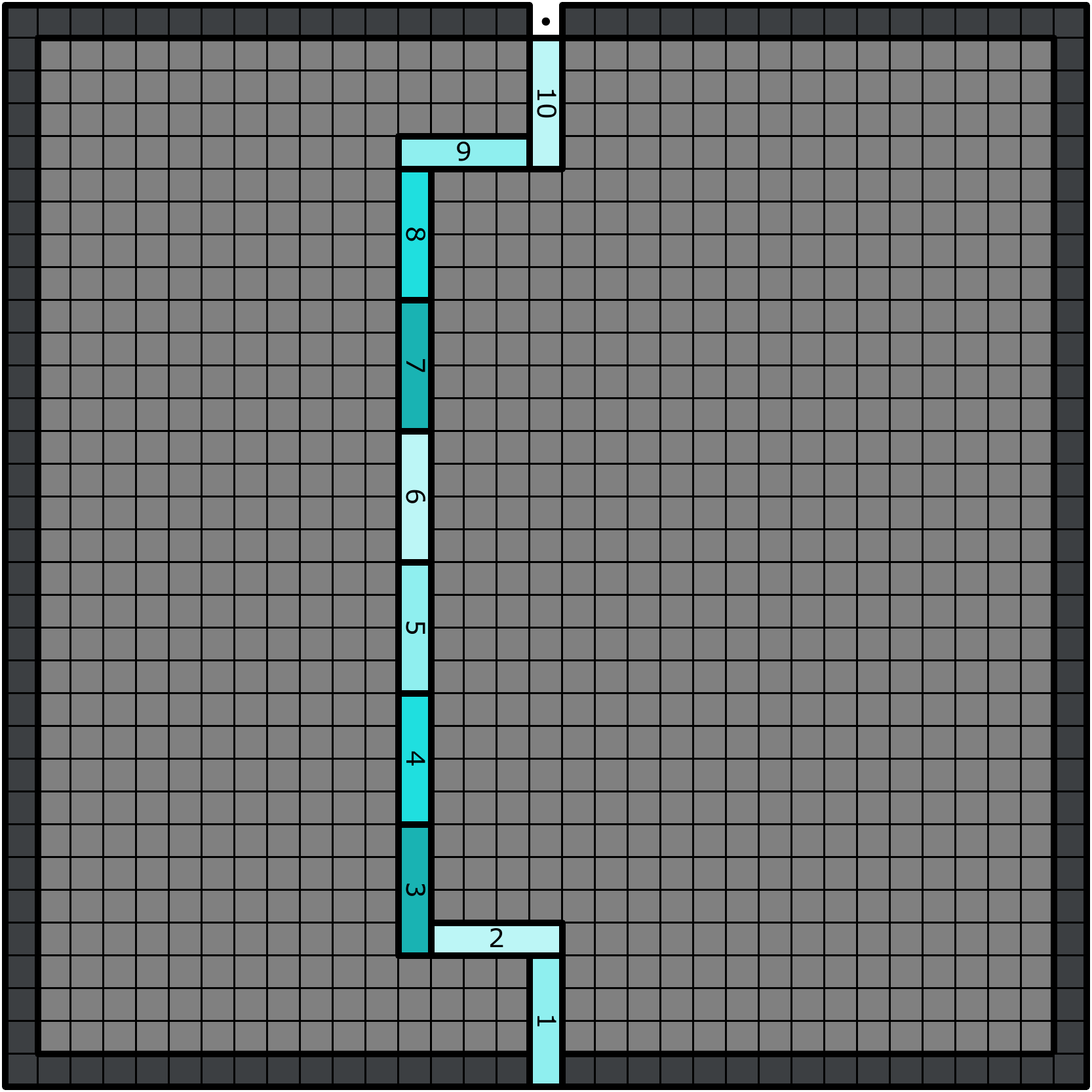}
    \caption{First tiling}
  \end{subfigure}
  \begin{subfigure}[b]{0.49\textwidth}
    \centering
    \includegraphics[width=220pt]{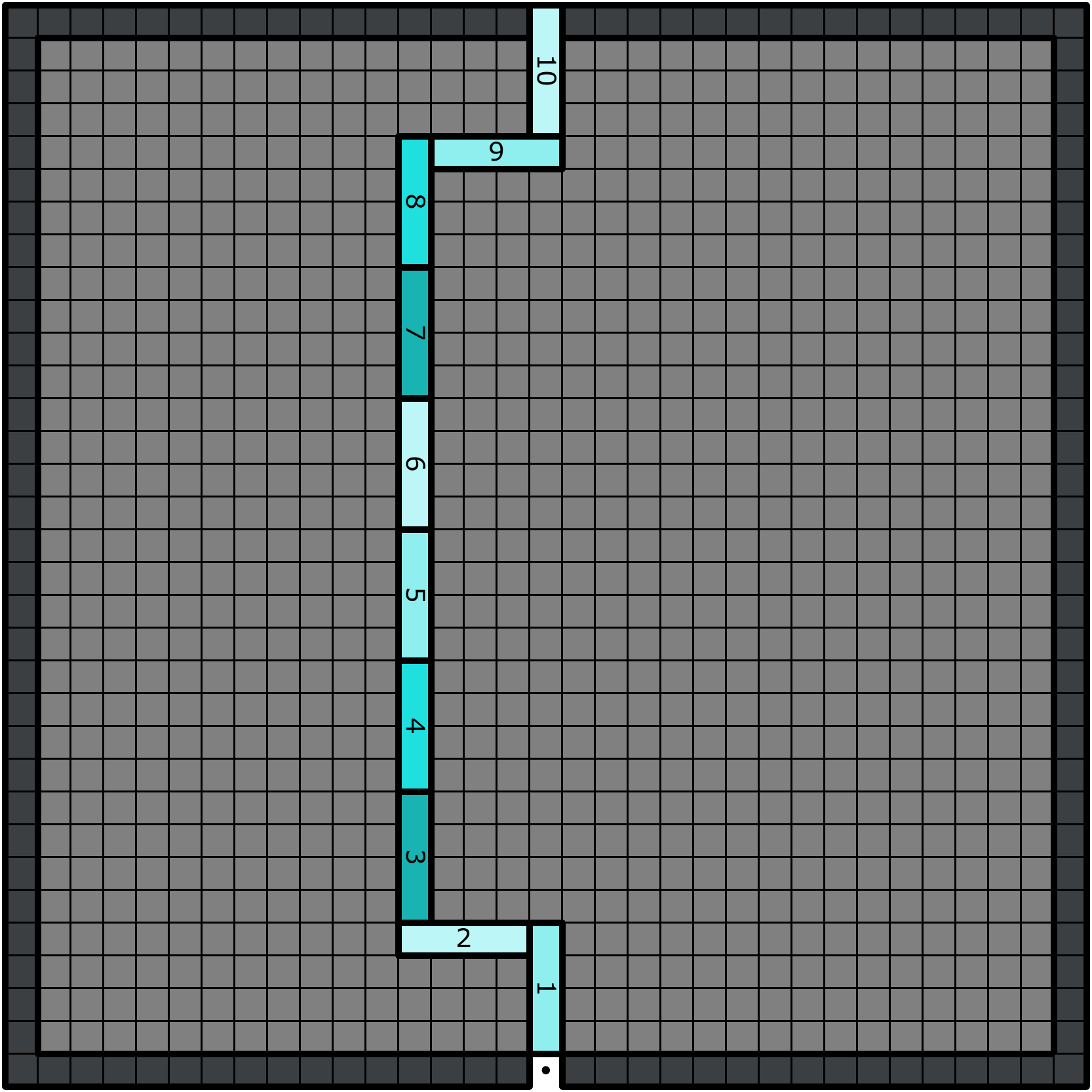}
    \caption{Second tiling}
  \end{subfigure}
  \caption{Tilings + placement orders for the vertical line gadget}
  \label{fig:i_vertline_tilings}
\end{figure}

\begin{figure}[!ht]
  \begin{subfigure}[b]{0.49\textwidth}
    \centering
    \includegraphics[width=220pt]{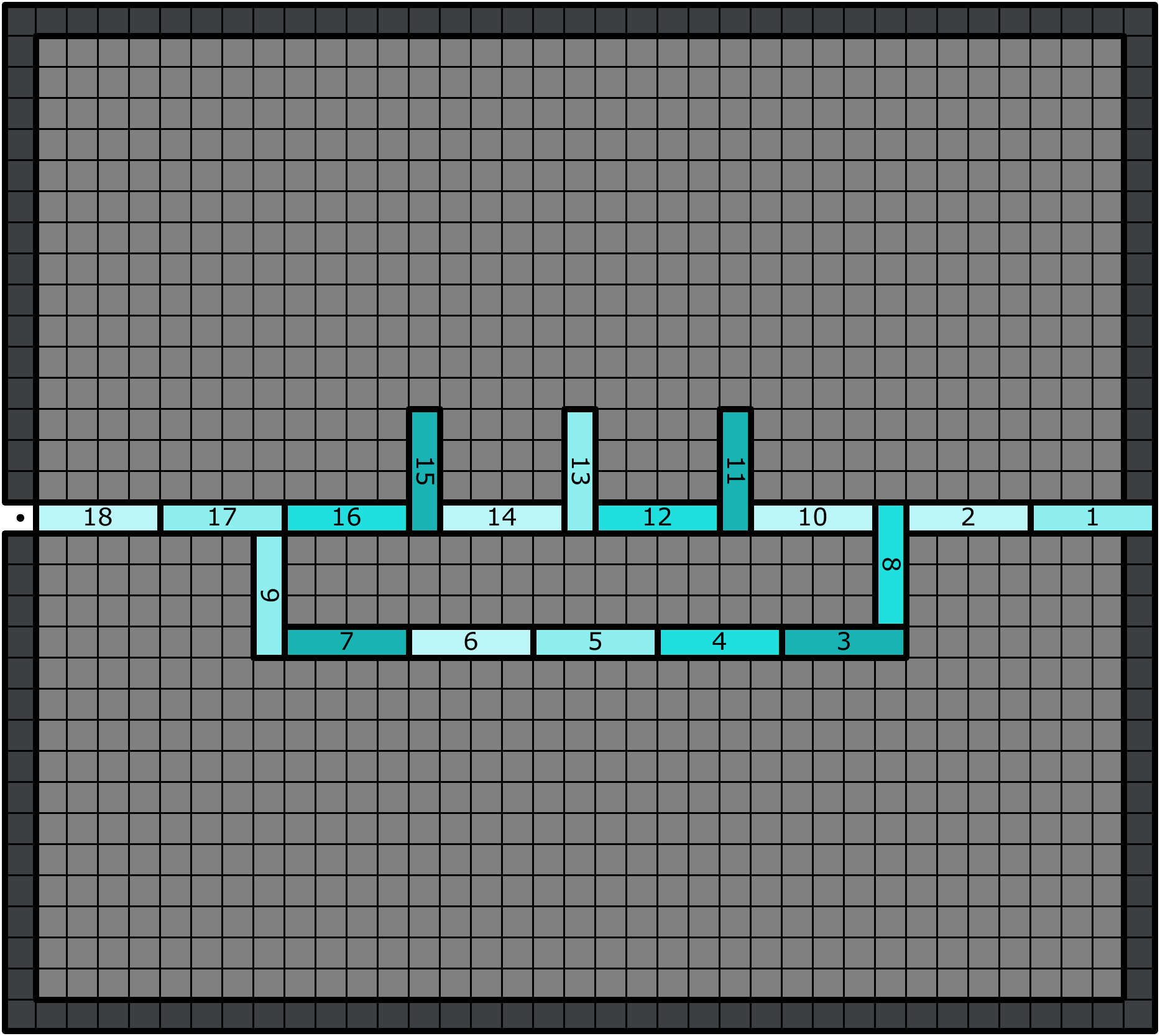}
    \caption{First tiling, left port}
  \end{subfigure}
  \begin{subfigure}[b]{0.49\textwidth}
    \centering
    \includegraphics[width=220pt]{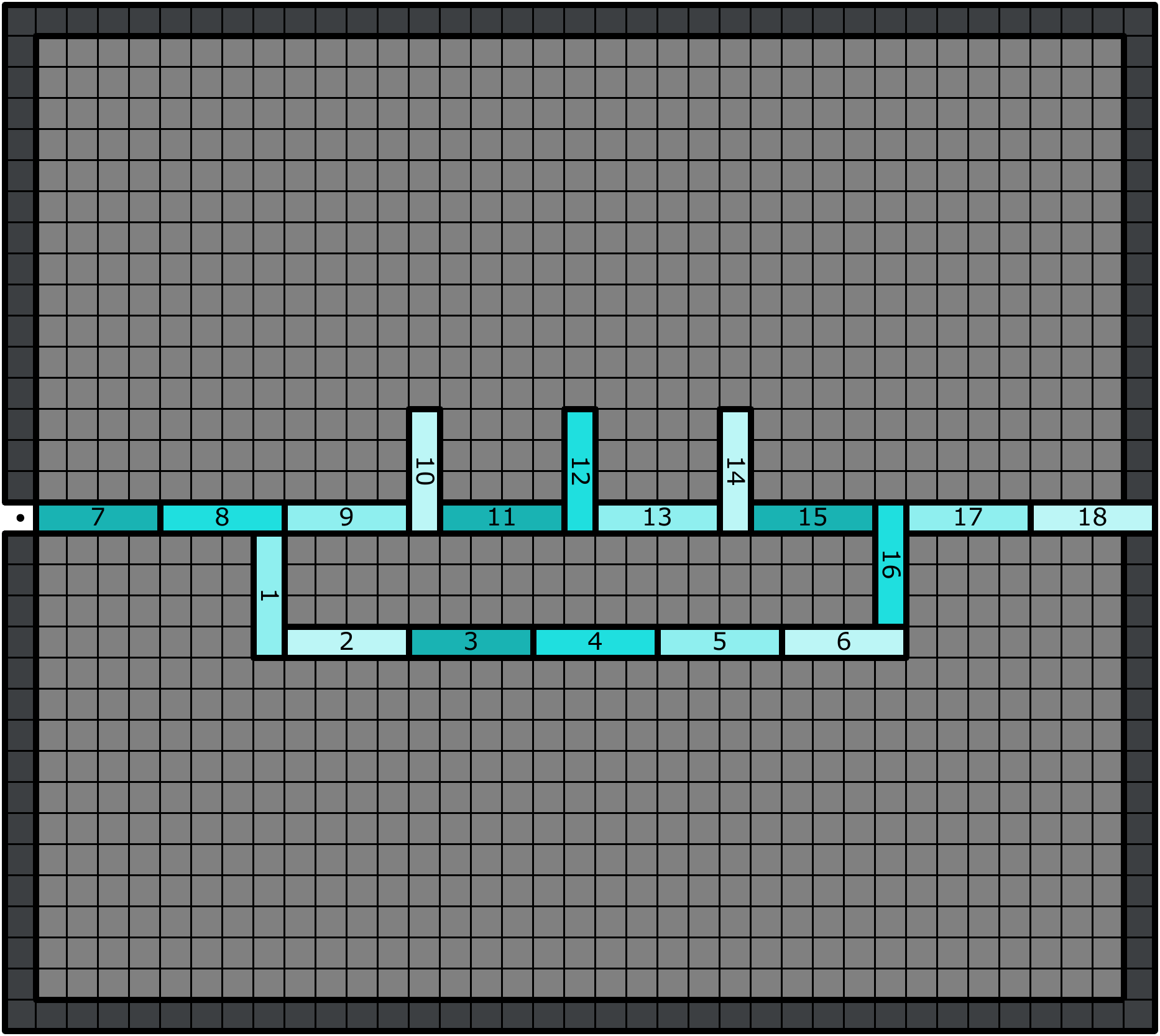}
    \caption{First tiling, right port}
  \end{subfigure}
  \begin{subfigure}[b]{0.49\textwidth}
    \centering
    \includegraphics[width=220pt]{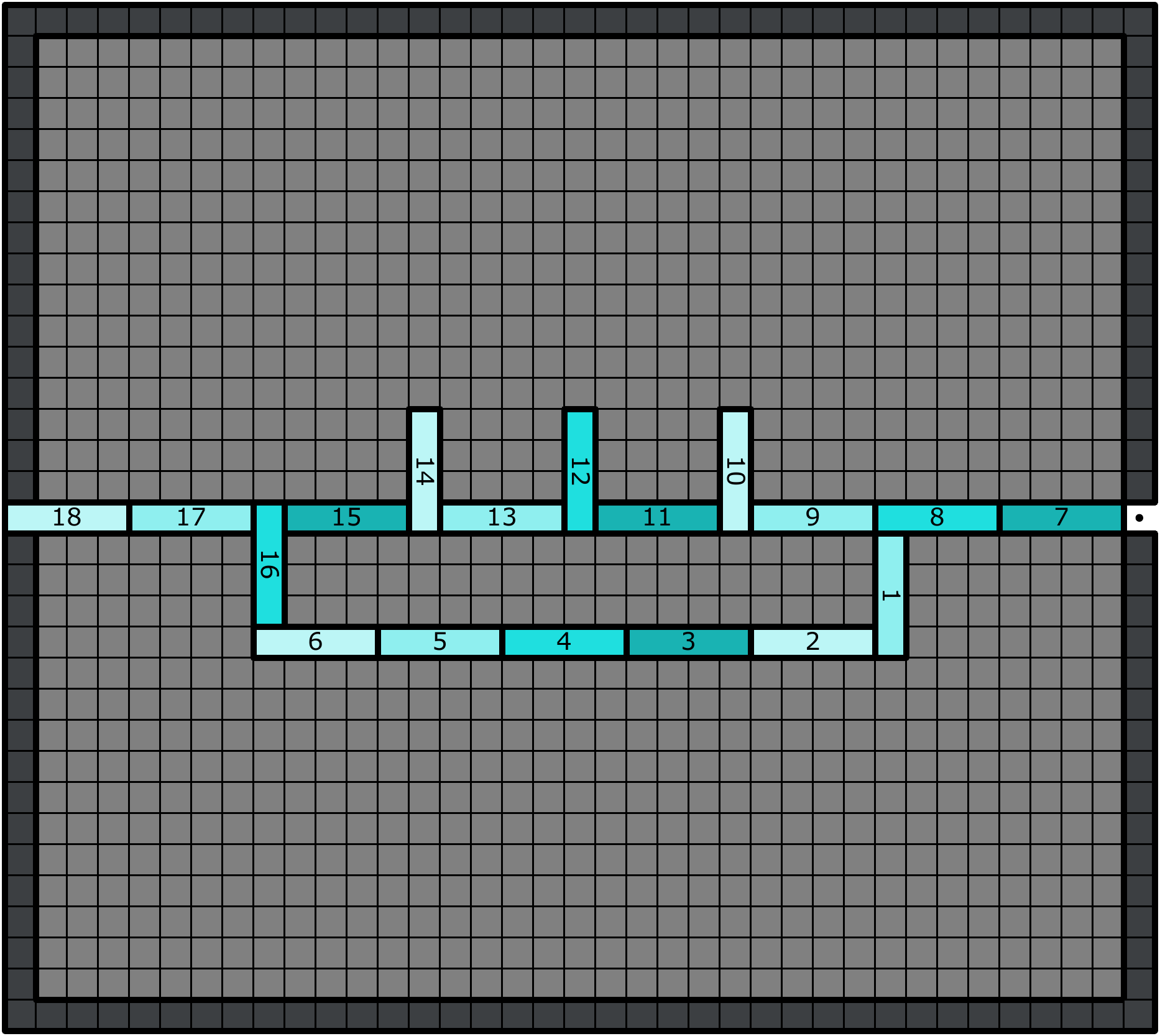}
    \caption{Second tiling, left port}
  \end{subfigure}
  \begin{subfigure}[b]{0.49\textwidth}
    \centering
    \includegraphics[width=220pt]{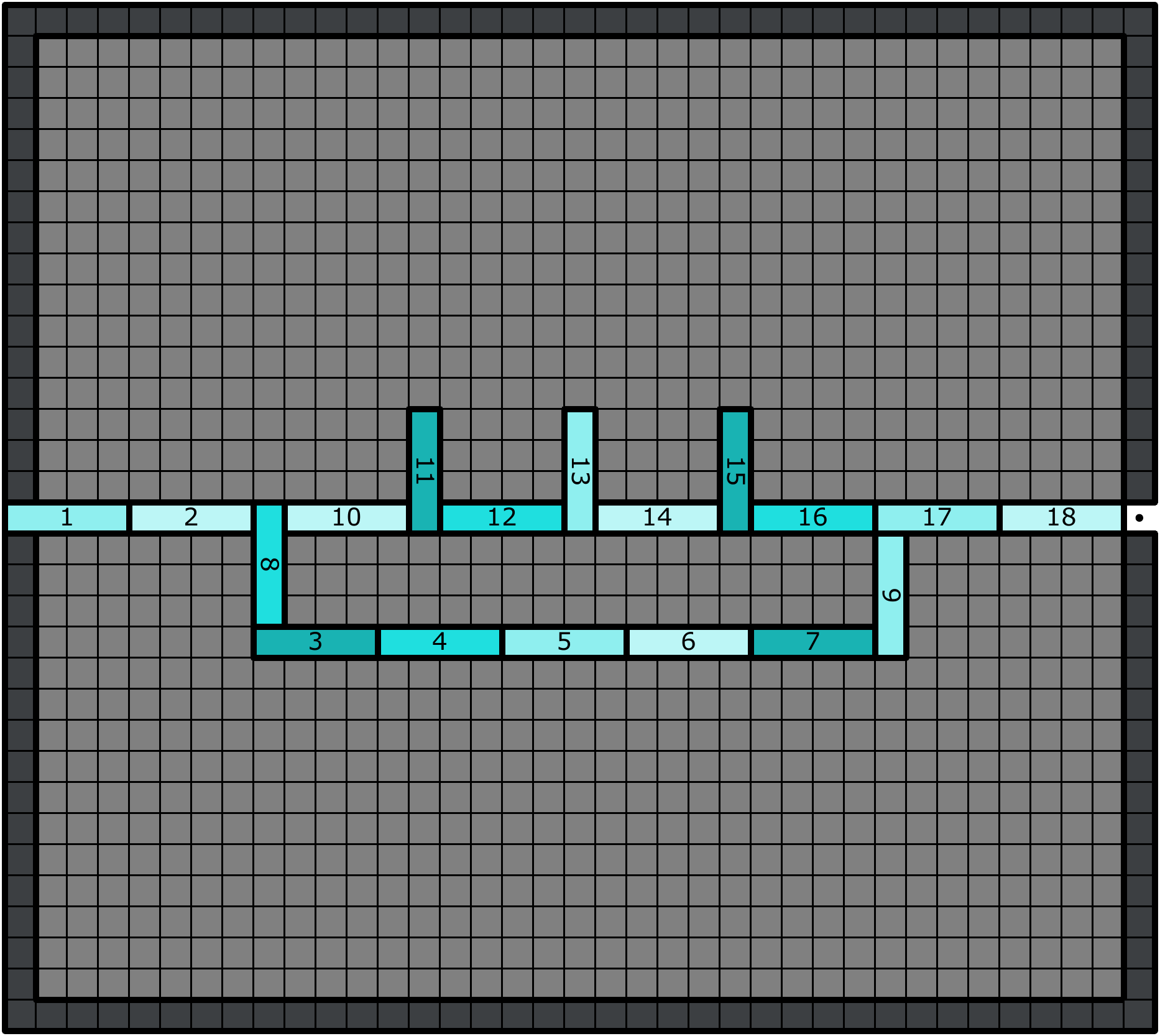}
    \caption{Second tiling, right port}
  \end{subfigure}
  \caption{Tilings + placement orders for the HL gadget}
  \label{fig:i_horline_tilings}
\end{figure}

\begin{figure}[!ht]
  \begin{subfigure}[b]{0.49\textwidth}
    \centering
    \includegraphics[width=220pt]{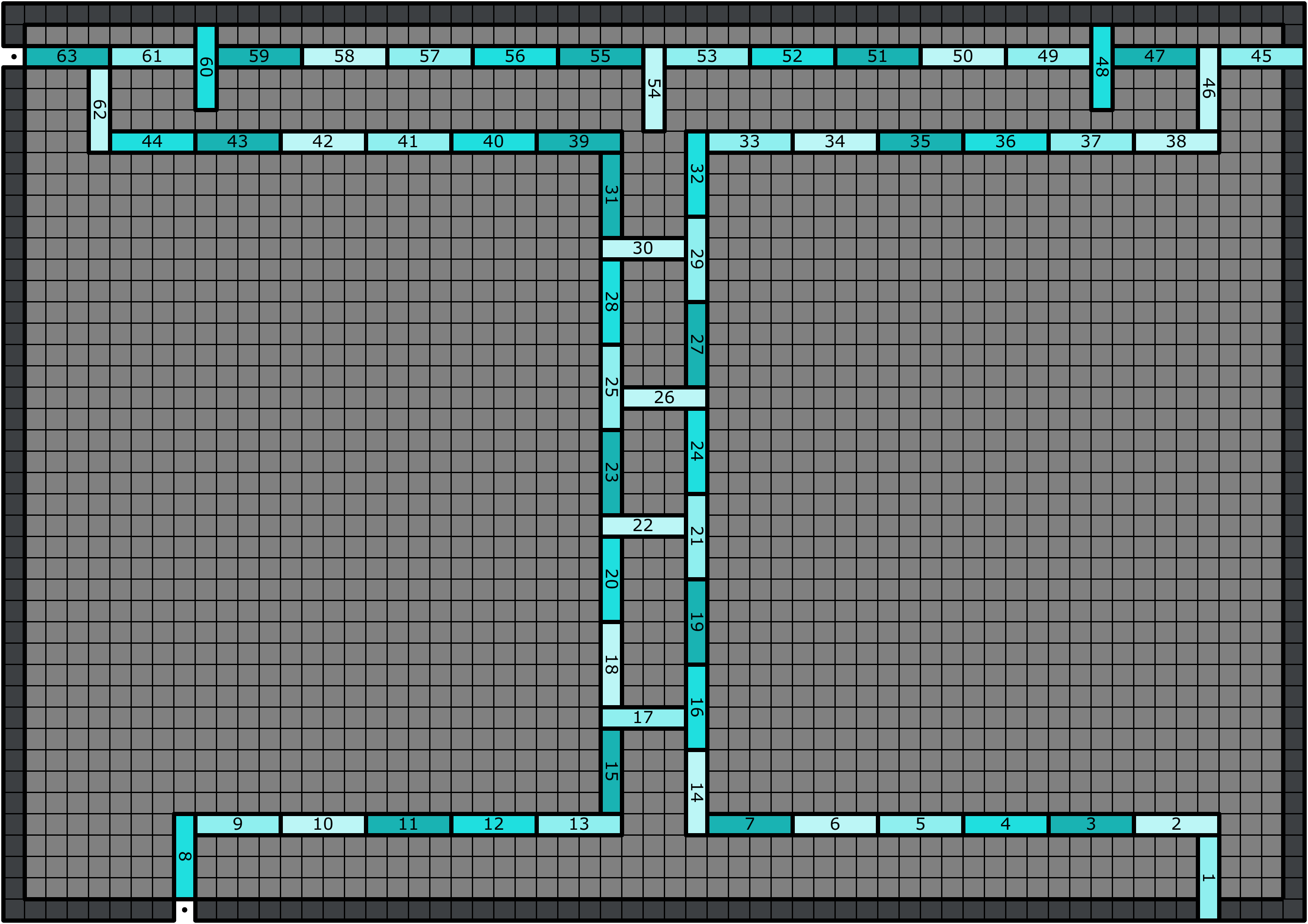}
    \caption{First tiling, left port}
  \end{subfigure}
  \begin{subfigure}[b]{0.49\textwidth}
    \centering
    \includegraphics[width=220pt]{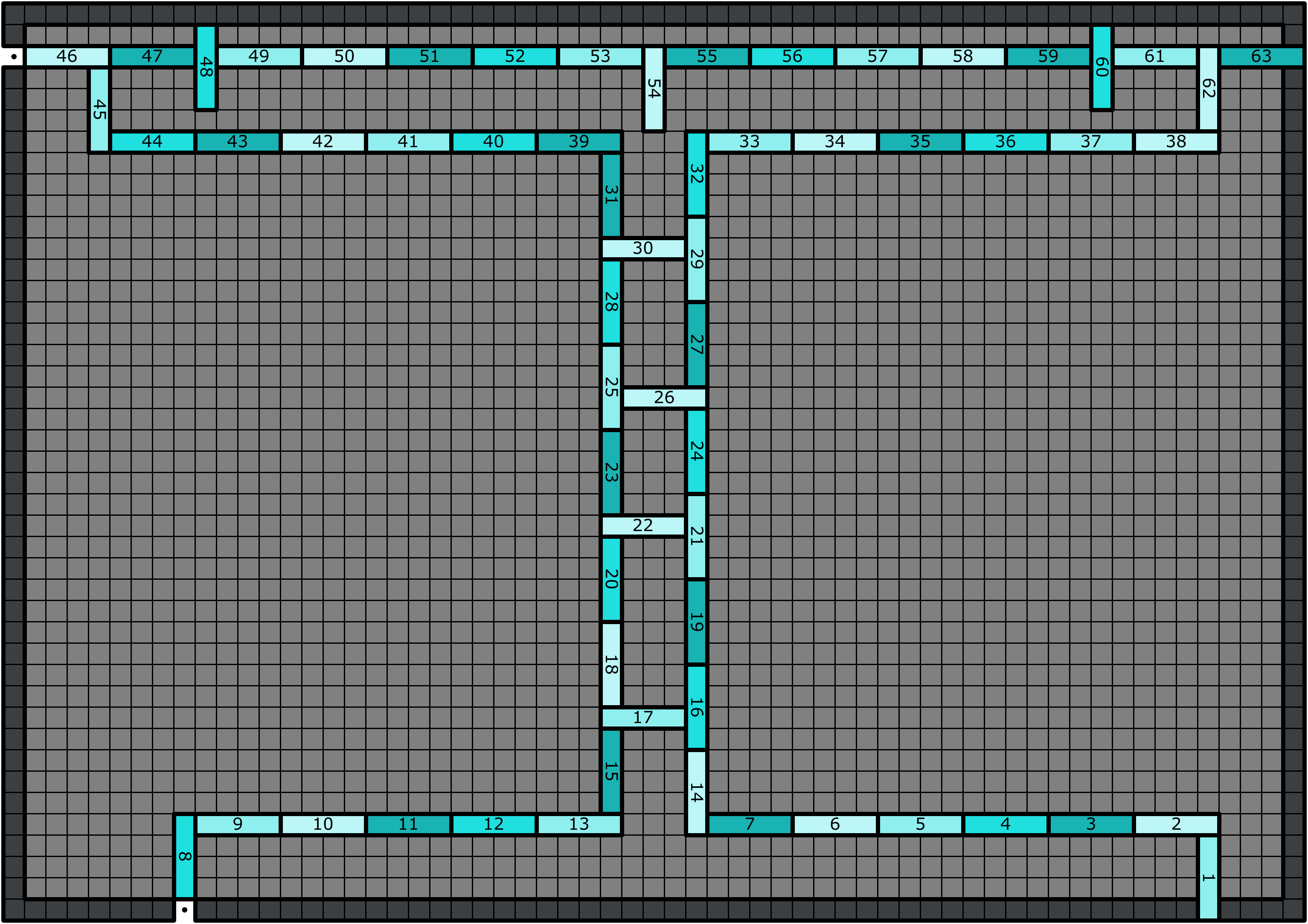}
    \caption{First tiling, right port}
  \end{subfigure}
  \begin{subfigure}[b]{0.49\textwidth}
    \centering
    \includegraphics[width=220pt]{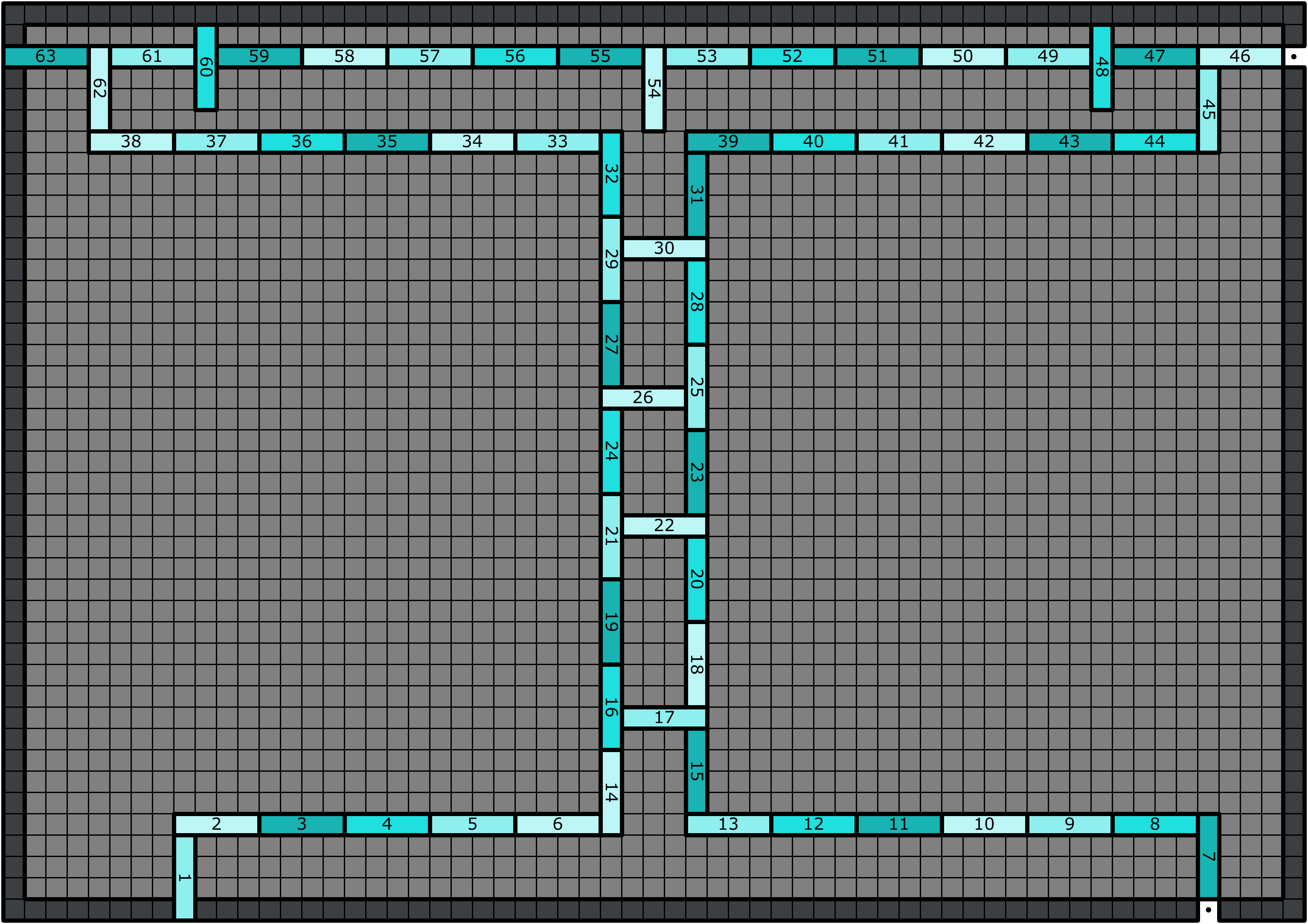}
    \caption{Second tiling, left port}
  \end{subfigure}
  \begin{subfigure}[b]{0.49\textwidth}
    \centering
    \includegraphics[width=220pt]{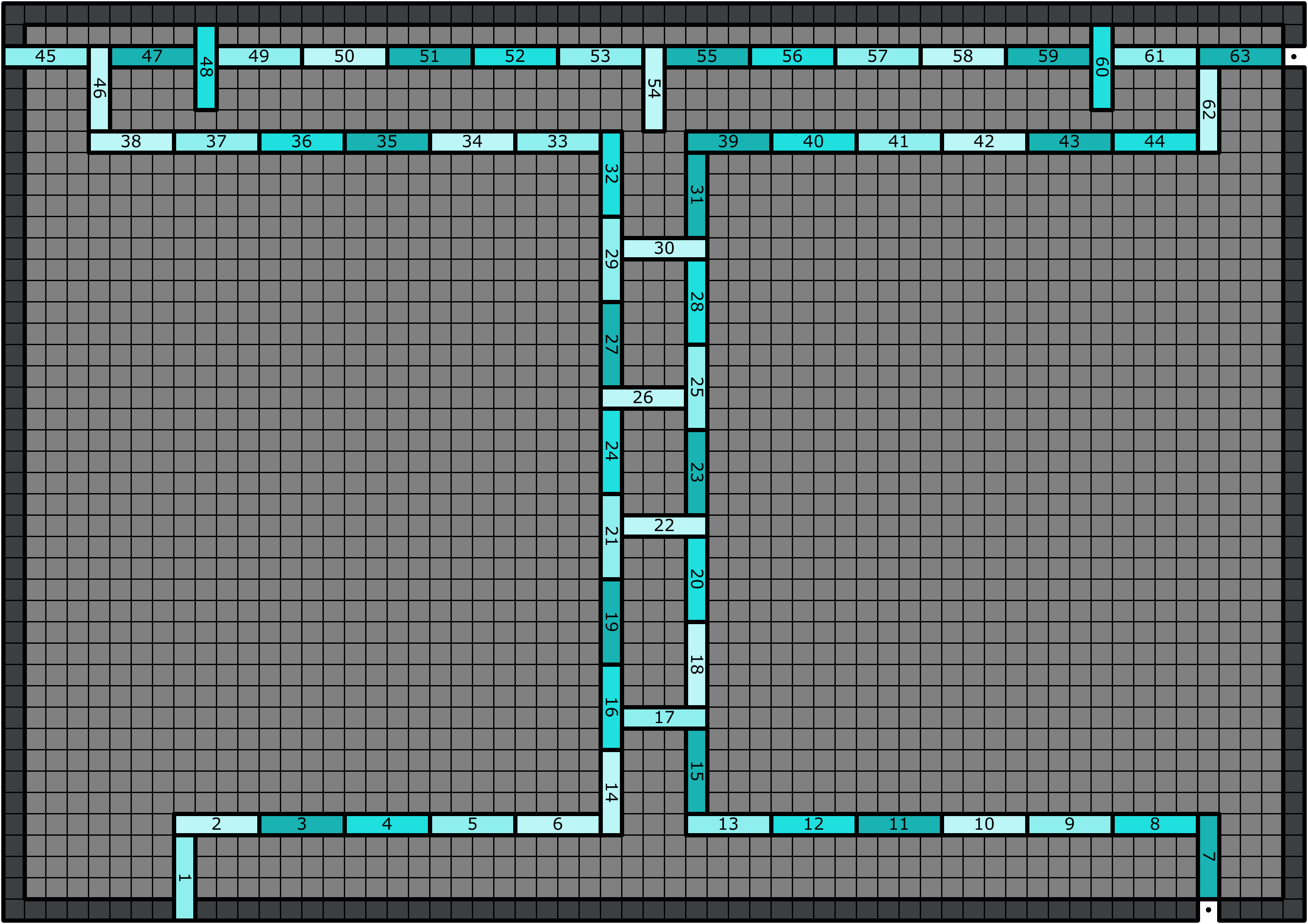}
    \caption{Second tiling, right port}
  \end{subfigure}
  \caption{Tilings + placement orders for the duplicator gadget}
  \label{fig:i_duplicator_tilings}
\end{figure}

\begin{figure}[!ht]
  \begin{subfigure}[b]{0.325\textwidth}
    \centering
    \includegraphics[width=140pt]{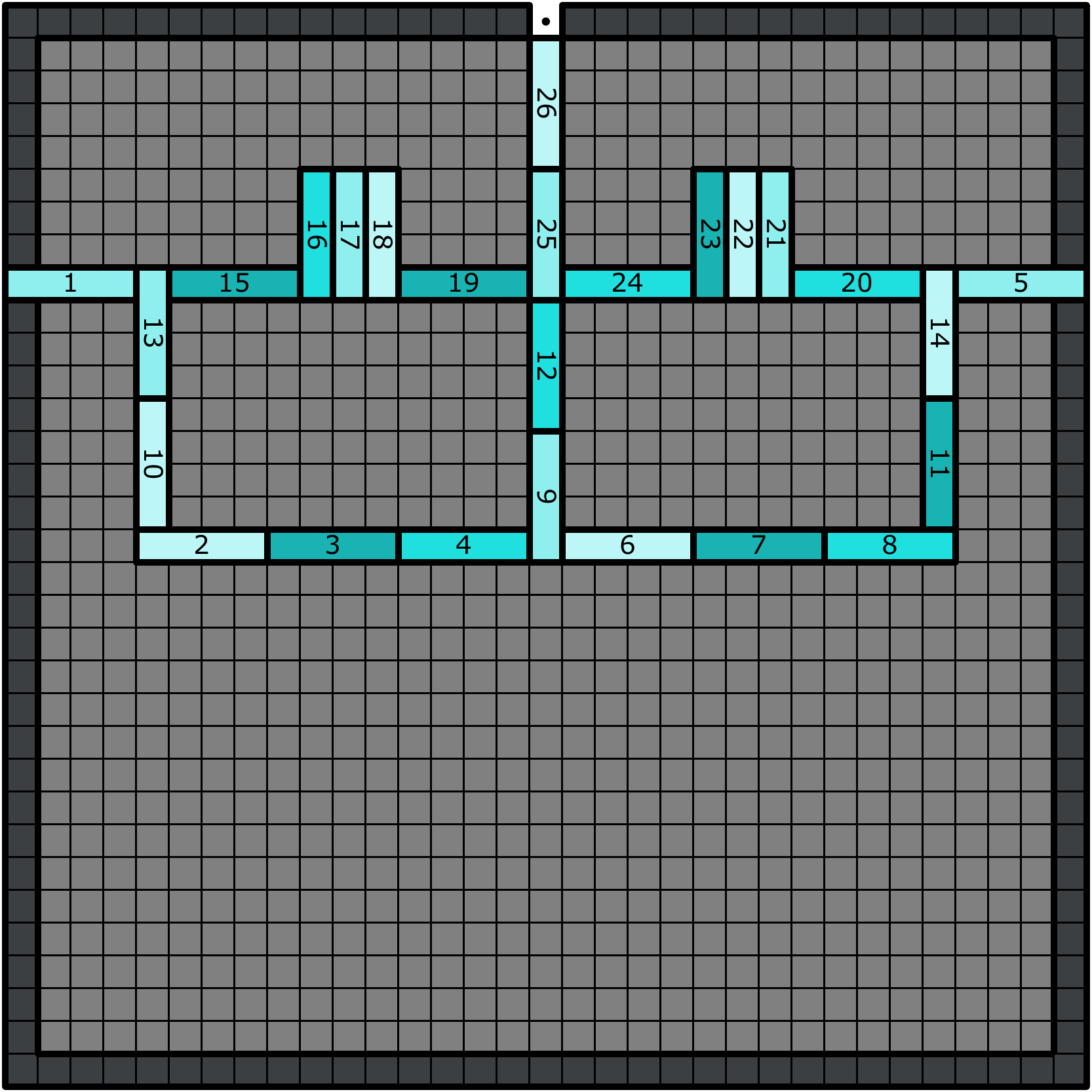}
    \caption{First tiling}
  \end{subfigure}
  \begin{subfigure}[b]{0.325\textwidth}
    \centering
    \includegraphics[width=140pt]{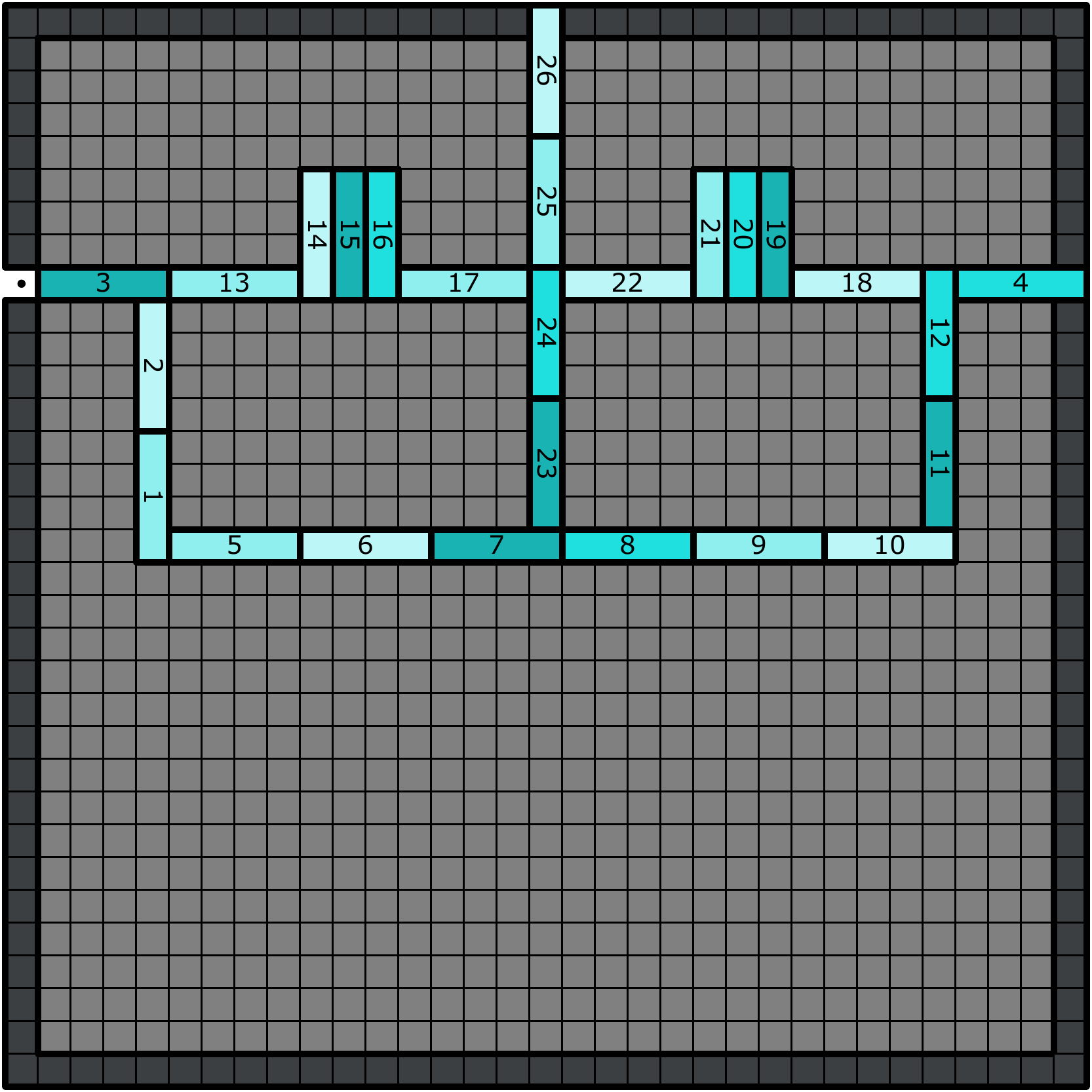}
    \caption{Second tiling}
  \end{subfigure}
  \begin{subfigure}[b]{0.325\textwidth}
    \centering
    \includegraphics[width=140pt]{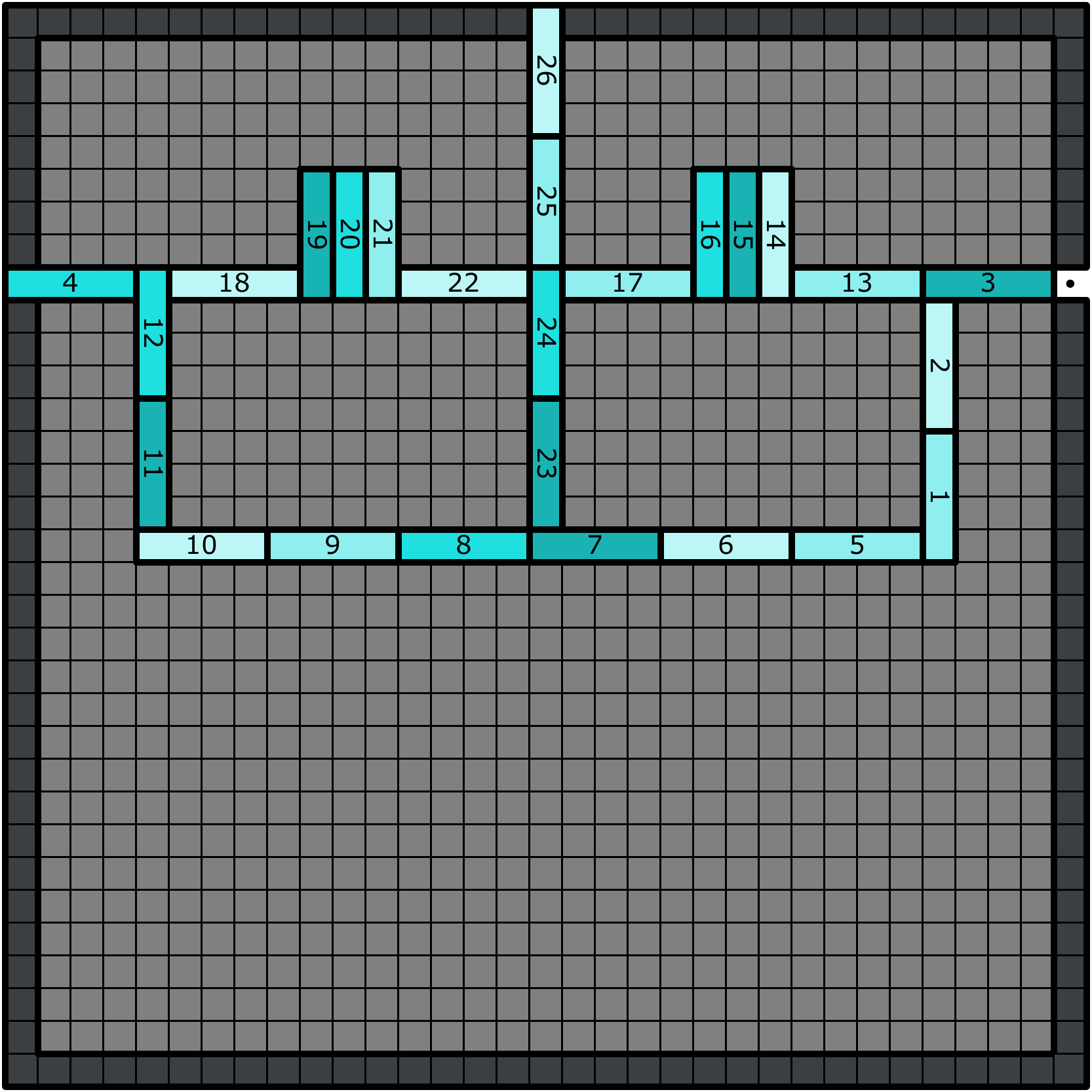}
    \caption{Third tiling}
  \end{subfigure}
  \caption{Tilings + placement orders for the clause gadget}
  \label{fig:i_clause_tilings}
\end{figure}

\begin{figure}[!ht]
  \begin{subfigure}[b]{0.325\textwidth}
    \centering
    \includegraphics[width=140pt]{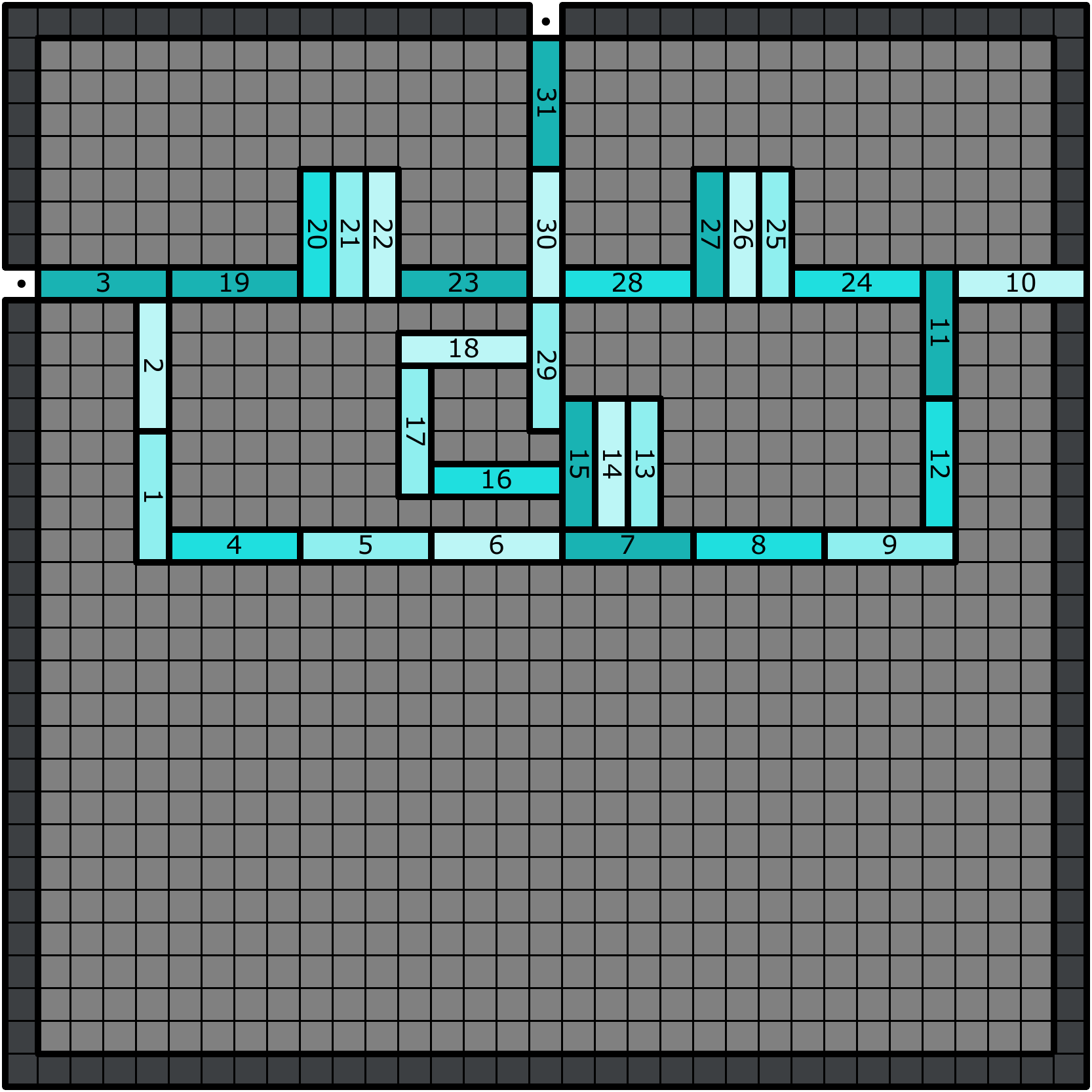}
    \caption{First tiling}
  \end{subfigure}
  \begin{subfigure}[b]{0.325\textwidth}
    \centering
    \includegraphics[width=140pt]{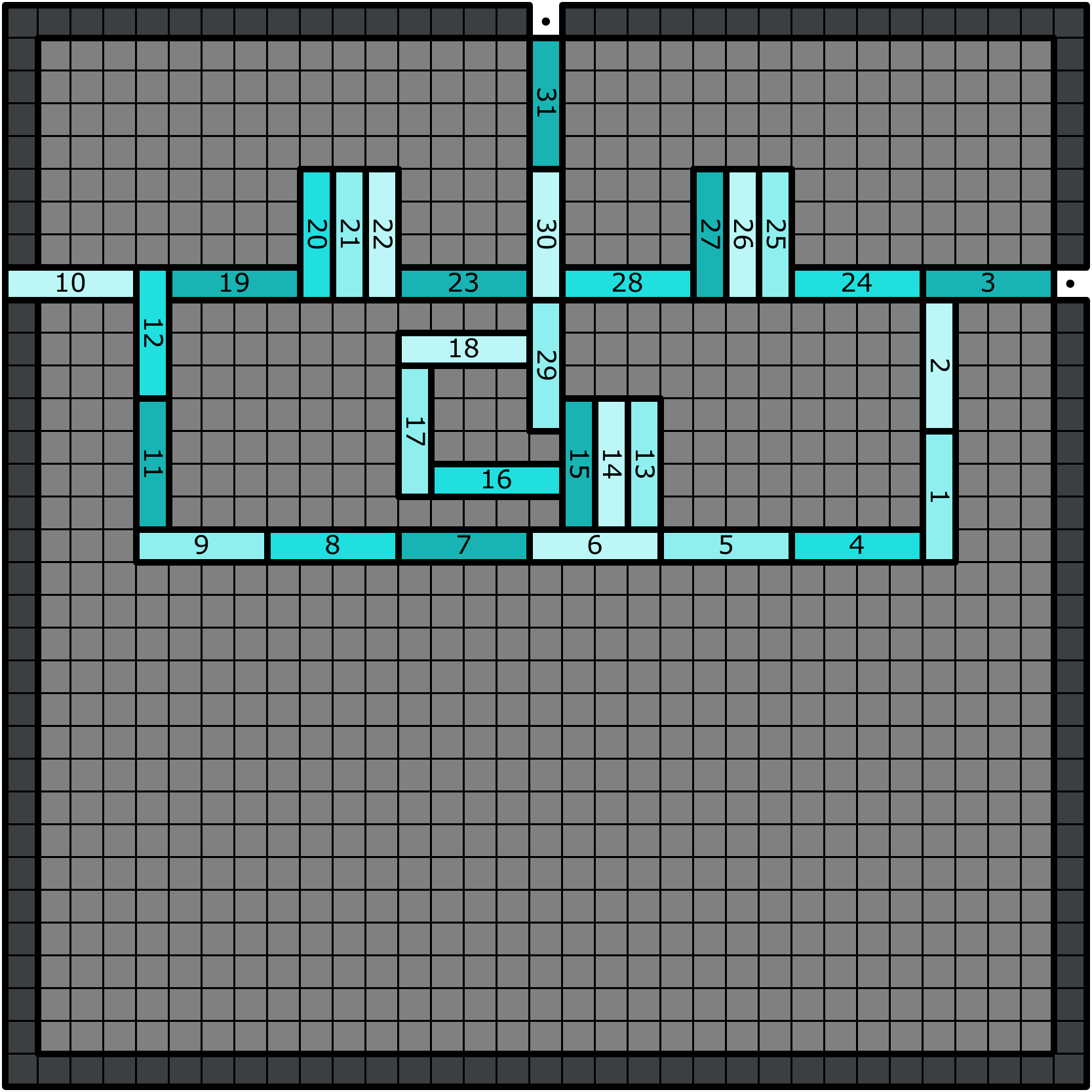}
    \caption{Second tiling}
  \end{subfigure}
  \begin{subfigure}[b]{0.325\textwidth}
    \centering
    \includegraphics[width=140pt]{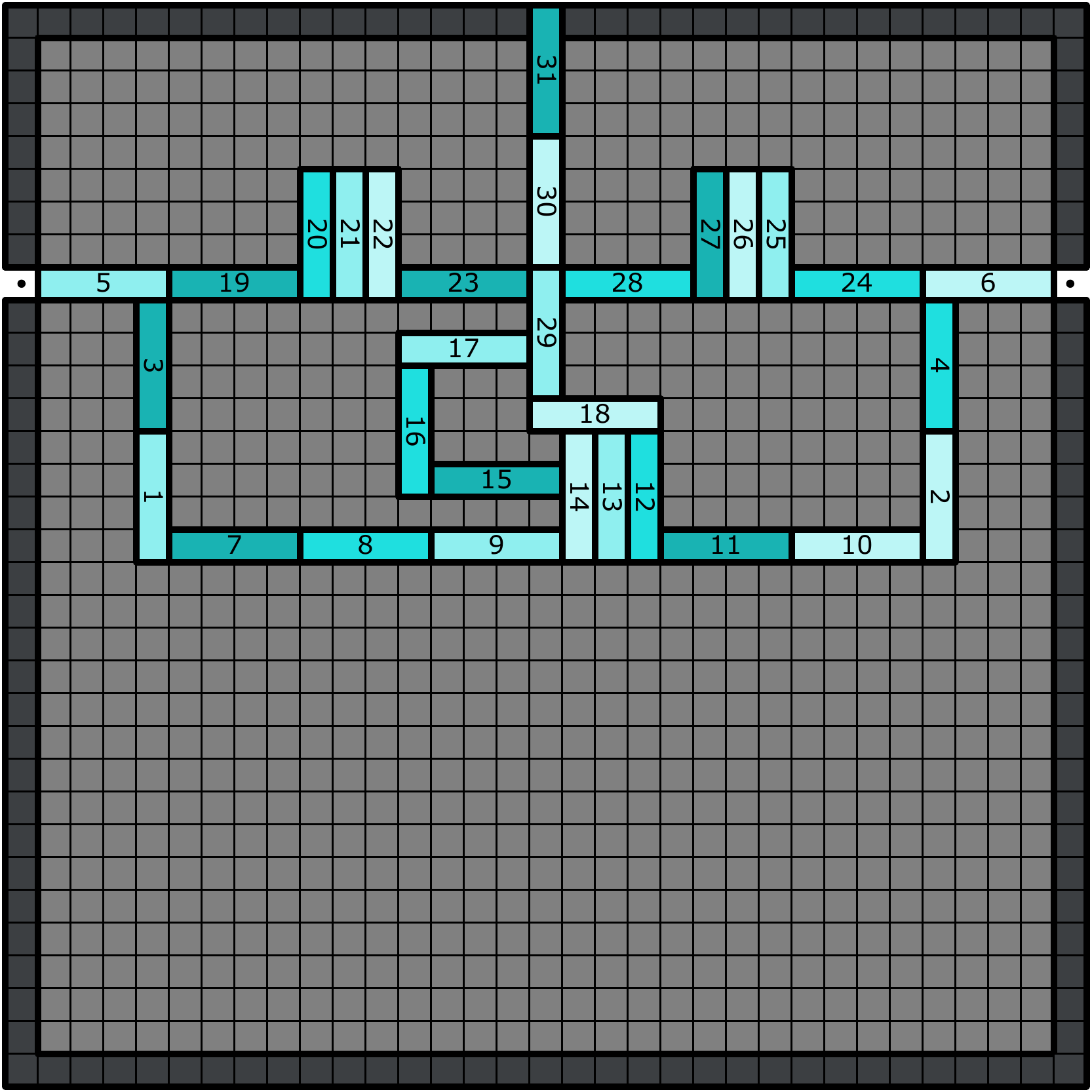}
    \caption{Third tiling}
  \end{subfigure}
  \caption{Tilings + placement orders for the negated-clause gadget}
  \label{fig:i_negclause_tilings}
\end{figure}

\begin{figure}[!ht]
  \begin{subfigure}[b]{0.49\textwidth}
    \centering
    \includegraphics[width=220pt]{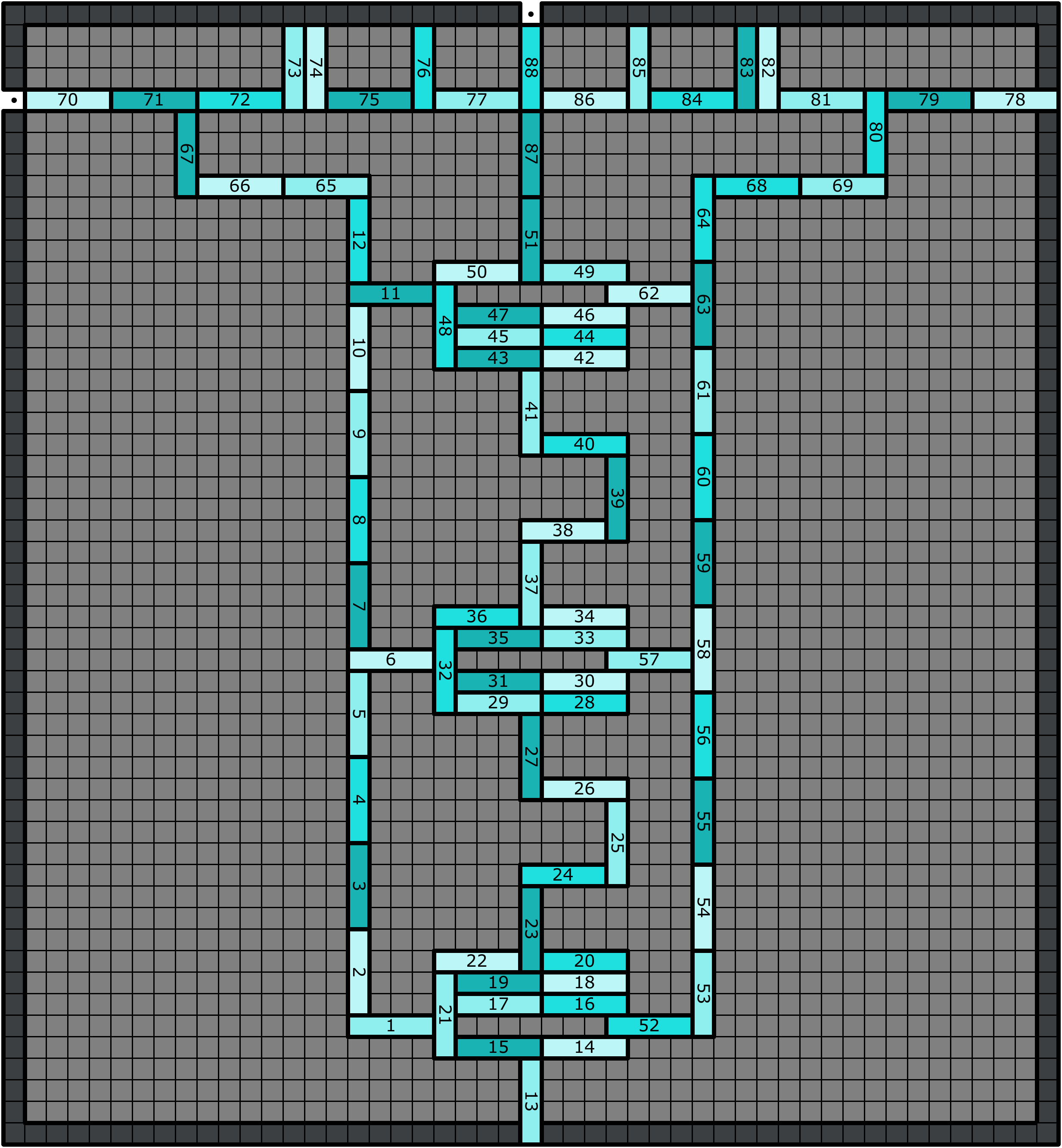}
    \caption{First tiling}
  \end{subfigure}
  \begin{subfigure}[b]{0.49\textwidth}
    \centering
    \includegraphics[width=220pt]{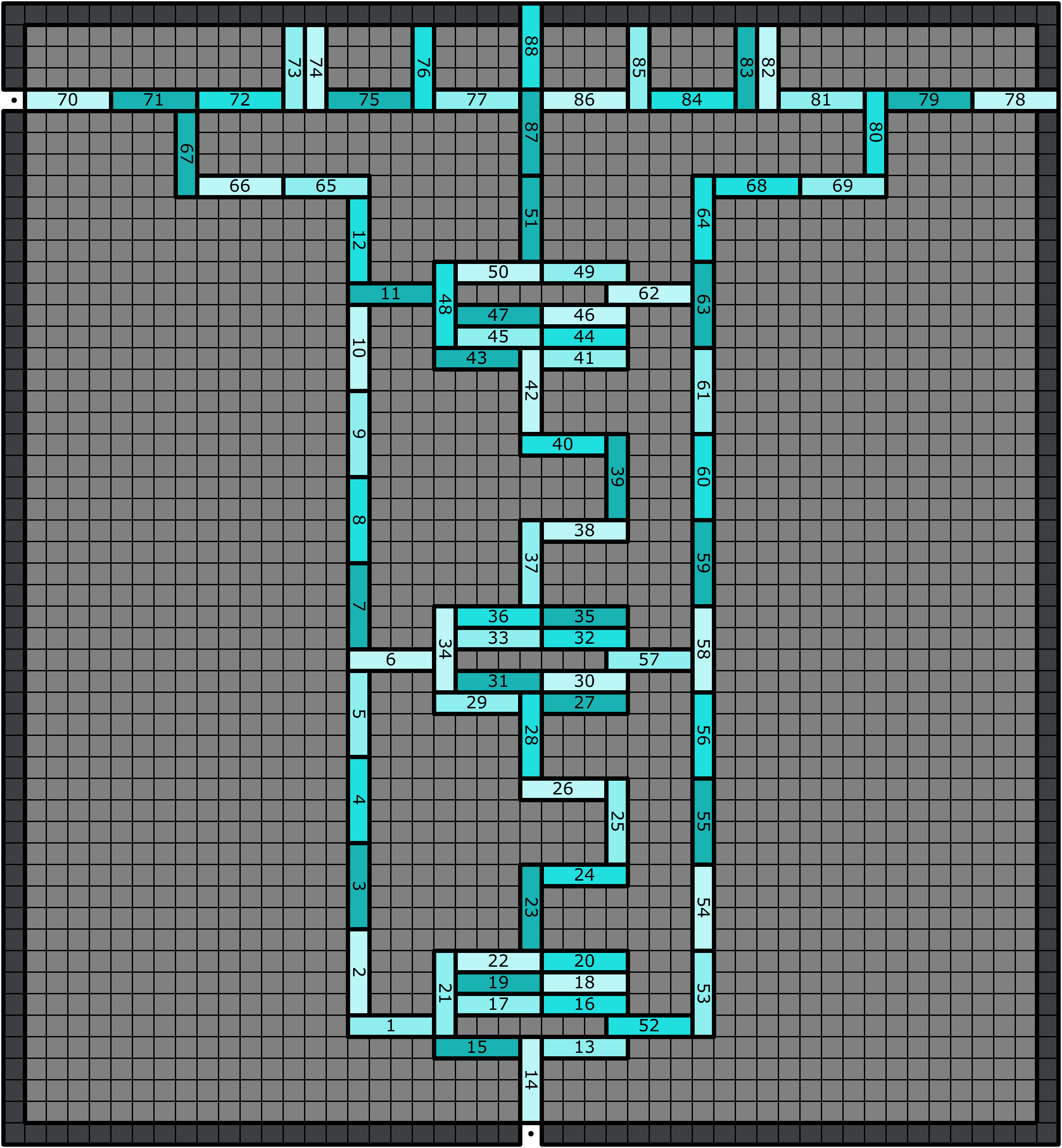}
    \caption{Second tiling}
  \end{subfigure}
  \begin{subfigure}[b]{0.49\textwidth}
    \centering
    \includegraphics[width=220pt]{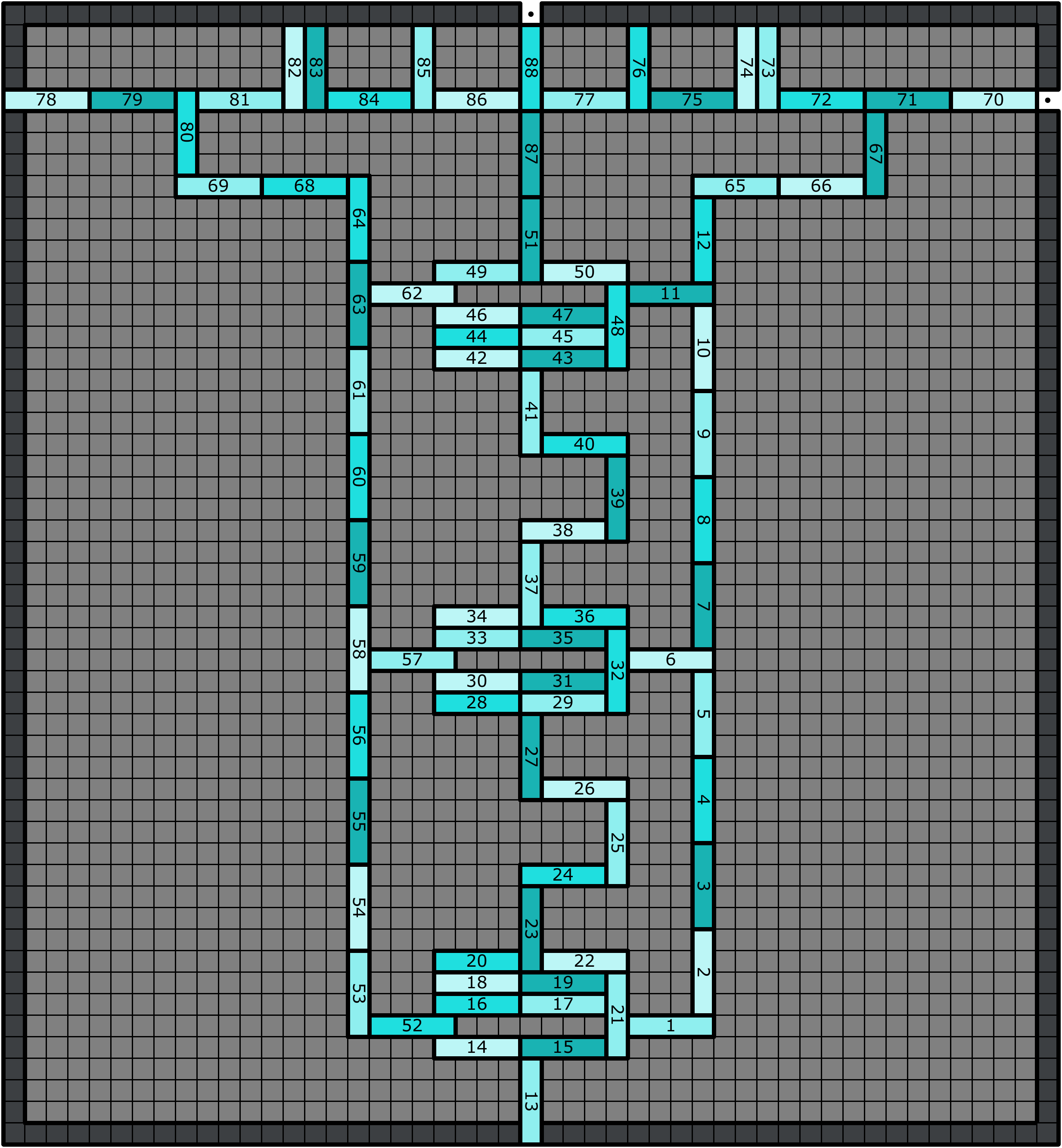}
    \caption{Third tiling}
  \end{subfigure}
  \begin{subfigure}[b]{0.49\textwidth}
    \centering
    \includegraphics[width=220pt]{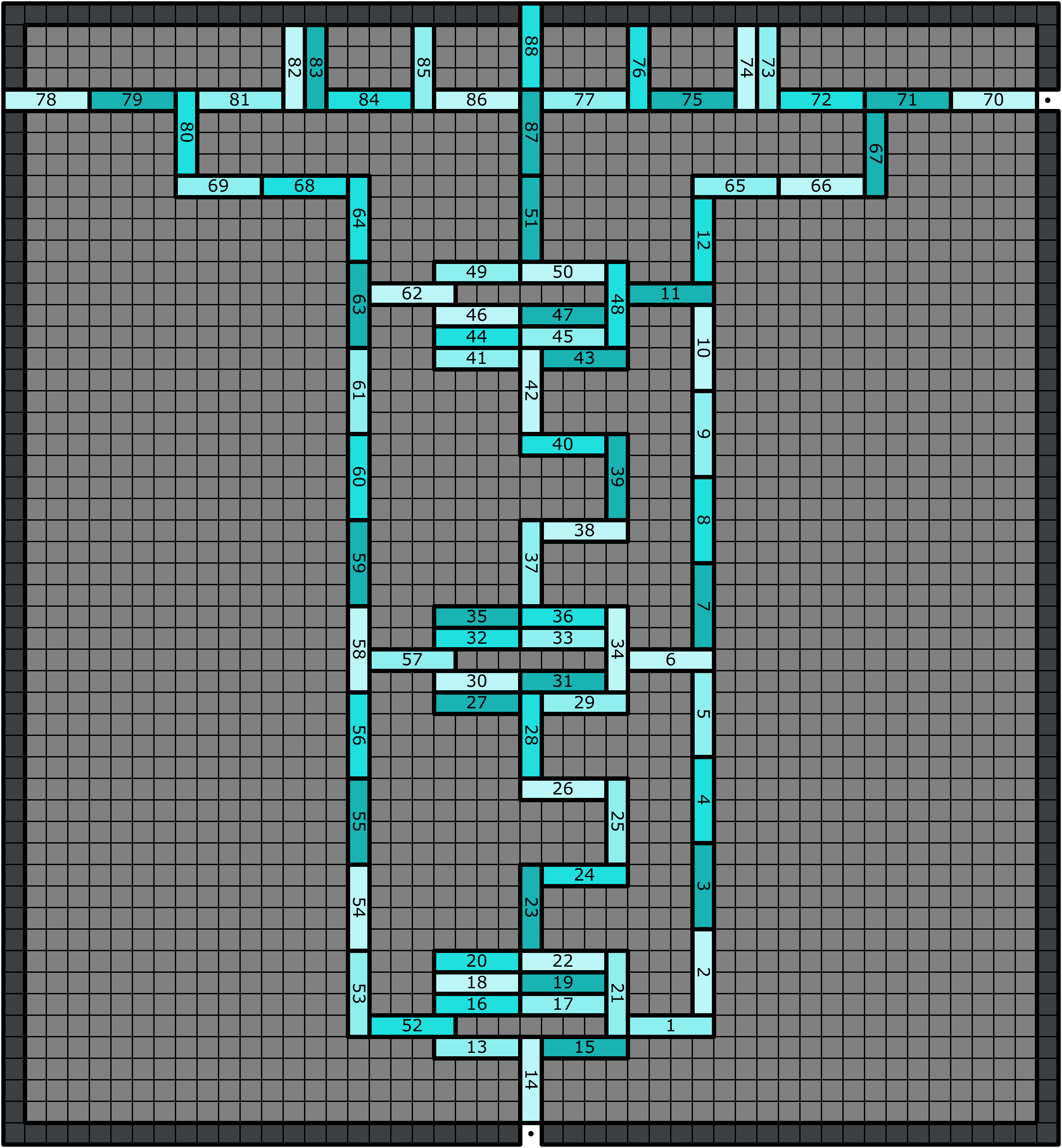}
    \caption{Fourth tiling}
  \end{subfigure}
  \caption{Tilings + placement orders for the crossover gadget}
  \label{fig:i_cross_tilings}
\end{figure}

\section{$\JJ$-tris Gadget Tilings/Fillings}

Like in $\II$-tris, numbers in pieces denote placement order, with smaller numbers placed first.

\begin{figure}[!ht]
  \centering
  \begin{subfigure}[b]{0.49\textwidth}
    \centering
    \includegraphics[width=200pt]{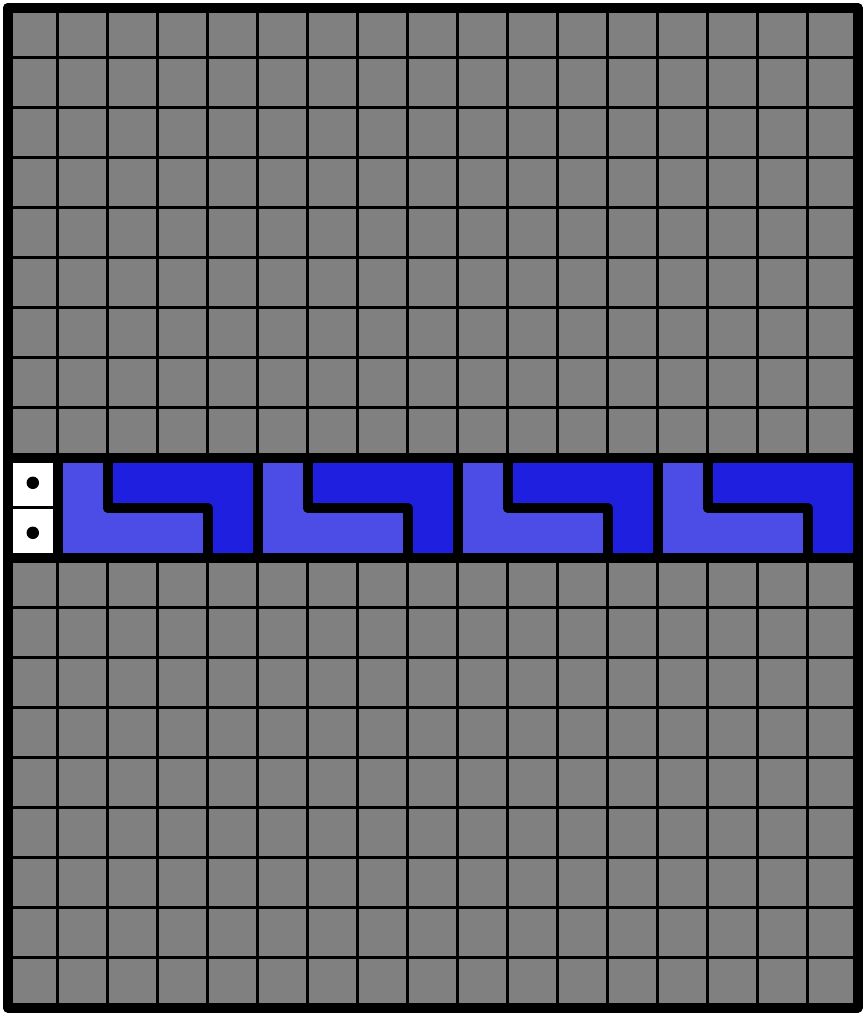}
    \caption{First tiling}
  \end{subfigure}
  \begin{subfigure}[b]{0.49\textwidth}
    \centering
    \includegraphics[width=200pt]{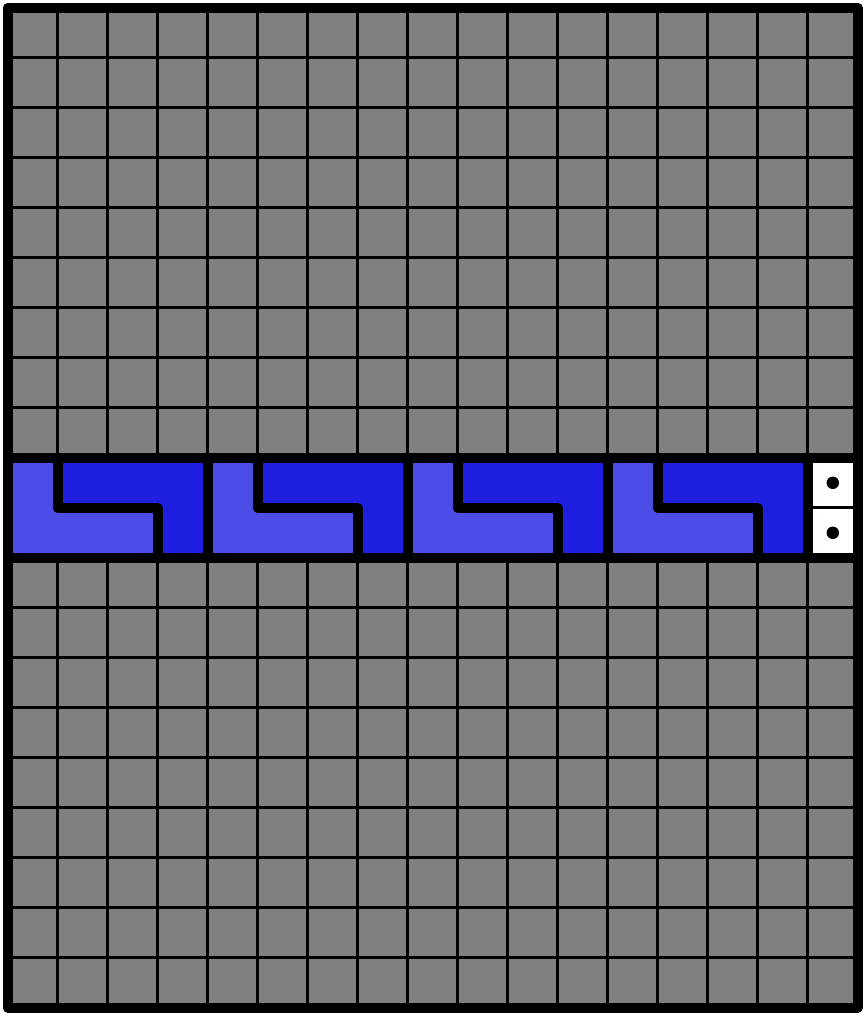}
    \caption{Second tiling}
  \end{subfigure}
  \caption{Tilings for an HL gadget}
  \label{fig:j_horline_tilings}
\end{figure}

\begin{figure}[!ht]
  \centering
  \begin{subfigure}[b]{0.49\textwidth}
    \centering
    \includegraphics[width=200pt]{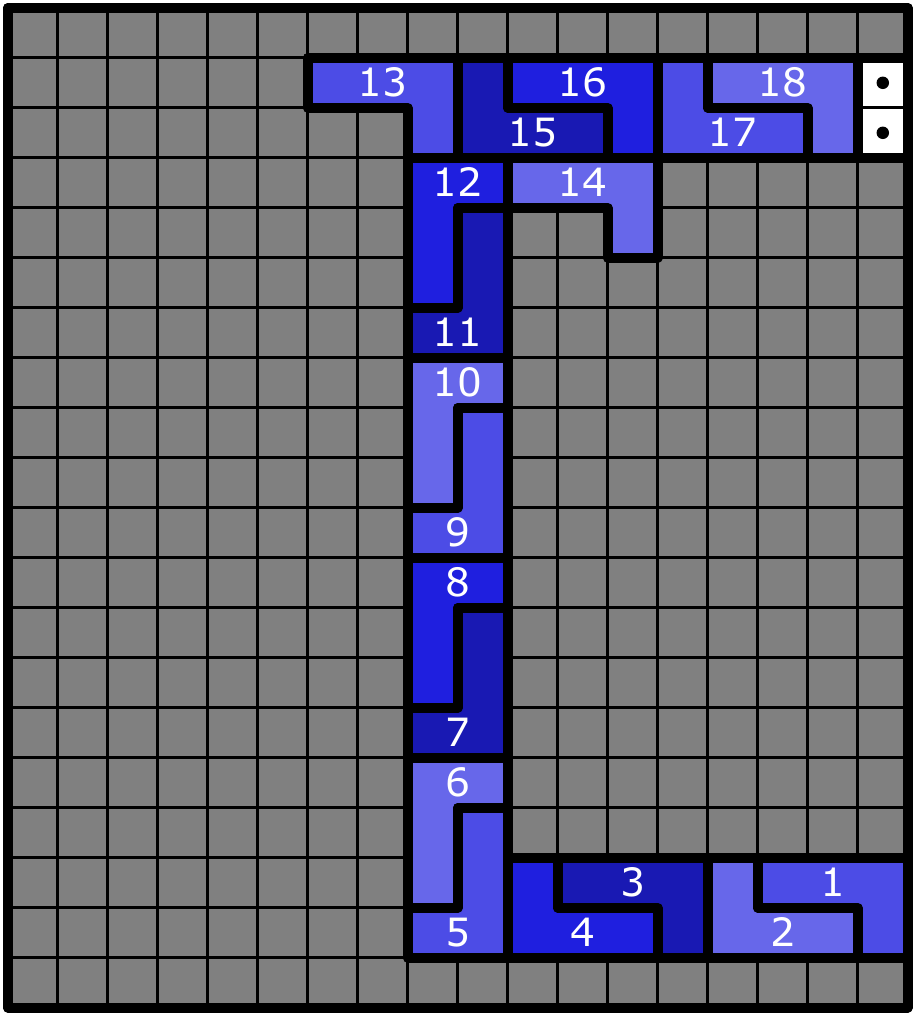}
    \caption{First tiling for a UL gadget}
  \end{subfigure}
  \begin{subfigure}[b]{0.49\textwidth}
    \centering
    \includegraphics[width=200pt]{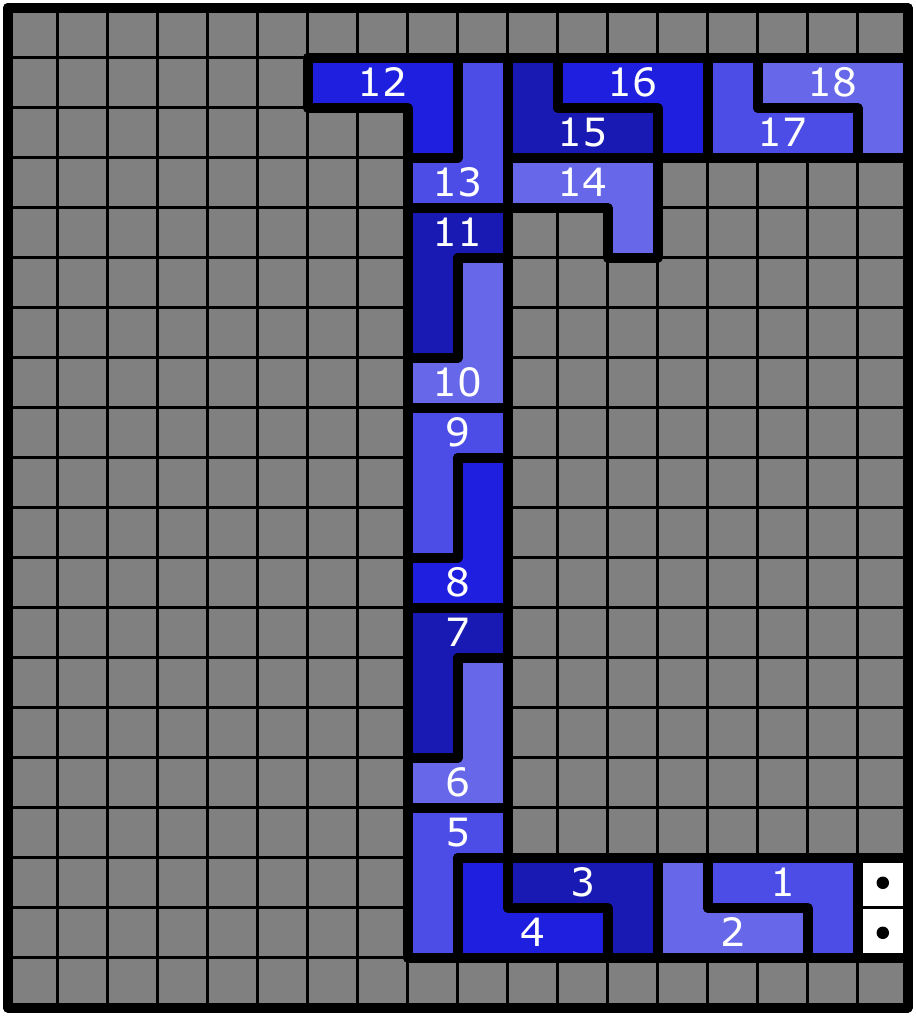}
    \caption{Second tiling for a UL gadget}
  \end{subfigure}
  \begin{subfigure}[b]{0.49\textwidth}
    \centering
    \includegraphics[width=200pt]{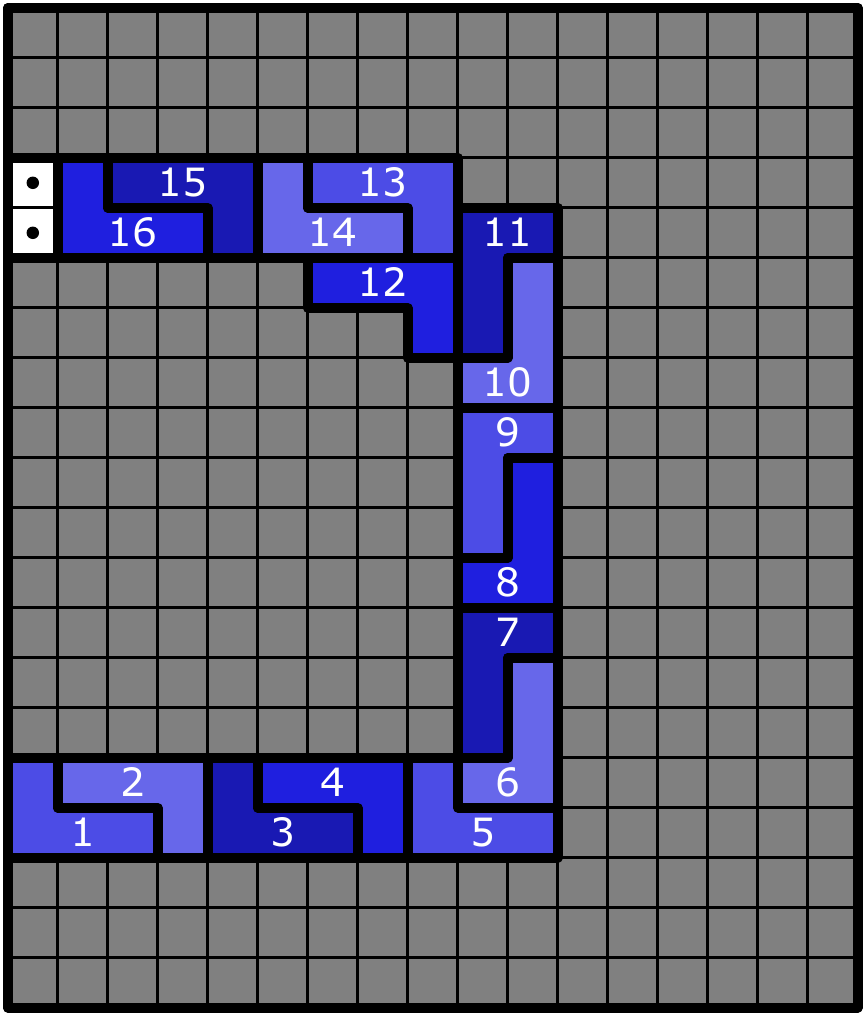}
    \caption{First tiling for a UR gadget}
  \end{subfigure}
  \begin{subfigure}[b]{0.49\textwidth}
    \centering
    \includegraphics[width=200pt]{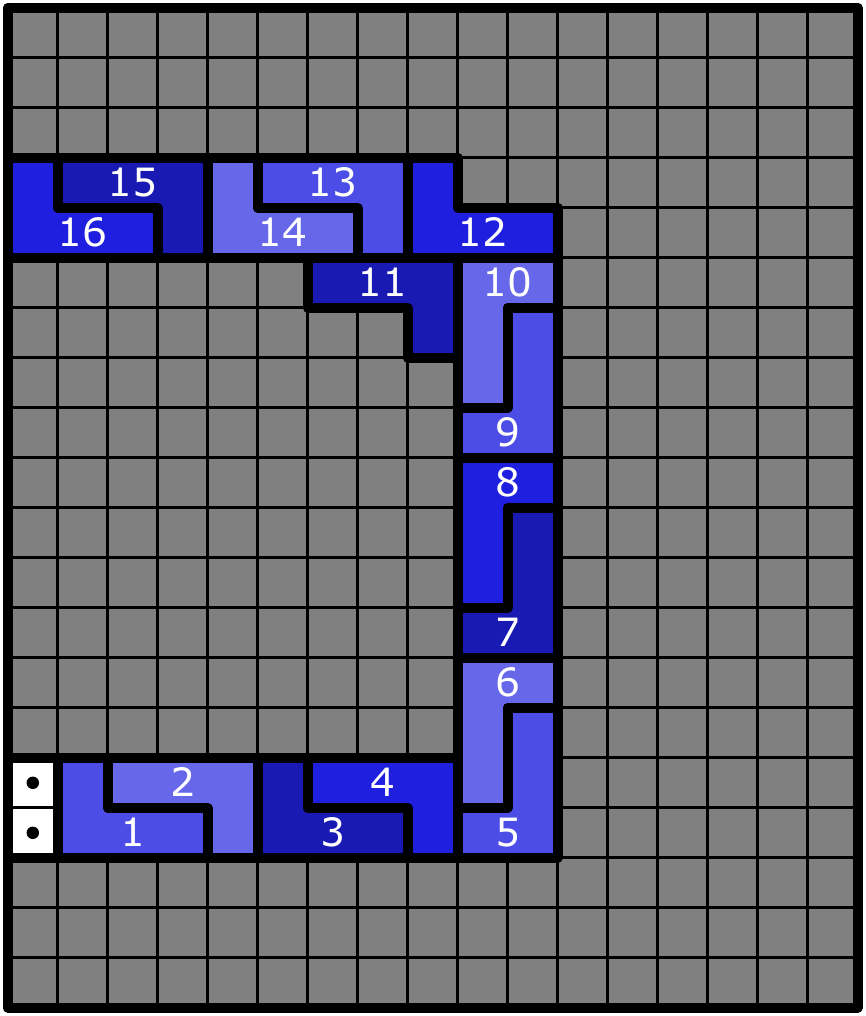}
    \caption{Second tiling for a UR gadget}
  \end{subfigure}
  \caption{Tilings + placement orders for U-turn gadgets}
  \label{fig:j_uturn_tilings}
\end{figure}

\begin{figure}[!ht]
  \centering
  \begin{subfigure}[b]{0.49\textwidth}
    \centering
    \includegraphics[width=200pt]{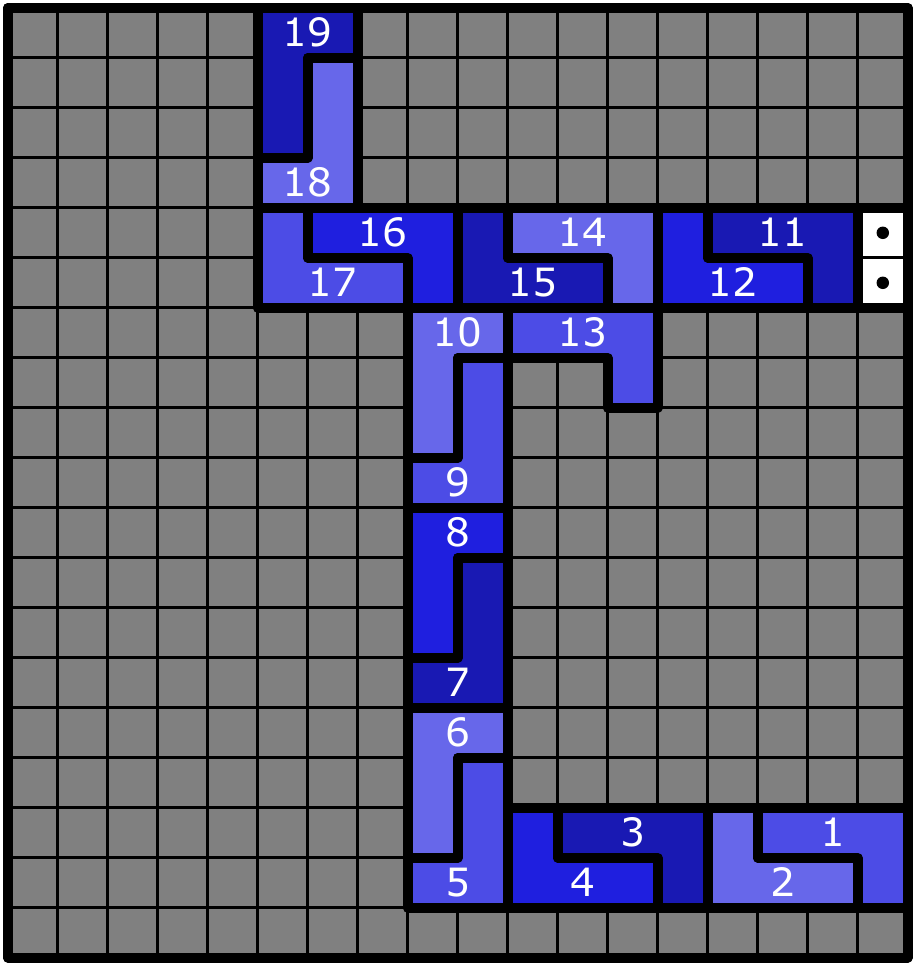}
    \caption{First tiling}
  \end{subfigure}
  \begin{subfigure}[b]{0.49\textwidth}
    \centering
    \includegraphics[width=200pt]{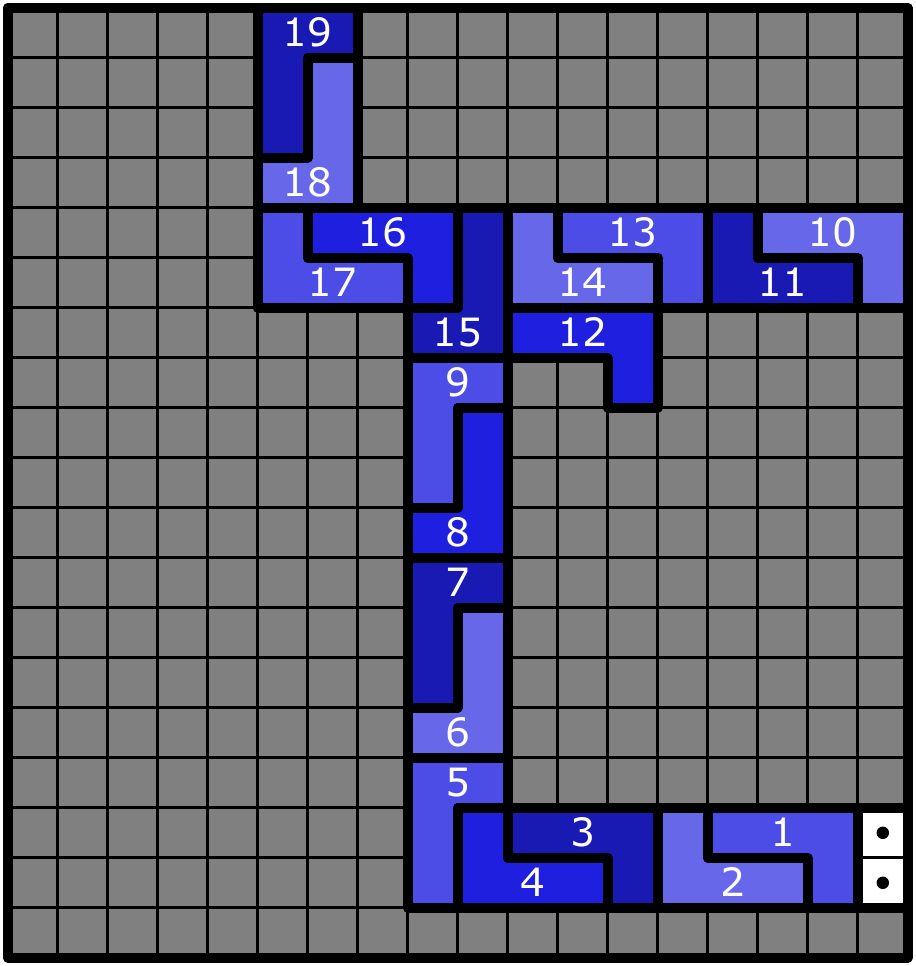}
    \caption{Second tiling}
  \end{subfigure}
  \caption{Tilings + placement orders for an EC gadget}
  \label{fig:j_entrycorner_tilings}
\end{figure}

\begin{figure}[!ht]
  \centering
  \begin{subfigure}[b]{0.49\textwidth}
    \centering
    \includegraphics[width=200pt]{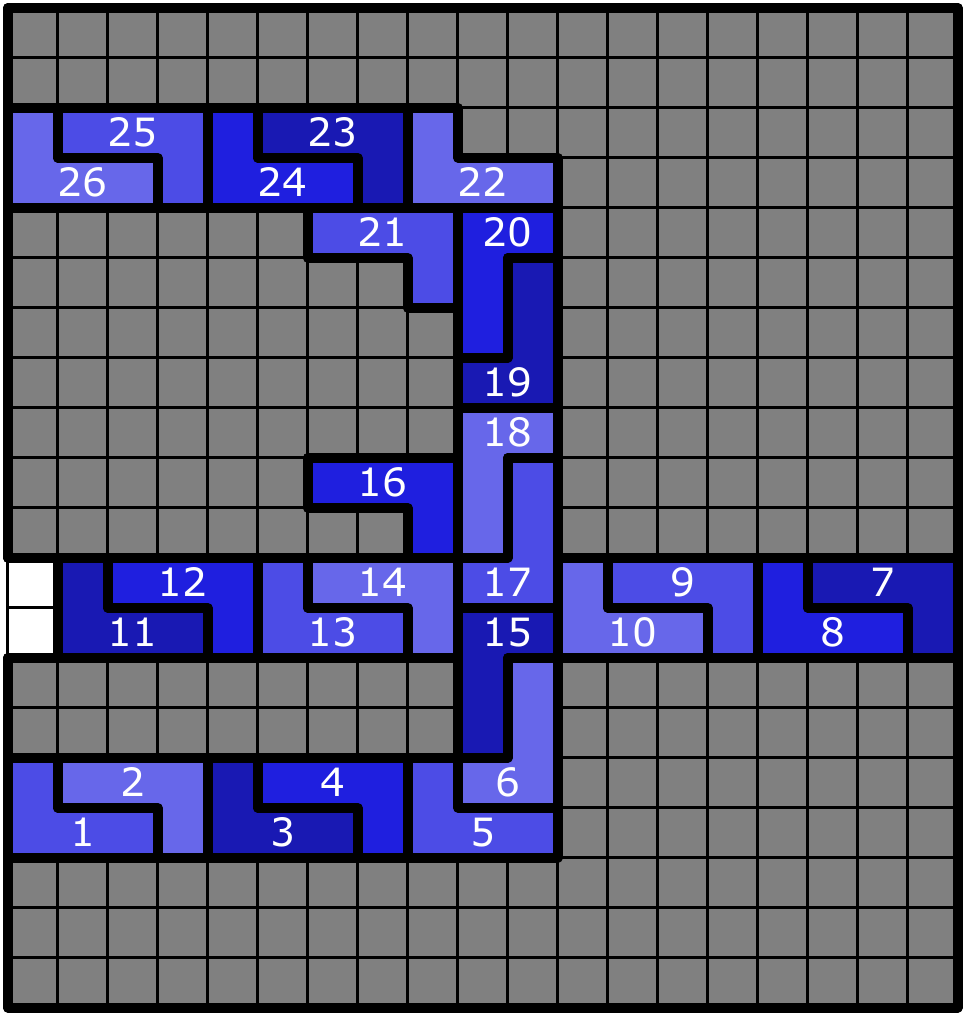}
    \caption{First tiling}
  \end{subfigure}
  \begin{subfigure}[b]{0.49\textwidth}
    \centering
    \includegraphics[width=200pt]{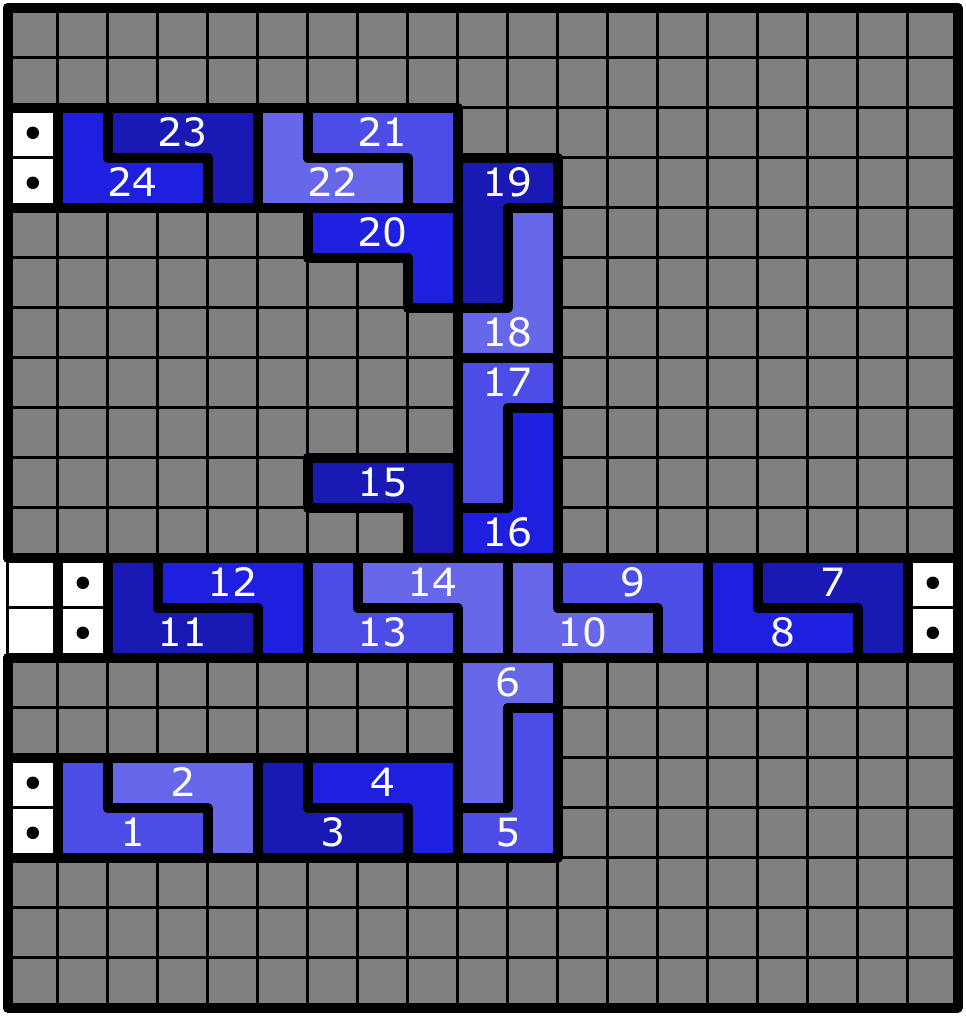}
    \caption{Second tiling}
  \end{subfigure}
  \caption{Tilings + placement orders for the $0$-or-$4$ gadget}
  \label{fig:j_0mod4_tilings}
\end{figure}

\begin{figure}[!ht]
  \centering
  \begin{subfigure}[b]{0.49\textwidth}
    \centering
    \includegraphics[width=150pt]{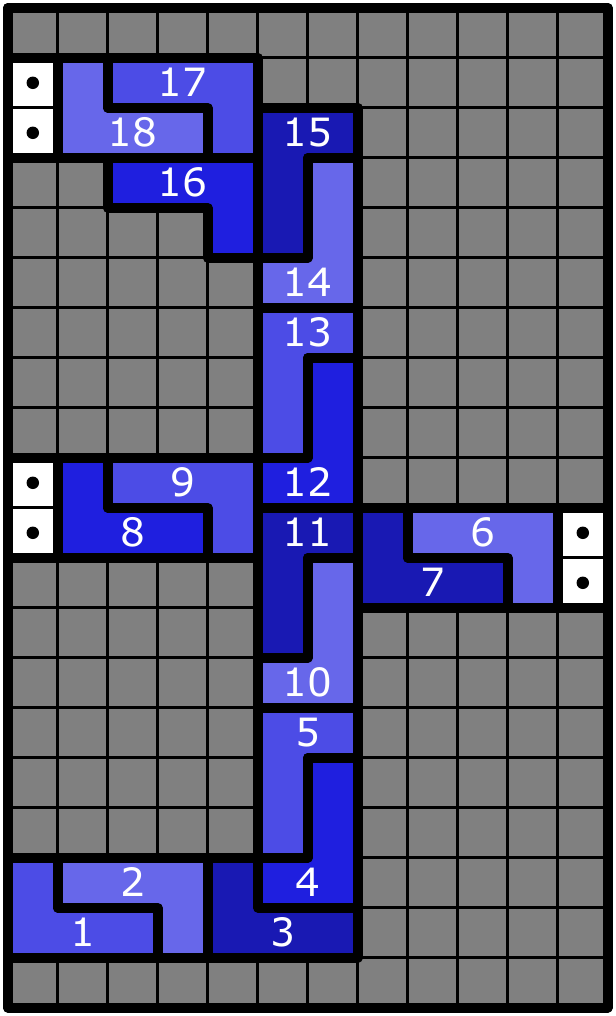}
    \caption{First tiling}
  \end{subfigure}
  \begin{subfigure}[b]{0.49\textwidth}
    \centering
    \includegraphics[width=150pt]{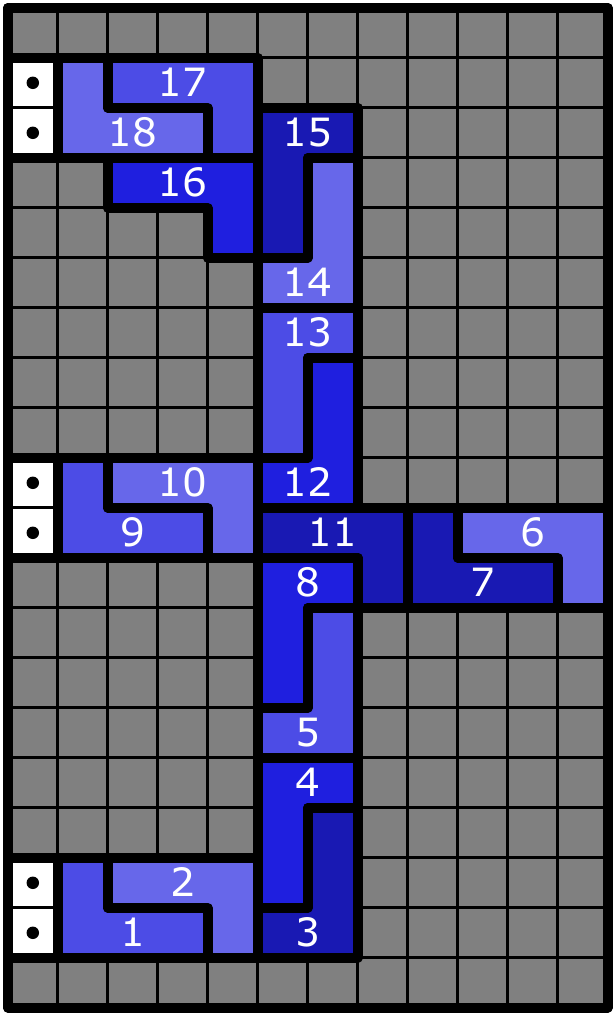}
    \caption{Second tiling}
  \end{subfigure}
  \begin{subfigure}[b]{0.49\textwidth}
    \centering
    \includegraphics[width=150pt]{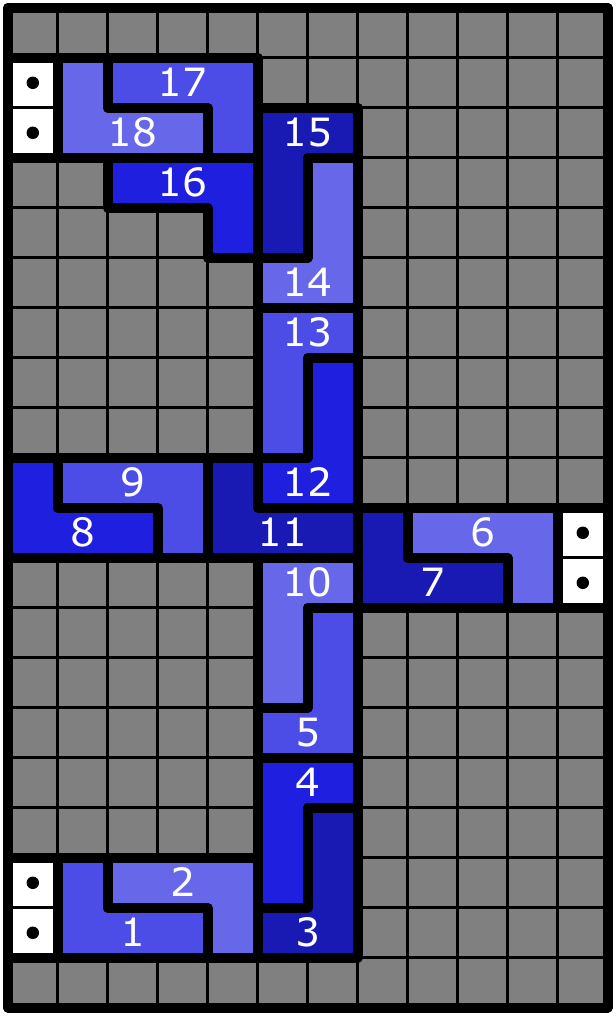}
    \caption{Third tiling}
  \end{subfigure}
  \begin{subfigure}[b]{0.49\textwidth}
    \centering
    \includegraphics[width=150pt]{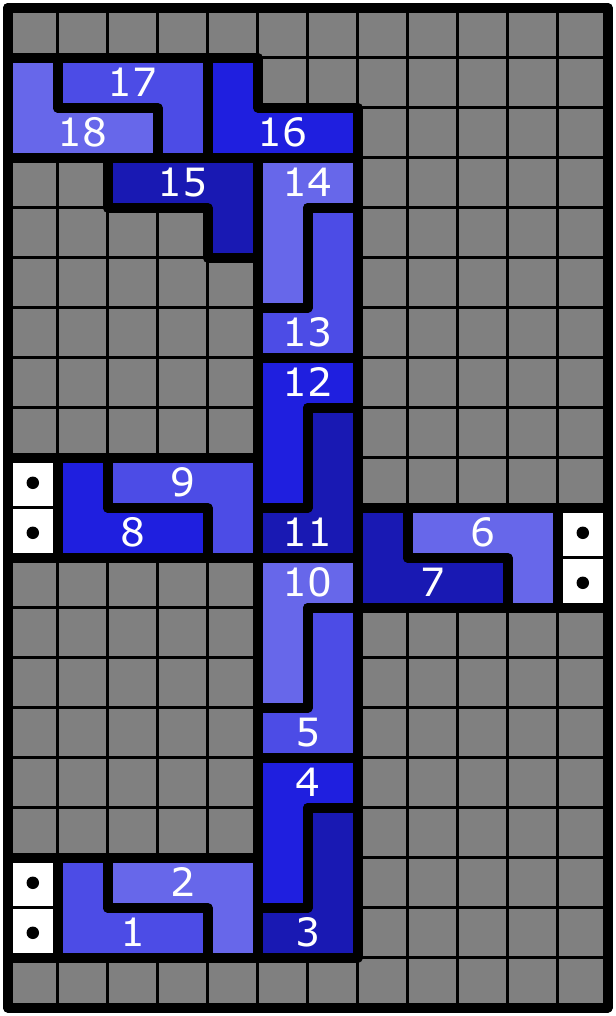}
    \caption{Fourth tiling}
  \end{subfigure}
  \caption{Tilings + placement orders for the $3$-in-$4$ gadget}
  \label{fig:j_3mod4_tilings}
\end{figure}

\begin{figure}[!ht]
  \centering
  \begin{subfigure}[b]{0.49\textwidth}
    \centering
    \includegraphics[width=120pt]{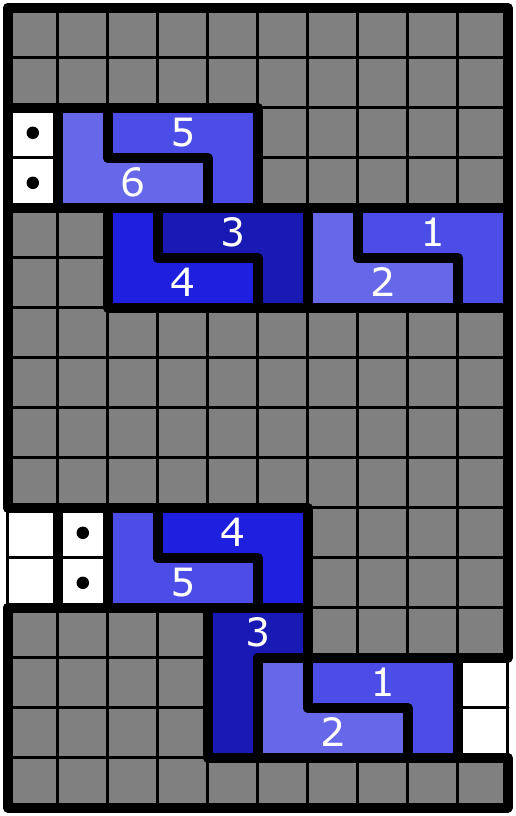}
    \caption{First tiling}
  \end{subfigure}
  \begin{subfigure}[b]{0.49\textwidth}
    \centering
    \includegraphics[width=120pt]{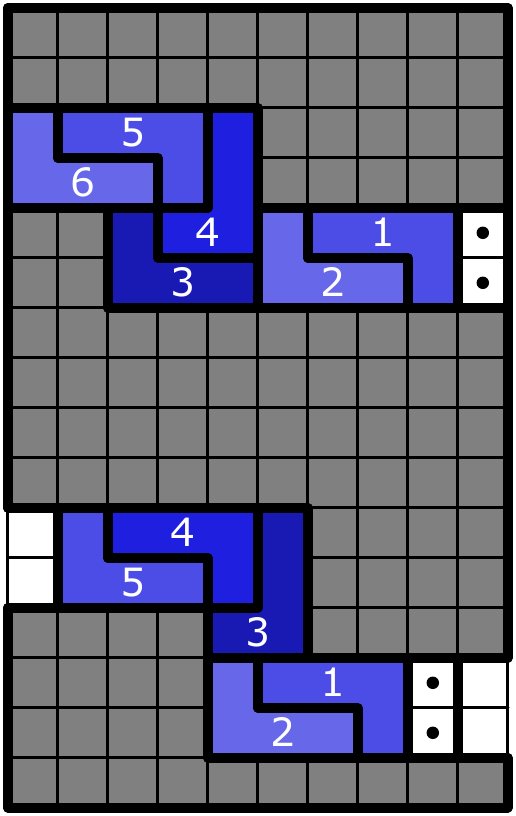}
    \caption{Second tiling}
  \end{subfigure}
  \caption{Tilings + placement orders for the PF gadgets}
  \label{fig:j_parityfixer_tilings}
\end{figure}

\section{$\TT$-tris Gadget Tilings/Fillings}

Like in $\II$-tris, numbers in pieces denote placement order, with smaller numbers placed first.

\begin{figure}[!ht]
  \centering
  \begin{subfigure}[b]{0.49\textwidth}
    \centering
    \includegraphics[width=200pt]{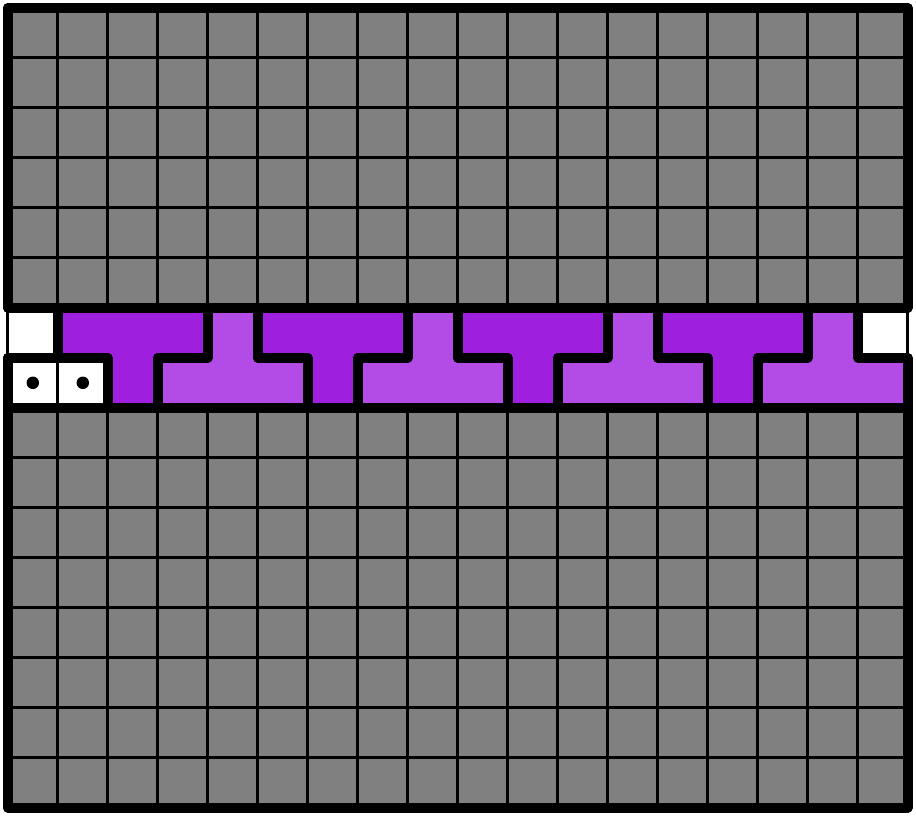}
    \caption{First tiling}
  \end{subfigure}
  \begin{subfigure}[b]{0.49\textwidth}
    \centering
    \includegraphics[width=200pt]{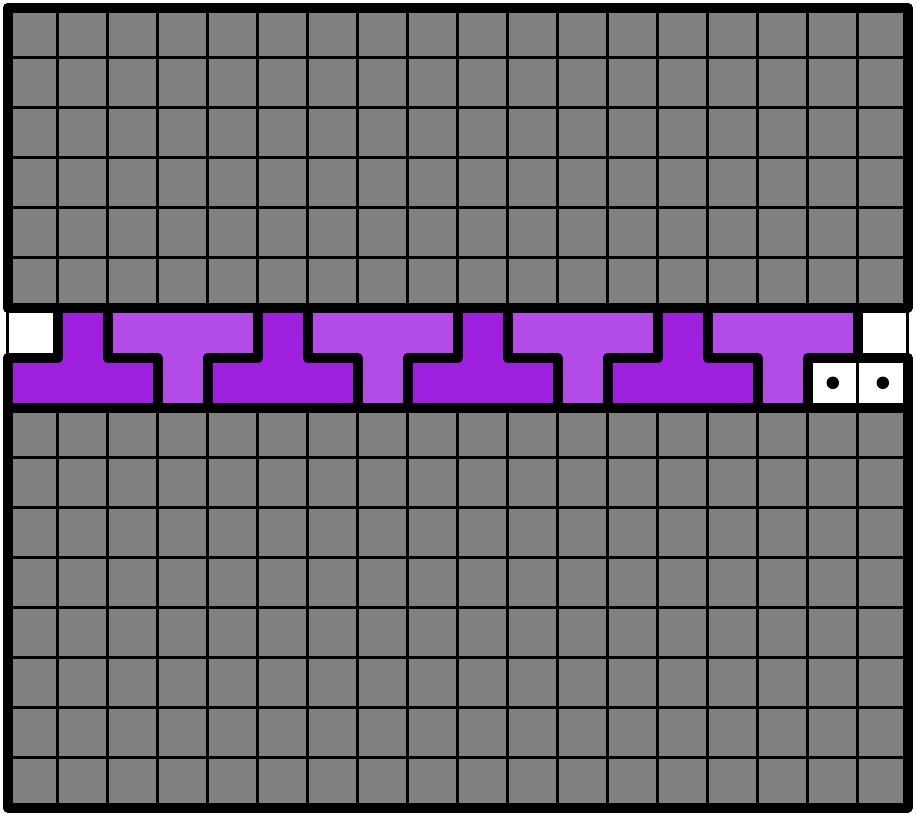}
    \caption{Second tiling}
  \end{subfigure}
  \caption{Tilings for an HL gadget}
  \label{fig:t_horline_tilings}
\end{figure}

\begin{figure}[!ht]
  \centering
  \begin{subfigure}[b]{0.49\textwidth}
    \centering
    \includegraphics[width=200pt]{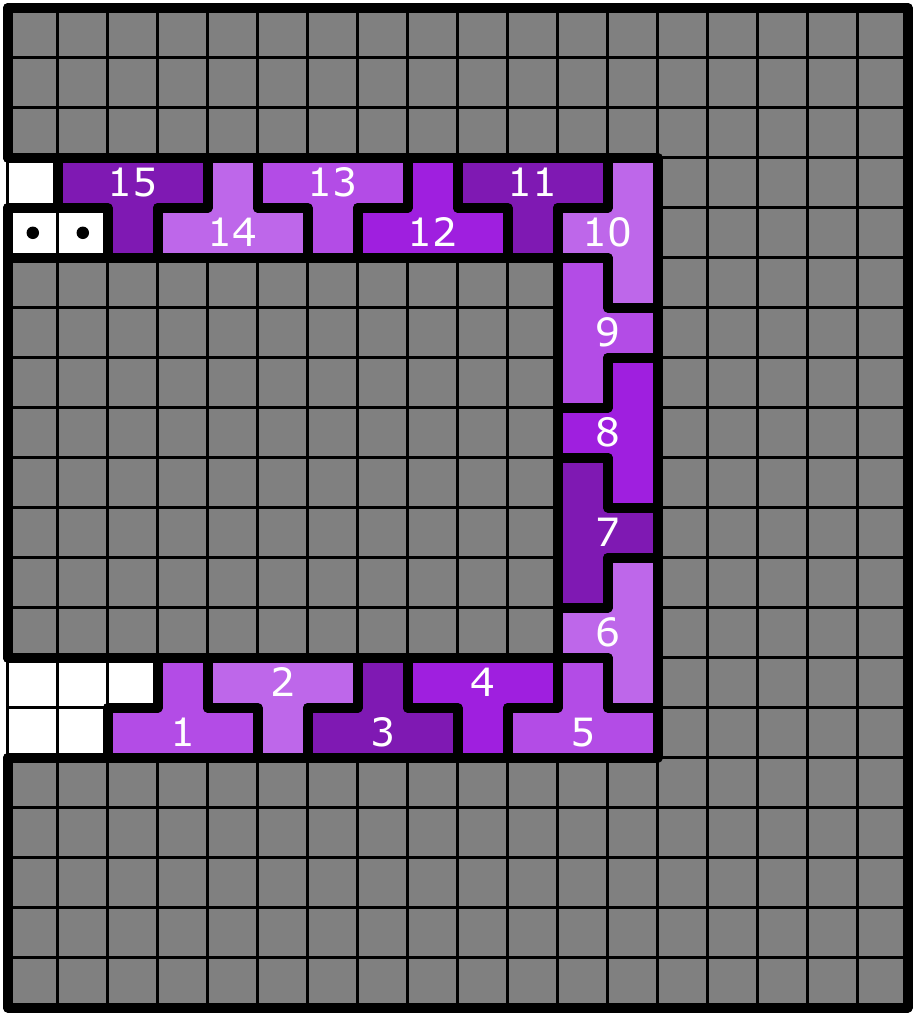}
    \caption{First tiling}
  \end{subfigure}
  \begin{subfigure}[b]{0.49\textwidth}
    \centering
    \includegraphics[width=200pt]{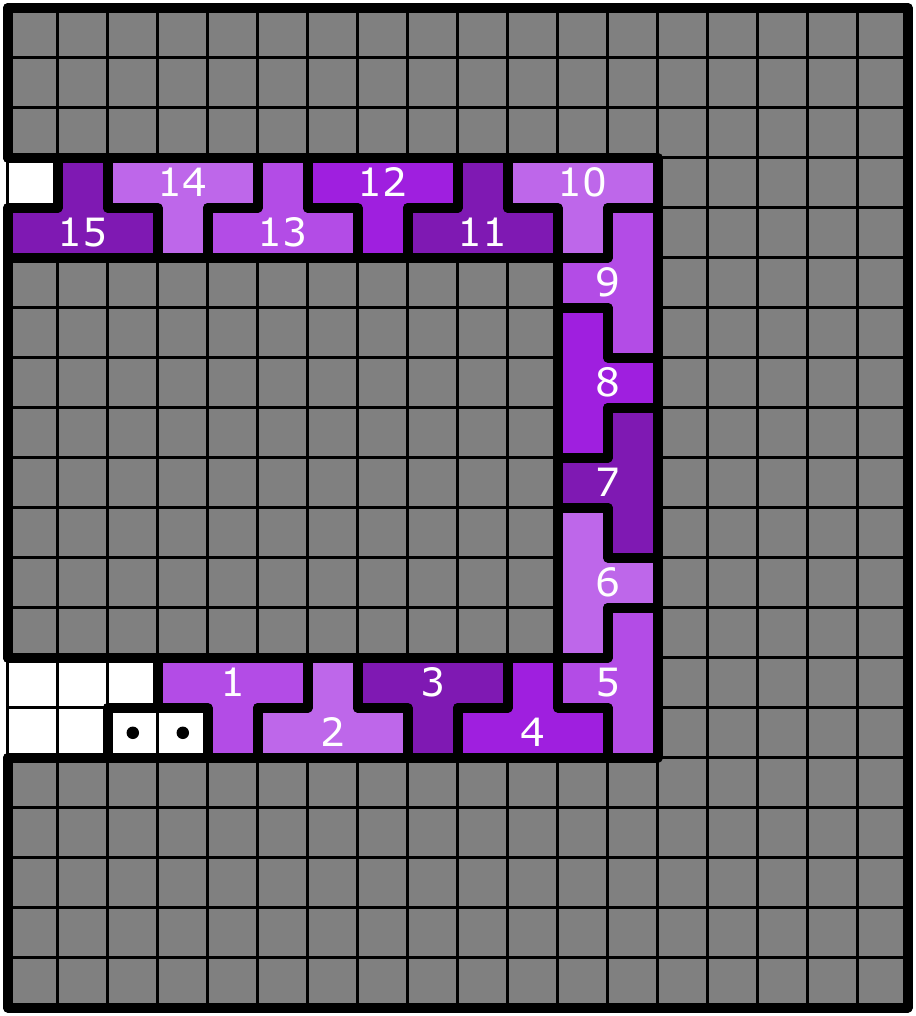}
    \caption{Second tiling}
  \end{subfigure}
  \caption{Tilings + placement orders for the U-turn gadget (the left-down-right case is symmetric)}
  \label{fig:t_uturn_tilings}
\end{figure}

\begin{figure}[!ht]
  \centering
  \begin{subfigure}[b]{0.49\textwidth}
    \centering
    \includegraphics[width=200pt]{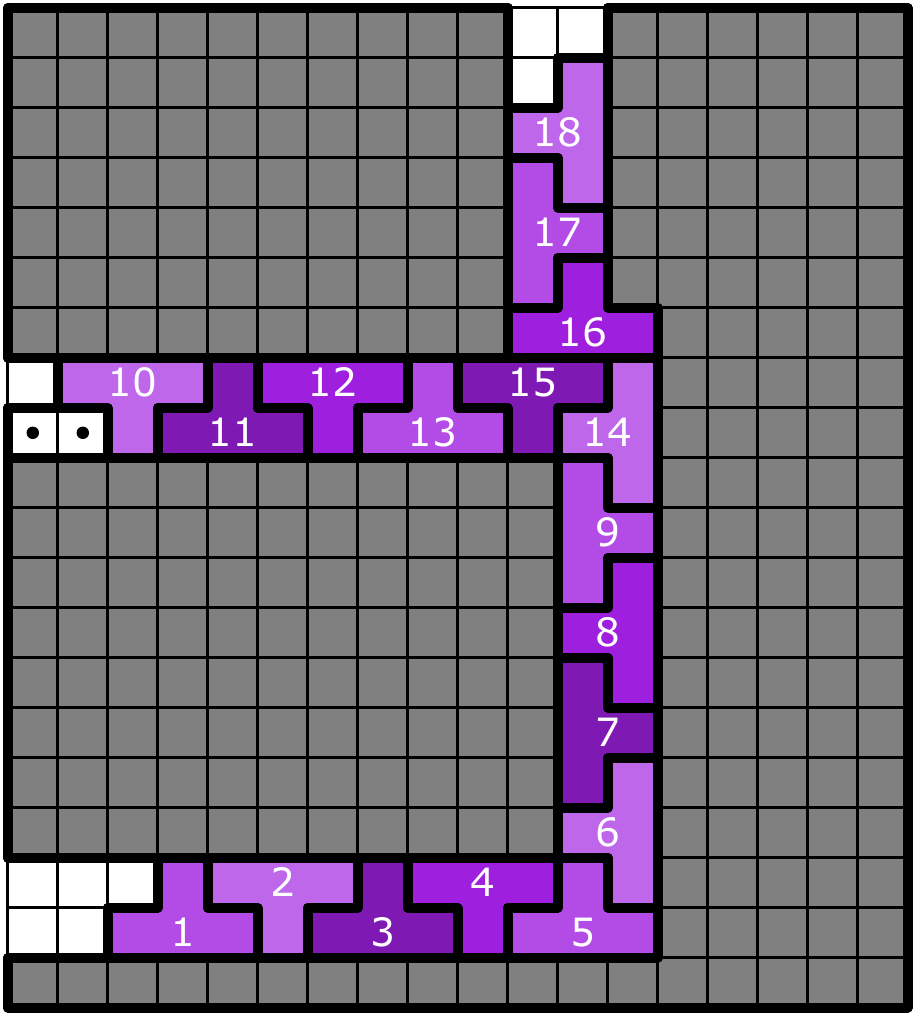}
    \caption{First tiling}
  \end{subfigure}
  \begin{subfigure}[b]{0.49\textwidth}
    \centering
    \includegraphics[width=200pt]{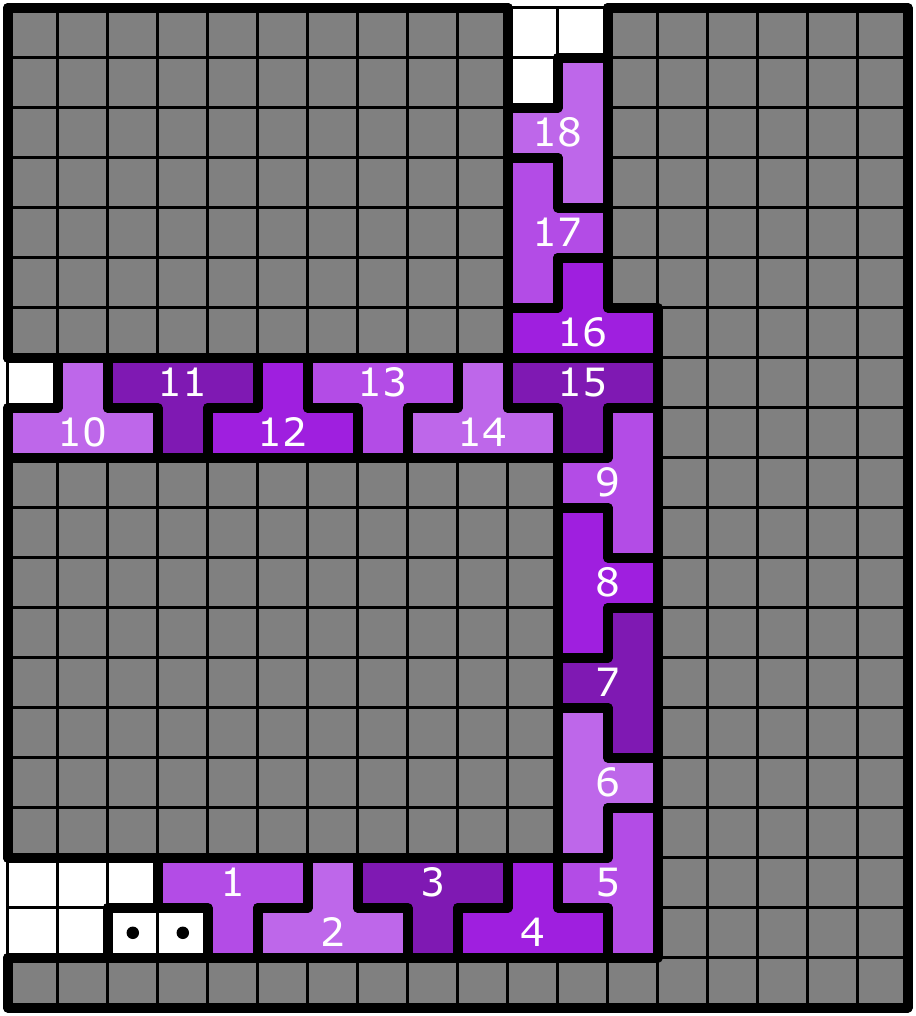}
    \caption{Second tiling}
  \end{subfigure}
  \caption{Tilings + placement orders for an EC gadget}
  \label{fig:t_entrycorner_tilings}
\end{figure}

\begin{figure}[!ht]
  \centering
  \begin{subfigure}[b]{0.49\textwidth}
    \centering
    \includegraphics[width=200pt]{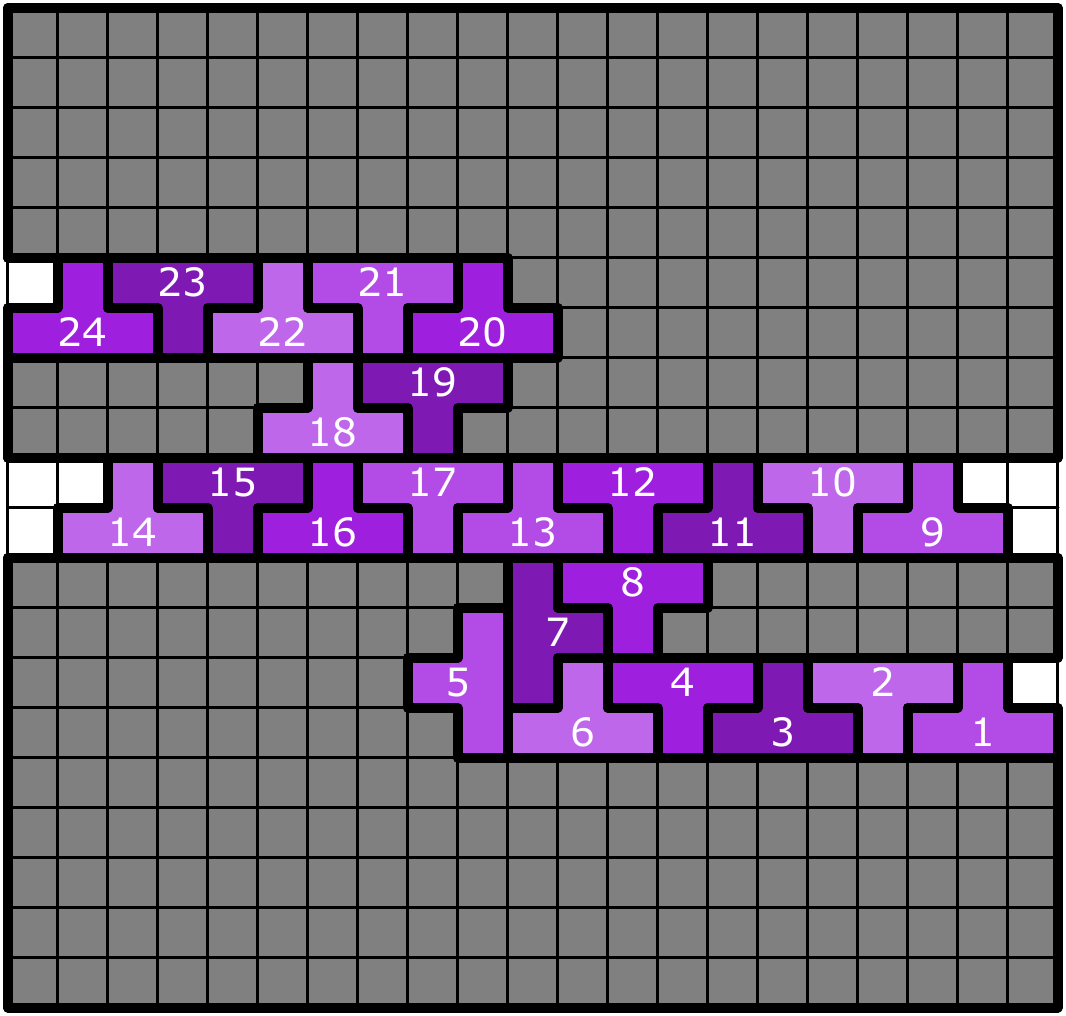}
    \caption{First tiling}
  \end{subfigure}
  \begin{subfigure}[b]{0.49\textwidth}
    \centering
    \includegraphics[width=200pt]{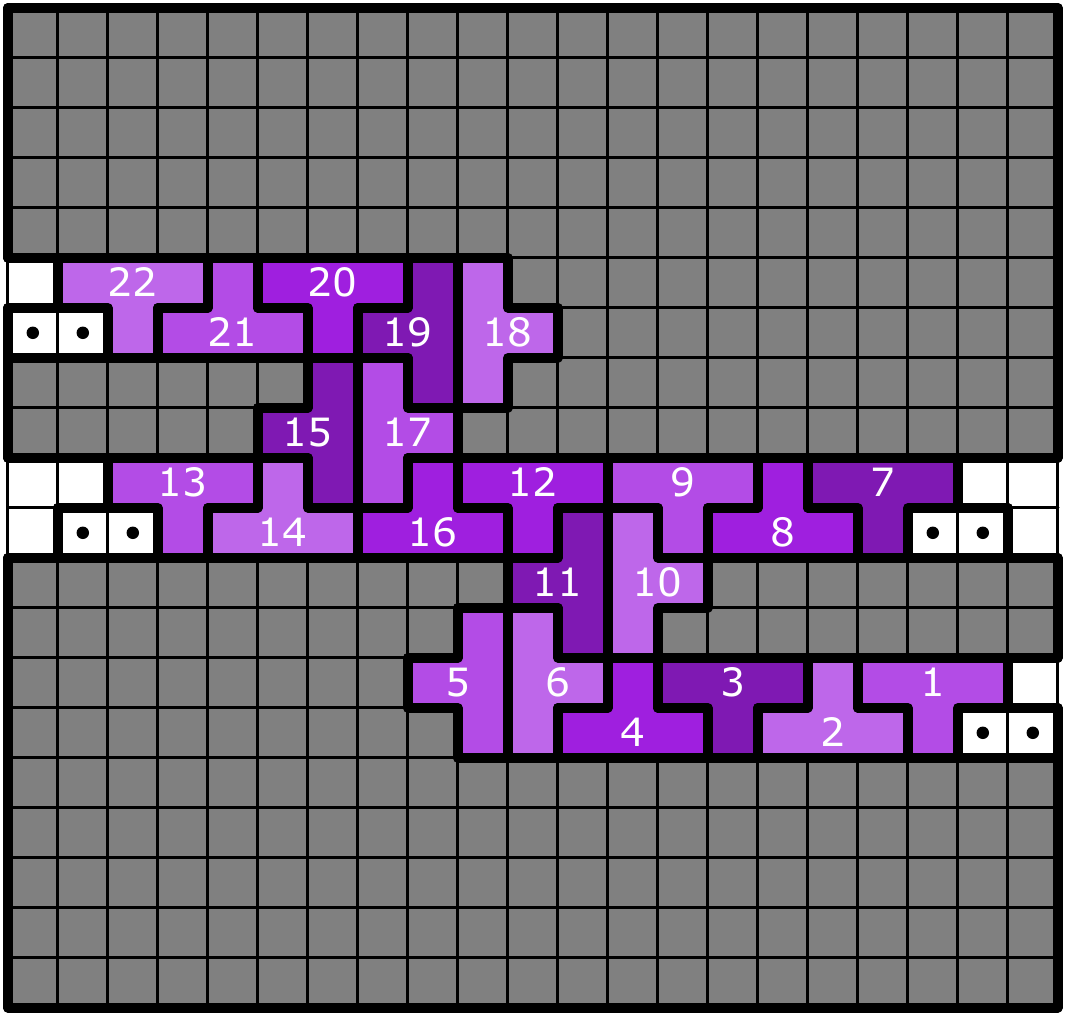}
    \caption{Second tiling}
  \end{subfigure}
  \caption{Tilings + placement orders for the $0$-or-$4$ gadget}
  \label{fig:t_0mod4_tilings}
\end{figure}

\begin{figure}[!ht]
  \centering
  \begin{subfigure}[b]{0.49\textwidth}
    \centering
    \includegraphics[width=200pt]{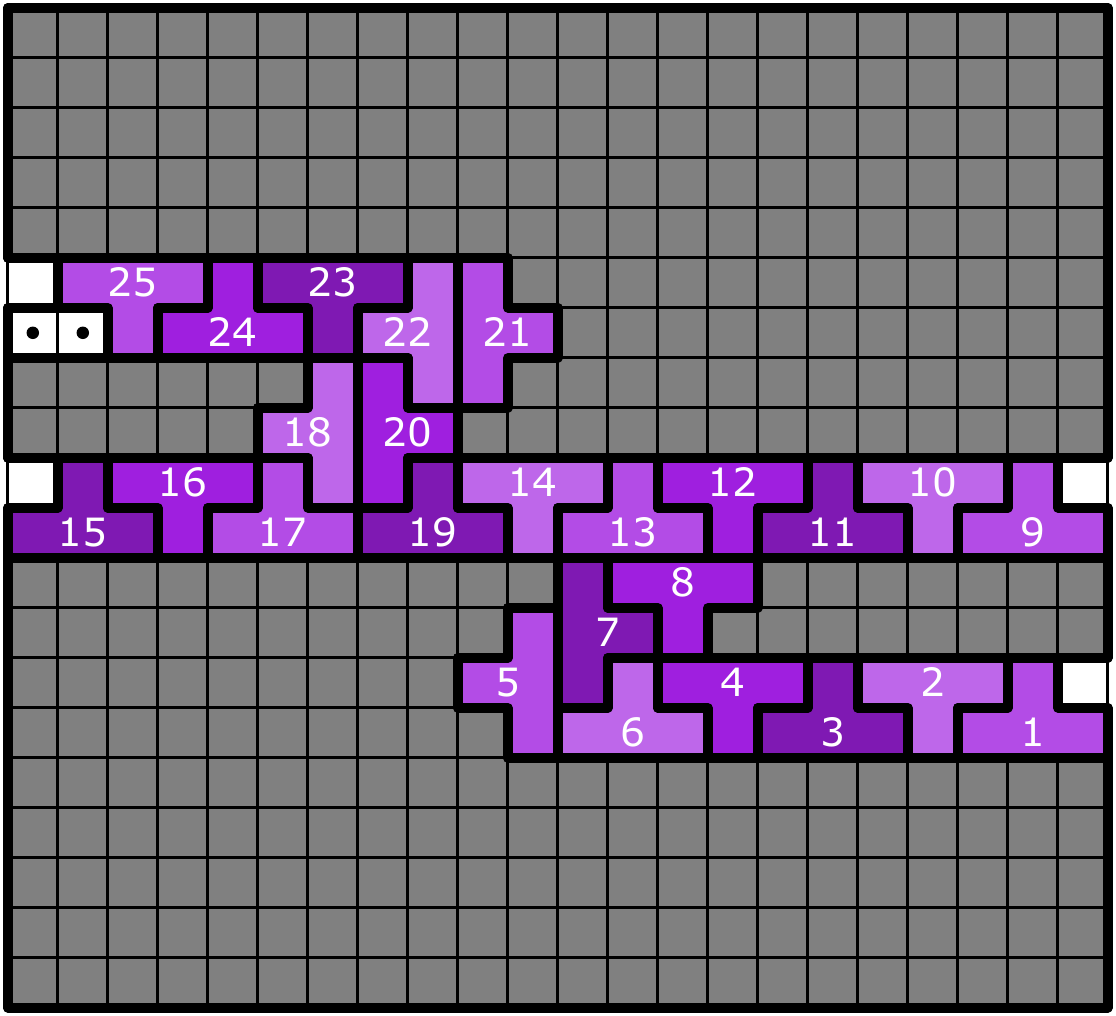}
    \caption{First tiling}
  \end{subfigure}
  \begin{subfigure}[b]{0.49\textwidth}
    \centering
    \includegraphics[width=200pt]{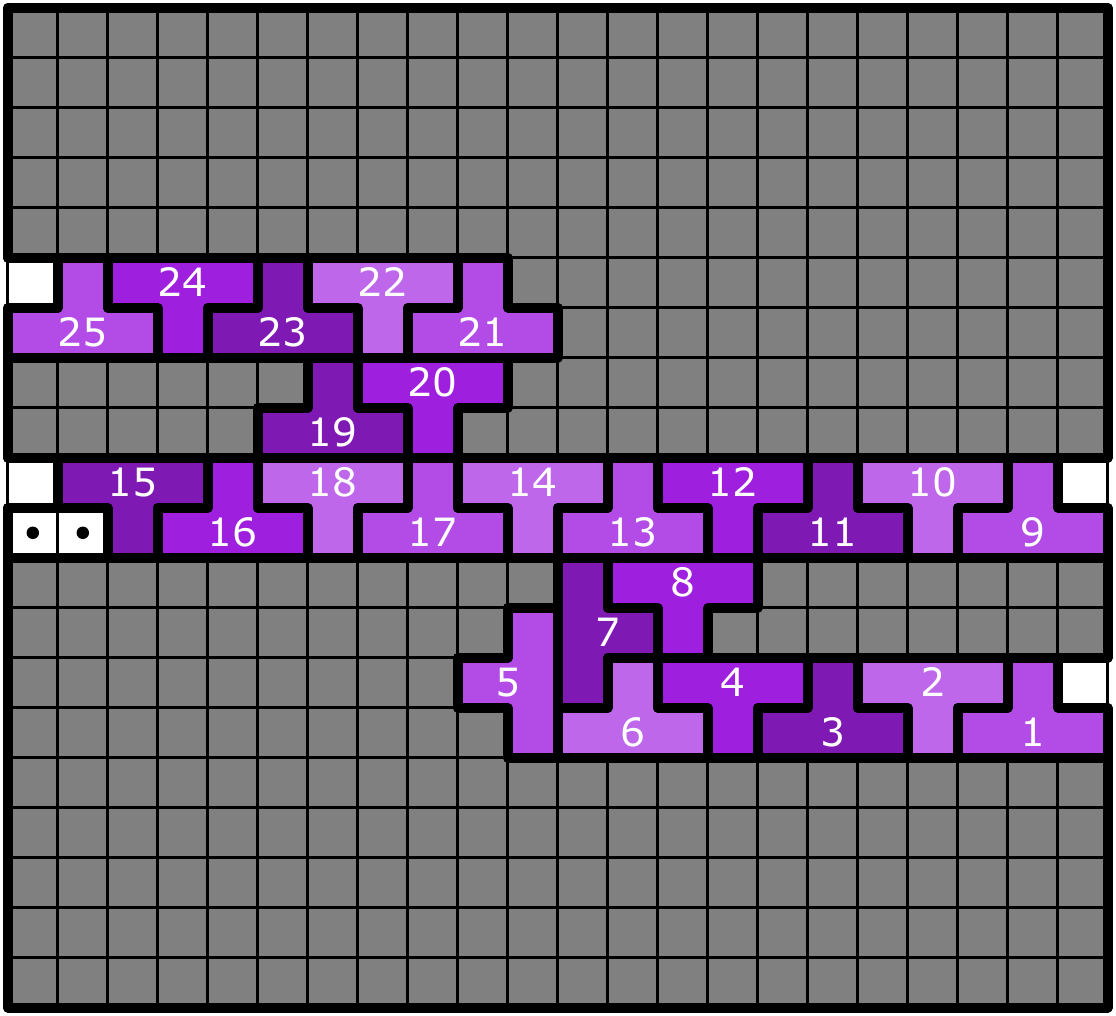}
    \caption{Second tiling}
  \end{subfigure}
  \begin{subfigure}[b]{0.49\textwidth}
    \centering
    \includegraphics[width=200pt]{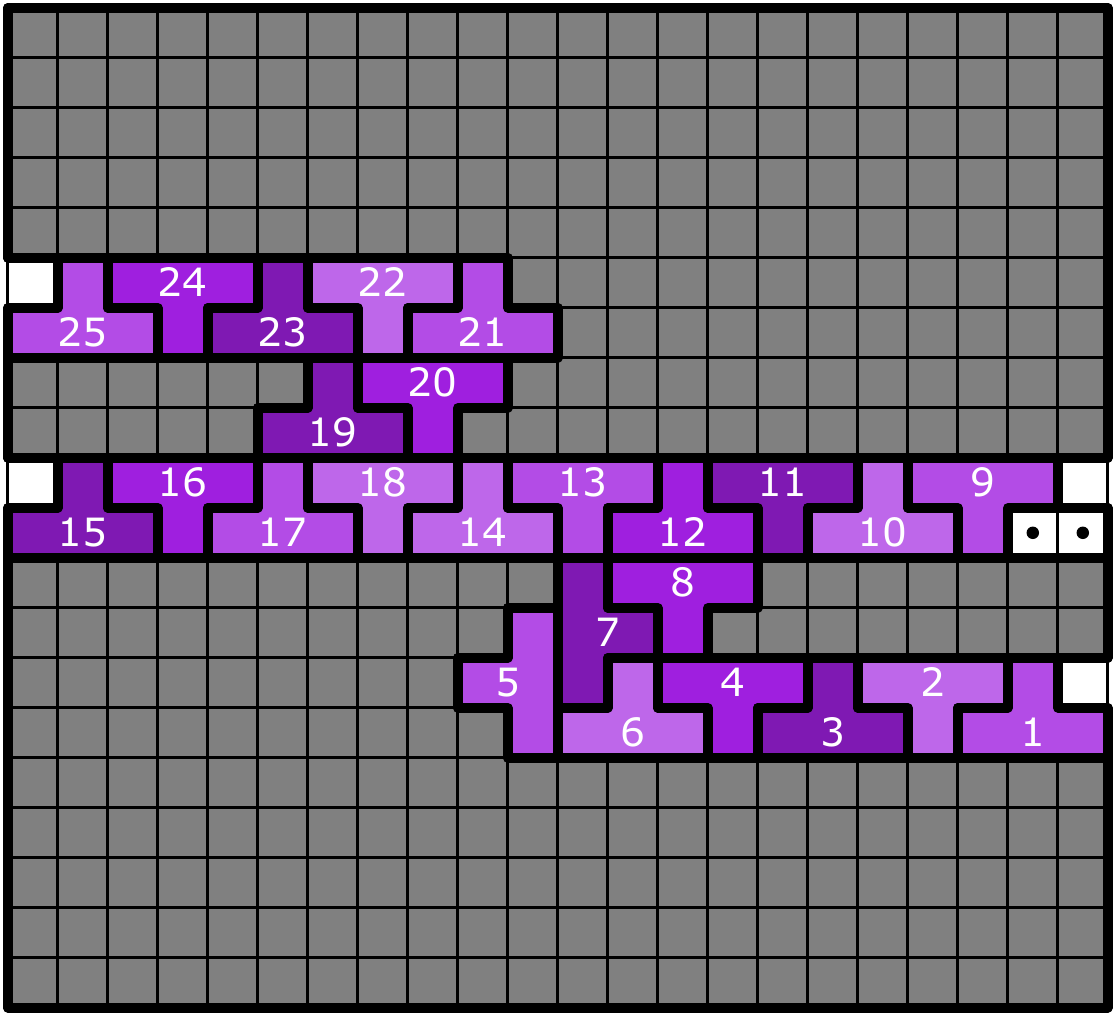}
    \caption{Third tiling}
  \end{subfigure}
  \begin{subfigure}[b]{0.49\textwidth}
    \centering
    \includegraphics[width=200pt]{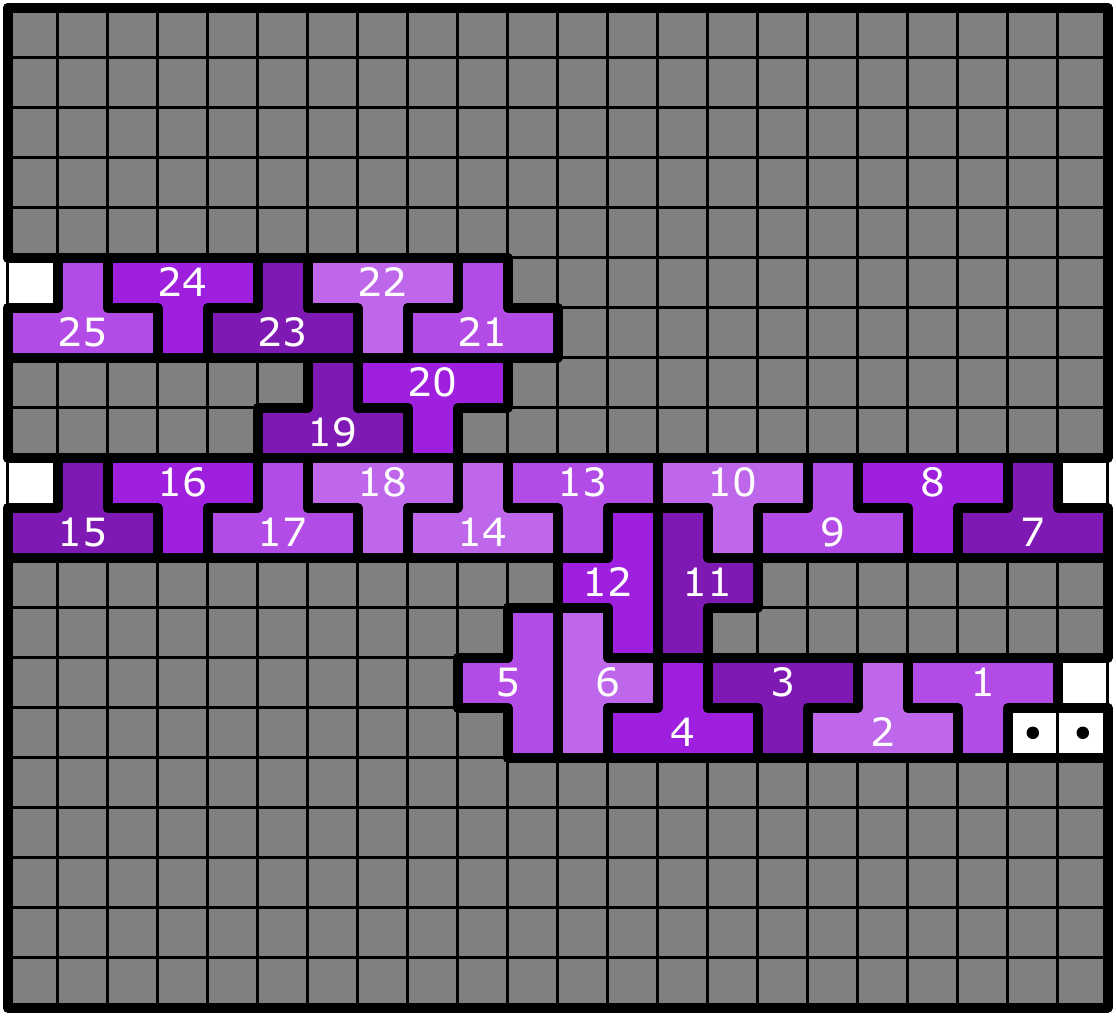}
    \caption{Fourth tiling}
  \end{subfigure}
  \caption{Tilings + placement orders for the $1$-in-$4$ gadget}
  \label{fig:t_1mod4_tilings}
\end{figure}

\begin{figure}[!ht]
  \centering
  \begin{subfigure}[b]{0.49\textwidth}
    \centering
    \includegraphics[width=100pt]{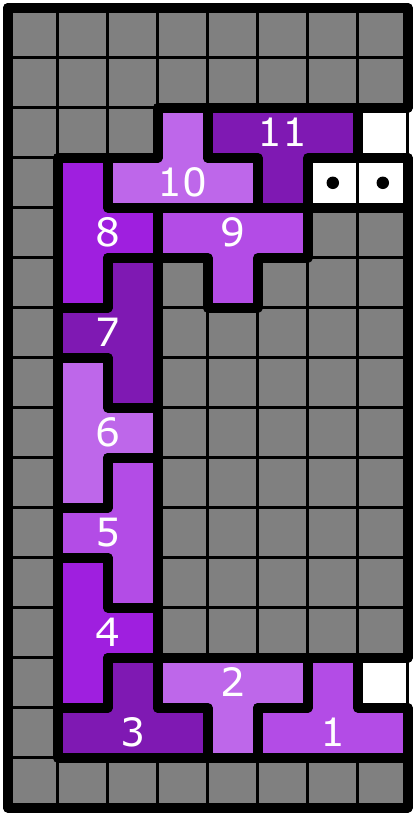}
    \caption{First tiling for first PF gadget}
  \end{subfigure}
  \begin{subfigure}[b]{0.49\textwidth}
    \centering
    \includegraphics[width=100pt]{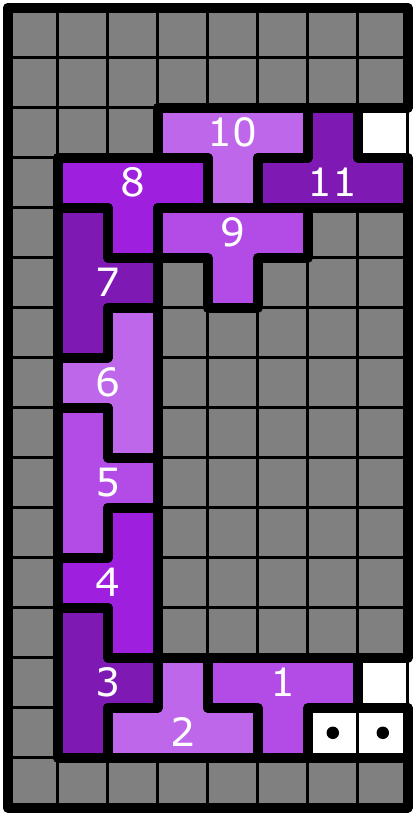}
    \caption{Second tiling for first PF gadget}
  \end{subfigure}
  \begin{subfigure}[b]{0.49\textwidth}
    \centering
    \includegraphics[width=220pt]{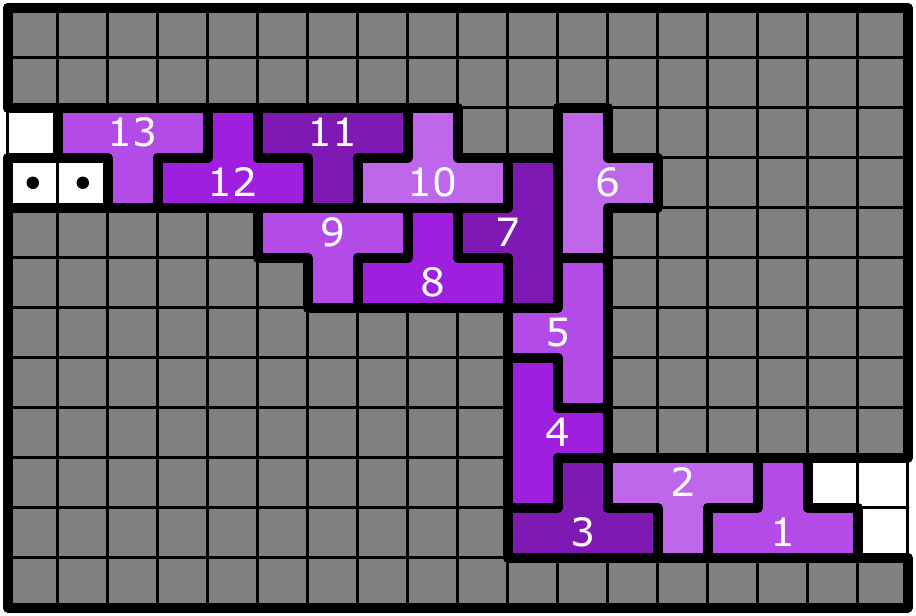}
    \caption{First tiling for second PF gadget}
  \end{subfigure}
  \begin{subfigure}[b]{0.49\textwidth}
    \centering
    \includegraphics[width=220pt]{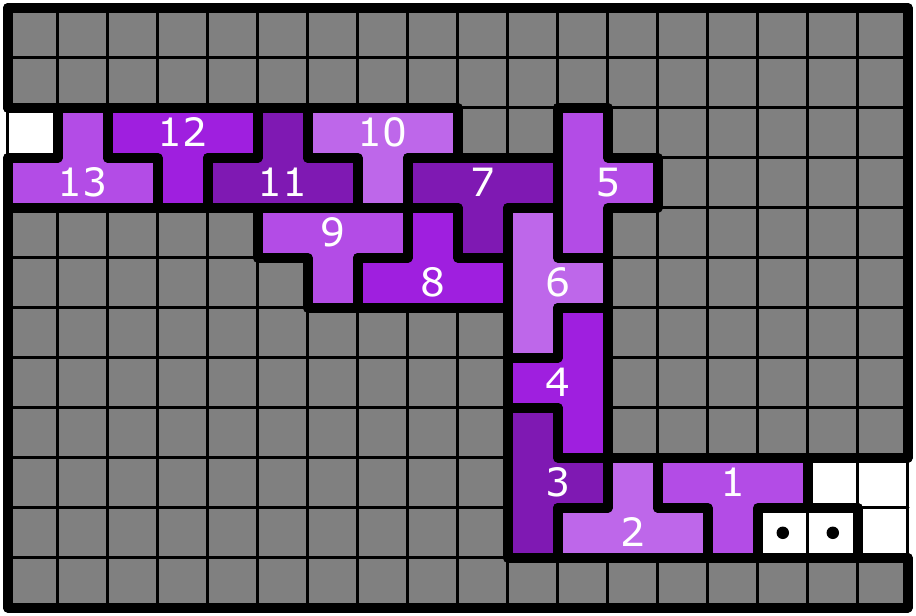}
    \caption{Second tiling for second PF gadget}
  \end{subfigure}
  \caption{Tilings + placement orders for the PF gadgets}
  \label{fig:t_parityfixer_tilings}
\end{figure}

\section{$\SS$-tris Gadget Tilings/Fillings}

Like in $\II$-tris, numbers in pieces denote placement order, with smaller numbers placed first.

\begin{figure}[!ht]
  \centering
  \begin{subfigure}[b]{0.49\textwidth}
    \centering
    \includegraphics[width=170pt]{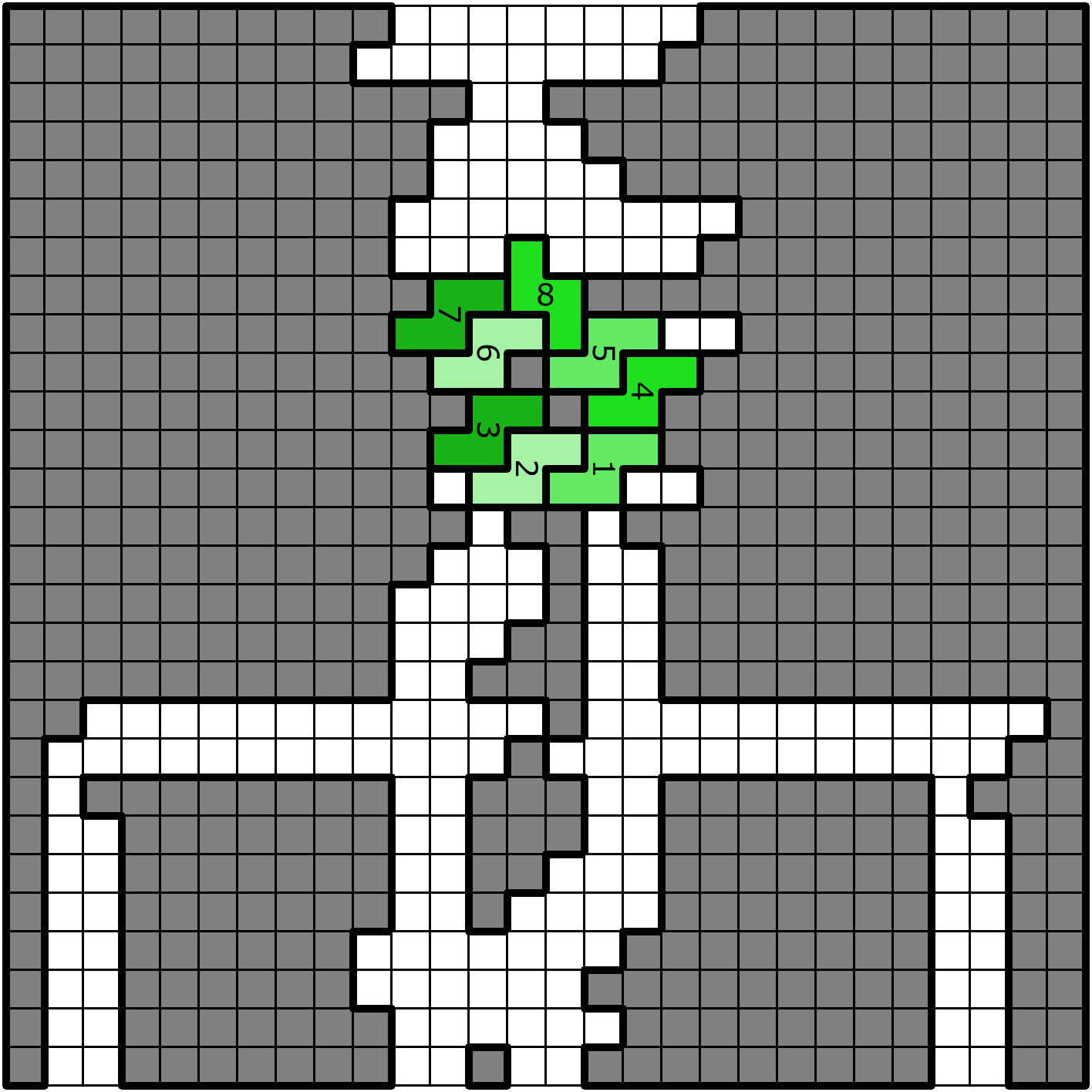}
    \caption{}
  \end{subfigure}
  \begin{subfigure}[b]{0.49\textwidth}
    \centering
    \includegraphics[width=170pt]{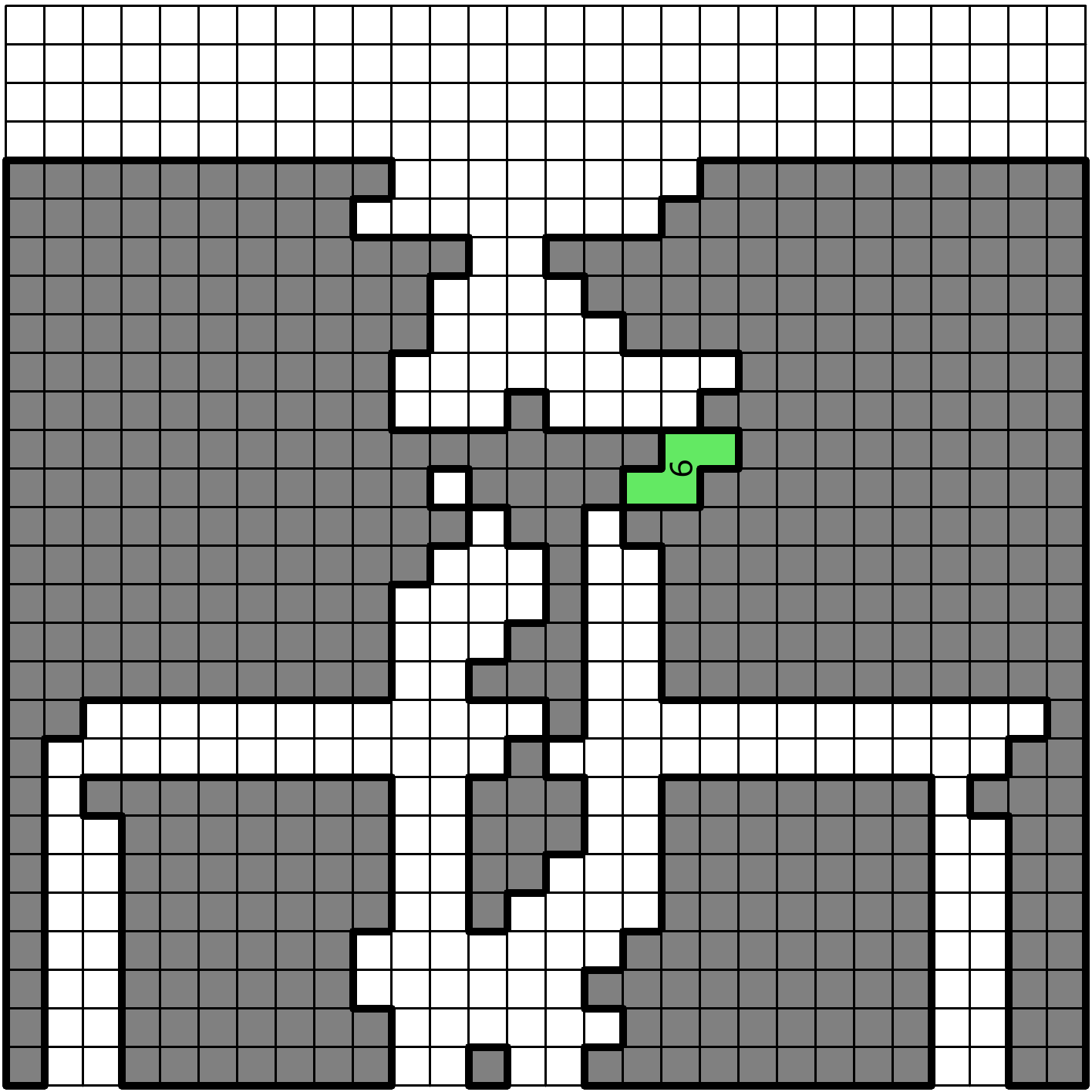}
    \caption{}
  \end{subfigure}
  \begin{subfigure}[b]{0.49\textwidth}
    \centering
    \includegraphics[width=170pt]{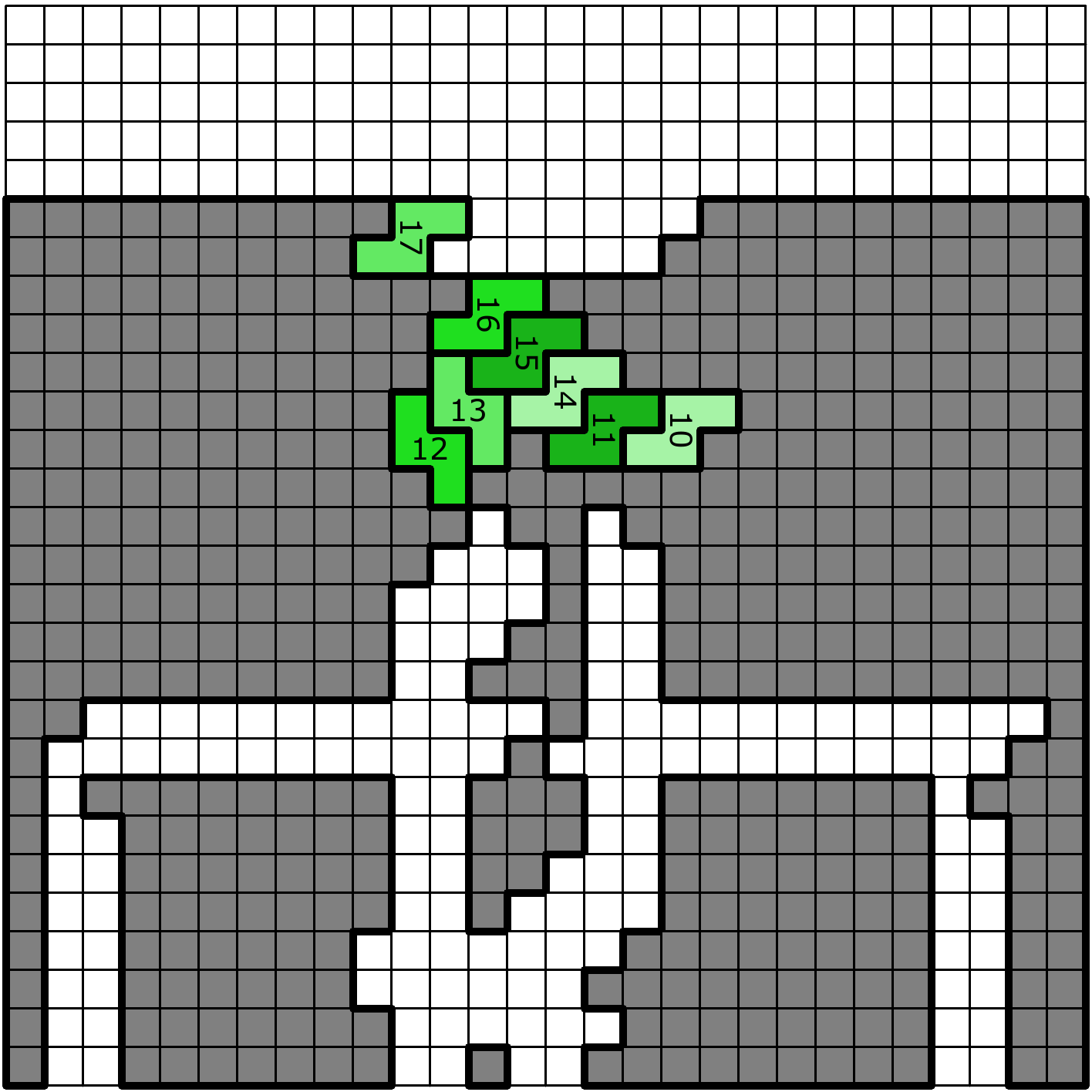}
    \caption{}
  \end{subfigure}
  \begin{subfigure}[b]{0.49\textwidth}
    \centering
    \includegraphics[width=170pt]{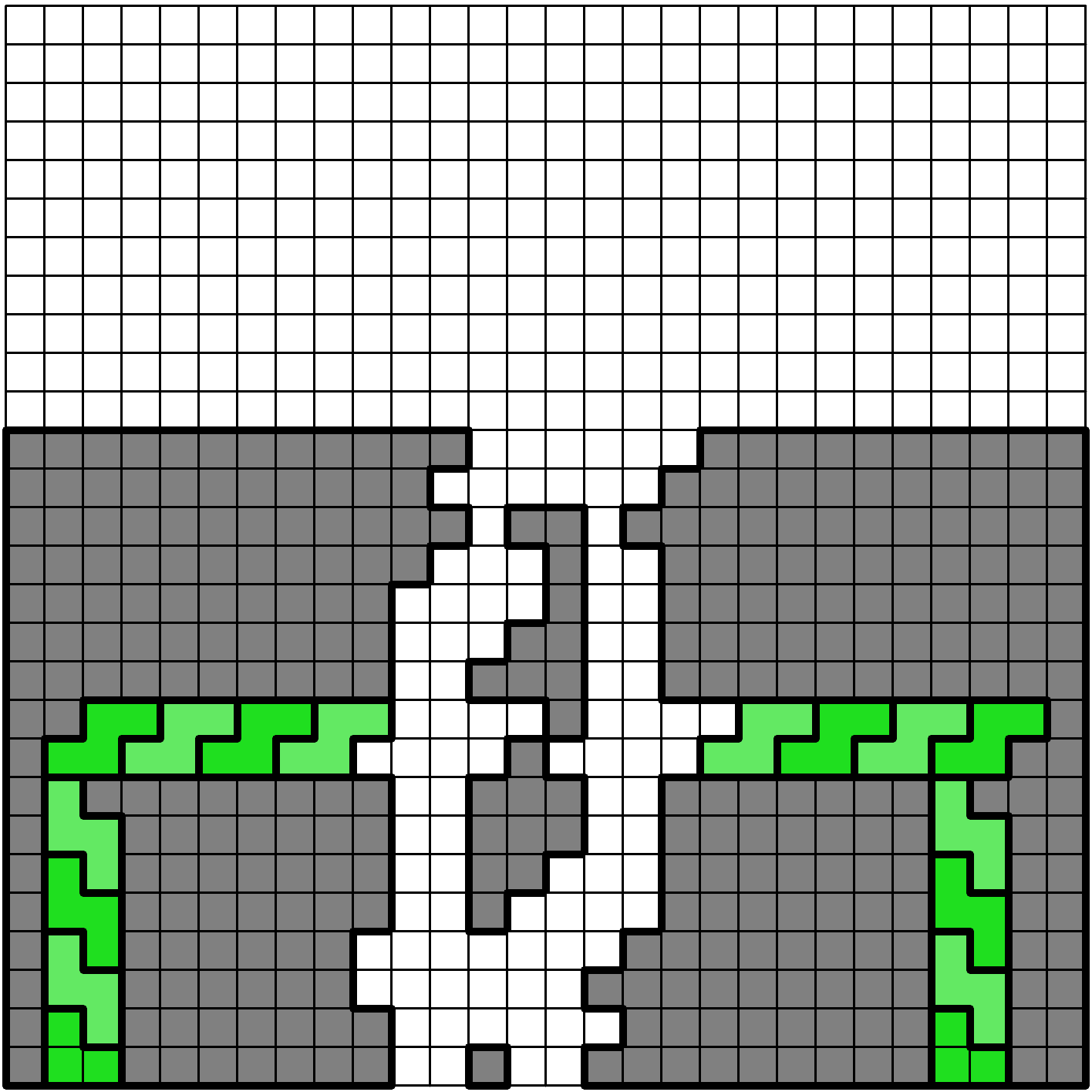}
    \caption{Wire-filling}
  \end{subfigure}
  \begin{subfigure}[b]{0.49\textwidth}
    \centering
    \includegraphics[width=170pt]{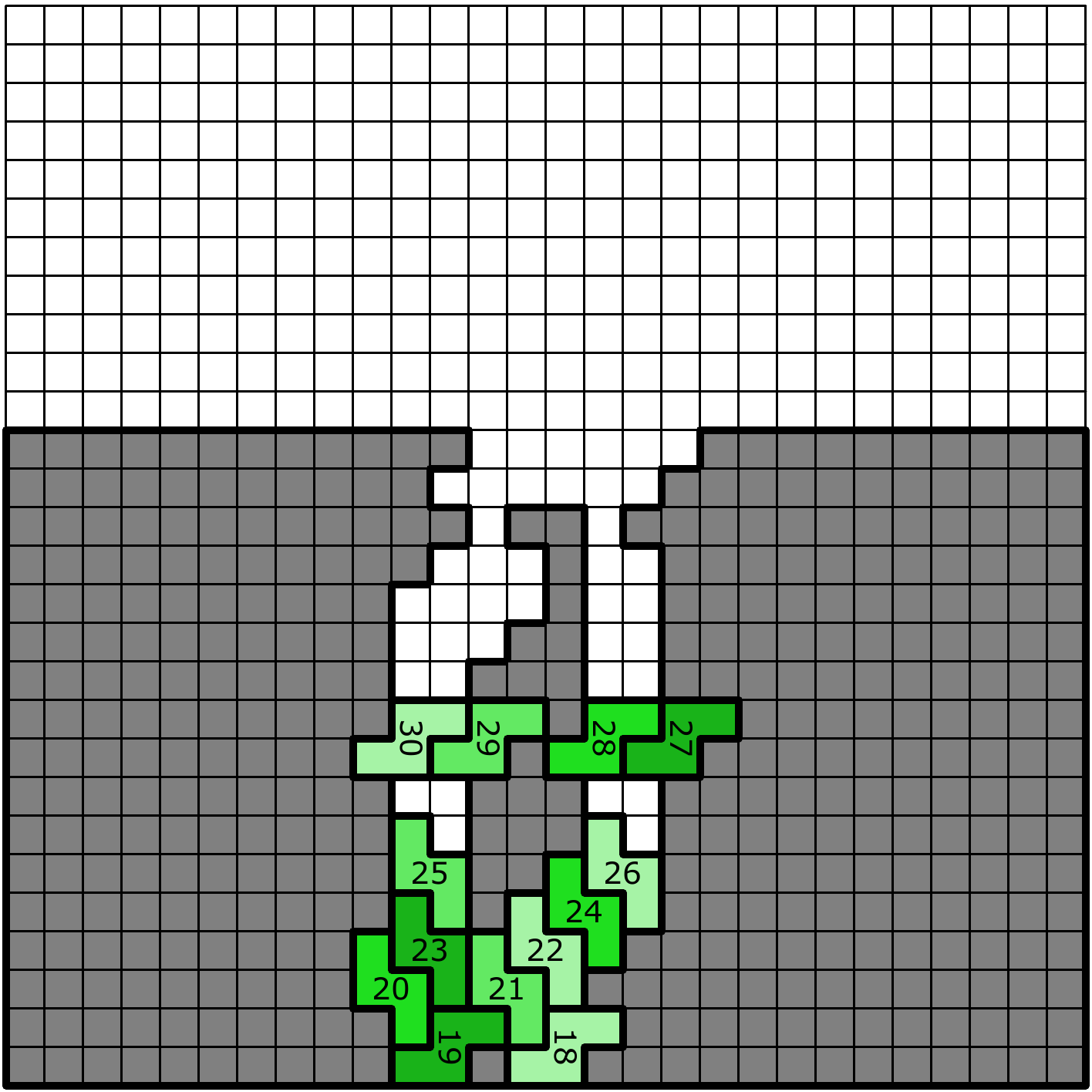}
    \caption{}
  \end{subfigure}
  \begin{subfigure}[b]{0.49\textwidth}
    \centering
    \includegraphics[width=170pt]{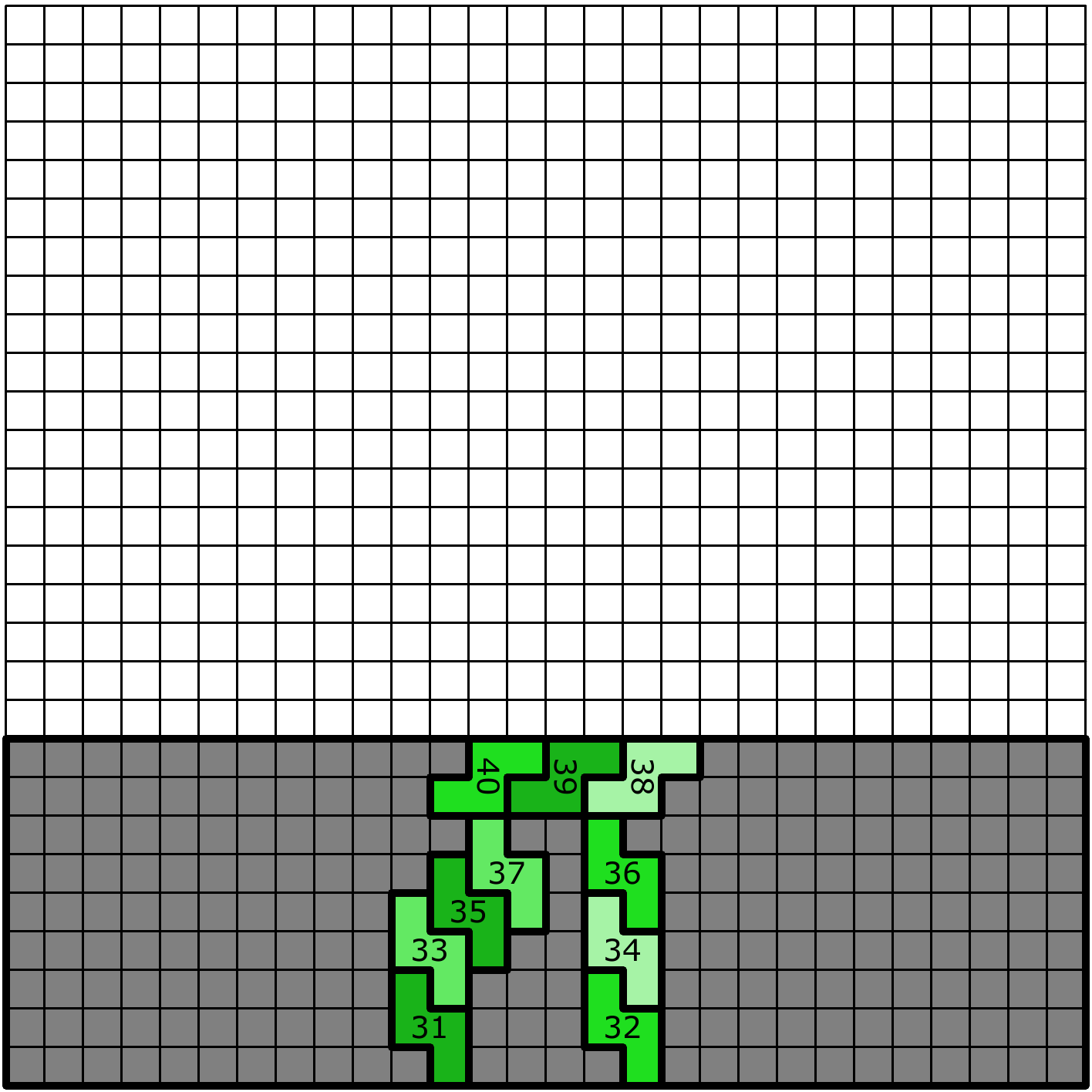}
    \caption{}
  \end{subfigure}
  \caption{Placement orders for variable gadgets, assuming the "True" setting (the "False" setting swaps "1" and "2" in (a)) and the variable appears once positively and once negatively (different numbers of appearances means additional pieces need to be placed in (d) and (e) to clear out the wires)}
  \label{fig:s_var_filling}
\end{figure}

\begin{figure}[!ht]
  \centering
  \begin{subfigure}[b]{0.9\textwidth}
    \centering
    \includegraphics[width=320pt]{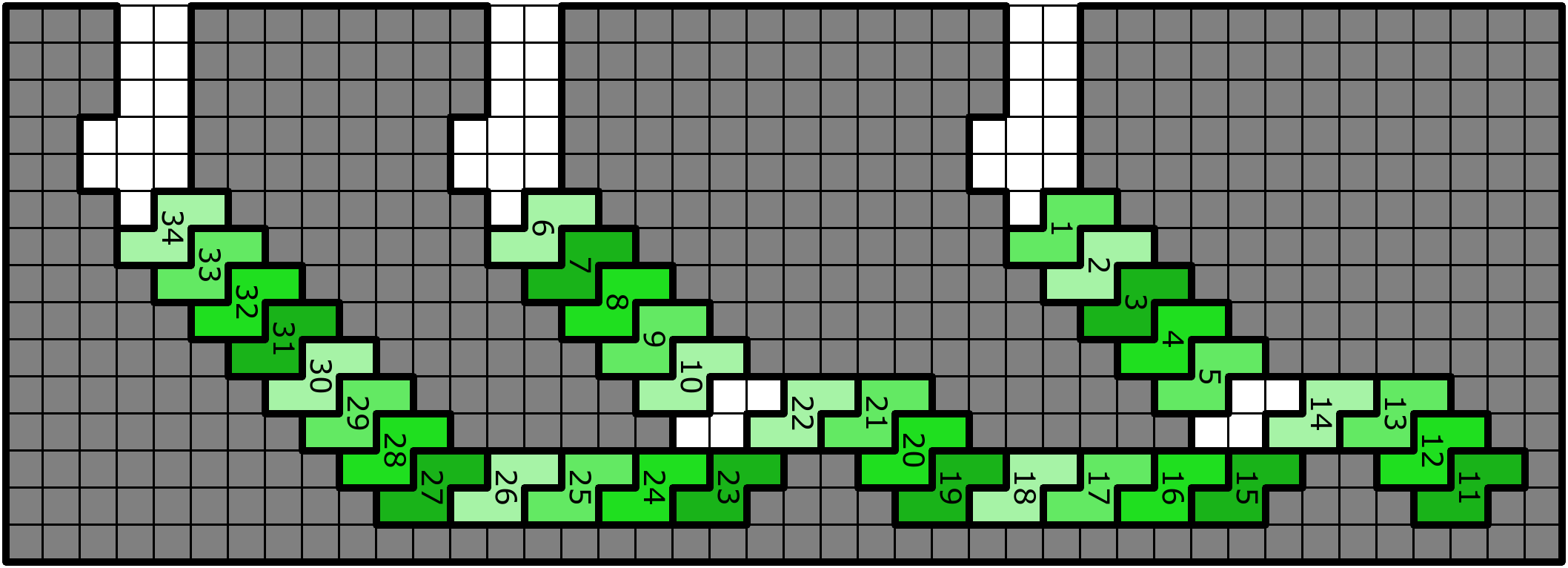}
    \caption{Filling through first entrance}
  \end{subfigure}
  \begin{subfigure}[b]{0.9\textwidth}
    \centering
    \includegraphics[width=320pt]{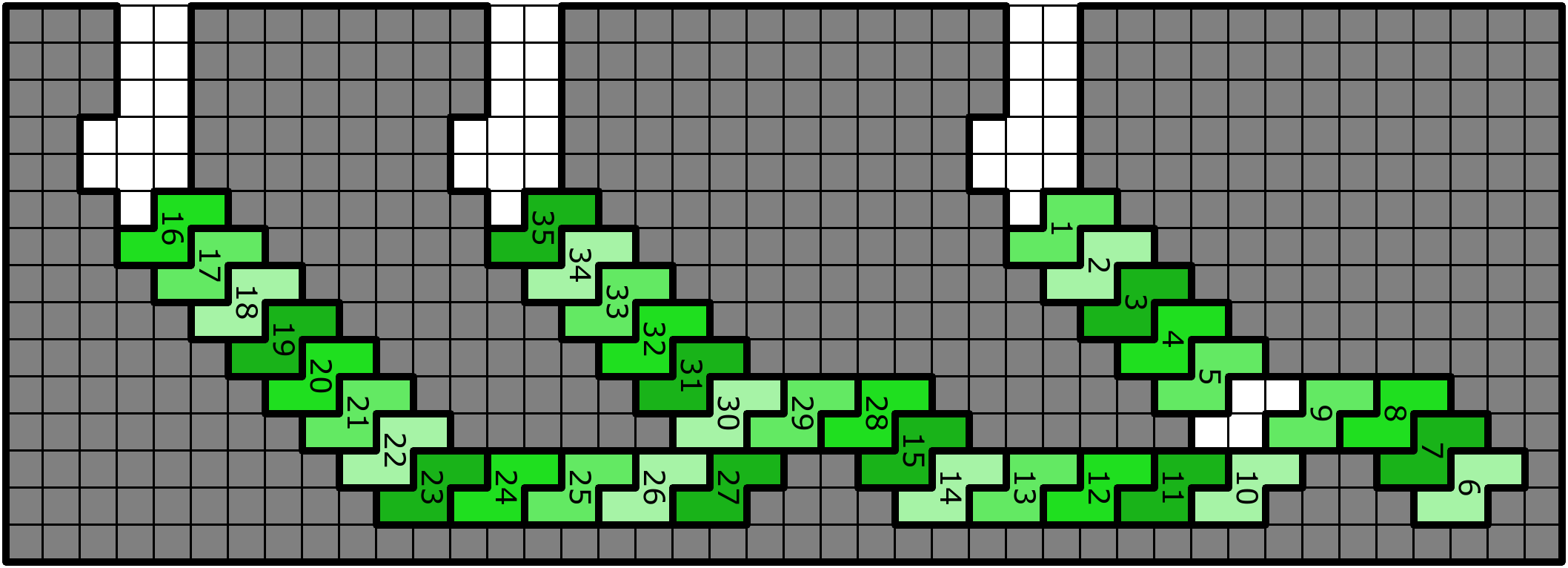}
    \caption{Filling through second entrance}
  \end{subfigure}
  \begin{subfigure}[b]{0.9\textwidth}
    \centering
    \includegraphics[width=320pt]{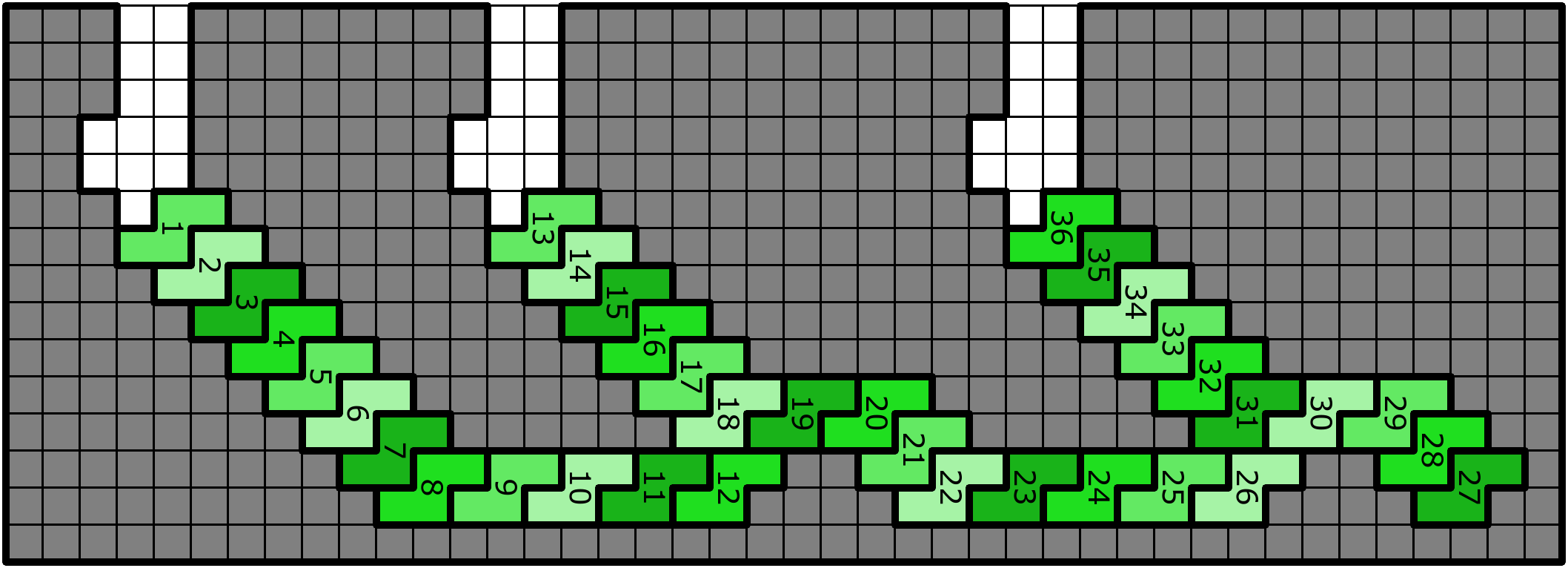}
    \caption{Filling through third entrance}
  \end{subfigure}
  \begin{subfigure}[b]{0.9\textwidth}
    \centering
    \includegraphics[width=320pt]{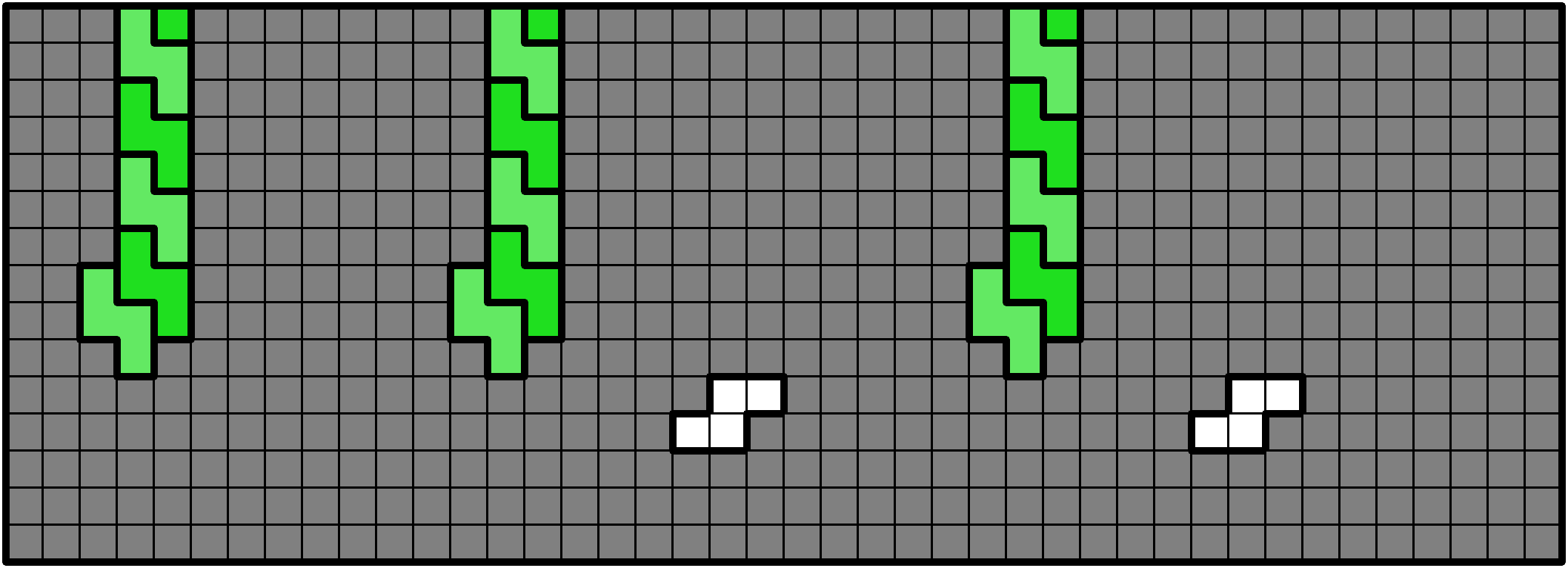}
    \caption{Wire-filling (after the climbable vertical structure is filled)}
  \end{subfigure}
  \caption{Fillings + placement orders for clause gadgets}
  \label{fig:s_cl_filling}
\end{figure}

\begin{figure}[!ht]
  \centering
  \includegraphics[width=320pt]{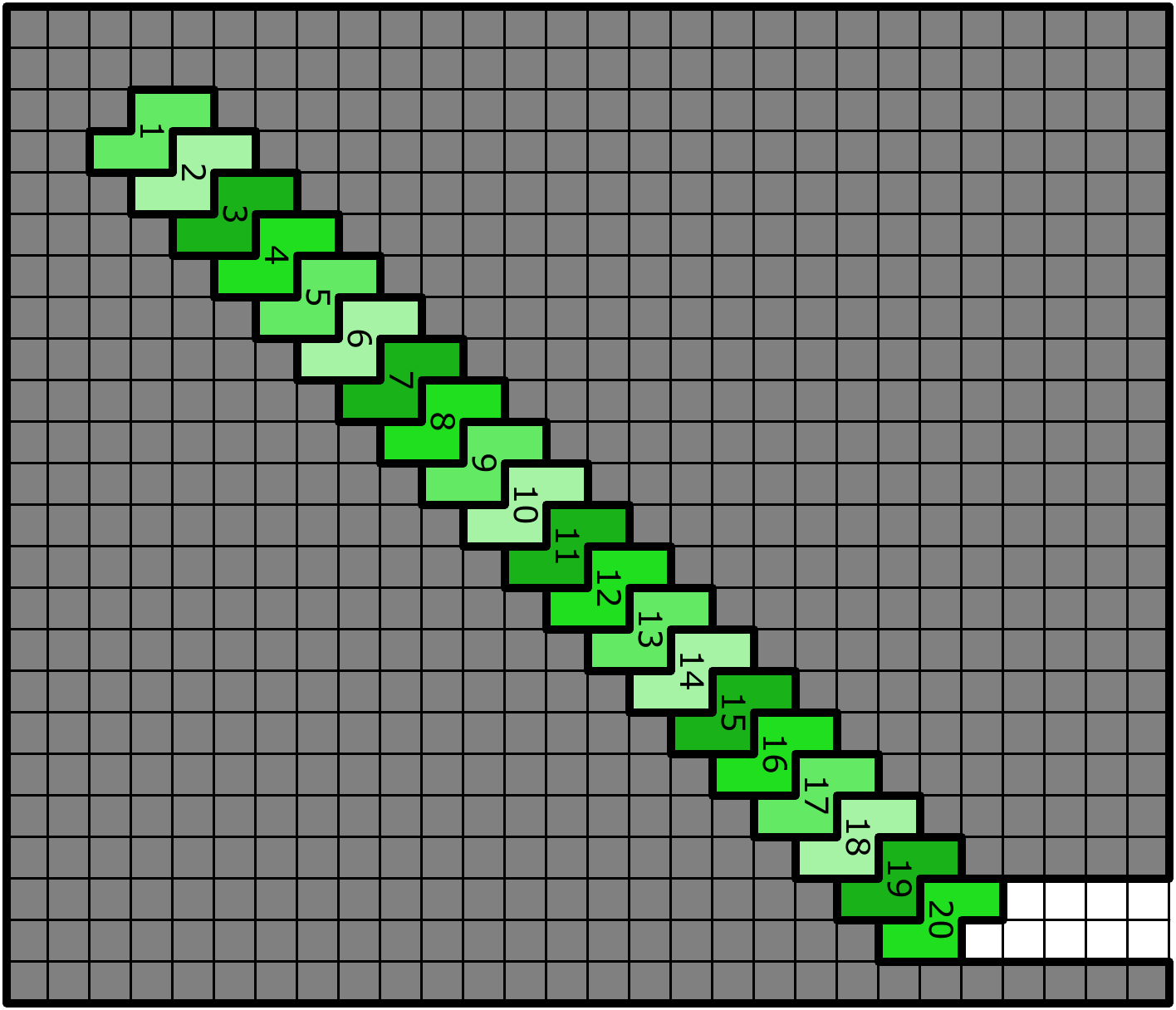}
  \caption{Tiling + placement order for climbable vertical structure}
  \label{fig:s_cvs_tiling}
\end{figure}


\end{document}